\PassOptionsToPackage{nameinlink}{cleveref}
\documentclass[UKenglish,cleveref,autoref,thm-restate]{lipics-v2021}
\usepackage{mathtools,aligned-overset,framed,bold-extra,xcolor}

\usepackage{aligned-overset}

\usepackage[scale=0.9]{sourcecodepro}
    \hypersetup{colorlinks,linkcolor={blue!70!black},citecolor={blue!70!black},urlcolor={black}}
    
    \renewcommand\path[1]{{\normalfont\small\detokenize{#1}}}
\usepackage[ruled,vlined,linesnumbered,scleft]{algorithm2e}

\crefalias{AlgoLine}{line}%
\makeatletter
\let\cref@old@stepcounter\stepcounter
\def\stepcounter#1{%
  \cref@old@stepcounter{#1}%
  \cref@constructprefix{#1}{\cref@result}%
  \@ifundefined{cref@#1@alias}%
    {\def\@tempa{#1}}%
    {\def\@tempa{\csname cref@#1@alias\endcsname}}%
  \protected@edef\cref@currentlabel{%
    [\@tempa][\arabic{#1}][\cref@result]%
    \csname p@#1\endcsname\csname the#1\endcsname}}
\makeatother

\usepackage[autostyle=true]{csquotes}
\usepackage{bm}

\usepackage{xspace}
\newtheorem{defn}[lemma]{Definition}
\newtheorem{prop}[lemma]{Proposition}
\newtheorem{cor}[lemma]{Corollary}

\numberwithin{equation}{subsection}
\newcommand{\floor}[1]{\lfloor{}#1\rfloor{}}
\newcommand{\ceil}[1]{\lceil{}#1\rceil{}}
\newcommand{\OO}{\mathcal{O}}

\newcommand{\E}{\mathop{\mathbb{E}}}
\renewcommand{\Pr}{\mathop{\mathbb{P}}}
\newcommand{\pr}{\Pr}
\newcommand{\R}{\mathbb{R}}

\newcommand{\N}{\mathbb{N}}
\newcommand{\Q}{\mathbb{Q}}
\newcommand{\supp}{\mathrm{supp}}

\newcommand{\eps}{\varepsilon}
\renewcommand{\epsilon}{\varepsilon}
\newcommand{\pout}{p_\mathsf{out}}

\DeclareMathOperator{\poly}{poly}
\newcommand{\indora}{\ensuremath{\mathtt{IND}}} 
\newcommand{\cindora}{\ensuremath{\mathtt{cIND}}}

\newcommand{\oldcoarsecount}{\textnormal{\texttt{ColourCoarse\_DLM22}}\xspace}
\newcommand{\combinecoarsecount}{\textnormal{\texttt{ColourCoarse}}\xspace}
\newcommand{\newcoarsecount}{\textnormal{\texttt{\hyperref[algo:coarsecount]{ColourCoarse\_New}}}\xspace}

\newcommand{\Refine}{\textnormal{\texttt{Refine}}\xspace}
\newcommand{\approxUncol}{\textnormal{\texttt{HelperCount}}\xspace}
\newcommand{\acc}{\textnormal{\texttt{Coarse}}\xspace}
\newcommand{\aau}{\textnormal{\texttt{Count}}\xspace}
\newcommand{\oldcount}{\textnormal{\texttt{Count\_DLM22}}\xspace}
\newcommand{\verifyguess}{\textnormal{\texttt{VerifyGuess}}\xspace}
\newcommand{\newverifyguess}{\textnormal{\texttt{\hyperref[algo:newverifyguess]{VerifyGuess\_New}}}\xspace}

\newcommand{\dlmimprove}{\textnormal{\texttt{CoarseLargeCore}}\xspace}
\newcommand{\dlmrecurse}{\textnormal{\texttt{CoarseSmallCore}}\xspace}
\newcommand{\helperdlmimprove}{\textnormal{\texttt{\hyperref[algo:helperdlmimprove]{CoarseLargeCoreHelper}}}\xspace}
\newcommand{\helperdlmrecurse}{\textnormal{\texttt{\hyperref[algo:helperdlmrecurse]{CoarseSmallCoreHelper}}}\xspace}

\newcommand{\Torig}{T_\mathrm{orig}}
\newcommand{\Corig}{C_\mathrm{orig}}
\newcommand{\Tcost}{T^\mathrm{comp}_\mathrm{cost}}
\newcommand{\Ttotals}{T^\mathrm{comp}_\mathrm{totals}}
\newcommand{\Csample}{C_{\mathrm{sample}}}

\newcommand{\Bin}{\mathrm{Bin}}

\newcommand{\yes}{\textnormal{\texttt{Yes}}\xspace}
\newcommand{\no}{\textnormal{\texttt{No}}\xspace}

\DeclarePairedDelimiter\paren{\lparen}{\rparen}
\DeclarePairedDelimiter\abs{\lvert}{\rvert}
\DeclarePairedDelimiter\set{\{}{\}}
\DeclarePairedDelimiterX\setc[2]{\{}{\}}{\,#1 \;\colon\; #2\,}
\DeclarePairedDelimiterX\parenc[2]{\lparen}{\rparen}{\,#1 \;\delimsize\vert\; #2\,}

\def\cost{\mathrm{cost}}

\def\cA{\mathcal{A}}

\def\cG{\mathcal{G}}
\def\cX{\mathcal{X}}

\newcommand{\calA}{\mathcal{A}}
\newcommand{\calB}{\mathcal{B}}
\newcommand{\calC}{\mathcal{C}}
\newcommand{\calD}{\mathcal{D}}
\newcommand{\calE}{\mathcal{E}}
\newcommand{\calF}{\mathcal{F}}
\newcommand{\calG}{\mathcal{G}}

\newcommand{\calI}{\mathcal{I}}

\newcommand{\calQ}{\mathcal{Q}}
\newcommand{\calR}{\mathcal{R}}
\newcommand{\calT}{\mathcal{T}}
\newcommand{\calX}{\mathcal{X}}
\newcommand{\calZ}{\mathcal{Z}}
\def\zo{\set{0,1}}
\newcommand{\flalpha}{\floor{\alpha}}
\newcommand{\flalphak}{\floor{\alpha_k}}

\newcommand{\var}{\mathrm{Var}}

\newcommand{\SampleSubset}{\textnormal{\texttt{\hyperref[algo:sample-random-subset]{SampleSubset}}\xspace}}
\newcommand{\SparseCount}{\textnormal{\texttt{\hyperref[lem:sparse-count]{SparseCount}}}\xspace}

\newcommand{\UncolApprox}{\textnormal{\texttt{\hyperref[algo:uncolapprox]{UncolApprox}}}\xspace}
\newcommand{\TooDense}{\textnormal{\texttt{TooDense}}\xspace}
\newcommand{\RTE}{\textnormal{\texttt{RTE}}\xspace}
\newcommand{\RecEnum}{\textnormal{\texttt{RecEnum}}\xspace}
\DeclareMathOperator{\polylog}{polylog}

\newcommand{\EInacc}{\hyperref[def:EInacc]{\calE_{\mathrm{inacc}}}}
\newcommand{\ECost}{\hyperref[def:ECost]{\calE_{\mathrm{cost}}}}
\newcommand{\EEdge}{\hyperref[def:EEdge]{\calE_{\mathrm{edge}}}}
\newcommand{\EEdgeSub}[1]{\hyperref[{def:EEdgeSub}]{\calE_{\textnormal{edge},\,#1}}}
\newcommand{\ERoot}{\hyperref[def:ERoot]{\calE_{\mathrm{root}}}}
\newcommand{\ENroot}{\hyperref[def:ENroot]{\calE_{\mathrm{nonroot}}}}
\newcommand{\yset}{\mathcal{Y}}
\newcommand{\nset}{\mathcal{N}}

\definecolor{costwedge2}{RGB}{ 255  0  0}
\definecolor{costwedge3}{RGB}{ 0 0  255}
\definecolor{costwedge4}{RGB}{  0 128 0}
\definecolor{costwedge5}{RGB}{  0 0   0}

\definecolor{wedge0}{RGB}{ 190  30  46}
\definecolor{wedge1}{RGB}{ 240  65  54}
\definecolor{wedge2}{RGB}{ 241  90  43}
\definecolor{wedge3}{RGB}{ 247 148  30}
\definecolor{wedge4}{RGB}{  43  56 144}
\definecolor{wedge5}{RGB}{  28 117 188}
\definecolor{wedge6}{RGB}{  40 170 225}
\definecolor{wedge7}{RGB}{ 119 179 225}
\definecolor{wedge8}{RGB}{ 181 212 239}
\definecolor{wedge9}{RGB}{  0 104  56}
\definecolor{wedge10}{RGB}{  0 148  69}
\definecolor{wedge11}{RGB}{ 57 181  74}
\definecolor{wedge12}{RGB}{141 199  63}
\definecolor{wedge13}{RGB}{215 244  34}
\definecolor{wedge14}{RGB}{249 237  50}
\definecolor{wedge15}{RGB}{248 241 148}
\definecolor{wedge16}{RGB}{242 245 205}
\definecolor{wedge17}{RGB}{123  82  49}
\definecolor{wedge18}{RGB}{104  73 158}
\definecolor{wedge19}{RGB}{102  45 145}
\definecolor{wedge20}{RGB}{148 149 151}

\usepackage{tikz}
\usepackage{pgfplots}
\pgfplotsset{compat=1.10}
\usetikzlibrary{calc}
\usetikzlibrary{math}

\tikzstyle{p}=[circle,minimum size=4pt,inner sep=0pt,fill]
\definecolor{middlegreen}{HTML}{5B8C5A}
\definecolor{russianviolet}{HTML}{410342}
\definecolor{blueyonder}{HTML}{576CA8}

\tikzmath{
    function L(\i) {
        return floor(\i*\k/\logn);
    };
    function gamma(\i) {
        return \i*\k/\logn - L(\i);
    };
    function fone(\Li,\gi,\a) {
        return (\Li+\gi)*(\k-\Li-1-\a);
    };
    function ftwo(\Li,\gi,\a) {
        return (\Li+\gi)*(\k-\Li-\a)-\k*\gi;
    };
    function exponFi(\i,\a) {
        \Li = L(\i);
        \gi = gamma(\i);
        return max(fone(\Li,\gi,\a), ftwo(\Li,\gi,\a))/\k;
    };
    function expon_uncolapprox_cost(\k,\a) {
        \rounded = round((\k - \a) / 2);
        return \rounded * (\k - \a - \rounded) / \k;
    };
}

\makeatletter
\@ifundefined{st}{%
  \newcommand{\st}{\textsuperscript{\textup{st}}\xspace}
}{}
\@ifundefined{rd}{%
  
}{}
\@ifundefined{nd}{%
  
}{}
\makeatother

\renewcommand{\th}{\textsuperscript{\textup{th}}\xspace}

\title{Nearly optimal independence oracle algorithms for edge estimation in hypergraphs}
\author{Holger Dell}{University of Frankfurt, Germany\and IT University of Copenhagen and Basic Algorithms Research Copenhagen (BARC), Denmark}{}{https://orcid.org/0000-0001-8955-0786}{}
\author{John Lapinskas}{University of Bristol, UK}{}{https://orcid.org/0000-0003-3197-0854}{}
\author{Kitty Meeks}{University of Glasgow, UK}{}{https://orcid.org/0000-0001-5299-3073}{Supported by EPSRC grant EP/V032305/1.}

\authorrunning{H.~Dell, J.~Lapinskas and K.~Meeks}

\Copyright{Holger Dell, John Lapinskas, and Kitty Meeks} %TODO mandatory, please use full first names. LIPIcs license is "CC-BY";  http://creativecommons.org/licenses/by/3.0/



\EventEditors{Karl Bringmann, Martin Grohe, Gabriele Puppis, and Ola Svensson}
\EventNoEds{4}
\EventLongTitle{51st International Colloquium on Automata, Languages, and Programming (ICALP 2024)}
\EventShortTitle{ICALP 2024}
\EventAcronym{ICALP}
\EventYear{2024}
\EventDate{July 8--12, 2024}
\EventLocation{Tallinn, Estonia}
\EventLogo{}
\SeriesVolume{297}
\ArticleNo{90}
\nolinenumbers
\hideLIPIcs

\begin{document}
\maketitle

\begin{abstract}
Consider a query model of computation in which an $n$-vertex $k$-hypergraph can be accessed only via its independence oracle or via its colourful independence oracle, and each oracle query may incur a cost depending on the size of the query.  Several recent results (Dell and Lapinskas, STOC 2018; Dell, Lapinskas, and Meeks, SODA 2020) give efficient algorithms to approximately count the hypergraph's edges in the colourful setting.  These algorithms immediately imply fine-grained reductions from approximate counting to decision, with overhead only $\log^{\Theta(k)} n$ over the running time $n^\alpha$ of the original decision algorithm, for many well-studied problems including $k$-Orthogonal Vectors, $k$-SUM, subgraph isomorphism problems including $k$-Clique and colourful-$H$, graph motifs, and $k$-variable first-order model checking.

We explore the limits of what is achievable in this setting, obtaining unconditional lower bounds on the oracle cost of algorithms to approximately count the hypergraph's edges in both the colourful and uncoloured settings.  In both settings, we also obtain algorithms which essentially match these lower bounds; 
in the colourful setting, this requires significant changes to the algorithm of Dell, Lapinskas, and Meeks (SODA 2020) and reduces the total overhead to $\log^{\Theta(k-\alpha)}n$.
Our lower bound for the uncoloured setting shows that there is no fine-grained reduction from approximate counting to the corresponding uncoloured decision problem (except in the case $\alpha \ge k-1$): without an algorithm for the colourful decision problem, we cannot hope to avoid the much larger overhead of roughly $n^{(k-\alpha)^2/4}$.
The uncoloured setting has previously been studied for the special case $k=2$ (Peled, Ramamoorthy, Rashtchian, Sinha, ITCS 2018; Chen, Levi, and Waingarten, SODA 2020), and our work generalises the existing algorithms and lower bounds for this special case to $k>2$ and to oracles with cost.
\nocite{BHRRS-oracle-intro}%
\nocite{CLW-graph-tight}%
\end{abstract}

\section{Introduction}

Many decision problems in computer science, particularly those in NP, can naturally be expressed in terms of determining the existence of a witness. For example, solving SAT requires determining the existence of a satisfying assignment to a CNF formula. All such problems $\Pi$ naturally give rise to a counting version $\#\Pi$, in which we ask for the number of witnesses. It is well-known that $\#\Pi$ is often significantly harder than $\Pi$; for example, Toda's theorem implies that it is impossible to solve $\#\mathrm{P}$-complete counting problems in polynomial time with access to an NP-oracle unless the polynomial hierarchy collapses. However, the same is not true for \textit{approximately} counting witnesses (to within a factor of two, say). For example, it is known that: if $\Pi$ is a problem in NP, then there is an FPRAS for $\#\Pi$ using an NP-oracle~\cite{VV}; if $\Pi$ is a problem in $W[i]$, then there is an FPTRAS for $\#\Pi$ using a $W[i]$-oracle~\cite{Muller}; and that the Exponential Time Hypothesis is equivalent to the statement that there is no subexponential-time approximation algorithm for \#3-SAT~\cite{DL}.

In this paper we are concerned with analogous results in the fine-grained setting, which considers exact running times rather than coarse-grained classifications such as polynomial, FPT, or subexponential; such results turn out to be inextricably bound to graph oracle results of independent interest. 

Past work in this area has focused on the family of \emph{uniform witness problems}~\cite{DLM}. Roughly speaking, these are problems which can be expressed as counting edges in a $k$-hypergraph $G$ in which the edges correspond to witnesses and induced subgraphs correspond to sub-problems. (See~\cref{sec:intro-applications} for a detailed definition.) Many of the most important problems in fine-grained and parameterised complexity can be expressed as uniform witness problems including \textsc{$k$-SUM}, \textsc{$k$-OV}, \textsc{$k$-Clique}, Hamming weight-$k$ solutions to CNFs, \textsc{Size-$k$ Graph Motif}, most subgraph detection problems (including weighted problems such as \textsc{Zero-Weight $k$-Clique} and \textsc{Negative-Weight Triangle}), and first-order model-checking~\cite{DLM}, in addition to certain database queries \cite{FGRK-databases} and patterns in graphs \cite{BR-patterns}. Here $k$ may be either a constant, as in the case of $k$-SUM, or a parameter, as in the case of \textsc{$k$-Clique}. In this setting, invoking a decision algorithm on a sub-problem of the original problem corresponds to invoking an oracle to test, given a set of vertices $S$, whether the induced subgraph $G[S]$ contains any edges; this oracle is called an \textit{independence oracle} for $G$ and is well-studied in its own right (see Section~\ref{sec:intro-related} for an overview).

Surprisingly, there is a partial analogue of the above reductions from approximate counting to decision in this setting. If the vertices of $G$ are coloured, given a set $S \subseteq V(G)$, a \textit{colourful independence oracle} tests whether $G[S]$ contains any edges with one vertex of each colour. This typically corresponds to a natural colourful variant of the original decision problem --- for example, for \textsc{$k$-Clique}, it corresponds to deciding whether a $k$-coloured graph contains a size-$k$ clique with one vertex of each colour. These oracles are again well-studied in their own right (see Section~\ref{sec:intro-related}), and for many but not all uniform witness problems they can be efficiently simulated using the independence oracle. Given access to a colourful independence oracle for a graph $G$, we can count $G$'s edges to within a factor of $1 \pm \eps$ using $\eps^{-2}k^{\OO(k)}\log^{\Theta(k)} n$ oracle queries~\cite{DLM}. (See~\cite{BBGM-hypergraph} for an improvement to the log factor.) In fact, we can say more --- if we can simulate the colourful independence oracle in time $n^{\alpha_k}$ with $\alpha_k \ge 1$, then these queries dominate the running time and we obtain an approximate counting algorithm with running time $n^{\alpha_k}\cdot \eps^{-2}k^{\OO(k)}\log^{\Theta(k)}n$ in the usual word-RAM model. Translating back out of the oracle setting, this means that if we simulate the oracle by running an algorithm for the colourful decision problem, then for constant $k$ and $\eps$, we obtain an approximate counting algorithm with only polylogarithmic overhead over that decision algorithm. This result has led to several improved approximate counting algorithms --- see~\cite{DLM} for applications to $k$-OV over finite fields and graph motifs,~\cite{FGRK-databases} for applications to database queries, and~\cite{BR-patterns} for applications to patterns in graphs.

We are left with two major open problems of concern to researchers in fine-grained complexity, parameterised complexity and graph oracles, and we expect our paper to be of interest to all three communities. First, can the result of~\cite{DLM} be generalised from colourful independence oracles to independence oracles? This would imply, for example, a fine-grained reduction from approximate induced sub-hypergraph counting to induced sub-hypergraph detection. In this setting, efficiently simulating the colourful independence oracle using the independence oracle requires solving a long-standing open problem --- see Section~\ref{sec:intro-applications} --- so the result of~\cite{DLM} does not straightforwardly apply. Second, in the parameterised setting, the factor of $\log^{\Theta(k)}n$ is not truly polylogarithmic, but equivalent to a factor of $k^{\OO(k)}n^{o(1)}$. Can it be improved to~$\log^{\OO(1)}n$? 

In this paper, we answer both questions, and in the process substantially generalise recent graph oracle results for the $k=2$ case~\cite{CLW-graph-tight}. In both the colourful and uncoloured settings, we pin down the optimal oracle algorithm almost exactly. In both cases this algorithm improves on the current state of the art, and it allows for the desired fine-grained reductions if and only if the cost of calling the oracle on an $x$-vertex set (corresponding to the run-time of a decision algorithm on an $x$-element instance) is close to $x^k$. Moreover, our lower bounds are unconditional --- they do not rely on conjectures such as SETH or $\mbox{FPT} \ne \mbox{W}[1]$.

In a little more detail, suppose for the moment that $\eps=1/2$, and that the cost of calling the oracle on an $x$-vertex set is $x^{\alpha_k}$ for some $\alpha_k \in [0,k]$. In the uncoloured setting, we define a function $g(k,\alpha_k) \approx (k-\alpha)^2/(4k)$ (see \eqref{eq:g-def}) and show that an overhead of $2^{\OO(k)}n^{g(k,\alpha_k)\pm o(1)}$ is both achievable and required; we have $g(k,\alpha_k) = 0$ when $\alpha_k \ge k-1$, so in this regime we obtain a fine-grained reduction. In the colourful setting, we show that the $\log^{\Theta(k)}n$ overhead of~\cite{DLM,BBGM-hypergraph} can be improved to $\log^{\Theta(k-\alpha_k)}n$, but no further; thus polylogarithmic overhead is possible if and only if $k-\alpha_k \in \OO(1)$ as $k\to\infty$. For general values of $\eps$, both of our upper bounds have an additional multiplicative overhead of $\OO(\eps^{-2})$, which is common in approximate counting~algorithms.

In the rest of the introduction, we state our results for graph oracles more formally in \cref{sec:intro-oracle}, followed by their (immediate) corollaries for uniform witness problems in \cref{sec:intro-applications}. We then give an overview of related work in \cref{sec:intro-related}, followed by a brief description of our proof techniques in \cref{sec:intro-proofs}. 

\subsection{Oracle results}\label{sec:intro-oracle}
Our results are focused on two graph oracle models on $k$-hypergraphs: independence oracles and colourful independence oracles. Both oracles are well-studied in their own right from a theoretical perspective, as they are both natural generalisations of group testing from unary relations to $k$-ary relations, and the apparent separation between them in power is already a source of substantial interest. They also provide a point of comparison for a rich history of sublinear-time algorithms for oracles which provide more local information, such as degree oracles. See the introduction of~\cite{CLW-graph-tight} for a more detailed overview of the full motivation, and \cref{sec:intro-related} for a survey of past results.

In both the colourful and uncoloured case, while formally the oracles are bitstrings and a query takes $\OO(1)$ time, in order to obtain reductions from approximate counting problems to decision problems in \cref{sec:intro-applications} we will simulate oracle queries using a decision algorithm. As such, rather than focusing on the \emph{number} of queries as a computational resource, we define a more general \textit{cost function} which will correspond to the running time of the algorithm used to simulate the query; thus the cost of a query will scale with its size. In our application, this allows for more efficient reductions by exploiting cheap queries, while also substantially strengthening our lower bounds. Indeed, simulating an oracle query typically requires between $\poly(k)$ and $\poly(n)$ time, so a lower bound on the total number of queries required would tell us very little; meanwhile, setting the cost of all queries to $1$ in our results yields tight bounds for the number of queries~required.

We are also concerned with the \textit{running times} of our oracle algorithms, again due to our applications in \cref{sec:intro-applications}.
We work in the standard RAM-model of computation with $\Theta(\log n)$ bits per word and access to the usual $\OO(1)$-time arithmetic and logical operations on these words; in addition, oracle algorithms can perform oracle queries, which are considered to take $\OO(1)$ time.

As shorthand, for all real $x,y > 0$ and $\eps\in(0,1)$, we say that $x$ is an \emph{$\eps$-approximation} to $y$ if $|x-y| < \eps y$. We define an \emph{$\eps$-approximate counting algorithm} to be an oracle algorithm that is given $n$ and $k$ as explicit input, is given access to an oracle representing an $n$-vertex $k$-hypergraph~$G$, and outputs an $\eps$-approximation to the number of edges of $G$, denoted by $e(G)$. We allow $\eps$ to be either part of the input (for upper bounds) or fixed (for lower bounds).

\subsubsection{Our results for the uncoloured independence oracle.}
Given a $k$-hypergraph $G$ with vertex set~$[n]$, the \textit{(uncoloured) independence oracle} is the bitstring $\indora(G)$ such that for all sets $S \subseteq [n]$, $\indora(G)_S = 1$ if $G[S]$ contains no edges and $0$ otherwise.
Thus a query to $\indora(G)_S$ allows us to test whether or not the induced subgraph $G[S]$ contains an edge. We define the \textit{cost} of an oracle call $\indora(G)_S$ to be a polynomial function of the form $\cost_k(S) = |S|^{\alpha_k}$, where the map $k\mapsto\alpha_k$ satisfies $\alpha_k \in [0,k]$ but is otherwise arbitrary.
(This upper bound is motivated by the fact that we can trivially enumerate all edges of $G$ by using $\OO(n^k)$ queries to all size-$k$ subsets of~$[n]$, incurring oracle cost at most $n^k \cdot k^{\alpha_k}$.)

It is not too hard to show that the naive $\OO(n^k)$-cost exact edge-counting algorithm of querying every possible edge and the naive $\OO(n^{\alpha_k})$-cost algorithm to decide whether any edge is present by querying $[n]$ are both essentially optimal. For approximate counting we prove the following, where for all real numbers $x$ we write $\lfloor x \rceil \coloneqq \lfloor x + 1/2\rfloor$ for the value of $x$ rounded to the nearest integer, rounding up in case of a tie.

\begin{theorem}[Uncoloured independence oracle, polynomial cost function]\label{thm:uncol-main-simple}
    Let $\alpha_k \in [0,k]$ for all $k \ge 2$, let $\cost_k(x) = x^{\alpha_k}$, and let
    \begin{equation}\label{eq:g-def}
        g(k,\beta) \coloneqq \frac{1}{k}\cdot \Big\lfloor\frac{k-\beta}{2}\Big\rceil \cdot \bigg(k-\beta-\Big\lfloor\frac{k-\beta}{2}\Big\rceil\bigg)\,.
    \end{equation}
    There is a randomised $\eps$-approximate counting algorithm  $\textnormal{\texttt{Uncol}}(\indora(G),\eps,\delta)$
    with failure probability at most $\delta$, worst-case running time
    \[
        \OO\Big(\log(1/\delta)\big(k^{5k}+\eps^{-2}2^{5k}\log^5n\cdot n^{g(k,1)}\cdot n\big)\Big)\,,
    \]
    and worst-case oracle cost
    \[
        \OO\Big(\log(1/\delta)\big(k^{7k}+\eps^{-2}2^{5k}\log^5n\cdot n^{g(k,\alpha_k)}\cdot n^{\alpha_k}\big)\Big)
    \]
    under $\cost_k$. Moreover, every randomised $(1/2)$-approximate edge-counting \indora-oracle algorithm with failure probability at most $1/10$ has worst-case expected oracle cost $\Omega((n^{g(k,\alpha_k)}/k^{3k}) \cdot n^{\alpha_k})$ under $\cost_k$.
\end{theorem}

Observe that the polynomial overhead $n^{g(k,\alpha_k)}$ of approximate counting over decision is roughly equal to $n^{(k-\alpha_k)^2/(4k)}$. If $\alpha_k = 0$, then the worst-case oracle cost of an algorithm is simply the worst-case number of queries that it makes. Thus \cref{thm:uncol-main-simple} generalises known matching upper and lower bounds of $\widetilde\Theta(\sqrt{n})$ queries in the graph case~\cite{CLW-graph-tight}, both by allowing $k > 2$ and by allowing $\alpha_k > 0$. (See \cref{sec:intro-related} for more details.) Moreover, if $\alpha_k \ge k-1$, then $g(k,\alpha_k) = 0$; thus in this case, \cref{thm:uncol-main-simple} shows that approximate counting requires the same oracle cost as decision, up to a polylogarithmic factor. Taking $k=2$ and $\alpha_k=1$, this implies that whenever we can simulate an edge-detection oracle for a graph in linear time, then we can also obtain a linear-time approximate edge-counting algorithm (up to polylogarithmic factors). Analogous upper bounds on the running time and oracle cost of \texttt{Uncol} also hold for any ``reasonable'' cost function of the form $\cost_k(n) = n^{\alpha_k+o(1)}$; for details, see \cref{sec:regularly-varying} and \cref{thm:uncolapprox-algorithm}.

\subsubsection{Our results for the colourful independence oracle.}
Given a $k$-hypergraph $G$ with vertex set $[n]$, the \textit{colourful independence oracle} is the bitstring $\cindora(G)$ such that for all disjoint sets $S_1,\dots,S_k \subseteq [n]$, $\cindora(G)_{S_1,\dots,S_k} = 1$ if $G$ contains no edge $e \in E(G)$ with $|S_i \cap e|=1$ for all $i$, and $0$ otherwise.
We view $S_1,\dots,S_k$ as colour classes in a partial colouring of $[n]$; thus a query to $\cindora(G)_{S_1,\dots,S_k}$ allows us to test whether or not $G$ contains an edge with one vertex of each colour. (Note that we do not require $S_1 \cup \dots \cup S_k = [n]$.) Analogously to the uncoloured case, we define the \textit{cost} of an oracle call $\cindora(G)_{S_1,\dots,S_k}$ to be a polynomial function of the form $\cost_k(S_1,\dots,S_k) = \cost_k(|S_1|+\dots+|S_k|) = (|S_1| + \dots + |S_k|)^{\alpha_k}$, where the map $k\mapsto\alpha_k$ satisfies $\alpha_k \in [0,k]$ but is otherwise arbitrary.

It is not too hard to show that the naive $\OO(n^k)$-cost exact edge-counting algorithm of querying every possible edge and the naive $\OO((k^k/k!)n^{\alpha_k})$-cost algorithm to decide whether any edge is present by randomly colouring the vertices are both essentially optimal, and indeed we prove as \cref{prop:col-dec} that any such decision algorithm requires cost $\Omega(n^{\alpha_k})$. For approximate counting, we prove the following.

\begin{theorem}[Colourful independence oracle, polynomial cost function]\label{thm:col-main-simple}
    Let $\alpha_k \in [0,k]$ for all $k \ge 2$, let $\cost_k(x) = x^{\alpha_k}$,
    and let
    \(
        T \coloneqq \log(1/\delta)\eps^{-2}k^{27k}\log^{4(k-\ceil{\alpha_k})+18} n
    \).
    There is a randomised $\eps$-approximate edge counting algorithm
    \(\aau(\cindora(G),\eps,\delta)\)
    with worst-case running time
    \(
        \OO(T\cdot n^{\max(1,\alpha_k)})
    \),
    worst-case oracle cost
    \(
        \OO(T\cdot n^{\alpha_k})
    \)
    under $\cost_k$, and failure probability at most $\delta$. Moreover, every randomised $(1/2)$-approximate edge counting \cindora-oracle algorithm with failure probability at most $1/10$ has worst-case oracle cost 
    \[
        \Omega\bigg(k^{-9k}\Big(\frac{\log n}{\log\log n}\Big)^{k-\floor{\alpha_k}-3} \cdot n^{\alpha_k} \bigg)
    \]
    under $\cost_k$.
\end{theorem}

The upper bound replaces a $\log^{\Theta(k)}n$ term in the query count of the previous best-known algorithm~(\cite{BBGM-runtime} for $\alpha_k = 0$) by a $\log^{\Theta(k-\alpha_k)}n$ term in the multiplicative overhead over decision, giving polylogarithmic overhead over decision when $k- \alpha_k=\OO(1)$. The lower bound shows that this term is necessary and cannot be reduced to $\log^{\OO(1)}n$; this is a new result even for $\alpha_k=0$. (See \cref{sec:intro-related} for more details.) Analogous upper bounds on the running time and oracle cost of \aau also hold for any ``reasonable'' cost function of the form $\cost_k(n) = n^{\alpha_k+o(1)}$; for details, see \cref{sec:regularly-varying} and \cref{thm:col-alg}. 

\subsubsection{Approximate sampling results.} There is a known fine-grained reduction from approximate sampling to approximate counting~\cite{DLM}. Strictly speaking this, reduction is proved for $\alpha_k = 1$ with a colourful independence oracle, but the only actual use of the oracle in the reduction is to enumerate all edges in a set $X$ with $\OO(|X|^k)$ size-$k$ queries, so it transfers immediately to our setting. The upper bounds of \cref{thm:uncol-main-simple,thm:col-main-simple} therefore also yield approximate sampling algorithms with overhead $2^{\OO(k)}\log^{\OO(1)}n$ over approximate counting.

\subsubsection{A parameterised complexity motivation for our lower bound results.}
To understand an important motivation for the lower bounds in our results, consider as an example the longest path problem: Given $(G,k)$, does there exist a simple path of length~$k$?
A long sequence of works in parameterised complexity led to a spectacular algorithm~\cite{DBLP:journals/jcss/BjorklundHKK17} for this problem in undirected graphs that runs in time $1.66^k\cdot\poly(n)$.
There is a somewhat shorter sequence of works for the corresponding approximate counting version of the problem, which culminated in a $4^{k}\poly(n)$-time algorithm \cite{DBLP:conf/stoc/BrandDH18,DBLP:journals/talg/LokshtanovBSZ21}.

Instead of designing ever-more sophisticated algorithms for approximately counting $k$-paths in order to get closer to the running time of the decision problem, our dream result would instead be a subexponential-time approximate-counting-to-decision reduction that uses the decision problem in a black-box fashion and causes only a factor $2^{o(k)}\poly(n)$ overhead in the running time.
This way, any improvements to the decision algorithm would automatically carry over.
One way to formalise what the black-box can do is captured by defining the $k$-hypergraph whose edges are the $k$-paths of the underlying graphs; using an algorithm for the $k$-path decision problem, it is trivial to simulate the independence oracle and easy to simulate the colourful independence oracle for this hypergraph.

\cref{thm:col-main-simple} implies that any decision algorithm for $k$-path can be turned into an approximate counting algorithm by paying a $\log^{\OO(k)} n$-factor overhead in the running time. While this is still a fixed-parameter tractable running time, it leads to a useless algorithm, since the running time is much worse than $c^k\poly(n)$.
The main consequence of \cref{thm:col-main-simple} for $k$-path stems not from this meaningless upper bound, but from the lower bound, which is new even for $\alpha_k=0$:
Our results imply that if the decision algorithm for $k$-path is formalized using the colourful independence oracle, then the overhead of the approximate-counting-to-decision reduction must be $\log^{\Omega(k)}n$, and so a subexponential-time reduction cannot exist.
Conversely, if a useful approximate-counting-to-decision reduction exists, it cannot merely be based on the hypergraph whose edges consist of all $k$-paths; instead, the reduction would have to have access to and exploit the underlying structure of the original graph. We believe that this is a useful insight for the design of future algorithms for approximate counting.

\subsection{Reductions from approximate counting to decision}\label{sec:intro-applications}

\cref{thm:uncol-main-simple,thm:col-main-simple} can easily be applied to obtain reductions from approximate counting to decision for many important problems in fine-grained and parameterised complexity.
The following definition is taken from~\cite{DLM}; recall that a counting problem is a function $\#\Pi\colon \zo^*\to\mathbb{N}$ and its corresponding decision problem is defined via $\Pi=\{x\in\zo^*\colon\#\Pi(x)>0\}$.
\begin{defn}\label{def:UWP}
    The decision problem~$\Pi$ is a \emph{uniform witness problem} if there is a function that maps instances $x\in\zo^*$ to uniform hypergraphs $G_x$ such that the following statements hold:
    \begin{enumerate}[(i)]
        \item $\#\Pi(x)$ is equal to the number~$e(G_x)$ of edges in~$G_x$;
        \item $V(G_x)$ and the size $k(G_x)$ of edges in $E(G_x)$ can be computed from $x$ in time $\widetilde\OO(|x|)$;
        \item there exists an algorithm which, given $x$ and $S \subseteq V(G_x)$, in time $\widetilde\OO(|x|)$ prepares an instance $I_x(S)\in\zo^*$ such that $G_{I_x(S)} = G_x[S]$ and $|I_x(S)| \in \OO(|x|)$.
    \end{enumerate}
    The set $E(G_x)$ is the set of \emph{witnesses} of the instance $x$.
\end{defn}

Intuitively, we can think of a uniform witness problem as a problem of counting witnesses in an instance~$x$ that can be naturally expressed as edges in a hypergraph~$G_x$, in such a way that induced subgraphs of~$G_x$ correspond to sub-instances of $x$. This allows us to simulate a query to $\indora(G_x)_S$ by running a decision algorithm for $\Pi$ on the instance $I_x(S)$, and if our decision algorithm runs on an instance $y$ in time $T(|y|)$ then this simulation will require time $\widetilde{\OO}(T(|S|))$. Typically there is only one natural map $x\mapsto G_x$, and so we consider it to be a part of the problem statement. Simulating the independence oracle in this way, the statement of \cref{thm:uncol-main-simple} yields the following.

\begin{theorem}\label{thm:intro-uncol-algo-simple}
    Suppose $\alpha_k \in [1,k]$ for all $k \ge 2$. Let $\Pi$ be a uniform witness problem. Suppose that given an instance $x$ of $\Pi$, writing $n=|V(G_x)|$ and $k=k(G_x)$, there is an algorithm to solve $\Pi$ on $x$ with error probability at most $1/3$ in time $\widetilde{\OO}(n^{\alpha_k})$. Then there is an $\eps$-approximation algorithm for $\#\Pi(x)$ with error probability at most $1/3$ and running time 
    \[
        k^{\OO(k)} + \eps^{-2}n^{\alpha_k}\cdot 2^{\OO(k)}n^{g(k,\alpha_k)}\,.
    \]
\end{theorem}

Note that the running time of the algorithm for $\#\Pi$ is the sum of the oracle cost and the running time of the algorithm of \cref{thm:uncol-main-simple}; by requiring $\alpha_k \ge 1$, we ensure this is dominated by the oracle cost. (Indeed, for most uniform witness problems it is very easy to prove that every decision algorithm must read a constant proportion of the input, and so we will always have $\alpha_k \ge 1$.) The lower bound of \cref{thm:uncol-main-simple} implies that the $n^{\alpha_k+g(k,\alpha_k)}$ term in the running time cannot be substantially improved with any argument that relativises; thus in simple terms, there is a generic fine-grained reduction from approximate counting to decision if and only if the decision algorithm runs in time $\Omega(n^{k-1})$.

It is instructive to consider an example. Take $\Pi$ to be the problem \textsc{Induced-$H$} of deciding whether a given input graph $G$ contains an induced copy of a fixed graph $H$. In this case, the hypergraph corresponding to an instance $G$ will have vertex set $V(G)$ and edge set $\{X\subseteq V(G)\colon G[X] \simeq H\}$; thus the witnesses are vertex sets which induce copies of $H$ in $G$. The requirements of \cref{def:UWP}(i) and (ii) are immediately satisfied, and \cref{def:UWP}(iii) is satisfied since deleting vertices from the hypergraph corresponds to deleting vertices of $G$. Thus writing $k=|V(H)|$, \cref{thm:intro-uncol-algo-simple} gives us a reduction from approximate \#\textsc{Induced-$H$} to \textsc{Induced-$H$} with overhead $\eps^{-2}2^{\OO(k)}n^{g(k,\alpha_k)}$ over the cost of the decision algorithm.
Moreover, on applying the fine-grained reduction from approximate sampling to counting in \cite{DLM} we also obtain an approximate uniform sampling algorithm with overhead $\eps^{-2}2^{\OO(k)}n^{g(k,\alpha_k)}$. 
Many more examples of uniform witness problems to which \cref{thm:intro-uncol-algo-simple} applies can be found in the introduction of~\cite{DLM}.

We now describe the corresponding result in the colourful oracle setting, which we now set out --- again, the following definition is taken from~\cite{DLM}.
\begin{defn}
    Suppose $\Pi$ is a uniform witness problem. \textsc{Colourful-$\Pi$} is defined as the problem of, given an instance $x \in \zo^*$ of $\Pi$ and a partition of $V(G_x)$ into disjoint sets $S_1,\dots,S_k$, deciding whether $\cindora_{G_x}(S_1,\dots,S_k)=0$ holds.
\end{defn}
Continuing our previous example, in \textsc{Colourful-Induced-$H$}, we are given a (perhaps improper) vertex colouring of our input graph $G$, and we wish to decide whether $G$ contains an induced copy of $H$ with exactly one vertex from each colour. Simulating an oracle call to $\cindora(G_x)_{S_1,\dots,S_k}$ corresponds to running a decision algorithm for \textsc{Colourful-$\Pi$} on the instance $I_x(S_1 \cup \dots \cup S_k)$ with colour classes $S_1,\dots,S_k$, and if this decision algorithm runs on an instance $y$ in time $T(|y|)$ then this simulation will require time $\widetilde\OO(T(|S_1| + \dots + |S_k|))$. Simulating the colourful independence oracle in this way, the statement of \cref{thm:col-main-simple} yields the following.
\begin{theorem}\label{thm:intro-col-algo-simple}
    Suppose $\alpha_k \in [1,k]$ for all $k \ge 2$. Let $\Pi$ be a uniform witness problem. Suppose that given an instance $x$ of $\Pi$, writing $n = |V(G_x)|$ and $k=k(G_x)$, there is an algorithm to solve \textsc{Colourful-$\Pi$} on $x$ with error probability at most $1/3$ in time $\OO(n^{\alpha_k})$. Then there is an $\eps$-approximation algorithm for $\#\Pi(x)$ with error probability at most $1/3$ and running time
    \[
        \eps^{-2}n^{\alpha_k} \cdot k^{\OO(k)}(\log n)^{\OO(k-\alpha_k)}\,.
    \]
\end{theorem}
As in the uncoloured case, the requirement $\alpha_k \ge 1$ ensures that the dominant term in the running time is the time required to simulate the required oracle queries, and the lower bound of \cref{thm:col-main-simple} implies that the $\log^{\OO(k-\alpha_k)}n$ term in the running time cannot be substantially improved with any argument that relativises. Again writing $k=|V(H)|$, \cref{thm:intro-col-algo-simple} gives us a reduction from approximate \#\textsc{Induced-$H$} to \textsc{Colourful-Induced-$H$} with overhead $\eps^{-2}k^{\OO(k)}(\log n)^{\OO(k-\alpha_k)}$ over the cost of the decision algorithm. This result improves on the reduction of~\cite[Theorem~1.7]{DLM} by a factor of $\log^{\Omega(\alpha_k)}n$, and using the fine-grained reduction from approximate sampling to counting in \cite{DLM} it can immediately be turned into an approximate uniform sampling algorithm.

Observe that in most cases, there is far less overhead over decision in applying \cref{thm:intro-col-algo-simple} to reduce \#\textsc{Induced-$H$} to \textsc{Colourful-Induced-$H$} than there is in applying \cref{thm:intro-uncol-algo-simple} to reduce \#\textsc{Induced-$H$} to \textsc{Induced-$H$}. In some cases, such as the case where $H$ is a $k$-clique, there are simple fine-grained reductions from \textsc{Colourful-Induced-$H$} to \textsc{Induced-$H$}, and in this case \cref{thm:intro-col-algo-simple} is an improvement.
However, it is not known whether the same is true of all choices of $H$, and indeed even an FPT reduction from \textsc{Colourful-Induced-$H$} to \textsc{Induced-$H$} would imply the long-standing dichotomy conjecture for the embedding problem introduced in~\cite{embedding-conjecture}. More generally, but still within the setting of uniform witness problems, the problem of detecting whether a graph contains a size-$k$ set which either spans a clique or spans an independent set is in FPT by Ramsey's theorem, but its colourful version is W[1]-complete~\cite{meeksunbddtw}.

While the distinction between colourful problems and uncoloured problems is already well-studied in subgraph problems, these results strongly suggest that it is worth studying in many other contexts in fine-grained complexity as well. Indeed, it is easy to simulate $\indora(G)$ from $\cindora(G)$ with random colouring; thus the lower bound of \cref{thm:uncol-main-simple} and the upper bound of \cref{thm:col-main-simple} imply that there is a fine-grained reduction from uncoloured approximate counting to colourful decision, but not to uncoloured decision. By studying the relationship between colourful problems and their uncoloured counterparts, we may therefore hope to shed light on the relationship between approximate counting and decision.

Finally, we observe that the set of running times allowed by \cref{thm:intro-uncol-algo-simple,thm:intro-col-algo-simple} may not be sufficiently fine-grained to derive meaningful results for some problems. In fine-grained complexity, even a subpolynomial improvement to a polynomial-time algorithm may be of significant interest --- the classic example is the \textsc{Negative-Weight-Triangle} algorithm of \cite{williams-nwt-fast}, which runs in $n^3/e^{\Omega(\sqrt{\log n})}$ time on an $n$-vertex instance, compared to the naive $\OO(n^3)$-time algorithm. In order to ``lift'' such improvements from decision problems to approximate counting, we must generalise the upper bounds of \cref{thm:uncol-main-simple,thm:col-main-simple} to cost functions of the form $\cost_k(x) = x^{\alpha_k\pm o(1)}$
while maintaining low overhead. We do so in \cref{thm:uncolapprox-algorithm,thm:col-alg} for all ``reasonable'' cost functions, including any function of the form $\cost_k(x) = n^{\alpha_k}\log^{\beta_k}n$ and any function of the form $\cost_k(x) = n^{\alpha_k}e^{\pm (\log n)^{\gamma_k}}$ where $\gamma_k < 1$.
A full list of technical requirements is given in \cref{sec:regularly-varying}, but the most important one is regular variation --- this is a standard notion from probability theory for ``almost polynomial'' functions, and requiring it avoids pathological cases where (for example) we may have $\cost_k(x) = \OO(x)$ as $x\to\infty$, but $\cost_k(2x_i) = \omega(\cost_k(x_i))$ as $i\to\infty$ along some sequence $(x_i\colon i \ge 1)$.

\subsection{Discussion of related work}\label{sec:intro-related}

In order to compare algorithms without excessive re-definition of notation, throughout this subsection we consider the problem of $\eps$-approximating the number of edges in an $m$-edge, $n$-vertex $k$-hypergraph.

Colourful and uncoloured independence oracles were introduced in~\cite{BHRRS-oracle-intro} in the graph setting, then first generalised to hypergraphs in~\cite{DBLP:conf/isaac/BishnuGKM018}.
Edge estimation using these oracles was first studied in the graph setting (i.e.\ for $k=2$) in~\cite{BHRRS-oracle-intro}, which gave an $\eps^{-4}\log^{\OO(1)}n$-query algorithm for colourful independence oracles and an $(\eps^{-4} + \eps^{-2}\min\{\sqrt{m},n^2/m\})\log^{\OO(1)}n = (\eps^{-4}+\eps^{-2}n^{2/3})\log^{\OO(1)}n$-query algorithm for uncoloured independence oracles. The connection to approximate counting in fine-grained and parameterised complexity was first studied in~\cite{DL}.

For colourful independence oracles in the graph setting, \cite{DL} (independently from \cite{BHRRS-oracle-intro}) gave an algorithm using $\eps^{-2}\log^{\OO(1)}n$ $\cindora$ queries and $\eps^{-2}n\log^{\OO(1)}n$ adjacency queries. \cite{AMM-non-adaptive} subsequently gave a \textit{non-adaptive} algorithm using $\eps^{-6}\log^{\OO(1)}n$ $\cindora$ queries.

The case of edge estimation in~$k$-hypergraphs (i.e.\ for arbitrary $k \ge 2$) was first considered independently by \cite{DLM,BBGM-hypergraph}; \cite{DLM} gave an algorithm using $\eps^{-2}k^{\OO(k)}\log^{4k+\OO(1)}n$ queries, while \cite{BBGM-hypergraph} gave an $\eps^{-4}k^{\OO(k)}\log^{4k+\OO(1)}$-query algorithm. \cite{DLM} also introduced a reduction from approximate sampling to approximate counting in this setting (which also applies in the uncoloured setting) with overhead $k^{\OO(k)}\log^{\OO(1)}n$. \cite{BBGM-runtime} then improved the query count further to $\eps^{-2}k^{\OO(k)}\log^{3k+\OO(1)}n$.

In this paper, we give an algorithm with total query cost $\eps^{-2}k^{\OO(k)}\log^{4(k-\alpha_k)+\OO(1)}n$ under $\cost_k(x) = x^{\alpha_k}$, giving polylogarithmic overhead when $\alpha_k \approx k$. We also give a lower bound which shows that a $\log^{\Theta(k-\alpha_k)}$ term is necessary; no lower bounds were previously known even for $\alpha_k = 0$ (i.e.\ the case where the total query cost equals the number of queries).

For uncoloured independence oracles of graphs, \cite{CLW-graph-tight} improved the algorithm of \cite{BHRRS-oracle-intro} to use
\begin{equation}\label{eq:clw}
    \eps^{-\Theta(1)} \min\{\sqrt{m},n/\sqrt{m}\} \polylog n= \eps^{-\Theta(1)}\sqrt{n}\polylog n
\end{equation}
queries and gave a matching lower bound. (It is difficult to tell the exact value of the $\Theta(1)$ term from the proof, but it is at least $9$ --- see the definition of $N$ in the proof of Lemma 3.9 on p. 15.)
It is worth noting that the bound in \eqref{eq:clw} is stated as a function of both~$n$ and~$m$. We believe that our results can be stated in such a way as well, but we defer doing so to the journal version of this paper.

To the best of our knowledge, no results on uncoloured edge estimation for $\alpha_k > 0$ or $k > 2$ have previously appeared in the literature. However, we believe it would be easy to partially generalise the proof of~\cite{BHRRS-oracle-intro} to this setting. Very roughly speaking, their argument works by running a naive sampling algorithm and a more subtle branch-and-bound approximation algorithm in parallel, with the sampling algorithm running quickly on dense graphs and the branch-and-bound algorithm running quickly on sparse graphs. The main obstacle to generalising this approach would be the branch-and-bound approximation algorithm; however, by replacing it with a slower branch-and-bound \textit{enumeration} algorithm for $k$-hypergraphs such as~\cite{DBLP:journals/algorithmica/Meeks19}, we believe we would obtain worst-case oracle cost $k^{\OO(k)}+\eps^{-2}2^{\OO(k)}n^{\alpha_k + (k-\alpha_k)/2}$ under $\cost_k(x) = x^{\alpha_k}$; this technique is well-known in the literature and also appears in e.g.~\cite{Thurley}.
By comparison (see~\cref{fig:uncolored-main-thm}), the algorithm of \cref{thm:uncol-main-simple} achieves a much smaller worst-case oracle cost of $k^{\OO(k)}+\eps^{-2}2^{\OO(k)}n^{\alpha_k + g(k,\alpha_k)}$, where $g(k,\alpha_k) \approx (k-\alpha_k)^2/(4k) < (k-\alpha_k)/2$ and where $g(k,\alpha_k) = 0$ for $\alpha_k \ge k-1$.
This substantially improves on the algorithm implicit in \cite{BHRRS-oracle-intro}, and indeed is optimal up to a factor of $\eps^{\Theta(1)}k^{\Theta(k)}$. Also, our algorithm has a better dependence on~$\epsilon$ compared with \cite{CLW-graph-tight} when $k=2$; however, we bound the cost only in terms of $n$ and not in terms of~$m$.

\begin{figure}[t]
    \begin{center}
        \includegraphics{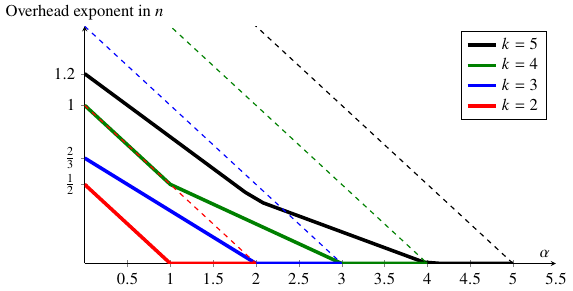}
        \includegraphics{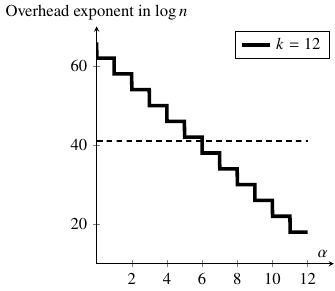}
    \end{center}
    \caption{\label{fig:uncolored-main-thm}%
    \emph{Left:} If each \indora-query of size~$x$ has cost $x^{\alpha_k}$, then up to $k^{\OO(k)}\log^{\OO(1)}n$ factors, we show in \cref{thm:uncol-main-simple} that $n^{g(k,\alpha_k)+\alpha_k}$ is the smallest possible \indora-oracle cost to (1/2)-approximate the number of edges.
    Plotted here in \emph{thick lines} is the overhead $\alpha\mapsto g(k,\alpha)$ in the exponent of $n$ for $k\in\set{2,3,4,5}$, and in \emph{dashed lines} is the larger overhead exponent $\alpha\mapsto(k-\alpha)/2$ obtained from a naive generalisation of the techniques of~\cite{BHRRS-oracle-intro}.
    \emph{Right:} If each \cindora-query of size $x$ has cost $x^{\alpha_k}$, then up to $k^{\OO(k)}(\log n)^{\OO(1)}(\log\log n)^{\OO(k-\alpha_k)}$ factors, we show in Theorem~\ref{thm:col-main-simple} that $n^{\alpha_k}\log^{4(k-\lceil\alpha_k\rceil)+18}n$ is the smallest possible \cindora-oracle cost to $(1/2)$-approximate the number of edges. Plotted in \emph{thick lines} is the overhead $\alpha\mapsto 4(k-\lceil \alpha\rceil ) + 18$ in the exponent of $\log n$ for $k=12$, and the \emph{dashed line} depicts the overhead $\alpha\mapsto 3k+5$ obtained by using the bound on the number of queries by \cite{BBGM-runtime}; our bound is better if $\ceil{\alpha_k}\ge k/4+\Theta(1)$.
    }\label{fig:runtime-exp}
\end{figure}

Although a full survey is beyond the scope of this paper, there are natural generalisations of (colourful and uncoloured) independence oracles \cite{RWZ-vmv-queries}, and edge estimation problems are studied for other oracle models including neighbourhood access~\cite{Tetek-neighbour}. Other types of oracle are also regularly applied to fine-grained complexity in other models, notably including cut oracles~\cite{MN-cut,AEGLMN-cut}. Perhaps surprisingly, we were unable to find any previous examples in the literature of unconditional oracle lower bounds relative to a general query cost function, and so our definitions and methods are novel in that sense. Of course, many existing works prove unconditional lower bounds in terms of query count~\cite{Tetek-neighbour}, or consider algorithmic construction of graph oracles with bounded query times~\cite{Le-planar-oracle}, or provide oracle algorithms with fast running times in addition to low query counts~\cite{PZ-ML-oracle}. In our setting, however, a lower bound in terms of query cost is absolutely necessary. Recall that for us, query cost is the running time of an algorithm for the decision problem we are reducing to, and the algorithmic results of \cref{thm:intro-col-algo-simple,thm:intro-uncol-algo-simple} all  rely on smaller queries running faster; thus to prove they are ``best possible'' in any meaningful sense, we absolutely require the formal notion of query cost set out in Section~\ref{sec:intro-oracle}.

Outside the oracle setting, it was recently proved~\cite{LSZ-rep-families} that any decision algorithm built around the representative family technique of~\cite{FLPS-rep-families} can be turned into an approximate counting algorithm with substantially lower overheads in $k$ than those of~\cite{DLM}. More recently, several important decision problems in the fine-grained setting with $k=3$ have been discovered to be ``equivalent'' to their \emph{exact} counting versions~\cite{CWX-exact-equivalence}. This work is not directly comparable to ours, as they work in a substantially stronger setting and use a correspondingly weaker notion of equivalence. Their equivalence is in the sense of equivalence of conjectures --- for example, they prove that if there is an $O(n^{2-\eps})$-time algorithm for 3-SUM, then there exists $0 < \eps' < \eps$ such that there is an $O(n^{2-\eps'})$-time algorithm for \#3-SUM. We stress that for exact counting problems as studied in~\cite{CWX-exact-equivalence}, such equivalence results are genuinely deep and surprising. However, for the approximate counting problems we study, analogous results are typically quite easy to prove via the standard combination of sampling and branch-and-bound approaches discussed above, and so we instead study the stronger notion of fine-grained reductions from approximate counting to decision. Where no such reductions exist, we aim to nail down the exact value of $\eps'$. (Recall also that there are uniform witness problems whose decision versions are in FPT but whose exact counting versions are W[1]-hard~\cite{meeksunbddtw}, so we cannot hope to extend our results to exact counting in this setting.)

\subsection{Proof techniques}\label{sec:intro-proofs}

\subsubsection{Colourful upper bound}

We first discuss the proof of the upper bound of \cref{thm:col-main-simple}, our \cindora-oracle algorithm for edge estimation using the colourful independence oracle of an $n$-vertex $k$-hypergraph~$G$. In~\cite{DLM}, it is implicitly proved that a \cindora-oracle algorithm that computes a ``coarse approximation'' to~$e(G)$, which is accurate up to multiplicative error~$b$, can be bootstrapped into a full $\eps$-approximation algorithm with overhead $\eps^{-2}2^{\OO(k)}\log^{\OO(1)}n \cdot b^2$. (See \cref{thm:use-coarse} of the present paper for details.) It therefore suffices to improve the coarse approximation algorithm of~\cite{DLM} from $k^{\OO(k)}\log^{\Theta(k)}n$ queries and $k^{\OO(k)}\log^{\Theta(k)}n$ multiplicative error to $k^{\OO(k)}\log^{\Theta(k-\alpha_k)}n$ query cost and $k^{\OO(k)}\log^{\Theta(k-\alpha_k)}n$ multiplicative error. Moreover, by a standard colour-coding argument, it suffices to make this improvement when $G$ is $k$-partite with vertex classes $V_1,\dots,V_k$ known to the algorithm. 

Oversimplifying a little, and assuming $n$ is a power of two, the algorithm of \cite{DLM} works by guessing a probability vector $(Q_1,\dots,Q_k) \in \{1,1/2,1/4,\dots,1/n\}^k$. It then deletes vertices from $V_1,\dots,V_k$ independently at random to form sets $\calX_1,\dots,\calX_k$, so that for all $v \in V_j$ we have $\pr(v \in \calX_j) = Q_j$. After doing so, in expectation, $Q_1Q_2\dots Q_k$ proportion of the edges of $G$ remain in the induced $k$-partite subgraph $G[\calX_1,\dots,\calX_k]$. If $e(G) \ll 1/(Q_1\dots Q_k)$, it is easy to show with a union bound that no edges are likely to remain. What is more surprising is that there exist $q_1,\dots,q_k$ with $q_1\dots q_k \approx 1/e(G)$ such that if $\vec{Q} = \vec{q}$, then at least one edge is likely to remain in $G[\calX_1,\dots,\calX_k]$. Thus the algorithm of \cite{DLM} iterates over all $\log^{\Theta(k)}n$ possible values of $\vec{Q}$, querying $\cindora(G)$ on $\calX_1,\dots,\calX_k$ for each, and then outputs the least value $m$ such that $e(G[\calX_1,\dots,\calX_k]) > 0$ for some $Q_1,\dots, Q_k$ with $1/(Q_1\dots Q_k) = m$.

Our algorithm improves on this idea as follows. First, \cite{DLM} does not actually prove the existence of the vector $\vec{q}$ described above --- it relies on a coupling between the different guesses of~$\vec{Q}$. We require not only the existence of~$\vec{q}$ but also a structural result which may be of independent interest. For all $I \subseteq [k]$ and all $\zeta \in (0,1]$, we define an \emph{$(I,\zeta)$-core} to be an induced subgraph $H=G[Y_1,\dots,Y_k]$ of $G$ such that:
\begin{enumerate}[(i)]
    \item $H$ contains at least $k^{-\OO(k)}$ proportion of the edges of $G$.
    \item For all $i \in I$, the set~$Y_i$ is very small, containing at most $2/\zeta$ vertices.
    \item For all $i \notin I$, every vertex of $Y_i$ is contained in at most $\zeta$ proportion of edges in $H$.
\end{enumerate}
As an example, the most extreme core is the \emph{rooted star}: It consists of some vertices $r_i\in Y_i$ for all $i\in I$ and \emph{all} $k$-partite edges~$e$ with $e\supseteq\setc{r_i}{i\in I}$.
We prove in \cref{lem:col-algo-core-exists} that, for all $\zeta \in (0,1]$, there is some $I \subseteq [k]$ such that $G$ contains an $(I,\zeta)$-core~$H$. 

Suppose for the moment that we are given the value of $I$, but not $Y_1,\dots,Y_k$. By property (i), it would then suffice to approximate $e(H)$ using $k^{\OO(k)}\log^{\OO(k-\alpha_k)}n$ query cost. If $|I| \ge \alpha_k$, then we can adapt the idea of the algorithm of \cite{DLM}, but taking $Q_j = 1$ for all $i \notin I$ to use only $\log^{\OO(k-\alpha_k)}n$ queries in total; intuitively, this is possible due to property (ii), which implies that this is the ``correct'' setting. We set this algorithm out as \dlmimprove in \cref{sec:dlm-coarse}. 

If instead $|I| \le \alpha_k$, then we will exploit the fact that query cost decreases polynomially with instance size by breaking $H$ into smaller instances. For all $i \notin I$, we randomly delete vertices from $V_i$ with a carefully-chosen probability $p$. Property (iii), together with a martingale bound (see \cref{lem:colourful-algo-conc}), guarantees that the number of edges in the resulting hypergraph $G'$ will be concentrated around its expectation of $p^ke(G)$. If we had access to $Y_1,\dots,Y_k$, we could then intersect~$V_i$ with~$Y_i$ for all $i \in I$ to obtain a substantially smaller instance, whose edges we could count with cheap queries; we could then divide the result by $p^k$ to obtain an estimate for $e(G)$. Unfortunately we do not have access to $Y_1,\dots,Y_k$, but we can still break $G'$ into smaller sub-hypergraphs by applying colour-coding to the vertex sets $V_i$ with $i \in I$, and as long as $|I| \le \alpha_k$ this still gives enough of a saving in query cost to prove the result. We set this algorithm out as \dlmrecurse in \cref{sec:dlm-recurse}.

Now, we are not in fact given the value of $I$ in the $(I,\zeta)$-core of $G$. But both \dlmimprove and \dlmrecurse fail gracefully if they are given an incorrect value of $I$, returning an underestimate of $e(G)$ rather than an overestimate. We can therefore simply iterate over all $2^k$ possible values of $I$, applying \dlmimprove or \dlmrecurse as appropriate, and return the maximum resulting estimate of $e(G)$. This proves the result.

\subsubsection{Colourful lower bound}

We now discuss the proof of the lower bound of \cref{thm:col-main-simple}. Using the minimax principle, to prove the bound for randomised algorithms, it is enough to give a pair of random graphs $G_1$ and $G_2$ with $e(G_2) \gg e(G_1)$ and prove that no \textit{deterministic} algorithm $A$ can distinguish between $G_1$ and $G_2$ with constant probability and worst-case oracle cost as in the bound. We base these random graphs on the main bottleneck in the algorithm described in the previous section: the need to check all possible values of $\calQ$ in a $k$-partite $k$-hypergraph with an $(I,\zeta)$-core where $|I| \approx k-\alpha_k$. 

We take $G_1$ to be an Erd\H{o}s-R\'{e}nyi $k$-partite $k$-hypergraph with edge density $p := t^{-(k-\flalphak-2)/2}$. We take the vertex classes $V_1,\dots,V_k$ of $G_1$ to have equal size $t$, so that $t=n/k$. We then define a random complete $k$-partite graph $H$ as follows. We first define a random vector $\vec{\calQ}$ of probabilities, then take binomially random subsets $\calX_1,\dots,\calX_{k-\flalphak-2}$ of $V_1,\dots,V_{k-\flalphak-2}$, with $\pr(v \in \calX_j) = \calQ_j$ for all $v \in V_j$. For $j \ge k-\flalphak-1$, we take $\calX_j \subseteq V_j$ to contain a single uniformly random vertex. We then define $H$ to be the complete $k$-partite graph with vertex classes $\calX_1,\dots,\calX_k$, and form $G_2 = G_1 \cup H$ by adding the edges of $H$ to $G_1$. We choose $\calQ$ in such a way that $\calQ_1 \cdot \ldots\cdot \calQ_{k-\flalphak-2}$ is guaranteed to be a bit larger than $pt^{\flalphak+2}$, so that $\E(e(H)) \gg \E(e(G_1))$. Intuitively, this corresponds to the situation of a randomly planted $(I,\zeta)$-core where $I = \{k-\flalphak-1,\dots,k\}$ --- we will show that the algorithm $A$ needs to essentially guess the value of $\calQ$ using expensive queries.

To show that a low-cost deterministic algorithm $A$ cannot distinguish $G_1$ from $G_2$, suppose for simplicity that $A$ is \emph{non-adaptive}, so that its future oracle queries cannot depend on the oracle's past responses. In this setting, it suffices to bound the probability of a fixed query $S=(S_1,\dots,S_k)$ distinguishing $G_1$ from $G_2$. 

It is not hard to show that without loss of generality we can assume $S_i \subseteq V_i$ for all $i \in [k]$. If $|S_j|\ll t$ for some $j \ge k-\flalphak-1 $, then with high probability $S_j$ will not contain the single ``root'' vertex of $\calX_j$, so $H[S_1,\dots,S_k]$ will contain no edges and $S$ will not distinguish $G_1$ from $G_2$. With some effort (a simple union bound does not suffice), this idea allows us to essentially restrict our attention to large, expensive queries $S$. However, if $|S_1|\dots|S_k|$ is large, then with high probability $G_1[S_1,\dots,S_k]$ will contain an edge, so again $S$ will not distinguish $G_1$ from $G_2 = G_1 \cup H$. With some more effort, this allows us to essentially restrict our attention to queries where for some possible value $\vec{q}$ of $\vec{\calQ}$ we have $|S_j| \approx 1/q_j$ for all $j \le k-\flalpha-2$; we choose these possible values to be far enough apart that such a query is only likely to distinguish $G_1$ from $G_2$ if $\vec{\calQ} = \vec{q}$. There are roughly $((\log n)/\log\log n)^{k-\flalphak-2}$ possible values of $\calQ_j$, so the result follows.

Of course, in our setting $A$ may be adaptive, and this breaks the argument above. Since the query $A$ makes depends on the results of past queries, we cannot bound the probability of a fixed query distinguishing $G_1$ from $G_2$ in isolation --- we must condition on the results of past queries. This is not a small technical point --- it is equivalent to allowing $A$ to be adaptive in the first place. The most damaging implication is that we could have a query $S = (S_1,\dots,S_k)$ with $|S_1|\dots |S_k|$ very large but such that $G_1[S_1,\dots,S_k]$ contains no edges, because most of the potential edges have already been exposed as not existing in past queries. We are able to deal with this while preserving the spirit of the argument above, by arguing based on the number of unexposed edges rather than $|S_1|\dots |S_k|$, but it requires significantly more effort and a great deal of care.

\subsubsection{Uncoloured upper bound}

We now discuss the proof of the upper bound of \cref{thm:uncol-main-simple}. We adapt a classic framework for approximate counting algorithms that originated in~\cite{VV}, and that was previously applied to edge counting in~\cite{DL}. We first observe that by using an algorithm from~\cite{DBLP:journals/algorithmica/Meeks19}, we can enumerate the edges in an $n$-vertex $k$-hypergraph $G$ with $2^{\OO(k)}\log^{\OO(1)}n\cdot e(G)$ queries to an independence oracle. Suppose we form an induced subgraph $G_i$ of $G$ by deleting vertices independently at random, keeping each vertex with probability $2^{-i}$; then in expectation, we have $e(G_i) = 2^{-ki}e(G)$. If $e(G_i)$ is small, and $e(G_i) \approx \E(e(G_i))$, then we can efficiently count the edges of $G_i$ using~\cite{DBLP:journals/algorithmica/Meeks19} and then multiply by $2^{ki}$ to obtain an estimate of $e(G)$. We can then simply iterate over all $i$ from $0$ to $\log n$ and return an estimate based on the first $i$ such that $e(G_i)$ is small enough for \cite{DBLP:journals/algorithmica/Meeks19} to return a value quickly.

Of course in general we do not have $e(G_i) \approx \E(e(G_i))$! One issue arises if, for some $r \in [k-1]$, every edge of $G$ contains a common size-$r$ set $R$ --- a ``root''. Then with probability $1 - 2^{-ri}$, at least one vertex in $R$ will be deleted and $G_i$ will contain no edges. We address this issue in the simplest way possible: by taking more samples. Roughly speaking, suppose we are given $i$, and that we already know (based on the failure of previous values of $i$ to return a result) that $e(G) > n^{k-r-1}$ for some $0 \le r \le k-1$. This implies that $G$ cannot have any ``roots'' of size greater than $r$. Rather than taking a single random subgraph $G_i$, we take $t_i \approx 2^{ri}$ independent copies of $G_i$ and sum their edge counts using \cite{DBLP:journals/algorithmica/Meeks19}; thus if $G$ does contain a size-$r$ root, we are likely to include it in the vertex set of at least one sample. Writing $\Sigma_i$ for the sum of their edge counts, we then return $\Sigma_i/(t_i2^{-ik})$ if the enumeration succeeds.

The exact expression we use for $t_i$ is more complicated than $2^{ri}$, due to the possibility of multiple roots --- see \cref{sec:uncolapprox} for a more detailed discussion --- but the idea is the same. The proof that $\Sigma_i$ is concentrated around its expectation is an (admittedly somewhat involved) application of Chebyshev's inequality, in which the rooted ``worst cases'' correspond to terms in the variance of~$\Sigma_i$. We consider it surprising and interesting that such a conceptually simple approach yields an optimal upper bound, and indeed gives us a strong hint as to how we should prove the lower bound of \cref{thm:uncol-main-simple}.

\subsubsection{Uncoloured lower bound}

We finally discuss the proof of the lower bound of \cref{thm:uncol-main-simple}. As in the colourful case, we apply the minimax principle, so we wish to define random $k$-hypergraphs $G_1$ and $G_2$ with $e(G_2) \gg e(G_1)$ which cannot be distinguished by a low-cost deterministic algorithm $A$. 

Our construction of $G_1$ and $G_2$ is natural from the above discussion, and the special case with $k=2$ and $\alpha_k=0$ is very similar to the construction used in \cite{CLW-graph-tight}. We take $G_1$ to be an Erd\H{o}s-R\'{e}nyi $k$-hypergraph with edge probability roughly $k!/n^r$. We choose a random collection of size-$r$ sets in $V(G_1)$ to be ``roots'', with a large constant number of roots present in expectation, and we define a $k$-hypergraph $H$ to include every possible edge containing at least one of these roots as a subset. We then define $G_2 := G_1 \cup H$.

Similarly to the colourful lower bound, in the non-adaptive setting, any fixed query $S_i$ with $|S_i|$ large is likely to contain an edge of $G_1$, and any fixed query with $|S_i|$ small is unlikely to contain a root and therefore unlikely to contain any edges of $H$; in either case, the query does not distinguish~$G_1$ from~$G_2$. Also as in the colourful case, generalising this argument from the non-adaptive setting to the adaptive setting requires a significant amount of care and effort.

\paragraph*{Organization.}
In \cref{sec:notation-defs}, we set out some standard conventions and our formal definitions of running time, oracle costs, and the most general cost functions to which our algorithmic results apply. We then recall some standard results from the literature (and folklore) in \cref{sec:standard-results}, and in \cref{sec:subset-sample} we apply a more recent result from~\cite{DBLP:phd/dnb/Bringmann15} to allow us to quickly sample binomially random subsets. We collect all our $\indora$-oracle results in \cref{sec:uncol}, proving various necessary properties of $g(k,\alpha_k)$ in \cref{sec:uncol-approx-algebra}, the upper bound in \cref{sec:uncol-upper}, and the lower bound in \cref{sec:uncol-lower}. Finally, we collect all our $\cindora$-oracle results in \cref{sec:col}, proving the upper bound in \cref{sec:col-upper} and the lower bound in \cref{sec:col-lower}.

\section{Preliminaries}\label{sec:prelims}

\subsection{Notation and definitions}\label{sec:notation-defs}

\subsubsection{Basic notation and conventions}\label{sec:notation}

We take $\N$ to be the set of natural numbers not including zero. For all $n \in \N$, we write $[n] \coloneqq \{1,\dots,n\}$. We write $\log$ for the base-2 logarithm, and $\ln$ for the base-$e$ logarithm. Given two sets $X$ and $Y$, we write $X \subset Y$ to mean that $X$ is a proper subset of $Y$ and $X \subseteq Y$ to mean that $X$ is a (possibly improper) subset of~$Y$. 

For all $x \in \R$, we write $\lfloor x \rceil \coloneqq \lfloor x + 1/2\rfloor$ for the value of $x$ rounded to the nearest integer, rounding up in the case of a tie. For all $\eps > 0$, we say that $x$ is an $\eps$-approximation to $y$ if $|x-y| \le \eps y$.

For all sets $S$ and all $r \in \N$, we write $S^{(r)}$ for the set of all $r$-element subsets of $S$. If $G$ is a $k$-uniform $k$-hypergraph with $S \subseteq V(G)$, we write $G[S]$ for the subgraph induced by $S$; thus $E(G[S]) \subseteq S^{(k)}$. Analogously, if $S = (S_1,\dots,S_k)$ is a $k$-tuple of disjoint sets, we write $G[S]=G[S_1,\dots,S_k]$ for the induced $k$-partite subgraph, which has vertex set $S_1 \cup \dots \cup S_k$ and edge set $\{e \in E(G) \colon |e \cap S_i|=1\mbox{ for all }i\in [k]\}$. We write $S^{(k)}$ for the set of all possible edges of $G[S]$, i.e.\ $\{\{s_1,\dots,s_k\}\colon s_i \in S_i\mbox{ for all }i\}$. If $G$ and $H$ are graphs on the same vertex set, we write $G \cup H$ for the graph on that vertex set with edge set $E(G) \cup E(H)$.

\subsubsection{Oracle algorithms}\label{sec:oracle-algorithms}

\paragraph*{Independence oracle.}
Let $G$ be a hypergraph with vertex set $[n]$ and $m$ edges.
Let $\indora(G)\in\zo^{2^{[n]}}$ be the string that satisfies the following for all $S \subseteq [n]$:
\begin{equation}
  \indora(G)_S =
  \begin{cases}
    1 & \text{if $G[S]$ contains no edge;}\\
    0 & \text{otherwise.}
  \end{cases}
\end{equation}
The bitstring $\indora(G)$ is called the \emph{independence oracle of $G$}.
Let us write $G_\indora(O)$ for the hypergraph belonging to the independence oracle~$O$.

\paragraph*{Colourful independence oracle.}
The \emph{colourful independence oracle~$\cindora(G)$ of~$G$} is the bitstring that is defined as follows for all disjoint sets $S_1,\dots,S_k\subseteq [n]$:
\begin{equation}
    \cindora(G)_{S_1,\dots,S_k} =
    \begin{cases}
        1 & \text{if $G$ contains no edge~$e\in E(G)$ with $\forall i\colon\abs{S_i\cap e}=1$;}\\
        0 & \text{otherwise.}
      \end{cases}
\end{equation}
We write $G_\cindora(O)$ for the hypergraph belonging to the colourful independence oracle~$O$.

\paragraph*{Oracle algorithm.}
In this paper, an \emph{oracle algorithm} is an algorithm with access to an oracle that represents the input $n$-vertex $k$-hypergraph implicitly, either as an independence oracle or as a colourful independence oracle.
We assume that the numbers~$n$ and~$k$ are always given explicitly as input and thus known to the algorithm at run-time; recall that the graph's vertex set is always $[n]$, so this is also known.
We consider \indora-oracle algorithms in \cref{sec:uncol} and \cindora-oracle algorithms in \cref{sec:col}, and we strictly distinguish between the worst-case running time of the algorithm (which does not include oracle costs) and the worst-case oracle cost of the queries (which does not include the running time of the algorithm).
Given an oracle $O$ (such as $\indora(G)$ for some $n$-vertex $k$-hypergraph $G$) we write $A(O,x)$ for the output produced by $A$ when supplied with $O$ and with additional input $x$. While~$n$ and~$k$ are always included as part of $x$, we typically do not write these out.
The algorithm can prepare a query by writing a set~$S$ or a tuple $(S_1,\dots,S_k)$ in the canonical encoding as an indicator string $q\in\set{0,\dots,k}^n$ to a special query array.
It can then issue the oracle query via a primitive $\mathtt{query}(q)$ and receive the answer~$0$ or~$1$.

\paragraph*{Induced subgraphs.}
Observe that given $\indora(G)$, it is trivial to simulate $\indora(G[X])$ for all $X \subseteq [n]$, as we have $\indora(G[X])_S = \indora(G)_S$ for all $S \subseteq X$. We will frequently use this fact implicitly to pass induced subgraphs of $G$ as arguments to subroutines without incurring overhead. To this end, we write $A(\indora(G[X]),x)$ as a shorthand for $A'(\indora(G),X,x)$, where $A'$ is the oracle algorithm that runs~$A$ and simulates oracle queries to $\indora(G[X])$ as just discussed using $X$ and $\indora(G)$.
Likewise, given $\cindora(G)$, it is trivial to simulate $\cindora(G[S_1,\dots,S_k])$ for all disjoint $S_1,\dots,S_k \subseteq [n]$. 

\paragraph*{Randomness.}
To model \emph{randomised algorithms}, we augment the RAM-model with a primitive~$\mathtt{rand}(R)$ that returns a uniformly random element from the set~$\set{1,\dots,R}$ for any given~$R\le \poly(n)$.
We view a randomised algorithm as a discrete probability distribution $\cA$ over a set $\supp(\cA)$ of deterministic algorithms based on the results of $\mathtt{rand}$.

\paragraph*{Running time.}
We measure the \emph{running time} in the standard RAM-model of computation with $\Theta(\log n)$ bits per word and the usual arithmetic and logical operations on these words.
Writing~$T(A,G,x)$ for the running time of $A(\indora(G),x)$, the \emph{worst-case running time} of a randomised algorithm~$\calA$ is defined as the function~$x\mapsto\max_G\sup_{A\in\supp(\calA)} T(A,G,x)$, where~$G$ ranges over all~$n$-vertex $k$-hypergraphs with~$n$ and~$k$ as specified in the input~$x$.
Similarly, for a randomised algorithm~$\calA$, the \emph{worst-case expected running time} is defined as the function~$x\mapsto\max_G(\E_{A \sim \cA} T(A,G,x))$, where again the maximum is taken over all $n$-vertex $k$-hypergraphs $G$.

\paragraph*{Oracle cost.}
We measure the oracle cost by a \textit{cost function} $\cost = \{\cost_k\colon k \ge 2\}$, where $\cost_k\colon(0,\infty)\to(0,\infty)$. We will typically require $\cost$ to be regularly-varying with parameter $k$ (see \cref{sec:regularly-varying}). In a $k$-uniform hypergraph $G$, the cost of an \indora-oracle query to a set $S \subseteq [n]$ is given by $\cost(S)\coloneqq\cost_k(|S|)$; note that this query cost depends only on the size of $S$.
By convention, we set $\cost_k(0)\coloneqq 0$ for all~$k$ as empty queries cannot provide any information about the graph.
Similarly, the cost of a \cindora-oracle query to a tuple~$(S_1,\dots,S_k)$ is given by $\cost(S_1,\dots,S_k)\coloneqq\cost_k(\abs{S_1\cup\dots\cup S_k})$.

Let $A$ be a deterministic oracle algorithm, and let~$X$ be the set of possible (non-oracle) inputs to $A$; recall that this always includes~$n$ and~$k$. Given an $n$-vertex $k$-uniform hypergraph $G$ and $x \in X$, let $S_1(G,x),\dots,S_{t_{G,x}}(G,x)$ be the sequence of queries that~$A$ makes when supplied with input $x$ and an oracle for~$G$. Then for all $n$-vertex $k$-uniform hypergraphs $G$ and all $x \in X$, the \emph{oracle cost} of $A$ on $G$ (with respect to $\cost$) is defined as $\cost(A,G,x) = \sum_{i=1}^{t_{G,x}} \cost_k(S_i(G,x))$. The oracle cost of a randomised oracle algorithm~$\cA$ on $G$ and $x$ is denoted by $\cost(\cA,G,x)$, and is the random variable $\cost(A,G,x)$ where~${A \sim \cA}$.

The \emph{worst-case oracle cost} of a randomised oracle algorithm $\cA$ (with respect to $\cost$) is defined as the function $x\mapsto\max_G(\sup_{A \in \supp(\cA)}\cost(A,G,x))$; here the maximum is taken over all $n$-vertex $k$-hypergraphs $G$, where $n$ and $k$ are as defined in the input~$x$. Similarly, the \emph{worst-case expected oracle cost} of $\cA$ is defined as the function $x\mapsto\max_G(\E_{A \sim \cA}(\cost(A,G,x)))$, where again the maximum is taken over all $n$-vertex $k$-hypergraphs $G$.

\subsubsection{Requirements on cost functions for upper bounds}\label{sec:regularly-varying}

In the introduction, we focused on polynomial cost functions of the form $\cost_k(S) = |S|^{\alpha_k}$ for some $\alpha_k \in [0,k]$, and these functions are easy to work with. However, even subpolynomial improvements to an algorithm's running time can be of interest and by considering more general cost functions we can lift these improvements from decision algorithms to approximate counting algorithms. To take a well-known example, in the negative-weight triangle problem, we are given an $n$-vertex edge-weighted graph $G$ and asked to determine whether it contains a triangle of negative total weight; this problem is equivalent to a set of other problems under subcubic reductions, including APSP~\cite{williams-nwt-apsp}. The naive $\Theta(n^3)$-time algorithm can be improved by a subpolynomial factor of $e^{\Theta(\sqrt{\log n})}$~\cite{williams-nwt-fast}, but as stated in the introduction our result would not lift this improvement from decision to approximate counting --- we would need to take $k=3$ and $\cost(n) = n^3/e^{\Theta(\sqrt{\log n})}$. (For clarity, in this specific case the problem is equivalent to its colourful version and a fine-grained reduction is already known~\cite{DL,DLM}.)

While we cannot hope to say anything meaningful in the algorithmic setting with a fully general cost function, our results do extend to all cost functions which might reasonably arise as running times. A natural first attempt to formalise what we mean by ``reasonable'' would be to consider cost functions of the form $\cost_k(n) = n^{\alpha_k+o(1)}$ as $n\to\infty$. Unfortunately, such cost functions can still have a pathological property which makes a fine-grained reduction almost impossible: the $o(1)$ term might vary wildly between different values of $n$. For example, if we take $\cost_k(n) = n^{\alpha_k + \sin(\pi n/2)/\sqrt{\log(n)}}$, then we have $\cost_k(n) = n^{\alpha_k+o(1)}$, but $\cost_k(2n)/\cost_k(n)$ could be $\omega(\polylog(n))$, $\Theta(1)$ or $o(1/\polylog(n))$ depending on whether $n$ is congruent to $3$, $2$ or $1$ modulo $4$ (respectively). It is even possible to construct a similar example which is monotonically increasing. We therefore require something slightly stronger, borrowing a standard notion from the probability literature for distributions which are ``almost power laws''.

\begin{defn}
    A function $f\colon\R\to\R$ is \emph{regularly-varying} if, for all fixed $a > 0$,
    \[
        \lim_{x\to\infty}|f(ax)/f(x)| \in (0,\infty).
    \]
\end{defn}

Observe that any cost function likely to arise as a running time is regularly-varying under this definition. We will use the following standard facts about regularly-varying functions.

\begin{lemma}\label{lem:regularly-varying-facts}
    Let $f$ be a regularly-varying function. Then there exists a unique $\alpha \in \R$,
    called the \emph{index} of $f$, with the following properties.
    \begin{enumerate}[(i)]
        \item For all fixed $A>0$, $\lim_{x\to\infty}(f(Ax)/f(x)) = A^\alpha$. Moreover, for all closed intervals $I \subseteq \R$, this limit is uniform over all $A \in I$.
        \item For all $\delta > 0$, there exists $x_0$ such that for all $x \ge x_0$ and all $A_x \ge 1$,
        \[
            A_x^{\alpha-\delta} \le f(A_xx)/f(x) \le A_x^{\alpha+\delta}.
        \]
        \item $f(x) = x^{\alpha + o(1)}$ as $x\to\infty$. 
        \item If $\alpha > 0$, then there exists $x_0$ such that $f$ is strictly increasing on $[x_0,\infty)$.
    \end{enumerate}
\end{lemma}

\begin{proof}
    Part (i) is proved in e.g.\ Feller~\cite[Chapter VIII.8 Lemmas 1--2]{feller-vol-2}. We now prove part (ii). Let $A_x \ge 1$.
    By part (i), there exists $x_0$ such that 
    \begin{equation}\label{eq:regularly-varying-facts-1}
        \mbox{for all $x \ge x_0$ and all $A \in [1,2]$, } A^{\alpha-\delta} \le f(Ax)/f(x) \le A^{\alpha + \delta}.
    \end{equation}
    Suppose $x \ge x_0$, and let $k_x$ be sufficiently large such that $A_x^{1/k_x} \in [1,2]$ holds. Then breaking $f(A_xn)/f(n)$ into a telescoping product and applying~\eqref{eq:regularly-varying-facts-1} to each term yields
    \begin{align}\nonumber
        \frac{f(A_xn)}{f(n)} &= \prod_{i=1}^{k_x}\frac{f(A_x^{i/k_x})}{f(A_x^{(i-1)/k_x})} \ge \prod_{i=1}^{k_x}A_x^{(\alpha-\delta)/k_x} = A_x^{\alpha-\delta},
    \end{align}
    as required. Similarly, $f(A_xn)/f(n) \le A_x^{\alpha+\delta}$ as required. Part (iii) follows immediately on taking $x=1$ in part (ii). To prove part (iv), take $x,y \in (0,\infty)$ with $x<y$ and observe from part (ii) that when $x$ and $y$ are sufficiently large, $\cost(y) \ge \cost(x)(y/x)^{\alpha/2} > \cost(x)$.
\end{proof}

We will only need our cost functions to vary regularly in $n$, not $k$; to facilitate this, we bring in the following standard result which allows us to ``confine the regular variation to the $n^{o(1)}$ term''.

\begin{defn}
    A regularly-varying function with index $0$ is called a \emph{slowly-varying function}.
\end{defn}

\begin{lemma}
    Let $f$ be a regularly-varying function with index $\alpha$. Then we have $f(x) = x^\alpha\sigma(x)$ for some slowly-varying function $\sigma$.
\end{lemma}
\begin{proof}
    This follows easily from Lemma~\ref{lem:regularly-varying-facts}; see Feller~\cite[Chapter VIII.8 (8.5)--(8.6)]{feller-vol-2}.
\end{proof} 

We are now ready to state the technical requirements on our cost functions. 

\begin{defn}\label{def:cost}
    For each $k \ge 2$, let $\cost_k\colon (0,\infty) \to (0,\infty)$. We say that $\cost = \{\cost_k\colon k\ge 2\}$ is a \emph{regularly-varying parameterised cost function} if $\cost_k(x) \le x^k$ for all $k$ and $x$, and there exists a slowly-varying function $\sigma\colon (0,\infty)\to(0,\infty)$ and a map $k\mapsto\alpha_k$ satisfying the following properties:
    \begin{enumerate}[(i)]
        \item for all $k$ and $x$, $\cost_k(x) = x^{\alpha_k}\sigma(x)$;
        \item for all $k$, $\alpha_k \in [0,k]$;
        \item for all $k$ and $x$, $\cost_k(x) \le x^k$;
        \item either $\liminf_{k\to\infty}\alpha_k > 0$ or there exists $x_0$ such that for all $k$, $\cost_k$ is non-decreasing on $(x_0,\infty)$;
        \item there is an algorithm to compute $\floor{\alpha_k}$ in time $\OO(k^{9k})$.
    \end{enumerate}
    We say that $k$ is the \emph{parameter} of $\cost$, $\alpha$ is the \emph{index} of $\cost$, and $\sigma$ is the \emph{slowly-varying component} of $\cost$.
\end{defn}

Point (i) is the main restriction, and the one we have been discussing until now. Points (ii) and (iii) have already been discussed in the introduction. In the colourful case we will need to compute $\floor{\alpha_k}$ in order to know which subroutines to use, so point (v) avoids an additive term in the running time. Finally, point (iv) is a technical convenience which (together with point (i)) guarantees monotonicity, as we show in \cref{lem:cost-monotone} below. Requiring point (iv) is unlikely to affect applications of our results --- typically such applications would either satisfy $\alpha_k = 0$ (as we are concerned only with query count and not with query cost) or $\alpha_k \ge 1$ (as the decision algorithm used to simulate the oracle needs to read the entire input).

\begin{lemma}\label{lem:cost-monotone}
    Let $\cost$ be a regularly-varying parameterised cost function with parameter~$k$, index~$\alpha$ and slowly-varying component~$\sigma$. Then there exists $x_0 \in (0,\infty)$ such that for all~$k$, $\cost_k$ is non-decreasing on $[x_0,\infty)$.
\end{lemma}
\begin{proof}
    If $\liminf_{k\to\infty}\alpha_k = 0$ then this is immediate from \cref{def:cost}(iv), so suppose instead $\liminf_{k\to\infty}\alpha_k = \delta > 0$. Then the function $x\mapsto x^\delta\sigma(x)$ is regularly-varying, and hence by Lemma~\ref{lem:regularly-varying-facts}(iv) there exists $x_0$ such that for all $y>x\ge x_0$ we have $y^\delta\sigma(y) > x^\delta\sigma(x)$. It follows that for all $k$,
    \[
        \frac{\cost_k(y)}{\cost_k(x)} = (y/x)^{\alpha_k}\frac{\sigma(y)}{\sigma(x)} \ge \frac{y^\delta\sigma(y)}{x^\delta\sigma(x)} > 1
    \]
    as required.
\end{proof}

\subsection{Collected standard results}\label{sec:standard-results}

\subsubsection{Probabilistic results}

We will need the following two standard Chernoff bounds.
\begin{lemma}[Corollary 2.3 of Janson, {\L}uczak and Rucinski \cite{DBLP:books/daglib/0015598}]\label{lem:chernoff-low}
	Let $X$ be a binomial random variable, and let $0 < \delta < 3/2$. Then
	\(
		\Pr\big(|X-\E(X)| \ge \delta\E(X)\big) \le 2e^{-\delta^2\E(X)/3}.
	\)
\end{lemma}
\begin{lemma}[Corollary 2.4 of Janson, {\L}uczak and Rucinski \cite{DBLP:books/daglib/0015598}]\label{lem:chernoff-high}
	Let $X$ be a binomial random variable, and let $z \ge 7\cdot\E(X)$. Then
	\(
		\Pr\big(X \ge z\big) \le e^{-z}.
	\)
\end{lemma}

We will also make use of the following form of the FKG inequality.

\begin{lemma}[Theorem 6.3.2 of Alon and Spencer \cite{alonspencer}]\label{lem:FKG}
    Let $X_1,\dots,X_n$ be independent Bernoulli variables, let $\calA$ and $\calB$ be events determined by these variables, let $1_A$ be the indicator function of $\calA$ and let $1_B$ be the indicator function of $\calB$. If $1_A$ and $1_B$ are either both increasing or both decreasing as functions of $X_1,\dots,X_n$, then $\pr(\calA\cap \calB) \ge \pr(\calA)\pr(\calB)$. If $1_A$ is increasing and $1_B$ is decreasing or vice versa, then $\pr(\calA\cap\calB) \le \pr(\calA)\pr(\calB)$.
\end{lemma}

The following probability bound is folklore.

\begin{lemma}\label{lem:cond-exp}
    Let $t$ be an integer, let $Z_1,\dots,Z_t$ be discrete random variables, let $\calE_1,\dots,\calE_t$ be events such that $\calE_i$ is a function of $Z_1,\dots,Z_i$, and let $\calE = \calE_1 \cup \dots \cup \calE_t$. For all $i$, let $\calZ_i$ be the set of possible values of $Z_1,\dots,Z_i$ under which $\calE_1,\dots,\calE_i$ do not occur, that is,
    \[
        \calZ_i = \Big\{(z_1,\dots,z_i) \colon \pr\Big(\overline{\calE_1},\dots,\overline{\calE_i} \mid (Z_1,\dots,Z_i) = (z_1,\dots,z_i)\Big) = 1\Big\}\,.
    \]
    Suppose that for all $i$, every $(z_1,\dots,z_i) \in \calZ_i$ is a prefix to some $(z_1,\dots,z_t) \in \calZ_t$. Then
    \begin{equation}\label{eq:cond-exp}
        \pr(\calE) \le \max_{(z_1,\dots,z_{t-1}) \in \calZ_{t-1}} \sum_{i=1}^t \pr\Big(\calE_i \mid (Z_1,\dots,Z_{i-1})=(z_1,\dots,z_{i-1})\Big)\,.
    \end{equation}
\end{lemma}

We will use \cref{lem:cond-exp} when proving lower bounds. In doing so, we will take $Z_1,\dots,Z_t$ to be pairs of query results from running a deterministic algorithm on two random inputs, $\calE_i$ to be the event that we distinguish the two inputs on the $i$\th query, and $\calE$ to be the event that we distinguish the two inputs at all. Note the order of the maximum and the sum in~\eqref{eq:cond-exp}; in a simple union bound over $\calE_1,\dots,\calE_t$, they would occur in the opposite order for a weaker result.

\begin{proof}
    Consider the random variable
    \begin{align*}
        Y \coloneqq \sum_{i=1}^t \pr_{Z_i}\Big(\overline{\calE_1}, \dots, \overline{\calE_{i-1}}, \calE_i \mid Z_1,\dots,Z_{i-1}\Big)\,.
    \end{align*}
    Note that the $i$\th term in this sum is a random variable that is a function of~$Z_1,\dots,Z_{i-1}$.
    We first bound $Y$ above. Deterministically, on exposing the values of $Z_1,\dots,Z_t$
    we obtain
    \[
        Y \le \max_{z_1,\dots,z_{t-1}} \sum_{i=1}^t \pr_{Z_i}\Big(\overline{\calE_1},\dots,\overline{\calE_{i-1}},\calE_i \mid (Z_1,\dots,Z_{i-1})=(z_1,\dots,z_{i-1})\Big)\,.
    \]
    If $(z_1,\dots,z_{t-1}) \notin \calZ_{t-1}$, then let $k = \min\{\ell \le t-1 \colon (z_1,\dots,z_\ell) \notin \calZ_\ell\}$.
    Thus $(z_1,\dots,z_k) \notin \calZ_k$, but (if $k > 1$) we have $(z_1,\dots,z_{k-1}) \in \calZ_{k-1}$; it follows that $\overline{\calE_1},\dots,\overline{\calE_{k-1}}$ and $\calE_k$ all occur. The last $t-k$ probabilities in the corresponding sum are therefore zero (even if $k=1$), and we have
    \begin{align*}
        &\sum_{i=1}^t \pr_{Z_i}\Big(\overline{\calE_1},\dots,\overline{\calE_{i-1}},\calE_i \mid (Z_1,\dots,Z_{i-1})=(z_1,\dots,z_{i-1})\Big)\\
        &\qquad\qquad\qquad = \sum_{i=1}^k \pr_{Z_i}\Big(\overline{\calE_1},\dots,\overline{\calE_{i-1}},\calE_i \mid (Z_1,\dots,Z_{i-1})=(z_1,\dots,z_{i-1})\Big)\,.
    \end{align*}
    It follows that
    \[
        Y \le \max_{k \in [t]} \max_{(z_1,\dots,z_{k-1}) \in \calZ_{k-1}} \sum_{i=1}^k \pr_{Z_i}\Big(\overline{\calE_1},\dots,\overline{\calE_{i-1}},\calE_i \mid (Z_1,\dots,Z_{i-1})=(z_1,\dots,z_{i-1})\Big)\,,
    \]
    as we are maximising over the same terms --- each $(z_1,\dots,z_{t-1}) \notin \calZ_{t-1}$ corresponds to a shorter $(z_1,\dots,z_{k-1}) \in \calZ_{k-1}$.
    Moreover, since every vector in $\calZ_{k-1}$ is a prefix to some vector in $\calZ_{t-1}$ by hypothesis, the first maximum must be attained at $k=t$, and so
    \begin{equation}\label{eq:cond-exp-0}
        Y \le \max_{(z_1,\dots,z_{t-1}) \in \calZ_{t-1}} \sum_{i=1}^t \pr_{Z_i}\Big(\overline{\calE_1},\dots,\overline{\calE_{i-1}}, \calE_i \mid (Z_1,\dots,Z_{i-1})=(z_1,\dots,z_{i-1})\Big)\,.
    \end{equation}
    
    We now calculate $\E(Y)$. By linearity of expectation and since each $\calE_i$ is a function of $Z_1,\dots,Z_i$, we have
    \begin{align*}
        \E(Y) &= \sum_{i=1}^t \E_{Z_1,\dots,Z_{i-1}}\bigg(\pr_{Z_i}\Big(\overline{\calE_1},\dots,\overline{\calE_{i-1}}, \calE_i \mid Z_1,\dots,Z_{i-1}\Big)\bigg)\\
        &= \sum_{i=1}^t \pr_{Z_1,\dots,Z_i}\Big(\overline{\calE_1},\dots,\overline{\calE_{i-1}},\calE_i\Big) = \pr(\calE)\,.
    \end{align*}
    Since $\E(Y)$ is bounded above by the maximum possible value of $Y$, the result follows from~\eqref{eq:cond-exp-0}.
\end{proof}

\subsubsection{Algorithmic results}

In order to prove lower bounds on the cost of randomised oracle algorithms we will apply the minimax principle, a standard tool from decision tree complexity. There are many possible statements of this principle (see \cite[Theorem~3]{DBLP:conf/focs/Yao77} for the original one by Yao); for convenience, we provide a statement matched to our setting along with a self-contained proof.

\begin{theorem}\label{thm:algorithm-minmax}
  Let $n,k\in\N$ with $k \ge 2$. Let $\cA$ be a randomised oracle algorithm whose possible non-oracle inputs are drawn from a set $X$ and whose possible outputs lie in a set~$\Sigma$. Let $x \in X$, and let~$n$ and~$k$ be as determined by~$x$. Let $\calG$ be the set of all $k$-hypergraphs on~$[n]$. If $\calD$ is any probability distribution on $\calG$, and $F$ is any set of functions $\calG \times X\to\Sigma$ with $\pr_{A\sim\cA}[A\in F]\ge p$, then
  \begin{equation}\label{eq:minmax}
    \max_{G\in \calG}
    \E_{A\sim\cA}[\cost(A,G,x)]
    \ge
    p
    \inf_{A\in F}\E_{G\sim\calD}[\cost(A,G,x)]\,.
  \end{equation}
\end{theorem}

The expectation
\(
  \E_{A\sim\cA}[\cost(A,G,x)]
\)
on the left side of \cref{eq:minmax} is the expected cost of running a randomised algorithm $\cA$ on the graph $G$; thus the left side of \cref{eq:minmax} is the worst-case expected oracle cost of $\cA$ over all possible inputs $G$ (with non-oracle input~$x$). The set $F$ will be application-dependent, but will typically contain all deterministic algorithms that with suitably high probability exhibit the ``correct'' behaviour on the specific input distribution $\calD$ --- for example, returning an accurate approximation to the number of edges in the input graph. In our applications, the requirement $\pr_{A\sim \cA}[A\in F]\ge p$ will be immediate from the stated properties of $\calA$. The point of \cref{thm:algorithm-minmax}, then, is that we can lower-bound the cost of an arbitrary \textit{randomised} algorithm $\cA$ on a worst-case deterministic graph $G \in \calG$ by instead lower-bounding the cost of any \textit{deterministic} algorithm $A$ on a randomised input distribution $\calD$ of our choice.

\begin{proof}
  We derive the claim as follows:
  \begin{align*}
    \max_{G\in \calG}\E_{A\sim\cA}[\cost(A,G,x)]
    &\ge
    \E_{G\sim\calD}\E_{A\sim\cA}[\cost(A,G,x)]
    =
    \E_{A\sim\cA}\E_{G\sim\calD}[\cost(A,G,x)]
    \\
    &\ge
    \E_{A\sim\cA}\Big[\;\E_{x\sim\cX}[\cost(A,G,x)] \;\Big|\; A\in F\;\Big]
    \cdot \pr_{A\sim \cA}[A\in F]\\
    &\ge
    \tfrac12
    \inf_{A\in F}
    \Big[\,\E_{x\sim\cX}[\cost(A,G,x)]\,\Big]
    \,.
  \end{align*}
  The first and last inequalities are trivial, and the second inequality follows by conditioning and dropping the terms for $A\not\in F$, which is correct because the cost is non-negative.
\end{proof}

Finally, we set out two lemmas for passing between different performance guarantees on randomised algorithms. The ideas behind these lemmas (independent repetition and median boosting) are very standard, but we set them out in detail because we are working in the word-RAM model (see \cref{sec:oracle-algorithms}) and so we may not be able to efficiently compute our own stated upper bounds on oracle cost; thus there is more potential than usual for subtle errors.

We use the following standard lemma to pass from probabilistic guarantees on running time and oracle cost to deterministic guarantees.
The lemma assumes an implicit problem definition in the form of a relation between inputs and outputs, modelling which input-output pairs are considered to be correct.
In our application, an input-output pair $((G,x),A(\indora(G),x))$ is considered correct if $A(\indora(G),x)$ is an $\eps$-approximation to~$e(G)$.
\begin{lemma}\label{lem:expected-to-worst-case}
    Let $\cost=\{\cost_k\colon k\ge 2\}$ be a cost function as in \cref{sec:oracle-algorithms}.
    Let~$\cA$ be a randomised oracle algorithm whose possible non-oracle inputs are drawn from a set~$X$.
    Let $\Torig,\Corig,\Ttotals,\Tcost\colon X\to (0,\infty)$, and let $\delta > 0$ be a constant. Suppose that:
    \begin{itemize}
        \item given any possible oracle input~$O$ and non-oracle input~$x \in X$, with probability at least $1-\delta$, $\cA(O,x)$ has the correct input-output behaviour, runs in time $\OO(\Torig(x))$, and has oracle cost $\OO(\Corig(x))$;
        \item for all $x \in X$, $\Torig(x)$ and $\Corig(x)$ can be computed in time $\OO(\Ttotals(x))$; and
        \item for all $x \in X$, writing $k_x$ and $n_x$ for the values of $k$ and $n$ specified by $x$, for all $n \le n_x$, $\cost_{k_x}(n)$ can be computed in time $\OO(\Tcost(x))$.
    \end{itemize}
    Then there is a randomised oracle algorithm~$\cA'$ such that:
    \begin{enumerate}[(i)]
        \item $\cA'$ has the same inputs as $\cA$ together with a rational $\gamma \in (0,1)$;
        \item $\cA'$ has worst-case running time~$\OO(\Ttotals(x) + \log(1/\gamma) \cdot \Torig(x) \cdot \Tcost(x))$;
        \item $\cA'$ has worst-case oracle cost~$\OO(\log(1/\gamma) \cdot \Corig(x))$;
        \item $\cA'$ exhibits the same correct input-output behaviour as $\cA$ with probability $1-\gamma-\delta$.
    \end{enumerate}
\end{lemma}
\begin{proof}
    The algorithm~$\cA'(O,x,\gamma)$ works as follows:
    \begin{itemize}
        \item Compute $\Torig(x)$ and $\Corig(x)$.
        \item Repeat at most~$\lceil \log_{\delta}(\gamma) \rceil$ times:
        \begin{itemize}
            \item Run $\cA(O,x)$ and keep track of the running time incurred. Compute the cost of each invocation of the oracle before it happens, and keep track of the total cost incurred.
            \item If the total running time exceeds $\Torig(x)$ (not counting overhead for tracking the running time) or if the total oracle cost would exceed~$\Corig(x)$ on the next oracle query, abort this run.
            \item If~$\cA(O,x)$ halts within these resource constraints, return its output.
        \end{itemize}
        \item Return the error \RTE{} (\enquote{running time exceeded}).
    \end{itemize}
    
    Part (i) is immediate. Observe that $\cA'$ runs~$\cA(O,x)$ at most~$\OO(\log(1/\gamma))$ times, each time incurring $\OO(\Torig(x))$ running time and $\OO(\Corig(x))$ oracle cost.
    Moreover, since each oracle query takes $\Theta(1)$ time, each run of $\cA$ makes $\OO(\Torig(x))$ oracle queries; thus we spend $\OO(\Torig(x)\cdot \Tcost(x))$ time per run on keeping track of the total running time and oracle cost. Finally, we spend $\OO(\Ttotals(x))$ time computing $\Torig(x)$ and $\Corig(x)$. Thus (ii) and (iii) follow.
    
    Let $\calE$ be the event that at least one run of $\cA$ succeeds within the resource constraints. Since the runs are independent, we have $\pr(\calE) \ge 1 - \delta^{\lceil \log_{\delta}(\gamma)\rceil} \ge 1-\gamma$.
    Moreover, conditioned on a given run being within the resource constraints, this run exhibits the correct input-output behaviour for $\cA$ with probability at least $1-\delta$; thus conditioned on~$\calE$, the algorithm~$\cA'$ also exhibits the correct input-output behaviour with probability at least $1-\delta$. By a union bound, (iv) follows.
\end{proof}

We also use the following standard lemma to boost the success probability of approximate counting by taking the median output among independent repetitions.

\begin{lemma}\label{lem:median-boosting}
    Let $\cA$ be a randomised oracle algorithm. Suppose that given non-oracle input $x$, the algorithm~$\cA$ has worst-case running time $\OO(T(x))$ and worst-case oracle-cost $\OO(C(x))$. Suppose further that for some constant $\delta < 1/2$, for all possible oracle inputs~$O$ and non-oracle inputs $x$, with probability at least~$1-\delta$, the algorithm call~$\cA(O,x)$ outputs a real number which lies in some possibly-unbounded interval $I_{O,x}$. Then there is a randomised oracle algorithm $\cA'$ such that:
    \begin{enumerate}[(i)]
        \item $\cA'$ has the same inputs as $\cA$ together with a rational $\gamma \in (0,1)$;
        \item $\cA'$ has worst-case running time~$\OO(\log(1/\gamma) \cdot T(x))$;
        \item $\cA'$ has worst-case oracle cost~$\OO(\log(1/\gamma) \cdot C(x))$;
        \item For all possible inputs~$(O,x)$, with probability at least $1-\gamma$, the algorithm call~$\cA'(O,x)$ outputs a real number which lies in $I_{O,x}$.
    \end{enumerate}
\end{lemma}
\begin{proof}
    Write $\xi = 1 - 1/(2(1-\delta))$, and note that $\xi \in (0,1)$ since $\delta \in [0,1/2)$. Given oracle input $O$ and non-oracle inputs $x \in X$ and $\gamma \in (0,1)$, the algorithm~$\cA'$ runs~$\cA(O,x)$ a total of $N\coloneqq\lceil 6\ln(2/\gamma)/\xi^2\rceil$ times using independent randomness and returns the median of the outputs, treating any non-numerical outputs of $\cA$ as $-1$. Parts (i)--(iii) are immediate. Moreover, let $M$ be the number of outputs of $\cA$ which lie in $I_{O,x}$. If $M > N/2$, then the output of $\cA'$ also lies in $I_{O,x}$. Since $\E(M) \ge (1-\delta)N$, it follows that
    \begin{align*}
        \pr(\cA'(O,x,\gamma) \in I_{O,x}) &\ge \pr(M \ge N/2) \ge \pr\big(M \ge \E(M)/(2(1-\delta))\big)\\
        &\ge \pr\big(|M-\E(M)| \ge \xi \E(M)\big)
    \end{align*}
    Since $M$ is a binomial variable, by the Chernoff bound of \cref{lem:chernoff-low} it follows that
    \[
        \pr(\cA'(O,x,\gamma) \in I_{O,x}) \ge 1-2e^{-\xi^2\E(M)/3} \le 1-2e^{-\xi^2N/6} \ge 1-\gamma\,,
    \]
    as required.
\end{proof}

\subsubsection{Algebraic results}

In proving our lower bounds, we will lower-bound the cost of a deterministic approximate counting algorithm in terms of the sizes of the queries it uses; we will then need to maximise this bound over all possible sets of query sizes to put our lower bound into closed form. A key part of this argument will be \cref{cor:karamata}, stated below. In order to prove this \cref{cor:karamata}, we first state a standard generalisation of Jensen's inequality due to Karamata.

\begin{defn}
    Let $t$ be a positive integer and let $\vec{x},\vec{y} \in \R^t$. We say that $\vec{y}$ \emph{majorises}~$\vec{x}$ and write $\vec{y} \succ \vec{x}$ or $\vec{x} \prec \vec{y}$
    if
    \begin{enumerate}[(i)]
        \item $x_1 \ge \dots \ge x_t$ and $y_1 \ge \dots \ge y_t$;
        \item $x_1 + \dots + x_t = y_1 + \dots + y_t$; and
        \item for all $m \in [t-1]$, $x_1 + \dots + x_m \ge y_1 + \dots + y_m$.
    \end{enumerate}
\end{defn}

\begin{lemma}[Karamata's inequality, eg.~see~{\cite[Theorem~12.2]{convex}}]\label{lem:karamata}
    Let $I$ be an interval in~$\R$ and let $\phi\colon I\to\R$ be a convex function. Let $t$ be a positive integer and let $\vec{x},\vec{y} \in \R^t$. If $\vec{x} \prec \vec{y}$, then
    \[
        \sum_{i=1}^t \phi(x_i) \le \sum_{i=1}^t \phi(y_i)\,.
    \]
\end{lemma}

\begin{cor}\label{cor:karamata}
    Let $r \ge \alpha \ge 0$. Let $c>0$, let $t$ be a positive integer and let $\vec{s} \in [0,c]^t$. Let $W > 0$, and suppose that $\sum_{i=1}^t s_i^\alpha \le W$. Then we have
    \[
        \sum_{i=1}^t s_i^r \le Wc^{r-\alpha}\,.
    \]
\end{cor}
\begin{proof}
    First observe that if $\alpha = 0$ then $t \le W$ by hypothesis; thus $\sum_i s_i^r \le tc^r \le Wc^{r-\alpha}$ as required. For the rest of the proof, we assume that $\alpha > 0$. Further, without loss of generality, we may assume that $s_1 \ge \dots \ge s_t$ and that $\sum_i s_i^\alpha = W$ (by increasing~$\vec{s}$ and~$t$ if necessary). 
    
    We apply \cref{lem:karamata}, taking $\phi(x) = x^{r/\alpha}$, $x_i = s_i^\alpha$ for all $i \in [t]$, and
    \[
        y_i = \begin{cases}
          c^\alpha & \mbox{ if }i \le \lfloor W/c^\alpha \rfloor,\\
          W - \lfloor W/c^\alpha \rfloor c^\alpha & \mbox{ if }i = \lfloor W/c^\alpha \rfloor + 1,\\
          0\mbox{ otherwise}\,.
        \end{cases}
    \]
    Observe that $\phi$ is convex since $r \ge \alpha$ and that $\sum_i x_i = \sum_i y_i = W$ by hypothesis. Further, $\vec{y} \succ \vec{x}$ --- indeed, (i) and (ii) of the definition are immediate, and for all $m \in [t-1]$ we have $x_1 + \dots + x_m \le \min\{W, mc^\alpha\} = y_1 + \dots + y_m$. Thus \cref{lem:karamata} applies and we obtain
    \begin{equation}\label{eq:karamata}
        \sum_{i=1}^t s_i^r = \sum_{i=1}^t \phi(x_i) \le \sum_{i=1}^t \phi(y_i) = \lfloor W/c^\alpha \rfloor c^r + (W - \lfloor W/c^\alpha \rfloor c^\alpha)^{r/\alpha}\,.
    \end{equation}
    Observe that
    \[
        (W - \lfloor W/c^\alpha \rfloor c^\alpha)^{r/\alpha} = (W/c^{\alpha} - \lfloor W/c^\alpha \rfloor)^{r/\alpha}\cdot c^r \le (W/c^{\alpha} - \lfloor W/c^\alpha \rfloor) c^r\,,
    \]
    where the inequality follows since $W/c^{\alpha} - \lfloor W/c^\alpha \rfloor < 1$ by definition and $\alpha \ge r$. The result therefore follows from~\eqref{eq:karamata}.
\end{proof}

Finally, the following bound on binomial coefficients is folklore.

\begin{lemma}\label{lem:binom-bound}
    Let $S$ be an arbitrary set, and let $|S| \ge a \ge b \ge 0$. Then we have
    \[
        \binom{|S|-a}{b} \ge \frac{|S|^{b}}{(2a)^a}\,.
    \]
\end{lemma}
\begin{proof}
    If $|S| \ge 2a$, then we have
    \[
        \binom{|S|-a}{b} \ge \frac{(|S|-a)^{b}}{b^b} \ge \frac{(|S|/2)^b}{b^b} \ge \frac{|S|^b}{(2a)^a}\,.
    \]
    If instead $|S| < 2a$, then we have
    \[
        \binom{|S|-a}{b} \ge 1 \ge \frac{|S|^b}{(2a)^a}\,.\qedhere
    \]
\end{proof}

\subsection{Efficiently sampling small random subsets}\label{sec:subset-sample}

Let $U=\set{1,\dots,n}$ be a universe of size~$n$ and let~$i$ be an integer with~$0\le i\le \OO(\log n)$.
How can we sample a random set~$X\subseteq U$ such that each element~$u\in U$ is picked with independent probability~$2^{-i}$?
Naively, we would do this by flipping a $2^{-i}$-biased coin for each element~$u\in U$, which overall takes time~$\OO(n)$. This will be too slow for our purposes, so we give an improved algorithm in \SampleSubset{} with running time $\OO(i+n/2^i)$ and prove its correctness in \cref{lem:sample-random-subset} (the goal of this subsection). 

As a key subroutine of \SampleSubset{}, we will need to sample from a binomial distribution $\Bin(n,1/2^i)$. 
In the real-RAM model, there is a simple and highly-efficient algorithm for this (namely drawing a uniformly random element from $[0,1]$ and using the inverse method), and more sophisticated methods are available (e.g., see Fishman~\cite{10.2307/2286346}). However, as we work in the word-RAM model (see \cref{sec:oracle-algorithms}) and e.g.\ $\binom{n}{n/2}$ is a $\Theta(n)$-bit number, the problem requires more thought.
Bringmann~{\cite[Theorem~1.19]{DBLP:phd/dnb/Bringmann15}} provides a word-RAM algorithm to sample from $\Bin(n,\frac{1}{2})$ in expected constant time.
\begin{lemma}[Bringmann~{\cite[Theorem~1.19]{DBLP:phd/dnb/Bringmann15}}]
\label{lem:sample-binomial-half}
In the word-RAM model with $\Theta(\log n)$ bits per word, a binomial random variable $\Bin(n,\frac12)$ can be sampled in expected time~$\OO(1)$.
\end{lemma}
Farach{-}Colton and Tsai~\cite[Theorem~2]{DBLP:journals/algorithmica/Farach-ColtonT15} show how to turn a sampler for $\Bin(n,\frac12)$ into a sampler for $\Bin(n,p)$. We will only need the special case where $1/p$ is a power of two. In this case the methods of~\cite{DBLP:journals/algorithmica/Farach-ColtonT15} yield a much faster sampler, but this is not formally stated. Since the $p=1/2^i$ case is far simpler than the general case, for the convenience of the reader we prove our own self-contained corollary to \cref{lem:sample-binomial-half} rather than analysing~\cite{DBLP:journals/algorithmica/Farach-ColtonT15} in depth.
\begin{cor}\label{cor:sample-binomial}
On the word-RAM with $w=\Omega(\log n)$, given $n$ and $i$, a binomial random variable $\Bin(n,2^{-i})$ can be sampled in expected time $\OO(i)$.
\end{cor}
\begin{proof}
We can sample from $\Bin(n,2^{-i})$ by tossing $n$ fair coins in $i$ rounds:
In round $j=1$, all coins are tossed, and in round $j>1$, all coins that turned up heads in round~$j-1$ are tossed again.
After round $j=i$, we count how many coins show heads and this is our sample from $\Bin(n,2^{-i})$, because the probability that any specific coin survives all rounds is exactly $2^{-i}$.

We can efficiently simulate this procedure as follows:
In round $j=1$, we sample $n_1$ from $\Bin(n,\frac{1}{2})$.
In round $j>1$, we sample $n_j$ from $\Bin(n_{j-1},\frac{1}{2})$.
We return $n_i$ as our sample from $\Bin(n,2^{-i})$.
Thus, all that is needed is to draw $i$ samples from $\Bin(n',\frac{1}{2})$ using \cref{lem:sample-binomial-half} for various values of $n'\le n$.
Therefore, the overall expected running time is $\OO(i)$.
\end{proof}

We now apply \cref{cor:sample-binomial} to implement \SampleSubset{}.

\begin{algorithm}
    \caption{\label{algo:sample-random-subset}%
    \SampleSubset}
    \SetKwInOut{Input}{Input}
	\SetKwInOut{Output}{Output}
	\DontPrintSemicolon
	\Input{Integers~$n$ and~$i$.}
	\Output{A random set~$X\subseteq\set{1,\dots,n}$ that includes each element with independent probability~$2^{-i}$.}
	\Begin{
        \If{$i\le2$}{
            Sample~$X$ naively by independently flipping a~$2^{-i}$-biased coin for each element of~$\set{1,\dots,n}$\;
        }
        \Else{
            \label{line-sample-S}
        Sample $S$ from the binomial distribution~$\text{Bin}(n,2^{-i})$ using \cref{cor:sample-binomial}.\;
        \label{line-sample-X}
        Sample $X$ as a uniformly random size-$S$ subset of $\set{1,\dots,n}$.\;
        }
    }
\end{algorithm}
\begin{lemma}\label{lem:sample-random-subset}
    \SampleSubset{}
    is correct and runs in expected time~$\OO(i+2^{-i}n)$.
\end{lemma}
\begin{proof}
    We first show that the random set~$X$ is indeed distributed as claimed for~$i\ge 3$.
    To this end, let~$Y$ be a random subset~$\set{1,\dots,n}$ that is sampled by including each element~$u\in\set{1,\dots,n}$ with independent probability~$2^{-i}$. We claim that~$X$ and~$Y$ follow the same distribution.
    Indeed, $\abs{Y}$ follows $\mbox{Bin}(n,2^{-i})$. By symmetry, conditioned on~$|Y|$, the set~$Y$ is equally likely to be any size-$|Y|$ subset of~$\set{1,\dots,n}$. Hence, $X$ and $Y$ follow the same distribution, as required.

    Now we describe how exactly the algorithm is implemented in order to achieve the expected running time.
    For $i\le 2$, the algorithm runs in time~$\OO(n)$ as required.
    For~${i\ge 3}$, line~\ref{line-sample-S} takes time $\OO(i)$ by line~\ref{cor:sample-binomial}.
    In order to achieve expected time $\OO(2^{-i}n)$, line~\ref{line-sample-X} is implemented using a hash set: We sample uniformly random elements from~$\set{1,\dots,n}$ (with replacement) and one by one add them to the hash set until the hash set contains exactly $S$ distinct elements. Each hash set operation takes~$\OO(1)$ expected time, so it remains to argue how many elements we have to sample.

    Conditioned on $S\le \frac78 n$, at any point it takes at most $8$ samples in expectation until a new element is added to the hash set. Overall, the expected total running time to sample~$S$ distinct elements in line~\ref{line-sample-X} is thus at most $8S\le\OO(2^{-i}n)$.

    Conditioned on $S>\frac78 n$, a coupon collector argument shows that it takes an expected number $\OO(n\log n)$ of samples until we have seen~$S$ distinct elements.
    This would be too large, so it remains to prove that the probability that $S>\frac78 n$ holds is vanishingly small.
    Indeed, $S$ is binomially distributed with mean $\E(S)=2^{-i}n$. Let $z=\frac78 n$ and note that $z \ge 7\cdot2^{-3}\ge 7\E(S)$ holds by~$i\ge 3$.
    Thus, we have $\pr(S\ge z) \le e^{-z}=e^{-7n/8}$ by \cref{lem:chernoff-high}. By $\OO(e^{-7n/8}n \log n)\le\OO(1)$, the running time in the event~$S_{i,j}>\frac78 n$ vanishes in expectation.

    To conclude, the overall expected time of lines~\ref{line-sample-S} and~\ref{line-sample-X} is $\OO(i+2^{-i}n)$ as claimed.
\end{proof}

Later, it will be technically convenient to make the constant in the expected running time of \SampleSubset{} explicit.

\begin{defn}\label{def:Csample}
    Let $\Csample$ be such that for all $n$ and $i$, $\SampleSubset(n,i)$ runs in expected time at most $\Csample(i+2^{-i}n)$.
\end{defn}

\section{Independence oracle with cost}\label{sec:uncol}

In this section, we study the edge estimation problem for \indora-oracle algorithms, which are given access to the independence oracle $\indora(G)$ in $k$-uniform hypergraphs~$G$ with a cost function $\cost\colon 2^{V(G)}\to\R_{\ge 0}$.
The cost can be thought of as the running time of a subroutine that detects the presence of an edge, but it can also be used to model other types of cost.
We restrict our attention to cost functions that only depend on the size of the query, so we write $\cost(S)=\cost(\abs{S})$, and in this case we assume that the cost is regularly-varying with index $\alpha \in [0,k]$.
In particular, in the case $\alpha=0$, the cost corresponds to the number of queries and therefore yields bounds for independence oracles without costs.

In \cref{sec:uncol-approx-algebra}, we prove some simple algebraic lemmas that we will use to express our bounds in terms of the parameter $g(k,\beta)$ already introduced in \cref{thm:uncol-main-simple}. We then present an \indora-oracle algorithm in \cref{sec:uncol-upper} that approximately counts edges of a given $k$-uniform hypergraph, proving the upper bound part of \cref{thm:uncolapprox-algorithm} and, as a corollary, the upper bound part of \cref{thm:uncol-main-simple}.
Finally, in \cref{sec:uncol-lower} we prove that the total cost incurred by our \indora-oracle algorithm is optimal up to a polylogarithmic factor, establishing \cref{thm:uncol-lower-main-full} and, as a corollary, the lower bound part of \cref{thm:uncol-main-simple}.

\subsection{Algebraic preliminaries}\label{sec:uncol-approx-algebra}

Recall the following notation from the statement of \cref{thm:uncol-main-simple}.

\begin{defn}
    For all real numbers $k$ and $\beta$, we define 
    \[
        g(k,\beta) = \frac1k\cdot\left\lfloor \frac{k-\beta}{2} \right\rceil
        \cdot
        \paren*{
            k-\beta-
        \left\lfloor \frac{k-\beta}{2} \right\rceil}.
    \]
\end{defn}

In \cref{thm:uncol-main-simple}, $g(k,\beta)$ will arise from a maximum over $\log n$ iterations of the main loop in our oracle algorithm; in \cref{thm:uncol-lower-main-full}, it will arise from a maximum over $k$ possible choices of input distribution for our lower bound. The arguments required are very similar in both cases, so we prove the necessary lemmas in this combined section. Our goals will be \cref{lem:uncol-bound-ti} and \cref{cor:uncol-final-cost-bound}; these are bounds on algebraic expressions which arise naturally in bounding the total running time and oracle cost of \UncolApprox{} above in \cref{sec:uncol-upper} and in bounding the required oracle cost below in \cref{sec:uncol-lower}. The proofs are standard applications of algebra and elementary calculus, and may be skipped on a first reading. We first prove some simple bounds on $g$.

\begin{lemma}\label{lem:uncol-bound-g}
    For all $k \ge 0$ and $\beta \in [0,k]$, we have $g(k,\beta) \le (k-\beta)^2/(4k)$. In particular, $g(k,\beta) \le k/4$.
\end{lemma}
\begin{proof}
    Fix $k$ and $\beta \in [0,k]$, and let $h(x) = x(k-\beta-x)/k$. Then since $\beta \in [0,k]$, $h$ is a parabola maximised at $x=(k-\beta)/2$, so
    it follows that 
    \[
        h(k-\beta) \ge h(\lfloor (k-\beta)/2 \rceil) = g(k,\beta).
    \]
    We have $h(k-\beta) = (k-\beta)^2/(4k)$, so the first part of the result follows. The second part is then immediate on observing that $(k-\beta)^2$ is maximised over $\beta \in [0,k]$ at $\beta=k/2$.
\end{proof}

\begin{lemma}\label{lem:uncol-bound-Fi-grad}
    For all integers $k \ge 2$ and all $\beta \in [0,k]$, we have $g(k,\beta) \ge g(k,0) - \beta/2$.
\end{lemma}
\begin{proof}
    We first claim that for fixed $k$, $g(k,\beta)$ is continuous in $\beta$. This is immediate for all values of $\beta$ except those at which the value of $\lfloor(k-\beta)/2\rceil$ ``jumps'', i.e.\ those values of $\beta$ such that $k-\beta$ is an odd integer. For such values of $\beta$, we must show that the limit of $g(k,x)$ is the same as $x$ converges to $\beta$ from above and below. Suppose $k-\beta$ is indeed an odd integer; then we have
    \begin{align*}
        \lim_{x\uparrow\beta} g(k,x) &= \frac{1}{k}\cdot \frac{k-\beta-1}{2} \cdot \Big(k - \beta - \frac{k-\beta-1}{2}\Big) = \frac{(k-\beta-1)(k-\beta+1)}{4k},\\
        \lim_{x\downarrow\beta} g(k,x) &= \frac{1}{k}\cdot\frac{k-\beta+1}{2} \cdot \Big(k - \beta - \frac{k-\beta+1}{2}\Big) = \frac{(k-\beta+1)(k-\beta-1)}{4k}.
    \end{align*}
    Thus $\lim_{x\uparrow\beta}g(k,x) = \lim_{x\downarrow\beta}g(k,x)$ as required, and $g(k,\beta)$ is continuous in $\beta$ as claimed.
    
    We next observe that for all non-negative integers $x$ and all fixed $k \ge 2$, $g(k,\beta)$ is differentiable over $\beta \in (x,x+1)$, since $\lfloor (k-\beta)/2 \rceil$ is constant on this interval. The derivative is given by
    \[
        \frac{\partial}{\partial \beta}g(k,\beta) = -\frac{1}{k}\Big\lfloor\frac{k-\beta}{2}\Big\rceil \ge -\frac{1}{k}\cdot\lim_{\beta\downarrow0}\Big\lfloor\frac{k-\beta}{2}\Big\rceil \ge -\frac{1}{2},
    \]
    where the final inequality follows since $k \ge 2$ is an integer and $\beta > x \ge 0$. Since we have already shown that on fixing $k$ and viewing $g$ as a function of $\beta$, $g$ is continuous everywhere and only fails to be differentiable at integer values, it follows by the mean value theorem that for all $\beta \ge 0$, $g(k,\beta) \ge g(k,0) - \beta/2$ as required.
\end{proof}

We now prove our first key upper bound. We will use this both directly to bound the running time of \UncolApprox and in the proof of \cref{lem:final-F-bound}.

\begin{lemma}\label{lem:uncol-bound-ti}
    Fix integers $k \ge 2$ and $n \ge 1$. For all integers $i \ge 0$, let~$\hat L_i\in\N$ and $\hat\gamma_i\in[0,1)$ be the unique values with~$2^{ik}=n^{\hat L_i+\hat\gamma_i}$.
    Let
    \begin{equation}\label{eq:uncolapprox-F}
        F_i=n^{(\hat L_i+\hat\gamma_i)(k-\hat L_i)/k}
            \cdot \max\set{2^{-i},n^{-\hat\gamma_i}}
    \end{equation}
    For all non-negative reals $\beta$, we have
    \begin{equation}\label{eq:uncolapprox-rounding}
        \max_i F_i \cdot 2^{-i\beta} = \max_{L \in \{0,\dots,k-\ceil{\beta}\}} n^{L(k-L-\beta)/k} =
        n^{g(k,\beta)}
        \,.
    \end{equation}
    Moreover, if $\beta > k-1$ and $i \ge (\log n)/k$, then we have 
    \[
        F_i \cdot 2^{-i\beta} \le n^{-(\beta-(k-1))}.
    \]
\end{lemma}

\begin{figure}
\begin{center}
    \includegraphics{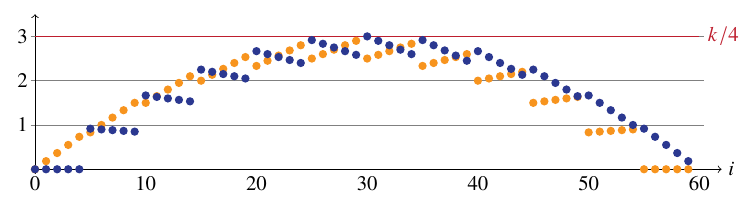}
\end{center}
\caption{\label{fig:uncolapprox_Fi}%
Depicted are the functions
$i\mapsto \tfrac1k f_1(\hat L_i,\hat\gamma_i)$ (\protect\tikz[baseline=-\the\dimexpr\fontdimen22\textfont2\relax]\protect\node[p,fill=wedge3] {};)
and
$i\mapsto \tfrac1k f_2(\hat L_i,\hat\gamma_i)$ (\protect\tikz[baseline=-\the\dimexpr\fontdimen22\textfont2\relax]\protect\node[p,fill=wedge4] {};)
from \cref{lem:uncol-bound-ti}
for~$n=2^{60}$ and $k=12$.
For each fixed integer~$L\in\set{0,\dots,k-1}$, the functions are linear in $\gamma$, which gives rise to the step artefacts; when~$\gamma$ is fixed, the functions are degree-2 polynomials in~$L$, which causes the overall parabolic shape.
When~$k$ and~$\log n$ are even, the overall maximum is equal to $\frac{k}{4}$ and achieved at~$i=\frac{\log n}{2}$.}
\end{figure}

\begin{proof}
    Using~$2^i=n^{(\hat L_i+\hat\gamma_i)/k}$ and taking base-$n$ logarithms of both sides of \cref{eq:uncolapprox-F}, we get $\log_n(F_i\cdot 2^{-i\beta})=\tfrac1k\max\set{f_1(\hat L_i, \hat\gamma_i), f_2(\hat L_i, \hat\gamma_i)}$ where~$f_1,f_2:\N\times[0,1]\to\R$ are the functions depicted in \cref{fig:uncolapprox_Fi} and defined via
    \begin{align*}
        f_1(L, \gamma)
        &=
        (L+\gamma)(k- L)
        - (L+\gamma)
        - (L+\gamma)\beta
        =
        (L+\gamma)(k-L-1-\beta)
        \,,\\
        f_2(L, \gamma)
        &=
        (L+\gamma)(k- L)
        -k\gamma
        - (L+\gamma)\beta
        =
        (L+\gamma)(k-L-\beta)
        -k\gamma
        \,.
    \end{align*}
    We now provide an upper bound for both functions~$f_1$ and~$f_2$.
    Both functions are linear in~$\gamma$, and thus the maximum is achieved at~$\gamma=0$ or~$\gamma=1$.
    Thus it remains to upper bound the following four functions of~$L$:
    \begin{align*}
        f_1(L, 0)
        &=
        L(k-L-1-\beta)
        \,,\\
        f_1(L, 1)
        &=
        (L+1)(k-L-1-\beta)
        \,,\\
        f_2(L, 0)
        &=
        L(k-L-\beta)
        \,,\\
        f_2(L, 1)
        &=
        (L+1)(k-L-\beta)
        -k
        \,.
    \end{align*}
    Note that for all $L$, $f_1(L,1)=f_2(L+1,0)$,
    $f_1(L,0)\le f_2(L,0)$,
    and $f_2(L,1)\le f_2(L,0)$ hold; thus
    \[
        \max_i \log_n(F_i\cdot 2^{-i\beta}) = \frac{1}{k}\max_i\Big(\max\big\{f_1(\hat{L}_i,\hat{\gamma}_i), f_2(\hat{L}_i,\hat{\gamma}_i)\big\}\Big) = \frac{1}{k} \max_{L\in\{0,\dots,k\}} f_2(L,0).
    \]
    By elementary calculus, $f_2(L,0)$ is maximised over integer values of $L$ at $L = \lfloor (k-\beta)/2 \rceil$. Moreover, this value of $L$ is at most $k-\ceil{\beta}$; indeed, if $\beta \le k-1$ then we have $\lfloor (k-\beta)/2 \rceil = \lfloor (k-\beta+1)/2\rfloor \le (k-\beta+1)/2 \le k-\beta$, and if $\beta > k-1$ then we have $\lfloor (k-\beta)/2 \rceil = k-\ceil{\beta} = 0$. Thus we have
    \[
        \max_i F_i\cdot 2^{-i\beta}  
        =
        n^{\frac{1}{k}\max_{0 \le L \le k-\ceil{\beta}} (f_2(L,0))}
        =
        n^{\frac{1}{k}f_2\left(\left\lfloor\frac{k-\beta}{2}\right\rceil,\,0\right)}
        =
        n^{g(k,\beta)}\,,
    \]
    as required.
    
    Finally, suppose $\beta > k-1$ and $i \ge (\log n)/k$, so that $L \ge 1$. Arguing as before,
    \begin{align*}
        \log_n(F_i\cdot 2^{-i\beta}) 
        &\le \frac{1}{k}\max_{L \ge 1}f_2(L,0) 
        = \max_{L \ge 1}L(k-L-\beta)\\
        &= \max_{L \ge 1} \big(L(1-L) - L(\beta - (k-1))\big)\,.
    \end{align*}
    The function inside the maximum is decreasing in $L$ for $L \ge 1$, so we have $\log_n(F_i\cdot 2^{-i\beta}) \le -(\beta-(k-1))$ and the result follows.
\end{proof}

We now state what is essentially our second key lemma of the section; however, for technical convenience, we replace cost functions by their associated polynomials. We will then use this lemma to prove the actual result we need in \cref{cor:uncol-final-cost-bound}.

\begin{lemma}\label{lem:final-F-bound}
    Fix integers $k \ge 2$ and $n \ge 1$. For all integers $i \ge 0$, let~$\hat L_i\in\N$ and $\hat\gamma_i\in[0,1)$ be the unique values with~$2^{ik}=n^{\hat L_i+\hat\gamma_i}$.
    Let
    \begin{equation*}
        F_i=n^{(\hat L_i+\hat\gamma_i)(k-\hat L_i)/k}
            \cdot \max\set{2^{-i},n^{-\hat\gamma_i}}\,.
    \end{equation*}
    For all $\beta \in [0,k]$, we have
    \begin{equation*}
        \max_i F_i\cdot \big(\max\{n/2^i,n^{1/4}\}\big)^\beta = \max_{r \in \{\ceil{\beta},\dots,k\}} n^{r(k-r+\beta)/k} = n^{\beta + g(k,\beta)}\,.
    \end{equation*}
    Moreover, if $\beta > k-1$ and $i \ge (\log n)/k$, then we have
    \[
        F_i \cdot \big(\max\{n/2^i,n^{1/4}\}\big)^\beta \le n^{\max\{3\beta/4,k-1\}}\,.
    \]
\end{lemma}
\begin{proof}
    Observe that by \cref{lem:uncol-bound-ti} applied with $\beta=0$, $F_i \le n^{g(k,0)}$. By Lemma~\ref{lem:uncol-bound-Fi-grad}, it follows that $F_i \le n^{g(k,\beta)+\beta/2}$, and so 
    \begin{equation}\label{eq:final-F-bound-1}
        F_in^{\beta/4} \le n^{g(k,\beta)+3\beta/4} < n^{\beta+g(k,\beta)}\,.
    \end{equation}
    On the other hand, by Lemma~\ref{lem:uncol-bound-ti} we have 
    \[
        \max_i F_i (n/2^i)^\beta = n^{g(k,\beta) + \beta}\,.
    \]
    It follows that
    \[
        \max_i F_i\cdot\big(\max\{n/2^i,n^{1/4}\}\big)^\beta = n^{g(k,\beta)+\beta}\,,
    \]
    as required. Moreover, by substituting $L=k-r$ and applying \cref{lem:uncol-bound-ti}, we have
    \begin{align*}
        \max_{r \in \{\ceil{\beta},\dots,k\}} n^{r(k-r+\beta)/k} 
        &= \max_{L \in \{0,\dots,k-\ceil{\beta}\}} n^{(k-L)(L+\beta)/k}\\ 
        &= \max_{L \in \{0,\dots,k-\ceil{\beta}\}} n^{\beta+L(k-L-\beta)/k} 
        = n^{\beta+g(k,\beta)}\,,
    \end{align*}
    again as required.
    
    Finally, suppose $\beta > k-1$ and $i \ge (\log n)/k$. Then by \cref{lem:uncol-bound-ti} we have
    \[
        F_i (n/2^i)^\beta \le n^{\beta-(\beta-(k-1))} = n^{k-1}\,.
    \]
    Moreover, in this case we have $g(k,\beta)=0$ and so~\eqref{eq:final-F-bound-1} implies that $F_i n^{\beta/4} \le n^{3\beta/4}$. The result therefore follows.
\end{proof}

Finally, we restate \cref{lem:final-F-bound} in terms of a general slowly-varying cost function; this is the form of the result we will actually use. Here $C^+(n)$ will be the upper bound on oracle cost we derive in \cref{sec:uncol-upper}.

\begin{cor}\label{cor:uncol-final-cost-bound}
    Let $\cost$ be a regularly-varying parameterised cost function with parameter $k$ and slowly-varying component $\sigma$, and let $\alpha_k \in [0,k]$ be the index of $\cost_k$ for all $k \ge 2$. For each integer~$i$, let $\hat{L}_i(n) \in \N$ and $\hat\gamma_i(n) \in [0,1)$ be the unique values with $2^{ik} = n^{\hat{L}_i(n) + \hat\gamma_i(n)}$, and let
    \begin{align*}
        F_i(n)&=n^{(\hat L_i(n)+\hat\gamma_i(n))(k-\hat L_i(n))/k}
            \cdot \max\set{2^{-i},n^{-\hat\gamma_i(n)}}\,,\\
        C^+_k(n) &= \max_{0 \le i \le \log n-1}\Big(F_i(n)\cdot \cost_{k}(\max\{n/2^i,n^{1/4}\})\Big)\,.
    \end{align*}
    Then for all sequences $k=k(n)$ with $k(n) \in \{2,\dots,n\}$ for all $n$, we have:
    \begin{enumerate}[(i)]
        \item $C^+_k(n) = n^{g(k,\alpha_k)+o(1)}\cost_k(n)$ as $n\to\infty$; and
        \item if $\sigma$ is eventually non-decreasing or if there exists $\eta>0$ such that $\alpha_k \ge k-1+\eta$ for all $k$, then $C^+_k(n) = \OO(n^{g(k,\alpha_k)}\cost_k(n))$ as $n\to\infty$.
    \end{enumerate} 
\end{cor}
\begin{proof}
    Let $\sigma$ be the slowly-varying component of $\cost$, so that $\cost_k(n) = n^{\alpha_k}\ell(n)$. By~\cref{lem:final-F-bound} applied with $\beta=\alpha_k$, we have
    \begin{align}\nonumber
        C_k^+(n) &= \max_i \Big(F_i(n) \cdot (\max\{n/2^i,n^{1/4}\})^{\alpha_k} \cdot \sigma\big(\max\{n/2^i,n^{1/4}\}\big)\Big)\\\nonumber
        &\le n^{g(k,\alpha_k)+\alpha_k}\cdot\max_i\Big(\sigma\big(\max\{n/2^i,n^{1/4}\}\big)\Big)\\\label{eq:uncol-final-cost-upper}
        &= n^{g(k,\alpha_k)}\cost_k(n) \cdot \max_i\bigg(\frac{\sigma\big(\max\{n/2^i,n^{1/4}\}\big)}{\sigma(n)}\bigg)\,.
    \end{align}
    
    Since $\sigma$ is slowly-varying, part (i) follows immediately from~\eqref{eq:uncol-final-cost-upper}. Moreover, if $\sigma$ is eventually non-decreasing, then since $\max\{n/2^i,n^{1/4}\} \le n$, (ii) follows from~\eqref{eq:uncol-final-cost-upper}. Finally, suppose $\alpha_k \ge k-1+\eta$ for some fixed $\eta>0$. We split into two cases depending on $i$.
    
    If $i < (\log n)/k$, then we have $\hat{L}_i(n) = 0$, so $F_i(n) = 1$. By \cref{lem:regularly-varying-facts}(ii) and (iii), for sufficiently large $n$ we have $\cost_k(n/2^i) \le \cost_k(n)/2^{i\alpha_k/2} = \OO(\cost_k(n))$ and $\cost_k(n^{1/4}) \le n^{\beta/2} = o(\cost_k(n))$; it follows that
    \[
        F_i(n)\cost_k(\max\{n/2^i,n^{1/4}\}) = \OO(\cost_k(n))\,.
    \]
    
    If instead $i \ge (\log n)/k$, then by \cref{lem:final-F-bound} we have
    \begin{align*}
        F_i(n)\cost_k(\max\{n/2^i,n^{1/4}\}) &= F_i(n) \cdot (\max\{n/2^i,n^{1/4}\})^{\alpha_k} \cdot \sigma\big(\max\{n/2^i,n^{1/4}\}\big) \\
        &\le n^{\max\{3\alpha_k/4,k-1\}} \cdot\sigma\big(\max\{n/2^i,n^{1/4}\}\big)\\
        &= \cost_k(n) \cdot n^{-\alpha_k+\max\{3\alpha_k/4,k-1\}} \cdot \frac{\sigma\big(\max\{n/2^i,n^{1/4}\}\big)}{\sigma(n)}\,.
    \end{align*}
    Observe that $\alpha_k/4 \ge (k-1)/4 \ge 1/4$, and that $\alpha_k - k-1 \ge \eta$, which is a constant; by \cref{lem:regularly-varying-facts}(iii), it follows that the $n^{-\alpha_k+\max\{3\alpha_k/4,k-1\}}$ term dominates the $\sigma(\cdot)$ terms and so the left-hand side is $\OO(\cost_k(n))$. We have therefore shown that $C^+_k(n) = \OO(\cost_k(n))$ in all cases, as required. 
\end{proof}

\subsection{Oracle algorithm for edge estimation}
\label{sec:uncol-upper}

Our \indora-oracle algorithm for the edge estimation problem has two components:
\begin{enumerate}
    \item A randomised \indora-oracle algorithm \SparseCount{} by Meeks~{\cite[Theorem~6.1]{DBLP:journals/algorithmica/Meeks19}} to enumerate all edges of a given hypergraph when given access to the independence oracle. \SparseCount{} is an exact enumeration algorithm, and its running time can be tightly controlled by setting a threshold parameter $M$ and aborting it after $M$ edges are seen. This version of the algorithm requires only roughly $\OO(M)$ queries, but may return \TooDense{} to indicate that the hypergraph has more than~$M$ edges. 
    \item A randomised \indora-oracle algorithm \UncolApprox{} that invokes \SparseCount{} on smaller and smaller random subgraphs until they are sparse enough so that \SparseCount{} succeeds at enumerating all edges under the time constraints. We describe \UncolApprox{} in \cref{sec:uncolapprox}.
    From the exact numbers of edges that were observed in several independently sampled random subgraphs, the algorithm then calculates and returns an approximation to the number of edges in the whole hypergraph.
\end{enumerate}

We give a self-contained overview of \SparseCount{} in \cref{sec:sparsecount}. We then set out \UncolApprox{} in \cref{sec:uncolapprox}, describe the intuition behind it, and sketch a proof of its correctness. We then prove its correctness formally in \cref{sec:uncol-approx-correct}, and bound its running time in Sections~\ref{sec:uncol-approx-algebra} and~\ref{sec:uncol-run-time}.

\subsubsection{SparseCount: Enumerate all edges in sparse hypergraphs}\label{sec:sparsecount}
All the properties of \SparseCount{} we need will follow from \cite{DBLP:journals/algorithmica/Meeks19}, but for completeness we sketch how \SparseCount{} works and give some intuition.
Note that the version of \SparseCount{} described in~\cite[Algorithm~1]{DBLP:journals/algorithmica/Meeks19} gives a deterministic guarantee of correctness when supplied with a deterministic oracle. For our purposes this is unnecessary, and so our sketch will be of a slightly simpler version; in \cite{DBLP:journals/algorithmica/Meeks19} our uniformly random sets are replaced with a deterministic equivalent using $k$-perfect hash functions.

\SparseCount{} has access to a $k$-uniform hypergraph through its independence oracle. Moreover, $\SparseCount$ is given a set $U\subseteq V(G)$, the integer~$k$, and a threshold parameter $M\ge 1$ as input.
$\SparseCount(\indora(G),U,k,M)$ works as follows:
\begin{enumerate}
    \item Sample $t$ uniformly random colourings $c_1,\dots,c_t\colon U\to[k]$ for $t=\Theta(e^{2k}\log \abs{U})$.
    \item For each $i$ from $1$ to $t$, run a subroutine \RecEnum{} to recursively enumerate and output all edges $e$ that are colourful with respect to $c_i$.
    Keep track of the number of unique edges seen so far by incrementing a counter each time \RecEnum{} outputs an edge $e$ that is colourful with respect to $c_i$ and \emph{not} colourful with respect to any $c_1,\dots,c_{i-1}$.
    \item
    If at any point the counter exceeds $M$, abort the execution immediately and return \TooDense{}. Otherwise, return the final contents of the counter.
\end{enumerate}
Next, we describe the missing subroutine \RecEnum{} that is called by \SparseCount{}.
The input for \RecEnum{} is a tuple $(U_1,\dots,U_k)$ of disjoint subsets of $U$ as well as the threshold parameter~$M\ge 1$. In the initial call, $U_j$ is the $j$\th colour class of some colouring $c_i$ from step~1 of \SparseCount{}. 
\RecEnum{}$(\indora(G),U_1,\dots,U_k)$ then proceeds as follows:
\begin{enumerate}
    \item If the set $e$ with $e=U_1\cup\dots\cup U_k$  has size at most~$k$, we have
    $\abs{U_1}=\dots=\abs{U_k}=1$, and $e$ is an edge (as determined by a call~$\indora(G)(e)$ to the independence oracle of~$G$), then output $e$; otherwise, do nothing.
    \item Otherwise, independently and uniformly at random split each part $U_j$ into two disjoint parts $U_{j0}$ and~$U_{j1}$ of near-equal size, and recurse on all tuples $(U_{0b_0},\dots,U_{kb_k})$ for all bit vectors $b\in\{0,1\}^k$ for which $G[U_{0b_0}\cup\dots\cup U_{kb_k}]$ contains at least one edge (as determined by a call to the independence oracle on $G[U_{0b_0}\cup\dots\cup U_{kb_k}]$).
\end{enumerate}

The intuition for the correctness of the algorithm is as follows:
By standard colour-coding arguments, with high probability every edge is colourful with respect to at least one of the $t$ random colourings $c_i$ and so \RecEnum{} will find it (regardless of how parts are split in step~2). The running time of \RecEnum{} is not affected too much by edges that are not colourful, because with high probability a large proportion of edges are colourful and uncoloured edges get deleted quickly when the subsets are sampled. We encapsulate the relevant properties of \SparseCount{} in the following lemma.
\begin{lemma}
\label{lem:sparse-count}
    There is a randomised \indora-oracle algorithm $\SparseCount(\indora(G),U,k,M,\delta)$ with the following behaviour:
    \begin{itemize}
        \item $\SparseCount$ takes as input a set $U\subseteq V(G)$, integers $k$ and $M$, and a rational number~$\delta>0$.
        \item \SparseCount{} may output the integer $e(G[U])$, \TooDense{} (\enquote{The hypergraph $G[U]$ definitely has more than $M$ edges}), or \RTE{} (\enquote{Allowed running time exceeded}).
        \item If $e(G[U]) \le M$ holds, then $\SparseCount(U,k,M,\delta)$ outputs either $e(G[U])$ or \RTE{}; otherwise, it outputs either \TooDense{} or \RTE{}.
        \item $\SparseCount$ invokes the uncoloured independence oracle of~$G$ at most \[\OO\paren[\Big]{\log\tfrac1\delta \cdot e^{2k}\log^2 \abs{U}\cdot\min\{M,e(G[U])\}}\] times, and runs in time at most \[\OO\paren[\Big]{\log\tfrac1\delta \cdot e^{2k}\log^2 \abs{U}\cdot\min\{M,e(G[U])\}\cdot\abs{U}}\] aside from that.
        \item The probability that \SparseCount{} outputs \RTE{} is at most $\delta$.
    \end{itemize}
\end{lemma}
We omit a formal proof of \cref{lem:sparse-count}, since it follows easily from Meeks~{\cite[Theorem~1.1]{DBLP:journals/algorithmica/Meeks19}} with a very similar argument to \cite[Theorem~6.1]{DBLP:journals/algorithmica/Meeks19}, with the following minor caveats. 
First, we have applied the standard probability amplification result of \cref{lem:expected-to-worst-case} to pass from bounds on the expected running time to deterministic bounds that hold with probability at least $2/3$, and from there we have applied \cref{lem:median-boosting} to pass to bounds that hold with probability $1-\delta$. (We take the cost function to be the number of queries, and the interval of \cref{lem:median-boosting} to be $\{e(G[U])\}$.) 
Second, the theorem statements in~\cite{DBLP:journals/algorithmica/Meeks19} do not explicitly separate bounds on the number of oracle queries from the total running time without oracle queries, and in particular this means the dependence on $n$ in their running time bounds is stated as $n^{\OO(1)}$. However, from the proof of \cite[Theorem~1.1]{DBLP:journals/algorithmica/Meeks19} in Section 4 we see that there are $\OO(e^{k+o(k)}e(G[U])\log^2 n)$ total oracle calls, and that the running time without oracle calls is $\OO(e^{k+o(k)}e(G[U])n\log^2 n)$.

\subsubsection{UncolApprox: Approximately count all edges in hypergraphs}\label{sec:uncolapprox}

\begin{algorithm}
	\SetKwInput{Oracle}{Oracle}
    \SetKwInput{Input}{Input}
	\SetKwInput{Output}{Output}
	\DontPrintSemicolon
    \Oracle{Independence oracle $\indora(G)$ of an $n$-vertex $k$-hypergraph $G$.}
	\Input{$n,k\in\N$ and $\eps \in (0,1)$.}
	\Output{$x\in\Q$ that, with probability at least $2/3$, is an $\eps$-approximation to $e(G)$.}
	\Begin{
	    \If{$n^k\le \eps^{-2}$ \textbf{\upshape or} $n\le k^5$}{
	        Enumerate all size-$k$ subsets of $V(G)$, check the independence oracle for each to compute $e(G)$ by brute force, and \textbf{return} this value.
            \label{line:uncol-exhaustive}\;
	    }
        \If{\textup{$n$ is not a power of two}}{
            Set $n=2^{\lceil\log n\rceil}$ to add at most~$n-1$ padding vertices.
            Before sending any future query to the oracle, we remove all padding vertices from the query.
            \label{line:uncol-power-of-two}\;
        }
        \For{$i=0,\dots,\log n-1$}{\label{line:uncol-start-for}
            Set the vertex survival probability $p_i = 1/2^i$.\label{line:set-pi}
            \;
            Set $\hat L_i\in\set{0,1,\dots,k}$ and $\hat\gamma_i\in\set{\tfrac0{\log n},\tfrac1{\log n},\dots,\tfrac{\log n-1}{\log n}}$ to be the unique values that satisfy $n^{\hat L_i+\hat\gamma_i} = 2^{ik}$.\label{line:set-Li}\; 
            Set $F_i=n^{(\hat L_i+\hat\gamma_i)(k-\hat L_i)/k}
            \cdot \max\set{2^{-i},n^{-\hat\gamma_i}}$.\label{line:set-Fi}\;
            Set
            $t_i=\lceil \eps^{-2} 10k^2 2^k\log n \cdot F_i\rceil$\label{line:set-ti}
            and set $M_i=2^{k+1}\cdot t_i$.\label{line:set-Mi}\;
	        \For{$j=1,\dots,t_i$}{
                Use \SampleSubset{} (see \cref{lem:sample-random-subset}) to efficiently sample a random subset~$U_{i,j} \subseteq V(G)$ that includes each element with independent probability~$p_i$. If the total running time of all \SampleSubset{} calls ever exceeds $20\Csample\sum_{i=0}^{\log n-1}t_i(i+n/2^i)$, then abort and return \RTE{}. \label{line-call-samplesubset}\;
                \If{$|U_{i,j}| > \max\{7p_i n, 7k\ln n\}$}{Return \RTE{}.\label{line-RTE-large-set}\;}
				Let $C_{i,j} = \SparseCount(\indora(G),U_{i,j},k,M_{i,j},n^{-5k}/120)$.\label{line-call-sparsecount}\;
				\If{$C_{i,j} = \TooDense{}$}{
                    Continue in line~\ref{line:set-pi} with the next iteration of the outer for-loop.
				}
				\ElseIf{$C_{i,j} = \RTE{}$}{
				    \Return \RTE{}.\label{line-RTE-sparse}\;
				} 
				\Else{
                    Set $M_{i,j+1}$ to $M_{i,j} - C_{i,j}$.\;
                }
                }
                \If{none of $C_{i,1},\dots,C_{i,t_i}$ are \TooDense{}}{
                    \Return $\frac{1}{p_i^kt_i}\sum_{j=1}^{t_i} C_{i,j}$.\;
                }
		}\label{line:uncol-end-for}
	}
    \caption{\label{algo:uncolapprox}\UncolApprox\\
    \textit{This $\indora$-oracle algorithm applies $\SparseCount$ to carefully-chosen numbers of successively sparser random subgraphs of~$G$ until the samples become sparse enough so that $\SparseCount$ stops returning \TooDense{} and starts consistently returning accurate edge counts. At this point, the algorithm takes the average of the edge counts over all samples at that density and renormalises to obtain the final estimate.}}
\end{algorithm}

In this section, we analyse our main algorithm, \UncolApprox{}, which is laid out as \cref{algo:uncolapprox}. We write~$n$ for the number of vertices of~$G$ and $m=e(G)$ for the number of edges.

The basic idea of the algorithm is both simple and standard. In the main loop over $i=0,\dots,\log n$ of lines \ref{line:uncol-start-for}--\ref{line:uncol-end-for}, for a suitably-chosen integer $t_i$, we sample $t_i$ independent random subsets $U_{i,j}$ of $V(G)$, including each element with probability $p_i=1/2^i$. We then count the edges in each $e(G[U_{i,j}])$ using \SparseCount{} with a suitably-chosen threshold~$M_{i,j}$. It is easy to see by linearity of expectation that $\E(e(G[U_{i,j}])) = p_i^k e(G)$, so if our calls to \SparseCount{} succeed then in expectation we return $e(G)$ in line~\ref{line:uncol-end-for}.

The main subtlety of \UncolApprox{} is in optimising our choice of parameters to minimise the running time while still ensuring correctness with high probability. The intuition here will tie into our lower bound proof in \cref{sec:uncol-lower}, so we go into detail rather than simply presenting calculations. 

By line~\ref{line:uncol-power-of-two}, we can assume without loss of generality that~$n$ is a power of two, which simplifies the notation.
Moreover, we write~$m=n^\delta$ and split~$\delta$ into its integral part~$L=\floor{\delta}$ and its rational part~$\gamma=\delta-L$, and we remark that~$m\le\binom{n}{k}<n^k$ and thus~$0\le L\le k-1$ and $\gamma\in[0,1)$ holds.

The first parameter in \UncolApprox{} is~$p_i=2^{-i}$, which is the probability of independently including each vertex~$v$ in the set~$U_{i,j}$.
Note that each edge $e\in E(G)$ has $k$ elements and thus survives in $G[U_{i,j}]$ with probability $p_i^k=2^{-ik}$.
Therefore, the expected number of edges in $G[U_{i,j}]$ is $p_i^km$, and for a given iteration $i$, the expected total number of edges $T_i \coloneqq \sum_j e(G[U_j])$ is $t_ip_i^km$. In the following definition, we define $i^*$ as the largest integer such that this expected value remains at least~$1$.
\begin{defn}\label{def:istar}
    Given an $m$-edge graph $G$, let $i^*=i^*(G)$ be the largest value of $i$ such that $p_i^k m \ge 1$ holds (where $p_i=2^{-i}$ as in \UncolApprox{}).
\end{defn}
Note that $i^*\le\log n-1$ follows from $m<n^k$ and $p_{i^*}^k m\le 2^{k+1}$. We will show in \cref{lem:uncoloured-chebyshev} using Chebyshev's inequality that $T_i$ is very likely to be an $\eps$-approximation of $\E(T_i)$ whenever $i \le i^*$. Observe that if this holds, then whenever $T_i \le M_i$, all iterations of \SparseCount{} succeed and \UncolApprox{} outputs a valid $\eps$-approximation of $m$.

Given Lemma~\ref{lem:uncoloured-chebyshev}, the reason \UncolApprox{} is correct will be as follows: Whenever $i \le i^*$, if all calls to $\SparseCount{}$ succeed in iteration $i$, we output an $\eps$-approximation to $m$ as required. The parameter $M_i$ is chosen in such a way that if the number of edges is roughly $2^{ik}$, then all calls to $\SparseCount{}$ are indeed likely to succeed. When $i=i^*$ we do indeed have $m \approx 2^{ik}$, so we are very likely to output a valid $\eps$-approximation in this iteration if we have not already done so. 

For the correctness of \UncolApprox{}, the concentration analysis of \cref{lem:uncoloured-chebyshev} is therefore crucial. The values $\hat L_i$ and $\hat\gamma_i$ with $2^{ik}=n^{\hat L_i+\hat\gamma_i}$ arise naturally in this analysis. Recall that $t_i$ is the number of random graphs~$G[U_{i,j}]$ we count to estimate $m$; to minimise the running time, we want to choose $t_i$ as small as possible while maintaining concentration of $T_i$, so we choose it to scale with the maximum possible value of $\var(T_i)$ on the assumption that $m \lessapprox 2^{ik}$. This maximum value always arises when $m \approx 2^{ik}$, and there are two possible ``extreme cases'' of input graphs $G$ with this edge count for which $\var(T_i)$ could be near-maximum. The value of $\hat\gamma_i$ determines which case dominates, and hence which value we should take for $t_i$. (This is the source of the maximum in line~\ref{line:set-Fi}.)

Suppose for simplicity that $m$ is a power of $2^k$, writing $m = 2^{ik} = n^{\hat{L}_i + \hat{\gamma}_i}$. The first extreme case is a \emph{single-rooted star graph}. Such a graph $G$ contains a single size-$(k-\hat{L}_i-1)$ \emph{root}~$R\subseteq V(G)$ such that every edge of $G$ contains $R$. Note that there are roughly $n^{\hat{L}_i+1}$ possible edges containing $R$, so $G$ can indeed have roughly $m$ edges since $0 \le \hat\gamma_i < 1$. In order for $T_i$ to be an $\eps$-approximation to $\E(T_i)$, it is necessary that the algorithm finds \emph{some} edge during iteration~$i$, so at least one set~$U_{i,j}$ must contain~$R$.
Since~$R\subseteq U_{i,j}$ happens with probability~$p_i^{k-\hat{L}_i-1}$, we need $1/p_i^{k-\hat{L}_i-1}$ samples in expectation until we see even a single edge, where
\begin{align}
    1/p_i^{k-\hat{L}_i-1}
    &= 2^{i(k-\hat{L}_i-1)}
    = n^{(L+\hat\gamma_i)(k-\hat{L}_i)/k}\cdot 2^{-i}
\end{align}
This shows why~$F_i$ in line~\ref{line:set-Fi} can't be much smaller if the first term of the maximum dominates. (Our choice of~$t_i$ then includes some additional minor factors to amplify the probability of producing an~$\eps$-approximation.)

The second extreme case is a \emph{many-rooted star graph}. Such a graph $G$ contains $\lceil n^{\hat\gamma_i}\rceil $ disjoint \emph{roots}~$R_x$ of $k-\hat{L}_i$ vertices each. The edges of the graph are precisely the size-$k$ sets $e$ that contain at least one $R_x$. Observe that each root has roughly $n^{\hat{L}_i}$ incident edges, so $G$ has roughly $n^{\hat{L}_i+\hat\gamma_i} = m$ edges. Again, in order for $T_i$ to be an $\eps$-approximation to $\E(T_i)$ the algorithm must find some edge, so some set $U_{i,j}$ must contain a root. Each set~$U_{i,j}$ has probability $p_i^{k-\hat{L}_i}$ to contain any particular root, and thus probability at most~$p_i^{k-\hat{L}_i}\cdot n^{\hat\gamma_i}$ to contain at least one root.
Thus, the expected number of samples until we have seen a single edge is at least:
\begin{align}
    1/\paren[\big]{p_i^{k-\hat{L}_i}\cdot n^{\hat\gamma_i}}&
    =2^{i(k-\hat{L}_i)}\cdot n^{-\hat\gamma_i}
    =n^{(\hat{L}_i+\hat\gamma_i)(k-\hat{L}_i)/k}\cdot n^{-\hat\gamma_i}
\end{align}
This shows why $F_i$ in line~\ref{line:set-Fi} can't be much smaller if the second term of the maximum dominates.

We refrain from making these claims on the optimality of the choice of~$t_i$ (and thus the running time of~\UncolApprox{}) any more formal, because we provide lower bounds for \emph{any} algorithm in \cref{sec:uncol-lower}. However, we think it is worth giving the intuition since the proof of the lower bound uses the first extreme case as a source of hard input graphs. (The calculations in the proof of \cref{lem:uncol-bound-ti} show that while the second case is necessary for correctness --- without it we will take too few samples to obtain concentration on our output when $\hat\gamma_i$ is close to 1 --- it does not end up affecting the running time.)

\subsubsection{Correctness of UncolApprox}\label{sec:uncol-approx-correct}

In the following lemma, we show that, for every iteration~$i\le i^*$ and with high probability, the average and normalised number of edges in the subgraphs~$G[U_{i,j}]$ is a good approximation to the number of edges in~$G$.
\begin{lemma}\label{lem:uncoloured-chebyshev}
    In $\UncolApprox$, suppose $i \le i^*(G)$ and $\eps \ge n^{-k/2}$.
    We have
    \[
        \pr\bigg(\bigg|e(G) - \frac{1}{p_i^kt_i}\sum_{j=1}^{t_i} e(G[U_{i,j}])\bigg| \ge \eps \cdot e(G)\bigg) \le \frac{1}{10k\log n}.
    \]
\end{lemma}
\begin{proof}
    We first set out some notation.
    For convenience, for all $0 \le i \le \log n -1$ and all $j \in [t_i]$, let $Z_{i,j} = e(G[U_{i,j}])$, let $Z_i = \frac{1}{p_i^kt_i}\sum_{j=1}^{t_i} Z_{i,j}$, and let $m=e(G)$.
    In this notation, our goal is to prove that $|Z_i - m| \le \eps m$ with probability at least $1 - 1/(10k\log n)$. We will prove the result by bounding the variance of $Z_i$ above and applying Chebyshev's inequality; observe by linearity of expectation that $\E(Z_i) = m$. 
    
    For all $e \in E(G)$, let $1_e$ be the indicator random variable of the event $e \subseteq U_{i,j}$.
    By linearity of expectation, for all $j$, we have
	\begin{equation}\label{var-Zij}
		\var(Z_{i,j}) = \E(Z_{i,j}^2) - \E(Z_{i,j})^2 \le \sum_{\substack{(e,f) \in E(G)^2\\e \cap f \ne \emptyset\\e \ne f}}\E(1_e1_f) = \sum_{\substack{(e,f) \in E(G)^2\\e \cap f \ne \emptyset\\e \ne f}}p_i^{|e \cup f|}.
    \end{equation}
    For any set $A \subseteq V(G)$, we write $d_A$ for the number of edges in $G$ which contain $A$. For all $e,f\in E(G)$ we have $|e\cup f|=2k-|e\cap f|$, so it follows from~\eqref{var-Zij} that
    \begin{equation}\label{var-Zij-2}
        \var(Z_{i,j}) \le \sum_{\ell=1}^{k-1} \sum_{\substack{A \subseteq V(G)\\|A|=\ell}} \sum_{\substack{(e,f) \in E(G)^2\\e\cap f=A}}p_i^{2k-|A|} = \sum_{\ell=1}^{k-1} \sum_{\substack{A \subseteq V(G)\\|A|=\ell}} d_A^2p_i^{2k-\ell}.
    \end{equation}
    Observe that for any set $A$ of size $\ell \in [k-1]$, we have $d_A \le \min\{m,n^{k-\ell}\}$; thus by~\eqref{var-Zij-2},
    \begin{equation}\label{var-Zij-3}
        \var(Z_{i,j}) \le \sum_{\ell=1}^{k-1} \min\{m,n^{k-\ell}\} p_i^{2k-\ell} \sum_{\substack{A \subseteq V(G)\\|A|=\ell}} d_A.
    \end{equation}
    
    For all $\ell \in [k-1]$, every edge of $G$ contains exactly $\binom{k}{\ell}$ sets $A$ of size $\ell$, and every set $A$ of size $\ell$ is contained in precisely $d_A$ edges of $G$. Thus by double-counting,
    \begin{equation}\label{eq:dA-double-counting}
        \sum_{\substack{A \subseteq V(G)\\|A| =\ell}} d_A = \binom{k}{\ell} m
        \le 2^km\,.
    \end{equation}
    It therefore follows from~\eqref{var-Zij-3} that
    \begin{equation}\label{var-Zij-4}
        \var(Z_{i,j}) \le 2^kp_i^{2k}m \sum_{\ell=1}^{k-1} \min\{m,n^{k-\ell}\} p_i^{-\ell}.
    \end{equation}
    
    Write $m=n^{L+\gamma}$, where $0\le\gamma<1$ and $L$ is an integer. If $\gamma = 0$, then the two terms in the minimum are equal. 
    Otherwise, the first term in the minimum is achieved for~$L+\gamma < k-\ell$, which is equivalent to~$\ell\le k-L-1$, and the second term is achieved for~$\ell\ge k-L$. Substituting in $p_i = 2^{-i}$, it follows from~\eqref{var-Zij-4} that
    \begin{equation}\label{eq:chebyshev-split-sum}
        \var(Z_{i,j}) \le 2^{k-2ik}m^2 \sum_{\ell=1}^{k-L-1} 2^{i\ell} + 2^{k-2ik}mn^k\sum_{\ell=k-L}^{k-1}(2^i/n)^\ell.
    \end{equation}

    Recall~$0\le i \le \log n$.
    We observe that the terms in the first sum in \eqref{eq:chebyshev-split-sum} are non-decreasing in~$\ell$, so the largest term is achieved for~$\ell=k-L-1$.
    Conversely, the terms of the second sum are non-increasing in~$\ell$, which means that the largest term is achieved for~$\ell=k-L$.
    Thus, \eqref{eq:chebyshev-split-sum} implies that
	\begin{align}\label{eq:uncol-chebyshev-2}
        \var(Z_{i,j})&\le
        k
        2^{-iL-ik+k} m
        \cdot
        \max\set{2^{-i}m,n^{L}}
        =k2^{-iL-ik+k}m^2
        \max\set{2^{-i},n^{-\gamma}}
        \,.
    \end{align}

    For all $i$, since the $Z_{i,j}$'s are i.i.d.\ and $\E(Z_i) = m$ holds, by~\eqref{eq:uncol-chebyshev-2} we have
	\begin{align}\label{eq:uncol-chebyshev-3}
	    \var(Z_i) &= \frac{\var(Z_{i,1})}{t_ip_i^{2k}} \le
        \frac{k2^{i(k-L)+k}}{t_i}
        \cdot\E(Z_i)^2
        \cdot
        \max\set{2^{-i},n^{-\gamma}}.
	\end{align}
    Now, recall from line~\ref{line:set-Li} of \UncolApprox{} that~$\hat L\coloneqq\hat L_i\in\N$ and $\hat\gamma\coloneqq\hat\gamma_i\in[0,1)$ are defined in such a way that~$n^{\hat L+\hat \gamma}=2^{ik}$ holds.
    Since $i\le i^*$ we have $2^{ik}\le m$, so either we have $\hat L=L$ and $\hat\gamma\le\gamma$ or we have $\hat L\le L-1$.
    We now claim that in either case, $2^{i(k-L)}\cdot\max\set{2^{-i},n^{-\gamma}} \le F_i$, where $F_i$ is defined as in line~\ref{line:set-Fi} of \UncolApprox{}.
    
    In the first case, where $\hat L=L$ and $\hat\gamma\le \gamma$, we have:
    \begin{equation}\label{eq:2ikl-bound-first}
        2^{i(k-L)} \cdot \max\set{2^{-i},n^{-\gamma}}
        \le 2^{i(k-\hat L)} \cdot \max\set{2^{-i},n^{-\hat\gamma}}
        =F_i
    \end{equation}
    as claimed. In the second case, we have $\hat L \le L-1$ and $2^i = n^{(\hat{L}+\hat{\gamma})/k}$, so
    \begin{align}\nonumber
        2^{i(k-L)} \cdot \max\set{2^{-i},n^{-\gamma}}
        &
        \le 2^{i(k-\hat L-1)} \cdot 1
        = n^{(\hat L+\hat\gamma)(k-\hat L)/k}\cdot 2^{-i}
        \\\label{eq:2ikl-bound-second}
        &\le n^{(\hat L+\hat\gamma)(k-\hat L)/k}\cdot \max\set{2^{-i},n^{-\hat\gamma}}=F_i
    \end{align}
    again as claimed. Combining \eqref{eq:2ikl-bound-first} and \eqref{eq:2ikl-bound-second} and using our choice of $t_i$ in line~\ref{line:set-ti} of \UncolApprox{}, we can continue our calculation from \eqref{eq:uncol-chebyshev-3} to arrive at our final variance bound of
    \begin{align*}
        \var(Z_i) &
        \le
        \frac{k2^k F_i}{t_i}
        \cdot\E(Z_i)^2
        \le\frac{\eps^2\cdot\E(Z_i)^2}{10k\log n}.
    \end{align*}
    By Chebyshev's inequality, it follows that
	\[
	    \pr\Big(|Z_i - \E(Z_i)| \ge \eps \E(Z_i)\Big) \le \frac{\var(Z_i)}{\eps^2\E(Z_i)^2} \le \frac{1}{10k\log n}.
	\]
	Since $\E(Z_i) = m = e(G)$, the claim follows from the definition of $Z_i$.
\end{proof}

We now prove correctness of \UncolApprox{}.

\begin{lemma}\label{lem:uncol-halt-in-istar}
    With probability at least $2/3$,
    $\UncolApprox(\indora(G),\eps)$ returns an $\eps$-approximation to $e(G)$.
\end{lemma}
\begin{proof}
    If $\eps < n^{-k/2}$ or $n\le k^5$, then the exhaustive search in line~\ref{line:uncol-exhaustive} produces the exact number~$e(G)$. Now suppose $\eps \ge n^{-k/2}$ and $n\ge k^5$. We consider the following events for an execution of $\UncolApprox(\indora(G),\eps)$:
    \begin{description}
        \item[$\calE_1$:] For all $i \le i^*$, either \UncolApprox{} does not reach iteration $i$ or the value $\frac{1}{p_i^kt_i}\sum_{j=1}^{t_i} e(G[U_{i,j}])$ is an $\eps$-approximation to $e(G)$.
        \item[$\calE_2$:] \SparseCount never returns \RTE, so that \UncolApprox{} does not return \RTE in line~\ref{line-RTE-sparse}.
        \item[$\calE_3$:] All calls to $\SampleSubset$ have combined runtime at most $R\coloneqq\Csample\sum_{i=0}^{\log n-1}t_i(i+n/2^i)$, so that \UncolApprox{} does not return \RTE{} in line~\ref{line-call-samplesubset}.
        \item[$\calE_4$:] All sets $U_{i,j}$ satisfy $|U_{i,j}| \le z_i \coloneqq \max\{7p_in, 7k\ln n\}$, so that \UncolApprox{} does not return \RTE in line~\ref{line-RTE-large-set}.
    \end{description}

    Let $\calE = \calE_1 \cap \calE_2 \cap \calE_3 \cap \calE_4$. We will first prove that $\pr(\calE) \ge 2/3$, and then show that if $\calE$ occurs then \UncolApprox halts by iteration $i^*(G)$ and returns a valid $\eps$-approximation of $e(G)$ as required. 
    
    For all $i \le \log n - 1$, observe that $(k-\hat{L}_i)/k \le 1$ and $\max\{2^{-i},n^{-\hat\gamma_i}\} \le 1$; thus $F_i \le e(G)\le n^k$. Since $\eps \ge n^{-k/2}$ and $k \le n^{1/5}$, it follows that
    \[
        t_i \le 12\eps^{-2}k^22^kn^k\log n \le 12n^{3k+2/5}\log n.
    \]
    Thus \SparseCount{} is called at most $12n^{5k}$ times in total. Since we call \SparseCount{} with parameter $\delta = 1/(120n^{5k})$, it follows from Lemma~\ref{lem:sparse-count} and a union bound that
    \begin{equation}\label{eq:uncol-halt-E1}
        \pr(\calE_1) \ge 9/10.
    \end{equation}
    
    By \cref{lem:uncoloured-chebyshev} and a union bound over all $i$ with $0\le i \le i^* \le \log n-1$, 
    \begin{equation}\label{eq:uncol-halt-E2}
        \pr(\calE_2) \ge 1 - \frac{i^*+1}{10\log n} \ge \frac{9}{10}.
    \end{equation}
    
    Recall from \cref{def:Csample} that the expected running time of a call to $\SampleSubset(n,i)$ is at most $\Csample(i+n/2^i)$; thus the expected running time of all such calls is at most $\Csample \sum_{i=0}^{\log n-1}t_i(i+n/2^i)$, so by Markov's inequality we have 
    \begin{equation}\label{eq:uncol-halt-E3}
        \pr(\calE_3) \ge 19/20.
    \end{equation}
    
    The number of queries that the algorithm makes to \SparseCount{} and the number of sets $U_{i,j}$ is at most $\sum_{i=0}^{\log n-1}t_i \le n^k\log n \le n^{2k}$.
    Moreover, by a Chernoff bound (\cref{lem:chernoff-high}), the probability that an individual $\abs{U_{i,j}}$ is larger than $z_i$ is at most $e^{-z_i}\le n^{-7k}$.
    Thus a union bound yields
    \[
        \Pr[\calE_4] \ge 1-n^{2k}\cdot n^{-7k}\ge 11/12\,.
    \]
    It follows from~\eqref{eq:uncol-halt-E1}, \eqref{eq:uncol-halt-E2}, \eqref{eq:uncol-halt-E3} and a union bound that
    \begin{equation}\label{eq:uncol-halt-E}
        \pr(\calE) \ge 2/3.
    \end{equation}
    
    If~$\calE$ occurs, then any iteration~$i\le i^*$ in which \UncolApprox{} halts will return an $\eps$-approximation of~$e(G)$. (Indeed, \UncolApprox{} cannot return $\RTE$, so it halts in a given iteration if and only if \SparseCount{} does not return \TooDense{}. In this case we have $C_{i,j} = e(G[U_{i,j}])$ for all $j$, so since $\calE_1$ occurs it outputs a valid $\eps$-approximation.) We now claim that if $\calE$ occurs and \UncolApprox{} reaches iteration $i^*$, then \SparseCount{} does not return \TooDense{} in iteration~$i^*$. Given this claim, it follows immediately that \UncolApprox{} halts by iteration $i^*(G)$ and returns a valid $\eps$-approximation; thus the result follows from~\eqref{eq:uncol-halt-E}.

    Suppose $\calE$ occurs, and that \UncolApprox{} reaches iteration $i^*$. Then by the definition of $i^*$, we have $2^{-i^*k}m=p_{i^*}^k m \ge 1$ and $2^{-(i^*+1)k} m < 1$. This implies $p_{i^*}^k m\le 2^{k}$. Since $\calE_1$ occurs and $\eps \le 1$, we have
    \[
        \sum_{j=1}^{t_{i^*}} e(G[U_{i^*,j}]) \le p_{i^*}^k t_{i^*}\cdot2 m \le 2^{k+1}t_{i^*} = M_{i^*}\,.
    \]
    In iteration $j$ of the inner for loop of \UncolApprox{}, we will therefore have 
    \[
        M_{i^*,j} = M_i^* - \sum_{\ell=1}^{j-1} e(G[U_{i^*,\ell}]) \ge e(G[U_{i^*,j}]),
    \]
    Thus the $j$\th \SparseCount{} call sees at most~$M_{i^*,j}$ edges and will thus avoid returning~\TooDense{} as required. This concludes the proof.
\end{proof}

\subsubsection{Running time and oracle cost of UncolApprox}\label{sec:uncol-run-time}

We now proceed to bound the running time and oracle cost of \UncolApprox{}. 
In the following, $\cost(\abs{S})$ will give the oracle cost of a set $S \subseteq V(G)$. Note that we do not require $\cost$ to be regularly-varying in the lemma below.

\def\uglyfactor{\eps^{-2} 2^{5k} \log^5 n}
\begin{lemma}\label{lem:approx-uncol-running}
    Let $\cost=\{\cost_k\colon k \ge 2\}$ be an arbitrary cost function with parameter~$k$. Suppose that $\cost_k(x) \le x^k$ for all~$k$ and~$x$ and that there exists $x_0$ such that for all~$k$, $\cost_k$ is non-decreasing on $[x_0,\infty)$. Let~$T(G,\eps)$ be the worst-case running time required to run~$\UncolApprox(\indora(G),\eps)$ and let~$C(G,\eps)$ be the worst-case oracle cost incurred under $\cost$. Then for every~$n$-vertex $k$-hypergraph~$G$ we have
    \begin{align}
        \label{eq:uncolapprox-time-lemma}
        T(G,\eps)&\le \OO\big(k^{5k} + \eps^{-2}2^{5k}\log^5n \cdot n^{1 + g(k,1)}\big),\\
        \label{eq:uncolapprox-cost-lemma}
        C(G,\eps)&\le \OO\Big(k^{7k} + \eps^{-2}2^{5k}\log^5n\cdot \max_i \big(F_i\cost_k(\max\{n/2^i,n^{1/4}\})\big)\Big)\,.
    \end{align}
\end{lemma}
As a special case, if $\cost_k(n)=1$ for all $k,n\in\N$, then the bound on the cost implies that \UncolApprox{} makes at most
\(\OO\paren*{\uglyfactor{}\cdot n^{k/4}}\)
queries to the oracle in the worst case.

\begin{proof}
    We first consider the resource requirements of line~\ref{line:uncol-exhaustive} by splitting into three subcases.
    
    \medskip\noindent\textbf{Case 1:} $n \le k^5$. In this case, the running time and query count of exhaustive search are both $\OO(\binom{n}{k})\le \OO(\paren*{en/k}^{k})=\OO(k^{5k})$, and each query has cost given by $\cost_k(k) = \OO(k^{k})$. This is accounted for by the additive $k^{5k}$ and $k^{7k}$ terms on the right sides of \cref{eq:uncolapprox-time-lemma,eq:uncolapprox-cost-lemma}.
    
    \medskip\noindent\textbf{Case 2:} $n^k \le \eps^{-2}$. In this case, the running time and query count of exhaustive search are both $\OO(n^k) = \OO(\eps^{-2})$, and each query is of size $k$. The running time is therefore dominated by the second summand on the right side of \cref{eq:uncolapprox-time-lemma}. Since $k \le n$ and cost is eventually non-decreasing, we have $\cost_k(k) = \OO(\cost_k(n))$;
    since $F_0 = 1$, it follows that the cost is dominated by the second summand on the right side of \cref{eq:uncolapprox-cost-lemma}.
    
    \medskip\noindent
    Thus in all cases our bounds on running time and cost are satisfied, and we may assume $n^{k}>\eps^{-2}$ and $n>k^5$ so that \UncolApprox does not halt in line~\ref{line:uncol-exhaustive}. 
    
    Observe that since $n$ is a power of two, we can quickly calculate $\hat{L}_i$ and $F_i = 2^{i(k-\hat{L}_i)}$, and hence also $t_i$ and $M_i$; we do not need to calculate $\hat{\gamma}_i$. The running time and cost are therefore dominated by lines~\ref{line-call-samplesubset} and~\ref{line-call-sparsecount}.
    
    For brevity, let $I = \log(n)-1$. In line~\ref{line-call-samplesubset}, we abort execution in line~\ref{line-call-samplesubset} if the running time gets too large, so the total running time of all calls to \SampleSubset{} is at most 
    \[
        \sum_{i=0}^{I} t_i \cdot \OO(i p_i n) = \OO\Big(n\log^2 n \cdot \max\big\{t_i p_i\colon i\in\{0,\dots,I\}\big\}\Big)\,.
    \]
    Expanding via $t_i p_i=\lceil \eps^{-2} 10k^2 2^k \log n \cdot F_i \rceil  2^{-i}$ and using the bound $F_i 2^{-i}\le n^{g(k,1)}$ from \cref{lem:uncol-bound-ti} (applied with $\beta=1$), we immediately see that this is at most 
    \[
        \OO(\eps^{-2}k^22^{k}\log^3n\cdot n^{1+g(k,1)})\,,
    \]
    which is dominated by the second summand on the right side of \cref{eq:uncolapprox-time-lemma}.

    It remains to analyse line~\ref{line-call-sparsecount}, starting with the $\min\{M_{i,j},e(G[U_{i,j}])\}$ term in both the cost and the running time of \SparseCount{} according to \cref{lem:sparse-count}. This running time is dominated by the case in which \UncolApprox does not halt early by returning \RTE{}, so assume this is the case. Line~\ref{line-call-sparsecount} calls $\SparseCount(\indora(G),U_{i,j}, k, M_{i,j},n^{-5k}/120)$, and the inner loop over~$j$ aborts once such a call returns~\TooDense{}. Let $j_i^*\in\set{1,\dots,t_i}$ be the last iteration of the inner loop.
    For all~$j$ with~$j<j_i^*$, the return value is $C_{i,j} = e(G[U_{i,j}])$. It follows that
    \begin{align*}
        \sum_{j=0}^{j_i^*} \min\{M_{i,j},e(G[U_{i,j}])\} 
        &\le \sum_{j=0}^{j^*-1} e(G[U_{i,j}]) + M_{i,j^*}\\
        &= \sum_{j=0}^{j^*-1} e(G[U_{i,j}]) + \Big(M_i - \sum_{j=0}^{j^*-1}C_{i,j}\Big)
        = M_i.
    \end{align*}
    Moreover, since we do not return \RTE{}, we have $|U_{i,j}| \le \max\{7p_in,7k\ln n\} \le n$ for all $i$. Thus by \cref{lem:sparse-count}, the total running time of all iterations of line~\ref{line-call-sparsecount} is
    \[
        \OO\Big(\sum_{i=0}^{I}\sum_{j=0}^{j^*}\log(n^{5k}/120) e^{2k}z_i\log^2 z_i \min\{M_{i,j},e(G[U_{i,j}])\}\Big) = \OO\Big(ke^{2k}\log^4 n \cdot \max_i (z_iM_i)\Big)\,.
    \]
    Substituting in the definitions of $z_i$ and $M_i$, we obtain a running time of at most
    \[
        \OO\Big(\eps^{-2}2^{5k}\log^5 n \cdot \max_i (F_i\max\{n/2^i, k\ln n\})\Big)\,.
    \]
    Recall that $k \le n^{1/5}$, so $k\ln n \le n^{1/4}$; it follows by \cref{lem:final-F-bound}, applied with $\beta=1$, that this expression is dominated by the right side of \cref{eq:uncolapprox-time-lemma}.
    
    Observe that line~\ref{line-call-samplesubset} has no oracle cost, so we now bound the total oracle cost of line~\ref{line-call-sparsecount}. Arguing exactly as with the running time, the total number of queries in the $i$\th iteration is at most
    \[
        \OO\Big(\sum_{j=0}^{j^*}\log(n^{5k}/120) e^{2k}\log^2 z_i \min\{M_{i,j},e(G[U_{i,j}])\}\Big) = \OO\Big(\eps^{-2}2^{5k}\log^4n\cdot F_i\Big)\,.
    \]
    Since $\cost_k$ is eventually non-decreasing and $|U_{i,j}| \le z_i$ for all $j$, the cost of any query in the $i$\th iteration is at most
    \[
        \max_j (\cost_k(|U_{i,j}|)) = \OO(\cost_k(z_i)) = \OO(\cost_k(\max\{n/2^i,k\ln n\})\,.
    \]
    It follows that the total oracle cost over all iterations is at most
    \begin{align*}
        &\OO\Big(\sum_{i=0}^{I}\eps^{-2}2^{5k}\log^4n\cdot  F_i\cost_k\big(\max\{n/2^i,k\ln n\}\big) \Big)\\
        &\qquad\qquad\qquad\qquad= \OO\Big(\eps^{-2}2^{5k}\log^5n \max_i\big(F_i\cost_k(\max\{n/2^i,n^{1/4}\})\big) \Big)\,.
    \end{align*}
    This is dominated by the right side of \cref{eq:uncolapprox-cost-lemma} as required.
\end{proof}

\begin{theorem}\label{thm:uncolapprox-algorithm}
    Let $\cost = \{\cost_k\colon k\ge 2\}$ be a regularly-varying parameterised cost function with parameter $k$, index $\alpha_k\in[0,k]$, and slowly-varying component $\sigma$.
    There is a randomised \indora-oracle algorithm $\textnormal{\texttt{Uncol}}(\indora(G),\eps,\delta)$
    with worst-case running time
    \begin{equation}\label{eq:uncolapprox-time-final}
        \OO\Big(\big(k^{5k}+\eps^{-2}2^{5k}\log^5n\cdot n^{1+g(k,1)}\big)\log(1/\delta)\Big),
    \end{equation}
    worst-case oracle cost
    \begin{equation}\label{eq:uncolapprox-cost-final-weak}
        \OO\Big(\big(k^{7k}+\eps^{-2}2^{5k}\log^5n\cdot n^{g(k,\alpha_k)+ o(1)}\cost_k(n)\big)\log(1/\delta)\Big)\,,
    \end{equation}
    and the following behaviour:
    Given an $n$-vertex $k$-hypergraph~$G$ and rationals $\eps,\delta \in (0,1)$,
    the algorithm outputs an integer $m$ that, with probability at least $1-\delta$, is an $\eps$-approximation to~$e(G)$.
    
    Moreover, if either $\sigma$ is eventually non-decreasing or there exists $\eta>0$ such that $\alpha_k \ge k-1+\eta$ for all $k$, then the worst-case oracle cost is
    \begin{equation}\label{eq:uncolapprox-cost-final-strong}
        \OO\Big(\big(k^{7k}+\eps^{-2}2^{5k}\log^5n \cdot n^{g(k,\alpha_k)} \cost_k(n)\big)\log(1/\delta)\Big)\,.
    \end{equation}
\end{theorem}
\begin{proof}
    By \cref{lem:uncol-halt-in-istar}, \UncolApprox{} has success probability at least~$2/3$.
    Applying standard boosting (\cref{lem:median-boosting}) to \UncolApprox{}, yields the algorithm with success probability~$1-\delta$ claimed here.
    By \cref{lem:cost-monotone}, the eventually non-decreasing cost assumption in \cref{lem:approx-uncol-running} is satisfied and thus the worst-case bounds on running time (\cref{eq:uncolapprox-time-lemma}) and cost (\cref{eq:uncolapprox-cost-lemma}) follow.
    The time bound in \cref{eq:uncolapprox-time-final} is identical to the one in \cref{eq:uncolapprox-time-lemma}, the cost bound \cref{eq:uncolapprox-cost-final-weak} follows from \cref{eq:uncolapprox-cost-lemma} by applying \cref{cor:uncol-final-cost-bound}(i), and the cost bound \cref{eq:uncolapprox-cost-final-strong} follows from \cref{eq:uncolapprox-cost-lemma} by applying \cref{cor:uncol-final-cost-bound}(ii).
\end{proof}

\begin{figure}
    \begin{center}
        \includegraphics{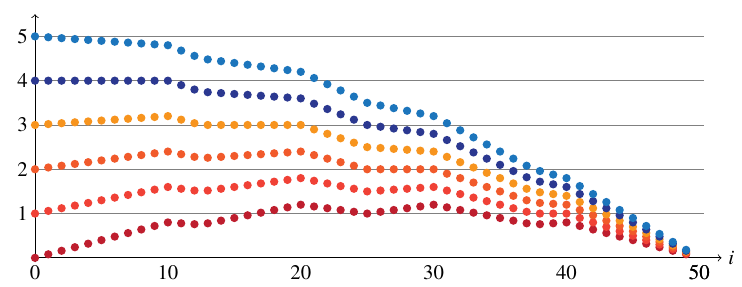}
    \end{center}
\caption{\label{fig:uncolapprox_overall_cost}%
The expected overall cost of the $i$\th iteration of \UncolApprox{} is depicted in the scenario where~$n=2^{50}$, $k=5$, and the cost function is one of six functions with~$\cost(\tau)=\tau^\alpha$ for $\alpha\in\set{0,\dots,k}$.
Plotted is the exponent~$g$ of the main~$n^g$ term of the cost of iteration~$i$; we have $g=\log_n(F_i 2^{-i\alpha} n^\alpha)$.
As can be readily seen from the figure, increasing $\alpha$ shifts the most expensive iteration to the left. Thus, if we have prior knowledge about the rough number of edges, we may be able to improve the cost of our algorithm. For example, if we know that the graph is very dense so that $i^*\ge 0.8\log n$ holds, then we can start the algorithm in iteration $i=0.8\log n$, which improves the cost.
We may similarly improve the cost analysis if $\alpha$ is small compared to $k$ and the graph is sparse, so that, for example, $i^*\le 0.1\log n$ holds.
Interestingly, when $\alpha\ge k-1$, the hardest instances seem to be the ones that are very sparse.}
\end{figure}

\subsection{Lower bounds on oracle algorithms for edge estimation}%
\label{sec:uncol-lower}

In this section, we unconditionally prove that the oracle cost achieved by \UncolApprox{} is essentially optimal.
We do this by proving a lower bound on the $\indora$-decision tree complexity of approximate edge-counting.
More specifically, we construct two correlated random $n$-vertex $k$-hypergraphs~$\calG_1$ and $\calG_2$ that, with high probability, have a significantly different number of edges.
This means that any approximate edge-counting algorithm~$\calA$ will distinguish them with high probability.
On the other hand, we also prove that any deterministic $\indora$-oracle algorithm~$A$ that can distinguish between $\calG_1$ and~$\calG_2$ must incur a large cost.
We formalize this discussion as follows:

\def\nkrepsassumption{Let $n,k,r$ be integers with $\sqrt{n}/10^4 \ge k \ge r \ge \alpha_k \ge 1$. Let $\eps\in(0,1)$ satisfy $240k!/n^r \le \eps$.}%

\def\nkepsassumption{Let $n,k$ be integers with $\sqrt{n}/10^4 \ge k \ge 1$. Let $\eps\in(0,1)$ satisfy $240k!/n \le \eps$.}

\begin{restatable}{theorem}{uncoldistinguisher}\label{thm:uncol-lower-distinguisher}
  Let $\cost_k(n) = n^{\alpha_k}$, where $\alpha_k \in [0,k]$. 
  \nkrepsassumption{}
  There exist two correlated distributions $\cG_1$ and $\cG_2$ on $n$-vertex $k$-hypergraphs with the following properties:
    \begin{enumerate}[(i)]
        \item\label{item:uncol-edge-gap}%
        We have $\Pr_{(G_1,G_2)\sim(\calG_1,\calG_2)} [e(G_2) \ge (1+\eps)e(G_1)] \ge 0.95$.
        \item\label{item:uncol-distinguisher}%
        If $A$ is a deterministic $\indora$-oracle algorithm with
        \begin{align}\label{eq:A-is-distinguisher}
            \pr_{(G_1,G_2)\sim(\calG_1,\calG_2)}
            \paren[\bigg]{A(\indora(G_1))\ne A(\indora(G_2))}
            &\ge 2/3\,,
        \end{align}
        then the expected oracle cost of $A$ (with respect to $\cost$) under random inputs~$G_1\sim\calG_1$ satisfies
        \[
            \E_{G_1\sim\calG_1}[\cost(A,G_1)]
            \ge n^{r(k-r+\alpha_k)/k}/(1080k^{3k}\eps^{(r-\alpha_k)/k})
            \,.
        \]
    \end{enumerate}
\end{restatable}
Before we prove this theorem, let us apply the minimax theorem (\cref{thm:algorithm-minmax}) to it, in order to derive our main lower bound for $\indora$-oracle algorithms.
\begin{theorem}\label{thm:uncol-lower-main-full}
    Let $\cost_k(n) = n^{\alpha_k}$, where $\alpha_k \in [0,k]$. Let $\calA$ be an \indora-oracle algorithm such that for all $k$-hypergraphs $G$, $\calA(\indora(G))$ is a $\tfrac12$-approximation to $e(G)$ with probability at least $9/10$. Then $\calA$ has worst-case expected oracle cost at least $\Omega(n^{\alpha_k+g(k,\alpha_k)}/k^{3k})$ as $n\to\infty$, where $k$ may depend on $n$ in an arbitrary fashion.
\end{theorem}
Observe that unlike in \cref{sec:uncol-upper}, we do not require $\{\cost_k\colon k \ge 2\}$ to be a regularly-varying parameterised cost function.

\begin{proof}
    First suppose that $k > \sqrt{n}/10^4$, so that \cref{thm:uncol-lower-distinguisher} does not apply. In this case, we have $k^{3k} = \Omega(n^{3k/2})$. Moreover, by \cref{lem:uncol-bound-g} and since $\alpha_k \le k$, we have
    \[
        \alpha_k + g(k,\alpha_k) \le \alpha_k + \frac{(k - \alpha_k)^2}{4k} = \frac{(k+\alpha_k)^2}{4k} \le \frac{(2k)^2}{4k} = k\,.
    \]
    It follows that $n^{\alpha_k + g(k,\alpha_k)}/k^{3k} = o(1)$. Since $\alpha_k \ge 0$, even a single query to $\indora(G)$ requires $\Omega(1)$ cost and so the result is vacuously true.
    
    For the rest of the proof, we may assume $k \le \sqrt{n}/10^4$. Suppose $G_1,G_2$ are $n$-vertex $k$-hypergraphs that satisfy $e(G_2)\ge 3e(G_1)/2$.
    With probability at least $(9/10)^2\ge 4/5$ over the random choices of~$\calA$, the algorithm~$\calA$ will correctly output a $(1/2)$-approximation for both graphs $G_1$ and $G_2$.
    This implies that~$\calA$ \emph{distinguishes} $G_1$ and $G_2$ in the sense that the following holds:
    \[
        \Pr_{A\sim\calA}
        \paren[\bigg]{
            A(\indora(G_1)) \ne A(\indora(G_2))
        }
        \ge 4/5\,.
    \]
    By a union bound, taking $\eps=1/2$ in \cref{thm:uncol-lower-distinguisher}\ref{item:uncol-edge-gap} implies
    \begin{align}\label{eq:mixed-probs}
        \Pr_{\substack{A\sim\calA\\(G_1,G_2)\sim(\calG_1,\calG_2)}}
        \paren[\bigg]{
            A(\indora(G_1)) \ne A(\indora(G_2))
        }
        &\ge 3/4\,.
    \end{align}
    
    Let $F$ be the family of \indora-oracle algorithms~$A$ that satisfy \eqref{eq:A-is-distinguisher}. By~\eqref{eq:mixed-probs}, we have
    \begin{align*}
        \tfrac34 
        &\le \Pr_{\substack{A\sim\calA\\(G_1,G_2)\sim(\calG_1,\calG_2)}}\Big(A(\indora(G_1)) \ne A(\indora(G_2))\Big) 
        \le \pr_{A\sim\calA}(A \in F) + \tfrac23\pr_{A\sim\calA}(A\notin F) \\
        &= \tfrac13\pr_{A\sim\calA}(A \in F) + \tfrac23\,,
    \end{align*}
    and so $\pr_{A\sim\cA}(A \in F) \ge 1/4$. Thus by minimax (i.e.\ \cref{thm:algorithm-minmax} taking $p=1/4$), we can lower-bound the worst-case expected oracle cost of $\cA$ by the worst-case expected oracle cost of any deterministic $A \in F$:
    \[
        \max_G \E_{A\sim\calA}[
            \cost(A,G)
        ]
        \ge
        \tfrac14 \inf_{A\in F}\E_{G_1\sim\calG_1}[\cost(A,G_1)]\,.
    \]
    By \cref{thm:uncol-lower-distinguisher}\ref{item:uncol-distinguisher}, it follows that
    \[
        \max_G \E_{A\sim\calA}[
            \cost(A,G)
        ] = \Omega(n^{r(k-r+\alpha_k)}/k^{3k})\,.
    \]
    The result now follows immediately from \cref{lem:final-F-bound}.
\end{proof}

\subsubsection{\texorpdfstring{$\bm{\cG_1}$}{G\_1} and \texorpdfstring{$\bm{\cG_2}$}{G\_2}: Choosing a hard input distribution}

Our first step in the proof of \cref{thm:uncol-lower-distinguisher} shall be to define the random graphs $\cG_1$ and~$\cG_2$ and prove that $e(G_2)\ge (1+\eps) e(G_1)$ holds with probability at least $0.95$.

Recall from the discussion in \cref{sec:uncolapprox} that our algorithmic approach is to randomly sample induced subgraphs and apply the independence oracle to these induced subgraphs. Recall also that one of the possible ``worst cases'' in which the number of samples required to see an edge is maximised is a single-rooted star graph, in which all edges intersect in a single $r$-vertex ``root''. Consider the effects of combining an Erd\H{o}s-R\'{e}nyi $k$-hypergraph $H_1$ with a single-rooted star graph $H_2$ of similar density.  Intuitively, we would expect that any large random set $S$ is likely to contain an edge from $H_1$, so independence oracles will return the same result for $H_1[S]$ and $(H_1\cup H_2)[S]$. At the same time, any small random set $S$ is unlikely to contain all of $R$, and will therefore also not distinguish $H_1$ from $H_1 \cup H_2$. This motivates our choice of $\cG_1$ and $\cG_2$, which we define as follows; note that for technical reasons we allow $H_2$ to contain multiple roots.

\begin{defn}\label{def:uncol-lb-graphs}
    \nkrepsassumption{}
    We define the following probabilities:
    \begin{align*}
        p_1(n,k,r,\eps) \coloneqq k!/(\eps n^r),\qquad\qquad
        p_2(n,k,r,\eps) \coloneqq 240 r!/n^r.
    \end{align*}
    We define $H_1(n,k,r,\eps)$ to be an Erdős--Rényi random $k$-hypergraph on the vertex set $[n]$, where each edge occurs independently with probability $p_1(n,k,r,\eps)$. We define $\calR(n,k,r,\eps)$ to be a random subset of $[n]^{(r)}$, where each size-$r$ set is included independently with probability $p_2(n,k,r,\eps)$; we refer to these size-$r$ sets as \emph{roots}. We define $H_2(n,k,r,\eps)$ to be the {$k$-hypergraph} with edge set 
    \[
        \setc[\Big]{e \in [n]^{(k)}}{e \supseteq R \text{ for some }R \in \calR}\,.
    \]
    Finally, we define
    \begin{align*}
        \calG_1(n,k,r,\eps) &= H_1(n,k,r,\eps),\\
        \calG_2(n,k,r,\eps) &= H_1(n,k,r,\eps) \cup H_2(n,k,r,\eps).
    \end{align*}
    Generally the values of $n$, $k$, $r$ and $\eps$ will be clear from context, and in this case we will omit them from the notation.
\end{defn}

The following lemma establishes \cref{thm:uncol-lower-distinguisher}\ref{item:uncol-edge-gap} for $\calG_1$ and $\calG_2$.
\begin{lemma}\label{lem:G1-G2-gap}
    \nkrepsassumption{}
    Then with probability at least $0.95$, we have $e(\calG_2) > (1+\eps)e(\calG_1)$.
\end{lemma}
\begin{proof}
    Observe that
    \begin{align}\label{eq:H1-H2-edges-gen-0}
        \frac{|e(\calG_2) - e(\calG_1)|}{e(\calG_1)} = \frac{e(H_2) - |E(H_1) \cap E(H_2)|}{e(H_1)}.
    \end{align}
    To prove the lemma, we must prove that the left side of~\eqref{eq:H1-H2-edges-gen-0} is larger than $\eps$ with probability at least $0.95$; to this end, we use Chernoff bounds to bound $e(H_1)$ and $|E(H_1) \cap E(H_2)|$ above, and $e(H_2)$ below. 
    
    \textbf{Upper bound on \boldmath$e(H_1)$.}
    By linearity of expectation, $\E(e(H_1)) = p_1\binom{n}{k} \le p_1n^k/k! = n^{k-r}/\eps$. Since $\eps< 1$, we have $n^{k-r}/\eps>n^{k-r}\ge 1$.
    By the Chernoff bound of \cref{lem:chernoff-high} applied with $z=7 n^{k-r}/\eps$, we have
    \begin{equation}\label{eq:H1-H2-edges-gen-1}
        \pr\big(e(H_1) \le 7 n^{k-r}/\eps \big) \ge 1 - e^{-7 n^{k-r}/\eps} \ge 1 - e^{-7} \ge 0.99\,.
    \end{equation}
    
    \textbf{Upper bound on \boldmath$|E(H_1) \cap E(H_2)|$ when \boldmath$e(H_2) \ge 30n^{k-r}$.}
    Any edge of $H_2$ is included in $E(H_1)$ independently with probability $p_1$. By assumption on $n,r,\epsilon$, we have $p_1\le 1/240\le 1/14$, so conditioned on $e(H_2)=M$ for any $M\in\N$, we have
    \[\E\Big[|E(H_1) \cap E(H_2)| \,\Big\vert\, e(H_2)=M\Big]=p_1 M \le M/14\,.\]
    By the Chernoff bound of \cref{lem:chernoff-high} applied with $z = M/2$, we have for all $M\ge 30n^{k-r}$:
    \begin{equation}\label{eq:H1-H2-edges-gen-3}
        \pr\Big(\abs{E(H_1) \cap E(H_2)} \le e(H_2)/2 \,\Big|\, e(H_2) = M\Big) \ge 1 - e^{-7 M/14} \ge 1 - e^{-15 n^{k-r}} \ge 0.99\,.
    \end{equation}

    \textbf{Lower bound on \boldmath$e(H_2)$.}
    Let $Z$ be the set of all edges of $H_2$ that contain exactly one root $R\in\calR$, and let $\calE$ be the event that all roots $R\in\calR$ are disjoint from each other.
    We have
    \begin{equation}\label{eq:H1-H2-edges-gen-1a}
        \mbox{when $\calE$ occurs, }e(H_2) \ge |Z| = |\calR|\cdot\binom{n-r|\calR|}{k-r},
    \end{equation}
    Note that $\abs{\calR}$ is binomially distributed, and since $\sqrt{n}/10\ge k \ge r$ we have
    \[
        \E(|\calR|) = p_2\binom{n}{r} = \frac{p_2}{r!}\prod_{i=0}^{r-1}(n-i) \ge \frac{p_2}{r!}\bigg(n^r - \sum_{i=0}^{r-1} i\bigg) \ge \frac{p_2n^r}{4r!} \ge 30\,.
    \]
    Moreover, $\E(|\calR|) \le p_2n^r/r! = 240$. By the Chernoff bound of \cref{lem:chernoff-low} applied with $\delta=1/2$, it follows that
    \begin{equation}\label{eq:H1-H2-edges-gen-1b}
        \pr\big(30 \le |\calR| \le 360) \ge 1 - 2e^{- \E(|\calR|)/12} \ge 1 - 2e^{-5}  \ge 0.98\,.
    \end{equation}
    Moreover, to bound $\pr(\calE)$, we observe that there are exactly $\binom{n}{r}\binom{r}{i}\binom{n-r}{r-i}$ pairs $(S,T)$ of size-$r$ sets $S$ and $T$ with $\abs{S\cap T}=i$.
    A union bound over all possible intersecting pairs of distinct roots yields
    \begin{align*}
        \pr(\overline\calE) &\le \sum_{i=1}^{r-1}\binom{n}{r}\binom{r}{i}\binom{n-r}{r-i} \cdot p_2^2
        \le \sum_{i=1}^{r-1}\frac{n^r \cdot r^i \cdot n^{r-i}}{r! \cdot i! \cdot (r-i)!} \cdot \frac{(240r!)^2}{n^{2r}}
        = 240^2\sum_{i=1}^{r-1}\frac{r^i r!}{n^i i! (r-i)!}\,.
    \end{align*}
    Observe that $r^2/n\le k^2/n < 10^{-8}$; thus
    \begin{align*}
        \pr(\overline\calE)
        &\le 240^2\sum_{i=1}^{r-1}(r^2/n)^i \le 240^2\sum_{i=1}^{\infty}(r^2/n)^i
        = 240^2\cdot \paren*{\frac{1}{1-r^2/n} - 1}
        \le 240^2 \cdot 2r^2/n < 0.01
        \,.
    \end{align*}
    Combining this with~\cref{eq:H1-H2-edges-gen-1a,eq:H1-H2-edges-gen-1b} and a union bound, we obtain
    \[
        \pr\Bigg(e(H_2) \ge 30\binom{n-360r}{k-r} \Bigg)\ge
        \pr\paren[\Big]{\calE \mbox{ and } 30\le\abs{\calR}\le 360}
        \ge
        1-(0.01+0.02)=0.97\,.
    \]
    Combined with $\binom{n-360r}{k-r}\le n^{k-r}$, we arrive at
    \begin{equation}\label{eq:H1-H2-edges-gen-2}
        \pr\big(e(H_2) \ge 30n^{k-r} \big) \ge 0.97\,.
    \end{equation}
        
    \textbf{Conclusion of proof.}
    Combining \cref{eq:H1-H2-edges-gen-2,eq:H1-H2-edges-gen-3}, we get
    \[
        \pr\paren*{e(H_2) \ge 30n^{k-r} \mbox{ and } \abs{E(H_1) \cap E(H_2)} \le e(H_2)/2} \ge 0.99 \cdot 0.97\ge 0.96\,.
    \]
    Combined with \cref{eq:H1-H2-edges-gen-1} and a union bound, all three bounds are likely to hold:
    \[
        \pr\paren*{e(H_2) \ge 30n^{k-r} \mbox{ and } \abs{E(H_1) \cap E(H_2)} \le e(H_2)/2 \mbox{ and } e(H_1)\le 7n^{k-r}/\eps} \ge 0.95\,.
    \]
    By~\cref{eq:H1-H2-edges-gen-0}, we arrive at the following with probability at least $0.95$:
    \[
        \frac{|e(\calG_2) - e(\calG_1)|}{e(\calG_1)} = \frac{e(H_2) - |E(H_1) \cap E(H_2)|}{e(H_1)} \ge \frac{15n^{k-r}}{7n^{k-r}/\eps} > \eps.
    \]
    The result therefore follows.
\end{proof}

\subsubsection{Bounding the cost of separating \texorpdfstring{$\bm{\cG_1}$}{G\_1} and \texorpdfstring{$\bm{\cG_2}$}{G\_2}}

Throughout this section, let $n$, $k$, $r$, $\eps$ and $A$ be as in the statement of \cref{thm:uncol-lower-distinguisher}, and let $\calG_1$, $\calG_2$, $H_1$ and $H_2$ be as in \cref{def:uncol-lb-graphs}.
Note that the number of queries carried out by $A$ on $\calG_1$ is a random variable. This would lead to some uninteresting and unpleasant technicalities in the probability bounds, so we first make some easy observations that allow us to assume this number is deterministic.
\begin{remark}\label{rem:uncol-trivial}\mbox{}
\begin{enumerate}[(i)]
    \item\label{item:wlog-no-trivial-queries}
    Without loss of generality, $A$ only ever queries sets which are either empty or have size at least $k$.
    \item\label{item:indora-G1-G2}
    For all vertex sets~$S$ with $\abs{S}\ge k$, we have:
    \begin{itemize}
        \item
            $\indora(\calG_1)_S=1$ if and only if $S$ contains no edge of $H_1$.
        \item
            $\indora(\calG_2)_S=0$ if and only if $S$ contains an edge of $H_1$ or a root from~$\calR$.
    \end{itemize}
    \item\label{item:wlog-same-number-queries}
    Without loss of generality, we may assume that there is some~$t=t_{n,k}\in\N$ such that $A$ makes exactly the same number~$t$ of queries for each~$n$-vertex~$k$-hypergraph oracle input.
\end{enumerate}
\end{remark}
\begin{proof}
    To see \ref{item:wlog-no-trivial-queries}, recall from~\cref{sec:oracle-algorithms} that the $\indora$-oracle algorithm~$A$ receives~$n$ and $k$ as explicit inputs. Any time $A$ is about to query a set~$S$ of size less than~$k$, it could instead avoid performing the query and correctly assume that the answer is~$\indora(G)_S=1$. Modifying~$A$ in this way can only reduce the cost of~$A$.
    
    \ref{item:indora-G1-G2} then follows directly from the definitions of~$\calG_1$, $\calG_2$, and~$\indora$, even when $S=\emptyset$.
    
    To see \ref{item:wlog-same-number-queries}, first observe that we can enumerate the edges of an $n$-vertex $k$-hypergraph $G$ using the oracle by querying all size-$k$ subsets of $[n]$ with total oracle cost $n^k k^{\alpha_k} \le (nk)^k$. Thus without loss of generality, we may assume $\cost(A,G)\le 3(nk)^k$ for all $n$-vertex $k$-hypergraphs $G$, by running $A$ in parallel with the simple enumeration algorithm and returning the result of whichever algorithm finishes first. Under this assumption, $A$ will always carry out $x_G < 4(nk)^k$ queries on $G$ before returning, and we simply ``pad'' $A$ by inserting $4(nk)^k-x_G$ zero-cost queries to the empty set before outputting the final answer.
\end{proof}

We now set out some notation to describe the sequence of oracle queries executed by~$A$ on a given input graph.
Recall from~\cref{sec:oracle-algorithms} that $A$ is explicitly given~$n$ and~$k$ as input, and can access~$G$ only by querying individual bits of~$\indora(G)$.

\begin{defn}\label{defn:uncol-query-1}
    Let $G$ be an arbitrary $n$-vertex $k$-hypergraph. Let $q(G)$ be the sequence of oracle queries that $A$ makes when given input ${\indora(G)}$, and write $q(G) \eqqcolon (S_1(G),\dots,S_{t}(G))$. For all $i \in [t]$, let $b_i(G) = {\indora(G)}_{S_i}$ be the bit returned by the $i$\th query, and let $b(G) = (b_1(G),\dots,b_t(G))$. We refer to $(q(G),b(G))$ as the \emph{transcript of $A$ on input~$G$}. For all $i \in [t]$, let $q_{< i}(G) = (S_1(G),\dots,S_{i-1}(G))$ and $b_{<i}(G) = (b_1(G),\dots,b_{i-1}(G))$.
\end{defn}

Observe that since $A$ is a general deterministic oracle algorithm, the $i$\th query it makes will be a function of the responses to its first $i-1$ queries. As such, the total oracle cost of running $A$ on $\indora(G)$ is a function of $b(G)$. We now define notation for these concepts.

\begin{defn}\label{defn:uncol-query-2}
    For all $i \le t$ and all bit strings $\vec{b}\in\zo^t$, let $S_i(\vec{b})$ be the $i$\th query that $A$ makes after it obtained the responses $b_1,\dots,b_{i-1}$ to the first $i-1$ queries. (Note in particular that $S_i(\vec{b})$ is a deterministic function of $\vec{b}$.)
    For all $C>0$, let $X_C$ be the set of all possible response vectors that lead to a total query cost of at most $C$ for $A$, that is,
    \[
        X_C = \bigg\{\vec{b}\in\zo^t \colon \sum_{i=1}^t \cost_k(S_i(\vec{b})) \le C\mbox{ and $\vec{b}=b(G)$ for some $G$}\bigg\}\,.
    \]
    Finally, write $\calI(\vec{b}) = \{i \in [t] \colon S_i(\vec{b}) \ne \emptyset\}$, so that $|S_i(\vec{b})| \ge k$ for all $i \in \calI(\vec{b})$ by \cref{rem:uncol-trivial}\ref{item:wlog-no-trivial-queries}\,.
\end{defn}

We now state and prove \cref{lem:uncol-size}, which is the heart of the proof of \cref{thm:uncol-lower-distinguisher}\ref{item:uncol-distinguisher} and bounds the probability of distinguishing between $\calG_1$ and $\calG_2$ with queries of a given cost~$C$ above in terms of the sizes of these queries. Turning this into an upper bound in terms of $C$ then requires optimising over possible query sizes; we do this in the final proof of \cref{thm:uncol-lower-distinguisher} at the end of this section.

The intuition behind \cref{lem:uncol-size} is that $A$ can only distinguish between~$\calG_1$ and~$\calG_2$ if it ever makes a query that distinguishes them, that is, if $b(\calG_1) \ne b(\calG_2)$ holds. We will consider queries individually. If a query contains too many vertices, then it is very likely to include an edge of $H_1$, which is thus contained in both~$\calG_1$ and~$\calG_2$; if a query contains too few vertices, then it is unlikely to pick up a root of~$H_2$, in which case it will have the same result in both~$\calG_1$ and~$\calG_2$.

It is important to note that these queries are heavily dependent on each other --- indeed, independence would correspond to the special case of a \emph{non-adaptive} algorithm~$A$, which must always use the same sequence of queries regardless of the responses. As such, turning this idea into a proof requires careful handling of the conditioning on past queries. 

\begin{lemma}\label{lem:uncol-size}
    For all $C > 0$, we have
    \begin{align*}
        &\pr\Big(A(\indora(\calG_1)) \ne A(\indora(\calG_2)) \mbox{\textnormal{ and }} \sum_{i=1}^t \cost_k(S_i(\calG_1)) \le C\Big)\\
	    &\qquad\qquad\qquad\qquad\qquad\qquad\le 4^k k^k p_2 \max_{\vec{b}\in X_C} \sum_{i \in \calI(\vec{b})} \min\set[\bigg]{|S_i(\vec{b})|^r,
        \frac{1}{p_1|S_i(\vec{b})|^{k-r}}}
        \,.
    \end{align*}
\end{lemma}
\begin{proof}
    For convenience, we define shorthand for the event of the lemma statement:
    \begin{align*}
        \calE \mbox{ occurs when } A(\indora(\calG_1)) \ne A(\indora(\calG_2)) \text{ and } \sum_{j=1}^t \cost_k(S_j(\calG_1)) \le C
        \,.
    \end{align*}
    Recall that $\sum_{j=1}^t\cost_k(S_j(\calG_1)) \le C$ if and only if $b(\calG_1) \in X_C$, and that $A(\indora(\calG_1)) \ne A(\indora(\calG_2))$ can only happen when $b(\calG_1) \ne b(\calG_2)$ since $A$ is deterministic. Thus
    \begin{align*}
        \pr(\calE) &\le \pr\bigg(b(\calG_1) \in X_C \mbox{ and }\bigvee_{i=1}^t \Big\{ b_i(\calG_1) \ne b_i(\calG_2)  \Big\}\bigg)\,.
    \end{align*}
    If $\pr(b(\calG_1) \in X_C) = 0$ then the result is immediate, so we may assume without loss of generality that $\pr(b(\calG_1) \in X_C) > 0$. It follows that
    \begin{align}\label{eq:uncol-size-E0}
        \pr(\calE)&\le \pr\bigg(\bigvee_{i=1}^t \Big\{ b_i(\calG_1) \ne b_i(\calG_2)  \Big\}\,\Big|\, b(\calG_1) \in X_C\bigg)\,.
    \end{align}
    
    We next bound $\pr(\calE)$ in terms of the probabilities of distinguishing between $\calG_1$ and $\calG_2$ at each query. We apply \cref{lem:cond-exp} to the event of~\eqref{eq:uncol-size-E0}, taking $Z_i = (b_i(\calG_1), b_i(\calG_2))$, taking $\calE_i$ to be the event that $b_i(\calG_1) \ne b_i(\calG_2)$, and working conditioned on $\vec{b}(\calG_1) \in X_C$. Observe that the condition on the sets~$\calZ_i$ of the lemma is satisfied --- indeed, any transcript with $b_{<i}(\calG_1)=b_{<i}(\calG_2)$ has non-zero probability to be extended to a transcript with $b_{<t}(\calG_1)=b_{<t}(\calG_2)$, since we have $\calR=\emptyset$ with non-zero probability and, conditioned on $\calR=\emptyset$, we have $b(\calG_1) = b(\calG_2)$ with certainty.
    Thus \cref{lem:cond-exp} applies, and we have
    \begin{equation*}
	    \pr(\calE) \le \max_{\vec{b} \in X_C} \sum_{i=1}^t
    	\Pr\paren[\Big]{b_{i}(\calG_1)\ne b_{i}(\calG_2)
    	\,\Big|\, b_{<i}(\calG_1) = b_{<i}(\calG_2) = (b_1,\dots,b_{i-1})}.
	\end{equation*}
	
    For brevity, for any bit string $\vec{b} = (b_1,\dots,b_t)$, we define the event $\calB_i(\vec{b})$ to occur when $b_{<i}(\calG_1) = b_{<i}(\calG_2) = (b_1,\dots,b_{i-1})$. Moreover, conditioned on $\calB_i(\vec{b})$ we have $S_i(\calG_1) = S_i(\calG_2) = S_i(\vec{b})$, so since $E(\calG_1) \subseteq \calE(\calG_2)$ the event $b_i(\calG_1) \ne b_i(\calG_2)$ can only occur if $b_i(\calG_1) = 1$ and $b_i(\calG_2) = 0$. Hence,
    \begin{equation*}
	    \pr(\calE) \le \max_{\vec{b}\in X_C} \sum_{i=1}^t
    	\Pr\Big(b_{i}(\calG_1)=0 \mbox{ and } b_{i}(\calG_2) = 1
    	\,\Big|\, \calB_i(\vec{b})\Big).
	\end{equation*}
	Further, note that conditioned on $\calB_i(\vec{b})$, if $S_i(\vec{b}) = \emptyset$ then we have $b_i(\calG_1) = b_i(\calG_2) = b_i = 1$ with certainty. It follows that
	\begin{equation*}
	    \pr(\calE) \le \max_{\vec{b}\in X_C} \sum_{i \in \calI(\vec{b})}
    	\Pr\Big(b_{i}(\calG_1)=0 \mbox{ and } b_{i}(\calG_2) = 1
    	\,\Big|\, \calB_i(\vec{b})\Big).
	\end{equation*}
	
    By~\cref{rem:uncol-trivial}\ref{item:indora-G1-G2}, we have $b_i(\calG_1) = 1$ if and only if $S_i(\vec{b})^{(k)} \cap E(H_1) = \emptyset$, and we have $b_i(\calG_2) = 0$ if and only if either $S_i(\vec{b})^{(k)} \cap E(H_1) \ne \emptyset$ or $S_i(\vec{b})^{(r)} \cap \calR \ne \emptyset$.
    This implies
	\begin{equation}\label{eq:uncol-size-1}
	    \pr(\calE) \le \max_{\vec{b} \in X_C} \sum_{i \in \calI(\vec{b})}
    	\Pr\Big(\big(S_i(\vec{b})^{(k)} \cap E(H_1) = \emptyset\big) \mbox{ and } \big(S_i(\vec{b})^{(r)} \cap \calR \ne \emptyset\big) \,\Big|\, \calB_i(\vec{b})\Big).
	\end{equation}

	We next decompose $\calB_i(\vec{b})$ into a conjunction of events depending either only on $H_1$ or only on $\calR$; this will allow us to split each summand of~\eqref{eq:uncol-size-1} into two independent events. We define
	\begin{align*}
	    \calB_{i,1}^-(\vec{b}) &\coloneqq \bigwedge_{\substack{j \le i-1\\b_j=1}} \big(S_i(\vec{b})^{(k)} \cap E(H_1) = \emptyset\big),\\
	    \calB_{i,1}^+(\vec{b}) &\coloneqq\bigwedge_{\substack{j \le i-1\\b_j=0}} \big(S_i(\vec{b})^{(k)} \cap E(H_1) \ne \emptyset\big),\\
	    \calB_{i,2}(\vec{b}) &\coloneqq \bigwedge_{\substack{j \le i-1\\b_j=1}} \big(S_i(\vec{b})^{(r)} \cap \calR = \emptyset\big).
	\end{align*}
	Then we have $\calB_i(\vec{b}) = \calB_{i,1}^-(\vec{b}) \wedge \calB_{i,1}^+(\vec{b}) \wedge \calB_{i,2}(\vec{b})$, where $\calB_{i,1}^-(\vec{b}) \wedge \calB_{i,1}^+(\vec{b})$ depends only on~$H_1$ and $\calB_{i,2}(\vec{b})$ depends only on~$H_2$. By~\eqref{eq:uncol-size-1}, it follows that
	\begin{align}\label{eq:uncol-size-2}
	    \pr(\calE)
	    \le \max_{\vec{b}\in X_C} \sum_{i \in \calI(\vec{b})}
    	\Pr\Big(S_i(\vec{b})^{(k)} \cap E(H_1) = \emptyset\,\Big|\,\calB_{i,1}^-(\vec{b}) \wedge \calB_{i,1}^+(\vec{b})\Big)\Pr\Big(S_i(\vec{b})^{(r)} \cap \calR \ne \emptyset \,\Big|\, \calB_{i,2}(\vec{b})\Big).
	\end{align}
	
	We next deal with the ``negative'' conditioning. For all $x \in \set{0,\dots,k}$, we define
    \[
        F_i^x(\vec{b}) = S_i(\vec{b})^{(x)} \setminus \bigcup_{\substack{j \in [i-1]\\b_j=1}} S_j(\vec{b})^{(x)}\,.
    \]
    Observe that $F_i^k(\vec{b})$ is precisely the set of size-$k$ subsets of $S_i(\vec{b})$ which can still be edges of~$H_1$ conditioned on $\calB_{i,1}^-(\vec{b})$, and that $F_i^r(\vec{b})$ is precisely the set of size-$r$ subsets of $S_i(\vec{b})$ which can still be roots in $\calR$ conditioned on $\calB_{i,2}(\vec{b})$. It follows from~\eqref{eq:uncol-size-2} that
    \begin{align}\label{eq:uncol-size-3}
    \pr(\calE)
	    \le \max_{\vec{b} \in X_C} \sum_{i \in \calI(\vec{b})}
    	\Pr\Big(F_i^k(\vec{b}) \cap E(H_1) = \emptyset\,\Big|\,\calB_{i,1}^+(\vec{b})\Big)\Pr\Big(F_i^r(\vec{b}) \cap \calR \ne \emptyset \Big).
    \end{align}
	
	We next deal with the ``positive'' conditioning. Observe that the indicator function of $\calB_{i,1}^+(\vec{b})$ is a monotonically increasing function of the indicator variables of $H_1$'s edges, and the indicator function of $F_i^k(\vec{b})\cap E(H_1) = \emptyset$ is a monotonically decreasing function of these variables.
    Thus the two events are negatively correlated and therefore, by the FKG inequality (\cref{lem:FKG}) combined with~\eqref{eq:uncol-size-3}, we obtain
	\begin{align}\nonumber
	    \pr(\calE)
	    &\le \max_{\vec{b}\in X_C} \sum_{i \in \calI(\vec{b})}
    	\Pr\Big(F_i^k(\vec{b}) \cap E(H_1) = \emptyset\Big)\Pr\Big(F_i^r(\vec{b}) \cap \calR \ne \emptyset \Big)\\\label{eq:uncol-size-4}
    	&= \max_{\vec{b}\in X_C} \sum_{i \in \calI(\vec{b})} (1-p_1)^{|F_i^k(\vec{b})|}\Big(1 - (1-p_2)^{|F_i^r(\vec{b})|}\Big)
    	\le \max_{\vec{b} \in X_C} \sum_{i \in \calI(\vec{b})} e^{-p_1|F_i^k(\vec{b})|}p_2|F_i^r(\vec{b})|.
	\end{align}
	
    By definition, we have $|F_i^r(\vec{b})| \le |S_i(\vec{b})^{(r)}| \le |S_i(\vec{b})|^r$, and so we can bound each term in the right-hand side of \eqref{eq:uncol-size-4} by
    \begin{equation}\label{eq:uncol-size-bound-1}
        e^{-p_1|F_i^k(\vec{b})|}p_2|F_i^r(\vec{b})| \le p_2|S_i(\vec{b})|^r\,.
    \end{equation}
    We also provide a second upper bound on each term which will be stronger when $|S_i(\vec{b})|$ is large. To this end, we double-count the set $Z$ of pairs $(X,Y) \in F_i^r(\vec{b}) \times F_i^k(\vec{b})$ with~$X \subseteq Y$.
    Each size-$k$ set~$Y\in F_i^k(\vec{b})$ contains at most $\binom{k}{r}$ size-$r$ sets~$X\in F_i^r(\vec{b})$, and each size-$r$ set~$X\in F_i^r(\vec{b})$ is contained in exactly $\binom{|S_i(\vec{b})|-r}{k-r}$ size-$k$ sets $Y \in F_i^k(\vec{b})$; hence we obtain
	\begin{equation}\label{eq:double-counting-Fi-sets}
	    \binom{|S_i(\vec{b})|-r}{k-r} |F_i^r(\vec{b})| = |Z| \le \binom{k}{r}|F_i^k(\vec{b})| \le 2^k|F_i^k(\vec{b})|\,.
    \end{equation}
    (Recall that for all $\vec{b}$ and all $i \in \calI(\vec{b})$, we have $|S_i(\vec{b})| \ge k$.) By \cref{lem:binom-bound} (taking $a=k$ and $b=k-r$), we have $\binom{|S_i(\vec{b})|-r}{k-r} \ge |S_i(\vec{b})|^{k-r}/(2k)^k$.
    Rearranging terms in \eqref{eq:double-counting-Fi-sets} yields
	\begin{align}\label{eq:Fi-upper-2}
	    |F_i^r(\vec{b})| \le \frac{4^k k^k |F_i^k(\vec{b})|}{|S_i(\vec{b})|^{k-r}}\,.
	\end{align}
    Using this, we can bound each term of 
    \eqref{eq:uncol-size-4}
    as follows:
    \begin{align}\label{eq:uncol-size-bound-2}
        e^{-p_1|F_i^k(\vec{b})|}p_2|F_i^r(\vec{b})|
        &\overset{\eqref{eq:Fi-upper-2}}{\le}
        4^k k^k p_2  \cdot \frac{e^{-p_1|F_i^k(\vec{b})|} \cdot |F_i^k(\vec{b})|}{|S_i(\vec{b})|^{k-r}}
        \le
        4^k k^k p_2  \cdot \frac{1}{p_1 \abs{S_i(\vec{b})}^{k-r}}\,.
    \end{align}
    The second inequality here follows by observing that the expression $xe^{-p_1 x}$ is maximised by $1/(p_1e)$ at $x=1/p_1$. Thus by applying the bounds of~\eqref{eq:uncol-size-bound-1} and~\eqref{eq:uncol-size-bound-2} to the terms of~\eqref{eq:uncol-size-4}, we arrive at the claimed bound of
	\begin{align*}
	    \pr(\calE)
	    &\le 4^k k^k p_2 \max_{\vec{b}\in X_C} \sum_{i \in \calI(\vec{b})} \min\set[\bigg]{|S_i(\vec{b})|^r,
        \frac{1}{p_1|S_i(\vec{b})|^{k-r}}} \,.
        \qedhere
    \end{align*}
\end{proof}

We are now ready to prove \cref{thm:uncol-lower-distinguisher}, from which our main result \cref{thm:uncol-lower-main-full} follows as discussed earlier at the start of the section.

\uncoldistinguisher*
\begin{proof}
    We take $\calG_1$ and $\calG_2$ as in \cref{def:uncol-lb-graphs}; these satisfy \ref{item:uncol-edge-gap} by \cref{lem:G1-G2-gap}. It remains to prove \ref{item:uncol-distinguisher}.
    
    Observe that by Markov's inequality, with probability at least $2/3$, $A(\indora(\calG_1))$ uses queries with total cost at most $3Z$. Write $\calE$ for the event that $A(\indora(\calG_1)) \ne A(\indora(\calG_2))$ and $\sum_{i=1}^t \cost(S_i(\calG_1)) \le 3Z$; it follows by a union bound that
    \[
        \pr(\calE) \ge \pr\big(A(\indora(\calG_1)) \ne A(\indora(\calG_2))\big) - 1/3 \ge 1/3\,.
    \]
    We now bound $\pr(\calE)$ above. By \cref{lem:uncol-size}, we have
    \[
        \pr(\calE) \le 4^k k^k p_2 \max_{\vec{b} \in X_{3Z}} \sum_{i \in \calI(\vec{b})} \min\bigg\{|S_i(\vec{b})|^r, \frac{1}{p_1|S_i(\vec{b})|^{k-r}}\bigg\}\,.
    \]
    Combining these two equations yields
    \begin{equation}\label{eq:uncol-cost-1}
        4^k k^k p_2 \max_{\vec{b} \in X_{3Z}} \sum_{i \in \calI(\vec{b})} \min\bigg\{|S_i(\vec{b})|^r, \frac{1}{p_1|S_i(\vec{b})|^{k-r}}\bigg\}
        \ge \frac{1}{3}\,.
    \end{equation}
    
    The remainder of the proof consists of bounding the left side of~\eqref{eq:uncol-cost-1} above in terms of~$Z$. We now write $y_i=\abs{S_i(\vec(b))}$ for all $i\in[t]$, and we define the set~$Y$ with
    \[
        Y = \setc[\Big]{
                \vec{y} \in [0,n]^t
            }{
                \sum_{i=1}^ty_i^{\alpha_k} \le 3Z
            }\,,
    \]
    where $[0,n]$ is the real interval.
    For all $\vec{b}\in X_{3Z}$, the vector $(\abs{S_1(\vec{b})}, \dots, \abs{S_t(\vec{b})})$ is contained in~$Y$, since $\cost(y_i)=y_i^{\alpha_k}$ holds. Moreover, $|S_i(\vec{b})| \ge 1$ for all $i \in \calI(\vec{b})$ by definition.
    By~\eqref{eq:uncol-cost-1}, we therefore have
    \[
        \max_{\vec{y} \in Y} \sum_{\substack{i \in [t]\\y_i \ge 1}}
        \min\set[\Big]{y_i^r, \frac{1}{p_1y_i^{k-r}}}
        \ge \frac{1}{3 p_2 4^k k^k} \ge \frac{1}{3p_2k^{k+2}}\,.
    \]
    
    Since $x\mapsto x^r$ is increasing and $x\mapsto 1/(p_1 x^{k-r})$ is decreasing, their minimum is maximised over the non-negative reals when they are equal. Over the interval $[0,n]$, the minimum is maximised for some $x\in[0,n]$ with $x^r\le 1/(p_1 x^{k-r})$, which is equivalent to $0\le x\le\min\set{n,p_1^{-1/k}}$. It follows that replacing the interval $[0,n]$ in the definition of $Y$ by the interval $[0, p_1^{-1/k}]$ will not affect the value of the maximum; thus writing
    \[
        Y' = \set[\Big]{\vec{y} \in [0,p_1^{-1/k}]^t\colon \sum_{i=1}^t y_i^{\alpha_k} \le 3Z}\,,
    \]
    it follows that
    \[
        \max_{\vec{y} \in Y'} \sum_{i=1}^t y_i^r \ge \max_{\vec{y} \in Y'} \sum_{\substack{i\in[t]\\y_i \ge 1}} \min\set[\Big]{y_i^r, \frac{1}{p_1y_i^{k-r}}} = \max_{\vec{y} \in Y} \sum_{\substack{i\in[t]\\y_i \ge 1}}\min\set[\Big]{y_i^r, \frac{1}{p_1y_i^{k-r}}} \ge \frac{1}{3 p_2 k^{k+2}}\,.
    \]
    
    We now apply Karamata's inequality in the form of \cref{cor:karamata}, taking $c = 1/p_1^{1/k}$ and $W=3Z$. Since $r \ge \alpha_k$, this yields
    \[
        \frac{3Z}{p_1^{(r-\alpha_k)/k}} \ge \frac{1}{3p_2k^{2k}}\,,
    \]
    Substituting in the definitions of $p_1$ and $p_2$ then yields
    \[
        Z \ge \frac{n^{r-(r/k)(r-\alpha_k)}}{1080k^{3k}\eps^{(r-\alpha_k)/k}} = \frac{n^{(kr + \alpha_k r - r^2)/k}}{1080k^{3k}\eps^{(r-\alpha_k)/k}}
    \]
    as required.
\end{proof}

\section{Colourful independence oracle with cost}\label{sec:col}

In this section, we study the edge estimation problem for \cindora-oracle algorithms, which are given access to the colourful independence oracle $\cindora(G)$ of a $k$-uniform hypergraph~$G$. Any query $\cindora(G)_{X_1, \dots, X_k}$ incurs a cost of $\cost_k(|X_1|+\dots+|X_k|)$, where
$\cost=\setc{\cost_k}{k\ge 2}$ is regularly-varying with parameter~$k$.
In \cref{sec:col-upper}, we prove the upper bound part of \cref{thm:col-main-simple} as a special case of \cref{thm:col-alg} by constructing a randomised \cindora-oracle algorithm for edge estimation.
In \cref{sec:col-lower}, we prove the lower bound part of \cref{thm:col-main-simple} as a special case of \cref{thm:col-lb} by showing that no randomised \cindora-oracle algorithm for edge estimation can have significantly smaller worst-case oracle cost than the one obtained in \cref{sec:col-upper}.

\subsection{Oracle algorithm for edge estimation}%
\label{sec:col-upper}

We will prove the following result. 

\begin{restatable}{theorem}{statecolalgo}\label{thm:col-alg}
    Let $\cost = \{\cost_k\colon k\ge 2\}$ be a regularly-varying parameterised cost function with parameter $k$ and index $\alpha_k\in[0,k]$, let
    \(\alpha_k' \coloneqq \ceil{\alpha_k} - 1\),
    and let
    \(
        T \coloneqq \log(1/\delta)\eps^{-2}k^{27k}\log^{4(k-\alpha_k')+14} n
    \).
    There is a randomised \cindora-oracle algorithm
    \(\aau(\cindora(G),\eps,\delta)\)
    with worst-case running time
    \(
        \OO(T\cdot (\cost_k(n)+n))
    \),
    worst-case oracle cost
    \(
        \OO(T\cdot \cost_k(n))
    \), and
    the following behaviour:
Given an $n$-vertex $k$-hypergraph~$G$ and rationals $\eps,\delta \in (0,1)$,
the algorithm outputs an integer $m$ that, with probability at least $1-\delta$, is an $\eps$-approximation to~$e(G)$.
%
% \holger{old statement below (Nov 4)}
%     Let $\cost = \{\cost_k\colon k\ge 2\}$ be a regularly-varying parameterised cost function with parameter $k$ and index $\alpha_k\in[0,k]$, and set $\alpha_k' = \ceil{\alpha_k} - 1$. 
%     There is a randomised \cindora-oracle algorithm $\aau(\cindora(G),\eps,\delta)$ with the following properties. Suppose~$G$ is an $n$-vertex $k$-hypergraph to which $\aau$ is given (only) colourful independence oracle access and that $\eps,\delta \in (0,1)$ are rational.
%     Let
%     \[
%         T = \log(1/\delta)\eps^{-2}k^{27k}\log^{4(k-\alpha_k')+14} n.
%     \]
%     Then $\aau(\cindora(G),\eps,\delta)$ runs in worst-case running time
%     \(
%         \OO(T\cdot (\cost_k(n)+n))\,,
%     \)
%     incurs worst-case oracle cost at most
%     \(
%         \OO(T\cdot \cost_k(n))\,,
%     \)
%     and, with probability at least $1-\delta$, outputs a valid $\eps$-approximation to $e(G)$.
\end{restatable}
% Added this paragraph because otherwise it's a serious "buh?" moment and it's otherwise pretty easy to miss that this dependence is a) present and b) significant. --John 2022/08/12
Recall from \cref{sec:oracle-algorithms} that we strictly separate our running times from our oracle costs, and note that the running time in \cref{thm:col-alg} still depends on the cost function. Indeed, $\aau$ will exploit a tradeoff between oracle cost and running time, and this is why we require the index of the cost function to be efficiently computable in \cref{def:cost}(v).

\begin{table}
    \centering\footnotesize
    \begin{tabular}{llcl}
        Algorithm & Statement & Approximation & Condition\\
        \footnotesize\oldcount & \cref{thm:old-counting} & $\eps$ & ---\\
        \footnotesize\aau & \cref{thm:col-alg} & $\eps$ & ---\\
        \footnotesize\oldcoarsecount & \cref{lma:old-coarse} & $(4k\log n)^k$ & large~$k$\\
        \footnotesize\newcoarsecount & \cref{lma:better-col-alg} & $k^{\OO(k)}\log^{k-\alpha'_k-1} n$ & small~$k$\\
        \footnotesize\helperdlmrecurse & \cref{lma:colourcoarse-correct} & $k^{\OO(k)}$ & small cores\\
        \footnotesize\dlmrecurse & \cref{lem:col-algo-small-core} &
        \multicolumn{2}{l}{boost of \footnotesize\helperdlmrecurse}\\
        \footnotesize\helperdlmimprove & \cref{lma:coarseguess} &
        $(k^{\OO(k)}\log^{k-|I|} n)$ & large cores\\
        \footnotesize\dlmimprove & \cref{lem:col-algo-large-core} & \multicolumn{2}{l}{ boost of \footnotesize\helperdlmimprove}\\
        \footnotesize\newverifyguess & \cref{lma:newverifyguess} & \multicolumn{2}{l}{\footnotesize distinguishes ``$e(G)\ll M$'' from ``core and $e(G)\ge M$''}
    \end{tabular}
    \caption{\label{fig:colourful-algo-table}%
    Overview of the algorithms used in \cref{sec:col-upper}.
    % \holger{in addition, \Refine is used, but the properties are not formally stated.}
    }
\end{table}

The structure of the proof of \cref{thm:col-alg} is roughly similar to the proof of \cite[Theorem~1.1]{DLM}:
We construct a \cindora-oracle algorithm $\combinecoarsecount$ with a coarse multiplicative approximation guarantee and use the following theorem, proved implicitly in~\cite{DLM}, to boost it into an $\eps$-approximation algorithm.
% That is, the following theorem is a reduction from approximate edge counting in the whole graph with arbitrarily low error to approximate edge counting in induced $k$-partite subgraphs with large multiplicative error.
See also \cite[Lemma~5.2]{BBGM-runtime}, where this reduction was formally stated for the worst-case \emph{number} of queries; in our formalisation, the oracle cost can be general, and we carefully distinguish the running time and the oracle cost.

% \john{Two paragraphs were removed from the proof below at some point in the last six months that were important to the meaning (not just sketches) --- I've added them back in and rewritten them a little to make them fit better. Also, it turns out we can safely remove the $k^2$ terms from the running time and oracle cost (they fit inside a $2^k = \OO(e^k)$ bound we already use). Could use a check-through from someone other than me, delete this comment after doing so. --John 2022/05/19}
\begin{theorem}[Turn coarse into fine approximation]\label{thm:use-coarse}
Let $b=b(n,k)$, $T = T(n,k)= \Omega(k^2n)$ and $C = C(n,k)$,
and suppose that
\(\combinecoarsecount(\cindora(G),X_1,\dots,X_k)\)
is a randomised \cindora-oracle algorithm with worst-case running time~$T$, worst-case oracle cost~$C$, and the following behaviour:
Given an $n$-vertex $k$-hypergraph~$G$ with~$n$ a power of two and disjoint sets ${X_1,\dots,X_k\subseteq V(G)}$,
the algorithm outputs an integer $m$ such that, with probability at least $2/3$, we have $m/b \le e(G[X_1,\dots,X_k]) \le mb$.

Then there is a randomised \cindora-oracle algorithm $\aau(\cindora(G),\eps,\delta)$ with worst-case running time
\[\OO\big(\log(1/\delta)\eps^{-2} e^{3k} b^2 \log^4 n \cdot T\big)\,,\]
worst-case oracle cost
\[ \OO\big(\log(1/\delta)\eps^{-2} e^{3k} b^2 \log^4 n \cdot C\big)\,, \]
and the following behaviour:
Given an $n$-vertex $k$-hypergraph~$G$ and rationals~$\eps,\delta\in(0,1)$,
the algorithm outputs an integer $m$ that, with probability at least $1-\delta$, is an $\eps$-approximation to~$e(G)$.
\end{theorem}
\begin{proof}
We first reduce from approximate edge counting in the whole graph with multiplicative error $2b$ to approximate edge counting in induced $k$-partite subgraphs with multiplicative error $b$; that is, we give a randomised \cindora-oracle algorithm $\acc(\cindora(G),\delta)$ with the following behaviour. Suppose $G$ is an $n$-vertex $k$-hypergraph to which \acc has (only) colourful oracle access, where $n$ is a power of two, and suppose $0 < \delta < 1$.  Then, in time $\OO(\log(1/\delta) k e^{2k} \cdot T)$ and with oracle cost $\OO(\log(1/\delta)ke^{2k}\cdot C)$, $\acc(\cindora(G),\delta)$ outputs a rational number $\hat{e}$.  Moreover, with probability at least $1 - \delta$, $\hat{e}/2b \le e(G) \le \hat{e} \cdot 2b$. 

To obtain this algorithm \acc from \combinecoarsecount, we carry out an intermediate colour-coding procedure ($\texttt{HelperCoarse}(\cindora(G))$ in \cite{DLM}), whose running time is dominated by $\OO(ke^{2k})$ invocations of \combinecoarsecount, then run \texttt{HelperCoarse} a total of $\OO(\log(1/\delta))$ times and output the median result.  The proofs of correctness of these steps are given as~\cite[Lemma 4.3]{DLM} and \cite[Lemma 3.3]{DLM}, respectively.  

The main counting algorithm in~\cite{DLM} is denoted there by \texttt{Count}; to avoid overloading notation, we will instead refer to it as \oldcount. Having defined \acc, we now define $\aau$ by running $\oldcount(\cindora(G),\eps,\delta)$, replacing the \acc subroutine from~\cite{DLM} with this new version and replacing the value of $b$ in \cite{DLM} (which corresponds to the error bound in \acc) with $b^* \coloneqq 2b(n,k)$. By exactly the same argument as in~\cite{DLM}, the output of \aau has the desired properties with the desired probability; it remains to bound the running time and query cost. Thankfully, most of these bounds are already given in terms of $b$, and so we only need to follow the analysis through.

% Added more detail about Count vs HelperCount. --John
From the proof of~\cite[Theorem 1.1]{DLM}, \oldcount is just a wrapper for the main algorithm $\approxUncol(\cindora(G),\eps)$, which serves to remove the requirement that $n$ be a power of two (by adding isolated vertices) and reduce the failure probability to $\delta$ by running $\approxUncol$ $\OO(1/\delta)$ times and outputting the median value. 

The running time of \approxUncol is analysed in \cite[Lemma 3.5]{DLM}; from this analysis, both the running time and oracle cost are dominated by $\OO(\log n)$ calls to a subroutine \Refine. 
% (as we can use \cref{cor:sample-binomial} to take binomial samples efficiently)

% Added more detail about what the heck \xi and \delta' and L are. --John
In \cite[Lemma 3.4]{DLM}, the running time of \Refine is given in terms of some of its arguments $L$, $\xi$ and $\delta'$. From steps (A2) and (A5) of \approxUncol in~\cite{DLM}, in every invocation of \Refine we have $\xi = \OO(\eps^{-1}\log n)$ and $1/\delta' = \OO(\log n)$. The value of $L$ is not fixed --- it is a so-called $(G,b,y)$-list $L$ of induced subgraphs whose weighted sum approximates $e(G)$ --- but from invariant (iii) in the proof of \cite[Lemma 3.5]{DLM}, in each invocation of \Refine we have
\[
    |L| = \OO\big(k\log(nb^*) + (b^*)^2\xi^{-2}\log(1/\delta')\big) = \OO(kb^2\eps^{-2}\log^3 n).
\]

With these bounds in place, we now bound the running time and oracle cost of \Refine. From~\cite[Lemma 3.4]{DLM}, the running time and oracle cost of \Refine are dominated by calls to \acc; the total number of such calls is upper bounded in \cite[Lemma 3.5]{DLM} by $\OO(\lambda)$, where
% Very slightly redefined \lambda to match the DLM definition. --John
\[
    \lambda \coloneqq |L| + 2^k \xi^{-2}(b^*)^2 \log(1/\delta') = \OO(\eps^{-2}2^kb^2\log^3 n).
\]
Each call to \acc has $\delta^{-1} = \OO(\lambda/\delta')$, so the running time of each call is $\OO(\log(\lambda\log n)\cdot 
% Added cdot just to fix bad hbox. --John
ke^{2k}T)$ and the query cost of each call is $\OO(\log(\lambda\log n)ke^{2k}C)$.
% = \OO(\eps^{-2} 2^k b^2 \log^4 n)$. 
We may assume that $\eps^{-1} < n^k$ (otherwise the algorithm counts exactly by brute force) and similarly that $b < n^k$, so $\log(\lambda\log n) = \OO(k\log n)$. It follows that the running time of each call to \Refine is $\OO(k^2e^{2k}\log n \cdot T)$ and the oracle cost is $\OO(k^2e^{2k}\log n \cdot C)$.
% \lambda is the number of calls, so shouldn't be included in this expression. --John
Multiplying by the number $\lambda$ of total calls, we see that the running time of \approxUncol is $\OO(\eps^{-2}e^{3k}b^2\log^4 n \cdot T)$, and the oracle cost is $\OO(\eps^{-2}e^{3k}b^2\log^4 n \cdot C)$. Since \UncolApprox runs \approxUncol $\OO(\log(1/\delta))$ times, the result follows.
\end{proof}

In order to apply \cref{thm:use-coarse} to prove \cref{thm:col-alg}, we need to provide an algorithm \combinecoarsecount with suitable properties.  In fact, we will use one of two different algorithms, depending on the size of $k$ relative to $n$.  Most of the rest of the section is devoted to proving the following lemma, which gives the existence of a suitable algorithm when $k$ is not too large.
\
\begin{restatable}{lemma}{stateColourCoarseNew}\label{lma:better-col-alg}
    Let $\cost = \{\cost_k\colon k\ge 2\}$ be a regularly-varying parameterised cost function with parameter $k$ and index $\alpha_k$.
    Let
    \begin{align*}
        \alpha_k' &\coloneqq \ceil{\alpha_k} - 1, \qquad\qquad 
        T \coloneqq k^{9k+2}\log^{2(k-\alpha_k')+9}n,\\
        b &\coloneqq (2k)^{5k}(\log n)^{k - \alpha_k'-1}\log^k(8k \log^{40k(k-\alpha_k')/(\alpha_k - \alpha_k')}(n)).
    \end{align*}

    There is a randomised \cindora-oracle algorithm
    \(\newcoarsecount(\cindora(G),X_1,\dots,X_k)\) 
    with worst-case running time~$\OO(T(\cost_k(n)+n))$, worst-case oracle cost~$\OO(T\cost_k(n))$, and
    the following behaviour:
Given an $n$-vertex $k$-hypergraph~$G$ with~$n$ a power of two, $k \le (\log n)/(\log\log n)^2$, and disjoint sets ${X_1,\dots,X_k\subseteq V(G)}$,
the algorithm outputs an integer $m$ such that, with probability at least $2/3$, we have $m/b \le e(G[X_1,\dots,X_k]) \le mb$.
%
% \holger{new proposal (Nov 4) is above, old proposal is below.}
%     There exists an algorithm \newcoarsecount with the following behaviour. Suppose $G$ is an $n$-vertex $k$-hypergraph to which \newcoarsecount has (only) colourful oracle access. Suppose $n$ is a power of two, 
%     $k \le (\log n)/(\log\log n)^2$, and that $X_1,\dots,X_k$ form a partition of $V(G)$. Let
%     \begin{align*}
%         \alpha_k' &\coloneqq \ceil{\alpha_k} - 1, \qquad\qquad 
%         T \coloneqq k^{9k+2}\log^{2(k-\alpha_k')+9}n,\\
%         b &\coloneqq (2k)^{5k}(\log n)^{k - \alpha_k'-1}\log^k(8k \log^{40k(k-\alpha_k')/(\alpha_k - \alpha_k')}(n)).
%     \end{align*}
%     Then $\newcoarsecount(\cindora(G),X_1,\dots,X_k)$ outputs a non-negative integer $m$ such that, with probability at least $2/3$, $m/b \le e(G[X_1,\dots,X_k]) \le mb$.
%     The worst-case running time of \newcoarsecount is $\OO(T(\cost_k(n)+n))$ and its worst-case oracle cost is $\OO(T\cost_k(n))$.
\end{restatable}
We defer the proof of this lemma to \cref{sec:coarse-approx-outline,sec:dlm-recurse,sec:dlm-coarse}.
When $k$ is large enough relative to $n$ that $\log^k n = k^{\OO(k)}$, the algorithm of \cref{lma:better-col-alg} may break down, but in this regime we can simply use the coarse approximation algorithm from~\cite{DLM}. We now state its properties.

\begin{lemma}[Dell, Lapinskas, and Meeks~{\cite[Lemma~4.2]{DLM}}]\label{lma:old-coarse}%\mbox{}\\
    Let $\cost = \{\cost_k\colon k\ge 2\}$ be a regularly-varying parameterised cost function with parameter $k$ and index $\alpha_k$, let
    \(
        b \coloneqq (4k\log n)^k
    \), and let
    \(
        T \coloneqq (8k\log n)^{2k+2}
    \).
    There is a randomised \cindora-oracle algorithm
    \[
    \oldcoarsecount(\cindora(G),X_1,\dots,X_k)
    \] 
    with worst-case running time~$\OO(T(\cost_k(n)+n))$, worst-case oracle cost~$\OO(T\cost_k(n))$, and
    the following behaviour:
Given an $n$-vertex $k$-hypergraph~$G$ with~$n$ a power of two and disjoint sets ${X_1,\dots,X_k\subseteq V(G)}$,
the algorithm outputs an integer $m$ such that, with probability at least $2/3$, we have $m/b \le e(G[X_1,\dots,X_k]) \le mb$.
%
% \holger{old statement below (Nov 4)}
%     There exists an algorithm \oldcoarsecount with the following behaviour. Let $\cost = \{\cost_k\colon k\ge 2\}$ be a regularly-varying parameterised cost function with parameter $k$. Suppose $G$ is an $n$-vertex $k$-hypergraph to which \oldcoarsecount has (only) colourful oracle access. Suppose $n$ is a power of two and that $X_1,\dots,X_k$ form a partition of $V(G)$. Let 
%     \begin{align*}
%         b &\coloneqq (4k\log n)^k, \qquad\qquad T \coloneqq (8k\log n)^{2k+2}.
%     \end{align*}
%     Then $\oldcoarsecount(\cindora(G),X_1,\dots,X_k)$ outputs a non-negative integer $m$ such that, with probability at least $2/3$, $m/b \le e(G[X_1,\dots,X_k]) \le mb$.
%     The worst-case running time of \oldcoarsecount is $\OO(T(\cost_k(n)+n))$ and its worst-case oracle cost is $\OO(T\cost_k(n))$.
\end{lemma}
\begin{proof}
    This is immediate from~\cite[Lemma~4.2]{DLM}, bounding the cost of each query above by $\OO(\cost_k(n))$. This bound is valid by \cref{lem:cost-monotone} since the cost function is regularly-varying.
\end{proof}

We now set out the proof of \cref{thm:col-alg} from \cref{thm:use-coarse,lma:better-col-alg,lma:old-coarse} (which consists of easy algebra), before devoting the rest of the section to proving \cref{lma:better-col-alg}.

\statecolalgo*

\begin{proof}
    We apply \cref{thm:use-coarse}, taking our coarse approximate counting algorithm \combinecoarsecount to be as follows. Given an $n$-vertex instance $(G,X_1,\dots,X_k)$, if $k \le (\log n)/(\log \log n)^2$ then we apply the algorithm \newcoarsecount of \cref{lma:better-col-alg} and return the results, and otherwise we apply the algorithm \oldcoarsecount of \cref{lma:old-coarse} and return the results.
    
    As in \cref{lma:better-col-alg}, let
    \begin{align*}
        \alpha_k' &\coloneqq \ceil{\alpha_k} - 1, \qquad\qquad 
        T_1 \coloneqq k^{9k+2}\log^{2(k-\alpha_k')+9}n,\\
        b_1 &\coloneqq (2k)^{5k}(\log n)^{k - \alpha_k'-1}\log^k(8k \log^{40k(k-\alpha_k')/(\alpha_k - \alpha_k')}(n)).
    \end{align*}
    As in \cref{lma:old-coarse}, let
    \[
        T_2 \coloneqq (8k\log n)^{2k+2}, \qquad\qquad b_2 \coloneqq (4k\log n)^k.
    \]
    Observe that when $k \ge (\log n)/(\log\log n)^2$, we have 
    \[
        k^{2k} = e^{2k\log k} = e^{2k(1-o(1))\log\log n} = \Omega(\log^k n),
    \]
    and hence
    \[
        T_2 = \OO(k^{7k}\log^2n) = \OO(T_1),\qquad\qquad b_2 = \OO(k^{4k}) = \OO(b_1).
    \]
    Thus on an $n$-vertex instance, our coarse approximation algorithm uses oracle queries of combined cost $\OO(T_1\cost_k(n))$, runs in time $\OO(T(\cost_k(n)+n))$, and has multiplicative error $\OO(b_1)$.
    
    We now take \aau to be the algorithm of \cref{thm:use-coarse}. Let
    \[
        T_3 = \log(1/\delta)\eps^{-2}e^{3k}b_1^2\log^4 n;
    \]
    then \aau uses oracle queries of combined cost $\OO(T_1T_3\cost_k(n))$ and runs in time $\OO(T_1T_3n)$. Observe that
    \begin{equation}\label{eq:col-algo-b1}
        b_1^2 = \OO\big(k^{13k}(\log n)^{2(k-\alpha_k')}(\log\log n)^{2k}\big).
    \end{equation}
    If $k \ge (\log \log n)/(\log \log \log n)^2$, then 
    \[
        k^{2k} = e^{2k\log k} = e^{2k(1-o(1))\log\log\log n} = \Omega((\log\log n)^k);
    \]
    if instead $k \le (\log\log n)/(\log\log\log n)^2$, then
    \[
        (\log\log n)^k \le e^{(\log\log n)/(\log\log\log n)} = (\log n)^{o(1)}.
    \]
    In either case, we have $(\log\log n)^{2k} = \OO(k^{4k}\log n)$, and hence by~\eqref{eq:col-algo-b1}, $b_1^2 = \OO(k^{17k}\log^{2(k-\alpha_k')+1}n)$.
    Thus
    \[
        T_1T_3 = \OO\big(\log(1/\delta)\eps^{-2}k^{27k}\log^{4(k-\alpha_k')+14} n\big) = \OO(T),
    \]
    and the result follows.
\end{proof}

% \begin{lemma}\label{lma:better-col-alg}
% There exists an algorithm \newcoarsecount with the following behaviour.  Suppose $G$ is an $n$-vertex $k$-hypergraph to which \newcoarsecount has (only) colourful oracle access, with an oracle call on $(X_1,\dots,X_k)$ taking time $\OO\left((|X_1| + \dots + |X_k|)^{\alpha}\right)$, that $n$ is a power of two, and that $X_1,\dots,X_k$ form a partition of $V(G)$.  Setting $\alpha' \coloneqq \ceil{\alpha} - 1$ and $b\coloneqq (2k)^{5k}(\log n)^{k - \alpha'-1}\log^k(8k \log^{40k(k-\alpha')/(\alpha - \alpha')}(n))$, in time $\OO(2^{900 k2^{k^2}/(\alpha-\alpha')} + \log(1/\delta) (\alpha - \alpha')^{-1} k^{23k} (\log n)^{2(k-\alpha')+2}n^{\max\{1,\alpha\}})$, \newcoarsecount outputs a non-negative integer $m$ such that, with probability at least $2/3$, $m/b \le e(G[X_1,\dots,X_k]) \le mb$.  
% \end{lemma}

The rest of this section is dedicated to proving \cref{lma:better-col-alg}.

\subsubsection{Coarse approximation: an overview}\label{sec:coarse-approx-outline}

In order to prove \cref{lma:better-col-alg}, given oracle access to an $n$-vertex $k$-hypergraph $G$ and vertex classes $X_1,\dots,X_k$, we will make use of two separate approximate counting algorithms. Writing $H \coloneqq G[X_1,\dots,X_k]$ for brevity, 
% Technically they're still random. --John
each algorithm will provide an estimate that is very likely not to be much larger than $e(H)$. Depending on $G$, one of these estimates is very likely to also not be much smaller than $e(H)$, so by returning whichever result is largest we obtain a coarse approximation to $e(H)$ as required. In this section, we set out the properties of these algorithms and explain how to use them to prove \cref{lma:better-col-alg}. We then set out the algorithms themselves and prove their stated properties in \cref{sec:dlm-coarse,sec:dlm-recurse}.

For motivation, we first sketch a possible proof of \cref{lma:better-col-alg} that does not work. Form a graph $H'$ from $H$ by deleting vertices independently at random, so that each vertex is retained with probability $p = 1/(\log n)^{k^2}$. Each edge of $H'$ is preserved with probability $1/(\log n)^{k^3}$, so in expectation, $H'$ contains $e(H)/(\log n)^{k^3}$ edges and $n/(\log n)^{k^2}$ vertices. By Chernoff bounds, we in fact have 
% Chernoff gives concentration in both directions, and I think this sentence makes more sense if we say so. --John
$|V(H')| = (1+o(1))n/(\log n)^{k^2}$ with high probability; suppose (incorrectly) that we also had $e(H') = (1+o(1))e(H)/(\log n)^{k^3}$ with high probability. Then we could simply run the full approximate counting algorithm of~\cite{DLM} on $H'$ and multiply the result by $(\log n)^{k^3}$. Indeed, the factor of $(\log |V(H')|)^{4k+7}$ in the query count of~\cite{DLM} would be dominated by the factor of $1/(\log n)^{\alpha k^2}$ saved in oracle cost by running it on an instance with fewer vertices.

Of course, we do not in general have such concentration of $e(H')$, as illustrated by the following definition.

\begin{defn}\label{def:root}
    Let $H$ be a $k$-hypergraph, and let $0 < \zeta \le 1$. We say that $v \in V(H)$ is a \emph{$\zeta$-root} of $H$ if $d(v) \ge \zeta e(H)$.
\end{defn}

% I think this paragraph might originally have been written for "roots" as sets of vertices? It felt odd to belabour the point that we can replace 1 by \zeta, so I've shortened this a bit. 
If $X_1$ contains a $\zeta$-root $v$ with $\zeta$ relatively large, then we should expect $e(H')$ to depend very strongly on whether or not $v \in V(H')$, and so we should not expect concentration of $e(H')$. It turns out, however, that in some sense this is the only barrier to concentration of $e(H')$: if $H$ were to contain no $\zeta$-roots for an appropriate choice of $\zeta$, then the above proof would work.

Despite the possibility of $\zeta$-roots, we can still recover part of this argument if we know enough about where $\zeta$-roots occur in the graph: if some vertex classes do not contain any $\zeta$-roots, then we can reasonably hope for concentration of the number of edges that survive when we delete vertices uniformly at random from these root-free classes. (See \cref{lem:colourful-algo-conc}.) In fact, we will make use of subgraphs of $H$ which contain a large proportion of the edges and in which all roots are contained in a collection of relatively small vertex classes; this notion is formalised in the following definition.

\begin{defn}\label{def:core}
    Let $H = (V,E)$ be a $k$-partite $k$-hypergraph with vertex classes ${X_1,\dots,\allowbreak X_k}$. Let $I \subseteq [k]$, and let $\zeta \in(0,1)$. For each $i \in [k]$, let $Y_i \subseteq X_i$. Then we say that $(Y_1,\dots,Y_k)$ is an \emph{$(I,\zeta)$-core} of $H$ if the following properties hold:
    \begin{enumerate}[(i)]
        \item $e(H[Y_1,\dots,Y_k]) \ge e(H)/(2k)^k$;
        \item for all $i \in I$, we have $|Y_i| \le 2/\zeta$;
        \item for all $i \in [k] \setminus I$, the set $Y_i$ contains no $\zeta$-roots of $H[Y_1,\dots,Y_k]$.
    \end{enumerate}
\end{defn}

We will show in \cref{lem:col-algo-core-exists} that, for any $\zeta \in (0,1)$, every $k$-partite $k$-hypergraph contains an $(I,\zeta)$-core for some~$I$. Before doing so, we set out the behaviour of two algorithms which exploit this fact. 

For simplicity, suppose $(Y_1,\dots,Y_k)$ is an $(I,\zeta)$-core of the $k$-partite $k$-hypergraph~$H$ for some suitably-chosen $\zeta$, and write $H' = H[Y_1,\dots,Y_k]$; we use this specific core to sketch proofs of correctness for the algorithms, but the algorithms themselves do not have access to the sets $(Y_1,\dots,Y_k)$. 
% Added a lead-in to the lemma like we have for the next lemma.
Our first algorithm is designed to run efficiently in the case where $G$ has an $(I,\zeta)$-core with $|I|$ small, and works by recursively applying the algorithm of \cite{DLM} to a collection of smaller instances, roughly as described above: vertices from classes outside $I$ are deleted randomly and classes in $I$ are partitioned into appropriately-sized subsets, each combination of which is considered in one of the smaller instances.  Formally, we will prove the following lemma in \cref{sec:dlm-recurse}. (The parameter~$t$ in this algorithm is a technical convenience which we will optimise later.)

\begin{lemma}\label{lem:col-algo-small-core}
    Let $\cost = \{\cost_k\colon k\ge 2\}$ be a regularly-varying parameterised cost function with parameter $k$.
    Let
    \[
        b\coloneqq (2k)^{k+2}\,,
        \qquad
        T \coloneqq \log(1/\delta)k^{6k+1}\log^{4k+8}(n)n^{|I|}/t^{|I|}.
    \]

    There is a randomised \cindora-oracle algorithm
    \[\dlmrecurse(\cindora(G), X_1,\dots,X_k, I, t, \delta)\]
    with worst-case running time~$\OO(n + T \cdot kt\log(n))$, worst-case oracle cost~$\OO(T\cost_k(2kt))$, and
    the following behaviour:
Given an $n$-vertex $k$-hypergraph~$G$ with~$n$ a power of two,
disjoint sets ${X_1,\dots,X_k\subseteq V(G)}$,
a set $I \subseteq [k]$,
an integer $t$ with $n \ge t \ge 12 \log k$,
and a rational $\delta > 0$,
the algorithm outputs a non-negative integer $m$ such that, with probability at least $1-\delta$:
\begin{enumerate}[(i)]
    \item $m \le b \cdot e(G[X_1,\dots,X_k])$;
    \item if $G$ has an $(I,\zeta)$-core for some $\zeta$ satisfying $t \ge (8k\zeta)^{1/(2(k-|I|))}n$, then $m \ge e(G[X_1,\dots,X_k])/b$.
\end{enumerate}
% %
%     \holger{old phrasing below (Nov 5)}
%     There exists an algorithm $\dlmrecurse(\cindora(G), X_1,\dots,X_k, I, t, \delta)$ with the following behaviour.
%     Let $\cost = \{\cost_k\colon k\ge 2\}$ be a regularly-varying parameterised cost function with parameter $k$.
%     Suppose $G$ is an $n$-vertex $k$-hypergraph to which \dlmrecurse has (only) colourful oracle access. Suppose that $n$ is a positive power of two, $X_1,\dots,X_k$ are disjoint subsets of $V(G)$, $I \subset [k]$, $t$ is an integer satisfying $n \ge t \ge 12 \log k$, and $\delta > 0$ is rational.  Let $b \coloneqq (2k)^{k+2}$, and let 
%     \[
%         T \coloneqq \log(1/\delta)k^{6k+1}\log^{4k+8}(n)n^{|I|}/t^{|I|}.
%     \]
%     Then in time at most $\OO(n + T \cdot kt\log(n))$ and with oracle cost at most $\OO(T\cost_k(2kt))$, \dlmrecurse outputs a non-negative integer $m$ such that, with probability at least $1-\delta$:
%     \begin{enumerate}[(i)]
%         \item $m \le b \cdot e(G[X_1,\dots,X_k])$;
%         \item if $G$ has an $(I,\zeta)$-core for some $\zeta$ satisfying $t \ge (8k\zeta)^{1/(2(k-|I|))}n$, then $m \ge e(G[X_1,\dots,X_k])/b$.
%     \end{enumerate}
\end{lemma}

We note explicitly that~\cite{DLM} is concerned only with the total number of queries (that is, $\cost_k(n)=1$), and that in this setting \cref{lem:col-algo-small-core} does not yield an improved oracle cost. Note also that \dlmrecurse does not require the value of sets~$(Y_1,\dots,Y_k)$ of the $(I,\zeta)$-core, only the value of $I$ itself (which we will simply guess, as there are only~$2^k$ possibilities).

Our second algorithm is designed to deal with the opposite situation, when there is an $(I,\zeta)$-core with $|I|$ large.  In this case, we can adapt the algorithm of \cite{DLM} to perform significantly better.  Very roughly speaking, the algorithm of~\cite{DLM} works by finding probabilities $p_1,\dots,p_k \in \{1,1/2,1/4,\dots,1/\log n\}$ such that:
\begin{enumerate}[(a)]
    \item $p_1p_2\dots p_k$ is as small as possible;
    \item on randomly deleting all but roughly $p_i$ proportion of vertices from each vertex class $X_i$, the resulting graph still contains at least one edge with high probability.
\end{enumerate}
We then conclude that $H$ contains at least roughly $1/(p_1\dots p_k)$ edges. (Note that this is inaccurate on a formal level, but it is close enough to give intuition.) The approach to finding such probabilities $p_1,\dots,p_k$ in \cite{DLM} is simply to check all $\log^k n$ possibilities. However, if $H$ contains an $(I,\zeta)$-core, then for all $i \in I$ we know that many edges are concentrated on a few $\zeta$-roots in $X_i$.  It follows that we must take $p_i$ very close to $1$ to satisfy property (b), substantially reducing the number of possible values of $p_1,\dots,p_k$ we need to check; this translates into a reduction in both the running time and oracle cost that scales with $|I|$. Formally, we will prove the following lemma in \cref{sec:dlm-coarse}.

\begin{lemma}\label{lem:col-algo-large-core}
    Let $\cost = \{\cost_k\colon k\ge 2\}$ be a regularly-varying parameterised cost function with parameter $k$. Let
    \[
        b\coloneqq (2k)^{5k}\log^{k-|I|}n\log^{|I|}(1/\zeta)\,,
        \qquad
        T \coloneqq 2^{7k}\log(1/\delta)\log^{2|I|}(1/\zeta)(\log n)^{2(k-|I|)+1}
        \,.
    \]

    There is a randomised \cindora-oracle algorithm
    \[\dlmimprove(\cindora(G), X_1,\dots,X_k, I, \zeta, \delta)\]
    with worst-case running time~$\OO(Tn\log n)$, worst-case oracle cost~$\OO(T \cost_k(n))$, and
    the following behaviour:
Given an $n$-vertex $k$-hypergraph~$G$ with~$n\ge 32$ a power of two,
disjoint sets ${X_1,\dots,X_k\subseteq V(G)}$,
a set $I \subseteq [k]$,
rationals $\zeta,\delta \in (0,\tfrac{1}{32})$ such that the denominator of $\zeta$ is $\OO(n^k)$,
the algorithm outputs a non-negative integer $m$ such that, with probability at least $1-\delta$:
\begin{enumerate}[(i)]
    \item $m \le b \cdot e(G[X_1,\dots,X_k])$; and
    \item if $G$ has an $(I,\zeta)$-core, then $m \ge e(G[X_1,\dots,X_k])/b$.
\end{enumerate}
%
% \holger{rephrased for consistency, old version below (Nov 5)}
%     There exists an algorithm $\dlmimprove(\cindora(G), X_1,\dots,X_k, I, \zeta, \delta)$ with the following behaviour. Let $\cost = \{\cost_k\colon k\ge 2\}$ be a regularly-varying parameterised cost function with parameter $k$. Suppose $G$ is an $n$-vertex $k$-hypergraph to which \dlmimprove has (only) colourful oracle access. Suppose that $n \ge 32$ is a power of two, that $X_1,\dots,X_k$ are disjoint subsets of $V(G)$, that $0 < \zeta,\delta < 1/32$ are rational, and that the denominator of $\zeta$ is $\OO(n^k)$. Let
%     \[
%         b=(2k)^{5k}\log^{k-|I|}n\log^{|I|}(1/\zeta), \qquad T = 2^{7k}\log(1/\delta)\log^{2|I|}(1/\zeta)(\log n)^{2(k-|I|)+1}.
%     \]
%     Then in time at most $\OO(Tn\log n)$ and with oracle cost at most $\OO(T \cost_k(n))$, \textnormal{\texttt{LargeCore\allowbreak Coarse}} outputs a non-negative integer $m$ such that, with probability at least $1-\delta$:
%     \begin{enumerate}[(i)]
%         \item $m \le b \cdot e(G[X_1,\dots,X_k])$; and
%         \item if $G$ has an $(I,\zeta)$-core, then $m \ge e(G[X_1,\dots,X_k])/b$.
%     \end{enumerate}
\end{lemma}

We now show that every $k$-hypergraph contains an $(I,\zeta)$-core for some $I$, ensuring that one of these two algorithms is always guaranteed to perform well for a suitably-chosen value of $\zeta$.

\begin{lemma}\label{lem:col-algo-core-exists}
    Let $0 < \zeta \le 1$, and let $H$ be a $k$-partite $k$-hypergraph with vertex classes $X_1,\dots,X_k$. Then there exists $I \subseteq [k]$ such that $H$ contains an $(I,\zeta)$-core.
\end{lemma}
\begin{proof}
    Informally, we will prove that we can find a set $I$ and an $(I,\zeta)$-core $(Y_1,\dots,Y_k)$ by applying the following algorithm. We start out with $I=\emptyset$ and $Y_i = X_i$ for all $i$. We look for a value of $i \notin I$ such that at least $1/(2k)$ proportion of edges of $J \coloneqq H[Y_1,\dots,Y_k]$ are incident to a $(\zeta/2)$-root of $J$ in $Y_i$. If such an $i$ exists, then we add $i$ to $I$, delete all non-$(\zeta/2)$-roots from $Y_i$, and repeat the process. Otherwise, for all $i \notin I$ we delete all $(\zeta/2)$-roots from $Y_i$, then halt; as we will prove, $(Y_1,\dots,Y_k)$ is then an $(I,\zeta)$-core of $H$.
    
    Formally, we proceed by induction, proving the following claim. 
    
    \medskip\noindent\textbf{Claim:} Suppose that there exist $J = H[Y_1,\dots,Y_k]$ and $I \subseteq [k]$ such that:
    \begin{enumerate}[(C1)]
        \item $|E(J)| \ge |E(H)|/(2k)^{|I|}$; and
        \item for all $i \in I$, $|Y_i| \le 2/\zeta$.
    \end{enumerate}
    Then there exist $J' = H[Y_1',\dots,Y_k']$ and $I' \subseteq [k]$ such that either:
    \begin{enumerate}[(D1)]
        \item $(Y_1',\dots,Y_k')$ is an $(I',\zeta)$-core of $H$; or
        \item $J'$ and $I'$ satisfy (C1) and (C2) and $|I'| = |I|+1$.
    \end{enumerate}
    
    \medskip\noindent\textbf{Proof of Lemma from Claim:} Since (C1) and (C2) are satisfied for $J=H$ and $I=\emptyset$, we can apply the claim repeatedly until (D1) holds. This process must terminate after at most $k+1$ applications, since we cannot have $|I| > k$.
    
    \medskip\noindent
    \textbf{Proof of Claim:} Let $J=H[Y_1,\dots,Y_k]$ and $I \subseteq [k]$ be as in the Claim.  For each $i \in [k]\setminus I$, let $Z_i$ be the set of $(\zeta/2)$-roots of $J$ in $X_i$. We split into two cases corresponding to (D1) and (D2) of the Claim.
    
    \medskip\noindent
    \textbf{Case 1:} For all $i \in [k]\setminus I$, we have $\sum_{v \in Z_i}d_J(v) \le e(J)/(2k)$. In this case, we take $Y_i' = Y_i$ for all $i \in I$ and $Y_i' = Y_i \setminus Z_i$ for all $i \in [k]\setminus I$. We claim that $(Y_1',\dots,Y_k')$ is an $(I,\zeta)$-core of $H$ (satisfying (i)--(iii) of \cref{def:core}), as in (D1). If $I=[k]$ then this is immediate from (C1) and (C2), so suppose $|I| \le k-1$. Let $J' = H[Y_1',\dots,Y_k']$. To see that (i) holds, observe that by the definition of the sets $Y_i'$ and by hypothesis,
    \begin{equation}\label{eq:col-algo-core-exists}
        e(J') \ge e(J) - \sum_{i \notin I} \sum_{v \in Z_i}d_J(v) \ge e(J) - (k-|I|)\frac{e(J)}{2k} \ge \frac{e(J)}{2}.
    \end{equation}
    By (C1) this is at least $e(H)/(2k)^{|I|+1} \ge e(H)/(2k)^k$, so (i) holds. Moreover, (ii) follows from (C2). Finally, observe that for all $i \notin I$ and all $v \in Y_i'$, we have $v \notin Z_i$; it follows from the definition of $Z_i$ and~\eqref{eq:col-algo-core-exists} that
    \[
        d_{J'}(v) \le d_J(v) < \zeta e(J)/2 \le \zeta e(J');
    \]
    thus $v$ is not a $\zeta$-root of $J'$, as required by (iii). We have shown that (D1) holds, so we are done.
    
    \medskip\noindent
    \textbf{Case 2:} There exists $\ell \in [k] \setminus I$ such that $\sum_{v \in Z_\ell}d_J(v) > e(J)/(2k)$. In this case, we choose an arbitrary such $\ell$, we take $Y_\ell' = Z_\ell$, $Y_i' = Y_i$ for all $i \ne \ell$, and $I' = I \cup \{\ell\}$.  We claim that $J' \coloneqq J[Y_1',\dots,Y_k']$ and $I'$ satisfy (C1) and (C2) with $|I'| = |I|+1$, as in (D2). We certainly have $|I'| = |I|+1$. To see (C1), observe that by hypothesis,
    % Changed \ge to > in this display to match the case statement. --John
    \[
        e(J') = \sum_{v \in Z_\ell}d_J(v) > e(J)/(2k);
    \]
    hence by (C1) of the inductive hypothesis, we have $|E(H')| \ge |E(G)|/(2k)^{|I|+1}$ as required. To see (C2), observe that, by the definition of $Z_\ell$,
    \[
        e(J) \ge \sum_{v \in Z_\ell}d_J(v) \ge |Z_\ell|\zeta e(J)/2,
    \]
    and hence $|Z_\ell| \le 2/\zeta$. It follows by (C2) of the inductive hypothesis that $|Y_i'| \le 2/\zeta$ for all $i \in I'$, as required by (C2). We have shown that (D2) holds, so we are done.
\end{proof}

With \cref{lem:col-algo-small-core,lem:col-algo-large-core,lem:col-algo-core-exists} in hand, we now set out the algorithm \newcoarsecount required by \cref{lma:better-col-alg} as \cref{algo:coarsecount}. (Recall from the statement of \cref{lma:better-col-alg} that $\alpha_k \in [0,k]$ is the index of our regularly-varying cost function.)
We now restate \cref{lma:better-col-alg} and prove it.

\begin{algorithm}
    \SetKwInput{Oracle}{Oracle}
    \SetKwInput{Input}{Input}
	\SetKwInput{Output}{Output}
	\DontPrintSemicolon
    \Oracle{Colourful independence oracle $\cindora(G)$ of an $n$-vertex $k$-hypergraph~$G$.}
	\Input{Integer~$n$ that is a power of two and satisfies $k \le (\log n)/(\log\log n)^2$, and disjoint subsets $X_1, \dots, X_k \subseteq V(G)$.}
	\Output{Non-negative integer $m$ such that, with probability at least $2/3$, $m/b \le e(G[X_1,\dots,X_k]) \le mb$, where $\alpha_k' \coloneqq \ceil{\alpha_k} - 1$ and
	\[
        b\coloneqq (2k)^{5k}(\log n)^{k - \alpha_k'-1}\log^k\big(8k \log^{40k(k-\alpha_k')/(\alpha_k - \alpha_k')}(n)\big).
    \]}
	\Begin{
	    Set $t\coloneqq \floor{n/(\log n)^{20k/(\alpha_k - \alpha_k')}}$ \label{line:coarse-setup1}.\; 
	    Set $\zeta \coloneqq (t/n)^{2k(k-\alpha_k')}/(8k)$ \label{line:coarse-setup2}.\;
	    \If{$n < 32$ \textbf{\upshape or} $t < 12\log k$ \label{line:coarse-bruteforce-cond}}{
	        Count edges of $G[X_1,\dots,X_k]$ by brute force, return the answer, and halt. \label{line:coarse-bruteforce}
	    }
	    \ForAll{$R \subseteq [k]$}{
	        \If{$|R|\le \alpha_k'$}{
	            Set $Z_R = \dlmrecurse(\cindora(G),X_1,\dots,X_k,R,t,1/2^{k+5})$.\; 
	       }
	       \If{$|R| \ge \alpha_k' + 1$}{
	            Set $Z_R = \dlmimprove(\cindora(G),X_1,\dots,X_k,R,\zeta,1/2^{k+5})$.\; \label{line:call-largecore}
	       }
	   }
	   \Return $\max\{Z_R: |R| \subseteq [k]\}$. \label{line:coarse-takemax}\;
	 }
    \caption{\label{algo:coarsecount}\newcoarsecount\\
    \textit{This algorithm applies either $\dlmrecurse$ or $\dlmimprove$ with each possible choice $R$ of root classes, depending on the size of the set $R$, and returns the maximum number of edges estimated by any of these calls.}}
\end{algorithm}

\stateColourCoarseNew*
\begin{proof}
    We first prove that \dlmimprove has the correct behaviour, before we turn to time and cost.

    \paragraph*{Behaviour.}
    We begin by noting that the conditions of \cref{lem:col-algo-small-core} (respectively \cref{lem:col-algo-large-core}) are satisfied each time we call the relevant subroutine. For \dlmimprove, it is immediate that when reaching line~\ref{line:call-largecore} we have $n \ge 32$. We trivially have $1/2^{k+5} < 1/32$, and since $\zeta = (t/n)^{2k(k-\alpha_k')}/(8k) \le 1/\log^{20} n$ and $n>4$, we also have $\zeta < 1/32$. For \dlmrecurse, we trivially have $n \ge t$ and that $t$ is an integer, and Lines \ref{line:coarse-bruteforce-cond}-\ref{line:coarse-bruteforce} guarantee that $t \ge 12 \log k$. Finally, we have $\alpha_k' \le k-1$, so \dlmrecurse is always called with $R\subset [k]$.
    
    %%% Old version of the t >= log(k) argument. --John
    %
    % However, this follows immediately from the fact that if we reach Step (A2) we have $n \ge \max\{2^{2^{k^2}},2^{30k/(\alpha_k - \alpha_k')}\}$, implying that
    % \[
    %     n > 2^{\left(\frac{20k + \alpha_k - \alpha_k'}{\alpha_k - \alpha_k'}\right)^2}
    % \]
    % and hence that $\log n / \log \log n > (20k + \alpha_k - \alpha_k')/(\alpha_k - \alpha_k')$; this implies that $n \ge (\log n)^{20k/(\alpha_k - \alpha_k') + 1}$ and hence that $t \ge \log n \ge 12 \log k$ (since $\log n \ge 2^{k^2} > 12 \log k$). 
    
    % Reworded a bit for clarity and hboxes. --John
    Observe that the number of calls made to \dlmrecurse and \dlmimprove combined is $2^k$. \cref{lem:col-algo-small-core,lem:col-algo-large-core} imply that the results of each such call lie within certain bounds with failure probability at most $1/2^{k+5}$; by a union bound, it follows that with probability at least $31/32$, no call to \dlmrecurse or \dlmimprove fails in this way. In this case, we say that \newcoarsecount \emph{succeeds}. It therefore suffices to prove that, whenever \newcoarsecount succeeds, its output $m$ satisfies $m/b \le e(G[X_1,\dots,X_k]) \le mb$.
    % This isn't a "b-approximation" in our terminology, since we already defined an \eps-approximation in a way that conflicts.
    
    First, note if \newcoarsecount succeeds, then by \cref{lem:col-algo-small-core}(i) and \cref{lem:col-algo-large-core}(i), every invocation of either \dlmrecurse or \dlmimprove returns a value that is at most $b \cdot e(G[X_1,\dots,X_k])$, so we have $m \le b \cdot e(G[X_1,\dots,X_k])$ as required.
    
    To see that \newcoarsecount returns a value that is not too small, recall from \cref{lem:col-algo-core-exists} that there exists $I \subseteq [k]$ such that $G[X_1,\dots,X_k]$ contains an $(I,\zeta)$-core.  Suppose first that $|I| \le \alpha_k'$. In this case we obtain $Z_I$ by invoking \dlmrecurse with $R=I$. Observe that, by definition, 
    \[  (8k\zeta)^{1/(2(k-|I|))}n = (t/n)^{2k(k-\alpha_k')/(2(k - |I|))} n \le (t/n)^kn \le t,  \]
    so it follows from \cref{lem:col-algo-small-core}(ii) that $Z_I \ge e(G[X_1,\dots,X_k])/b$.  Since our output $m$ is the maximum over all values $Z_R$, it follows that $m \ge e(G[X_1,\dots,X_k])/b$, as required.  Now suppose that $|I| \ge \alpha_k' + 1$.  In this case, we obtain $Z_I$ by invoking \dlmimprove with $R=I$, and it follows from \cref{lem:col-algo-large-core}(ii) that $Z_R \ge e(G[X_1,\dots,X_k])/b$.  As before, it is then immediate that $m \ge e(G[X_1,\dots,X_k])/b$, as required.

    \paragraph*{Running time and oracle cost.}
    We now prove the claimed bounds on the time and cost of \newcoarsecount.
    Since $k \le (\log n)/(\log\log n)^2$, we have $\log^k(n) \le 2^{(\log n)/(\log\log n)} = n^{o(1)}$ as $n\to\infty$, and so 
    \begin{equation}\label{eq:coarsecount-time-bounds-t}
        t = n^{1-o(1)}.
    \end{equation}
    Let $\eta = (\alpha_k-\alpha_k')/2$. Since $\cost_k$ is regularly-varying with index $\alpha_k$ and a slowly-varying component which does not depend on $k$, by \cref{lem:regularly-varying-facts}(ii) there exists $x_0$ such that 
    \begin{equation}\label{eq:coarsecount-reg-vary}
        \mbox{for all $k \ge 2$, all $x \ge x_0$ and all $A_x \ge 1$, }\cost_k(A_xx)/\cost_k(x) \ge A_x^{\alpha_k - \eta}.    
    \end{equation}
    Throughout the proof, we assume without loss of generality that $n$ is sufficiently large that $t \ge x_0$. 
    
    Suppose we solve the problem by brute force in line~\ref{line:coarse-bruteforce}, with time and oracle cost $\OO(n^k)$. Since $t \le 12\log k \le \log n$, and since $t=n^{1-o(1)}$ by~\eqref{eq:coarsecount-time-bounds-t}, we have $\log n \le n^{1-o(1)}$ and hence $n = \OO(1)$. It follows that $n^k = 2^{\OO(k)} = \OO(k^k)$. For the rest of the proof, suppose we do not solve the problem by brute force in line~\ref{line:coarse-bruteforce}.
    
    First, we bound the oracle cost of invocations of \dlmrecurse. Let
    \begin{align}\label{eq:coarsecount-T1}
        T_1 &= \max\big\{\log(2^{k+5})k^{6k+1}\log^{4k+8}(n)(n/t)^i\colon 0 \le i \le \alpha_k'\big\};
    \end{align}
    then by \cref{lem:col-algo-small-core}, each invocation of \dlmrecurse has oracle cost $\OO(T_1\cost_k(2kt))$. Observe that the maximum in~\eqref{eq:coarsecount-T1} is attained at $i=\alpha_k'$, so
    \begin{align}\label{eq:coarsecount-T1-2}
        T_1 &\le k^{6k+2}\log^{4k+8}(n)(n/t)^{\alpha_k'} \le (T/k^{3k})\log^{4k-1}(n)(n/t)^{\alpha_k'}
    \end{align}   
    Since \dlmrecurse is invoked at most $2^k$ times, by~\eqref{eq:coarsecount-T1-2} and~\eqref{eq:coarsecount-reg-vary} applied with $x=2kt$ and $A_x=n/(2kt)$, the total oracle cost of all such invocations is at most
    \begin{align*}
        \OO(T_1\cost_k(2kt)) &= \OO\bigg(2^k\frac{T}{k^{3k}}\log^{4k}(n)\Big(\frac{n}{t}\Big)^{\alpha_k'}\Big(\frac{2kt}{n}\Big)^{\alpha_k-\eta}\cost_k(n)\bigg)\\
        &= \OO\bigg(T\log^{4k}(n)\Big(\frac{t}{n}\Big)^{\alpha_k-\alpha_k'-\eta}\cost_k(n)\bigg).
    \end{align*}
    Substituting in the values of $t$ and $\eta$ yields an oracle cost of at most
    \[
        \OO\big(T\log^{4k}(n)\log^{-10k}(n)\cost_k(n)\big) = \OO\big(T\cost_k(n)\big),
    \]
    as required.
    
    We now bound the running time of invocations of \dlmrecurse. By \cref{lem:col-algo-small-core}, each invocation of \dlmrecurse has running time $O(n+T_1kt\log n)$. When $\alpha_k > 1$, we have $T_1kt\log n = \OO(T_1\cost_k(2kt))$, and so the running time is $\OO(T(\cost_k(n)+n))$ as above. Suppose instead $\alpha_k \le 1$, so that $\alpha_k' = 0$. Since \dlmrecurse is invoked at most $2^k$ times, by~\eqref{eq:coarsecount-T1-2} it follows that the total running time is at most
    \begin{align*}
        \OO\Big(2^kn+(T/k^k)\log^{4k}(n)t\Big) &=
        \OO\Big(n\Big(2^k+(T/k^k)\log^{4k}(n)(t/n)\Big)\Big)\\
        &=\OO\Big(n\Big(T+T\log^{4k - 20k/\alpha_k}(n)\Big)\Big).
    \end{align*}
    Since $\alpha_k \le 1$, this is $\OO(nT)$. Thus in both cases, the running time is $\OO(T(\cost_k(n)+n))$ as required.
    
    We now bound the oracle cost of invocations of \dlmimprove. Let
    \begin{equation}\label{eq:coarsecount-T2}
        T_2 = \max\big\{2^{7k}\log(2^{k+5})\log^{2i}(1/\zeta)\log^{2(k-i)+1}(n)\colon \alpha_k'+1 \le i \le k\big\};
    \end{equation}
    then by \cref{lem:col-algo-large-core}, each invocation of \dlmimprove has oracle cost $\OO(T_2\cost_k(n))$. Observe that, by the definitions of $\zeta$ and $t$,
    and since $t \ge 1$ by~\eqref{eq:coarsecount-time-bounds-t},
    \begin{equation}\label{eq:coarsecount-time-bounds-zeta}
        \log(1/\zeta) = \log\big(8k(n/t)^{2k(k-\alpha_k')}\big) = \OO(k^2\log(n/t)) = \OO(k^3 \log\log n).
    \end{equation}
    If $k \ge (\log \log n)/(\log \log \log n)^2$, then 
    \[
        k^{3k} = 2^{3k\log k} \ge 2^{3k(1-o(1))\log\log\log n} = \Omega((\log \log n)^k);
    \]
    if instead $k \le (\log\log n)/(\log\log\log n)^2$, then
    \[
        (\log\log n)^k \le 2^{(\log\log n)/(\log\log\log n)} = (\log n)^{o(1)}.
    \]
    In either case, by~\eqref{eq:coarsecount-time-bounds-zeta} we have $\log^{2k}(1/\zeta) = \OO(k^{6k}\log n)$. It follows
    from~\eqref{eq:coarsecount-T2} that
    \begin{equation}\label{eq:coarsecount-T2-2}
        T_2 = \OO\Big(2^{7k}k^{6k+1}\log^{2(k-\alpha_k')}(n)\Big) = \OO\big(T/(2^k\log n)\big).
    \end{equation}
    Thus the total oracle cost of all $2^k$ invocations of \dlmimprove is
    % display to avoid overflow
    \[\OO(2^kT_2\cost_k(n)) = \OO(T\cost_k(n))\]
    as required. 
    
    Similarly, by \cref{lem:col-algo-large-core} and~\eqref{eq:coarsecount-T2-2}, the running time of all $2^k$ invocations of \texttt{LargeCore{\allowbreak}Coarse} is $\OO(2^kT_2n\log n) = \OO(Tn)$; thus the total time is $\OO(T(\cost_k(n)+n))$ as required.
    
    Finally, we observe that lines \ref{line:coarse-setup1}, \ref{line:coarse-setup2} and \ref{line:coarse-takemax} take $\OO(k2^k)$ time (including arithmetic operations on $(k\log n)$-bit numbers) and make no oracle calls. The result therefore follows.
\end{proof}

\subsubsection{Counting edges with a small core}\label{sec:dlm-recurse}

As described in \cref{sec:coarse-approx-outline}, our second approximate counting algorithm (which will be efficient when the input graph has an $(I,\zeta)$-core with $|I|$ small) will make use of the main counting algorithm from \cite{DLM} as a subroutine; we can paraphrase this result as follows. 

% \kitty{TODO Change theorem number in citation when we can cite the journal version (it's Theorem 1.1 in the accepted version).}
% \john{I vote we just use the accurate theorem numbering and then have these particular citations go to the arXiv version if the journal version isn't ready yet --- the conference version doesn't contain the proofs we actually want anyway.}

\begin{theorem}[Dell, Lapinskas, and Meeks {\cite[Theorem 1.1 paraphrased]{DLM}}]
\label{thm:old-counting}
Let $\cost = \{\cost_k\colon k\ge 2\}$ be a regularly-varying parameterised cost function with parameter $k$.
There is a randomised \cindora-oracle algorithm $\oldcount(\cindora(G),\eps,\delta)$
with
worst-case running time
$\OO(\log(1/\delta)\eps^{-2}k^{6k}(\log n)^{4k+8}n)$,
worst-case oracle cost
$\OO(\log(1/\delta)\eps^{-2}k^{6k}(\log n)^{4k+7}\cost_k(n)$, and
the following behaviour:
Given an $n$-vertex $k$-hypergraph~$G$ and rationals~$\eps,\delta\in(0,1)$,
the algorithm outputs a rational number that, with probability at least $1-\delta$, is an $\eps$-approximaiton to~$e(G)$.
\end{theorem}

% Added this paragraph to motivate HelperSmallCoreCoarse versus SmallCoreCoarse. --John
We can now state the heart of our algorithm, \helperdlmrecurse. Given correctness and time and cost bounds on \helperdlmrecurse (which we will prove as \cref{lma:colourcoarse-correct}), \cref{lem:col-algo-small-core} will follow immediately by applying \cref{lem:median-boosting} to reduce the failure probability of \helperdlmrecurse from $1/3$ to $\delta$ with $\OO(\log(1/\delta))$ overhead.

\begin{algorithm}
    \SetKwInput{Oracle}{Oracle}
    \SetKwInput{Input}{Input}
	\SetKwInput{Output}{Output}
	\DontPrintSemicolon
    \Oracle{Colourful independence oracle $\cindora(G)$ of an $n$-vertex $k$-hypergraph~$G$.}
	\Input{Positive integer $n$ that is a power of two, disjoint subsets $X_1, \dots, X_k \subseteq V(G)$, set $I \subseteq [k]$, and integer $t$ satisfying $n \ge t \ge 12\log k$.}
	\Output{Non-negative integer $Z$ that satisfies the following properties with probability at least $2/3$, where $b \coloneqq (2k)^{k+2}$. Firstly, $e(G[X_1,\dots,X_k]) \le Zb$. Secondly, if $G[X_1,\dots,X_k]$ has an $(I,\zeta)$-core for some $\zeta$ satisfying $t \ge (8k\zeta)^{1/(2(k-|I|))}n$, then $e(G[X_1,\dots,X_k]) \ge Z/b$.}
	\Begin{
	    \ForAll{$i \in I$}{
	        Compute an arbitrary partition $X_{i,1},\dots,X_{i,x_i}$ of $X_i$ into $x_i \le \lceil n/t \rceil$ sets each of cardinality at most $t$.\; \label{line:dlmrecurse-A1}
	    }
	    \ForAll{$i \in [k] \setminus I$}{
	        Compute a random subset $X_{i,1} \subseteq X_i$ by retaining each element independently with probability $t/n$.\;
	        \If{$|X_{i,1}| > 2t$}{Return an arbitrary value and halt.\; \label{line:halt-bigset}}
	    }
	    Compute the set $\Pi$ of all functions $\pi$ with domain $[k]$ such that $\pi(i) \in [x_i]$ for each $i \in I$ and $\pi(i)=1$ for all $i \notin I$.\;
	    \ForAll{$\pi \in \Pi$}{
	        Set $H_\pi \coloneqq G[X_{1,\pi(1)},\dots,X_{k,\pi(k)}]$.\;\label{line:H-Pi}
	    }
	    \Return $Z\coloneqq\frac{n^{k - |I|}}{t^{k - |I|}}\sum_{\pi\in\Pi}\oldcount(\cindora(H_\pi),1/2,1/(12n^k))$.\; \label{line:smallcore-output}
    }
    \caption{\label{algo:helperdlmrecurse}\helperdlmrecurse\\
    \textit{
    This algorithm computes an estimate of the number of edges in $G$ which is unlikely to be much too large and, if $G$ has an $(I,\zeta)$-core satisfying certain properties, is also unlikely to be much too small.  Its running time increases with the size of the set $I$.}}
\end{algorithm}

Observe that line~\ref{line:halt-bigset} guarantees that $|V(H_\pi)| \le 2kt$, so the calls to \oldcount will run quickly -- this is where the fast running time will come from. Moreover, the graphs~$H_\pi$ of line~\ref{line:H-Pi} are all edge-disjoint, and that each edge of $G[X_1,\dots,X_k]$ survives in some~$H_\pi$ with probability $(t/n)^{k - |I|}$; thus we have $\E(Z)=e(G[X_1,\dots,X_k])$. Establishing concentration here will be central to the correctness proof. To this end, we first recall the following standard martingale concentration bound due to McDiarmid~\cite{mcdiarmid_1989}; we then apply it to prove concentration for specific graphs $H_\pi$ in \cref{lem:colourful-algo-conc}.
	
\begin{lemma}\label{lem:mcdiarmid}
	Let $f$ be a real function of independent real random variables $Z_1,\dots,Z_m$, and let $\mu = \E(f(Z_1,\dots,Z_m))$. Let $c_1,\dots,c_m \ge 0$ be such that, for all $i \in [m]$ and all pairs $(\mathbf{x},\mathbf{x'})$ differing only in the $i$'th coordinate, we have $|f(\mathbf{x}) - f(\mathbf{x'})| \le c_i$. Then for all $y > 0$,
	\[
		\pr\big(|f(Z_1,\dots,Z_m) - \mu| \ge y\big) \le 2e^{-2y^2/\sum_{i=1}^m c_i^2}.
	\]
\end{lemma}
	
% We say that $S \subset V$ is a \emph{$(\zeta,r)$-root} if $|S| = r$ and at least $\zeta|E|$ edges in $E$ contain $S$ as a subset.  Note that if $S$ is a $(\zeta,r)$-root then any subset $S'$ of $S$ is necessarily a $(\zeta,|S'|)$-root.
	
\begin{lemma}\label{lem:colourful-algo-conc}
	Let $H$ be an $n$-vertex $k$-partite $k$-hypergraph with vertex classes $X_1,\dots,X_k$. Let $0 < \zeta,p < 1$, let $I \subseteq [k]$, and suppose $H$ has an $(I,\zeta)$-core. For all $i \in I$, let $X_i' = X_i$; for all $i \in [k] \setminus I$, let $X_i'$ be a random subset of $X_i$ in which each element is included independently with probability $p$. Then, setting $H' = H[X_1',\dots,X_k']$, we have
	\[
	    \pr\big(e(H') < p^{k-|I|}e(H)/(2k)^{k+1}\big) \le 2\exp\big({-}p^{2(k-|I|)}/(2k\zeta)\big).
	\]
% 	let $y > 0$, let $r < k$ be a positive integer, and suppose that $G$ has no $(\zeta,1)$-roots in $X_{r+1} \cup \dots \cup X_k$. For all $1 \le i \le r$, let $X_i' = X_i$; for all $r+1 \le i \le k$, let $X_i'$ be a random subset of $X_i$ in which each element is included independently with probability $p$. Then, setting $G' = G[X_1',\dots,X_k']$, we have
% 	\[
% 		\pr\Big(\big|e(G') - p^{k-r}e(G)\big| > y\Big) \le 2e^{-2y^2/k\zeta e(G)^2}.
% 	\]
\end{lemma}
\begin{proof}
    Let $r\coloneqq|I|$. Without loss of generality, suppose $I = [r]$ (otherwise we can reorder $X_1,\dots,X_k$). By hypothesis, $H$ has an $(I,\zeta)$-core; denote this by $(Y_1,\dots,Y_k)$. For all $i \in [k]$, let $Y_i' = Y_i \cap X_i'$; thus, $Y_1',\dots,Y_r'$ are equal to $Y_1,\dots,Y_r$, and $Y_{r+1}',\dots,Y_k'$ are formed by randomly deleting vertices from $Y_{r+1},\dots,Y_k$. Let $J \coloneqq H[Y_1,\dots,Y_k]$ and $J' \coloneqq H[Y_1',\dots,Y_k']$.
    
    We next show that it suffices to prove concentration of $e(J')$. Observe by linearity of expectation applied to indicator random variables for each edge in $E(J)$ that $\E(e(J')) = p^{k-r}e(J)$. Moreover, observe that $e(H') \ge e(J')$, and recall from \cref{def:core}(i) that since $Y_1,\dots,Y_k$ is a core we have $e(J) \ge e(H)/(2k)^k$; thus whenever $e(J') \ge p^{k-r}e(J)/2$, we also have $e(H') \ge p^{k-r}e(H)/(2k)^{k+1}$. It follows that
    \begin{equation}\label{eq:col-algo-conc-J}
        \pr\big(e(H') < p^{k-r}e(H)/(2k)^{k+1}\big) \le \pr\big(e(J') < p^{k-r}e(J)/2\big).
    \end{equation}
    
    % Added more detail here to help the reader. --John
    We now bound the right-hand side of~\eqref{eq:col-algo-conc-J} above by applying \cref{lem:mcdiarmid} to $e(J')$. Observe that $e(J')$ is a function of the independent indicator variables for the events $\{v \in Y_j'\}$ for $j \ge r+1$, and that modifying any of those indicator variables --- that is, adding or removing a vertex $v$ from some $Y_j'$ --- affects $e(J')$ by at most $d_J(v)$. Thus by \cref{lem:mcdiarmid},
    % v_1, ..., v_\ell were undefined, but the meaning was clear.
    \begin{equation}\label{eq:col-algo-conc-mcd}
        \pr\Big(e(J') < p^{k-r}e(J)/2\Big) \le 2\exp\bigg({-}p^{2(k-r)}e(J)^2\Big/\Big(2\sum_{i=r+1}^k\sum_{v \in Y_i} d_J(v)^2\Big)\bigg).
    \end{equation}
    
    Now, since $(Y_1,\dots,Y_k)$ is an $(I,\zeta)$-core for $H$, by \cref{def:core}(iii) we have $d_J(v) \le \zeta e(J)$ for all $v \in Y_{r+1} \cup \dots \cup Y_k$. Moreover, since $H$ is $k$-partite we have $\sum_{v \in Y_i} d_J(v_\ell) = e(J)$ for all $i \in [k]$. Thus
    \[
        \sum_{i=r+1}^k\sum_{v \in Y_i} d_J(v)^2 \le \zeta e(J) \sum_{i=r+1}^k\sum_{v \in Y_i} d_J(v) \le k\zeta e(J)^2.
    \]
    By~\eqref{eq:col-algo-conc-J} and~\eqref{eq:col-algo-conc-mcd}, it follows that
    \[
        \pr\big(e(H') < p^{2(k-|I|)}e(H)/(2k)^{k+1}\big) \le 2\exp\big({-}p^{2(k-r) }/(2k\zeta)\big),
    \]
    as required.
\end{proof}

We now prove correctness of \helperdlmrecurse.

% Reworded slightly for hboxes. --John
\begin{lemma}\label{lma:colourcoarse-correct}
    \helperdlmrecurse behaves as stated. Moreover, let 
    \[
        T \coloneqq k^{6k+1}\log^{4k+8}(n)n^{|I|}/t^{|I|}.
    \]
    Then on an $n$-vertex $k$-hypergraph $G$, $\helperdlmrecurse(\cindora(G),X_1,\dots,X_k,I,t)$ has oracle cost $\OO(T\cost_k(2kt))$ and running time $\OO(Ttk\log n + n)$.
\end{lemma}
\begin{proof}
    \textit{Running time and oracle cost.} Since we halt early at line~\ref{line:halt-bigset} whenever any set $X_{i,j}$ contains more than $2t$ vertices, \oldcount is only invoked on graphs with at most $2kt$ vertices. By \cref{thm:old-counting}, it follows that each invocation of \oldcount runs in time $\OO(k^{6k+1}\log^{4k+9}(n)\cdot  2kt)$ and has oracle cost $\OO(k^{6k+1}\log^{4k+8}(n)\cost_k(2kt))$. Since $|\Pi| = \OO((n/t)^{|I|})$, it follows that line~\ref{line:smallcore-output} runs in time $\OO(Ttk\log n)$ and has oracle cost $\OO(T\cost_k(2kt))$. The rest of the algorithm runs in time $\OO(nk)$ and does not call the oracle, so the desired bounds follow.

    \textit{Correctness.} Let $r\coloneqq|I|$, and let $H \coloneqq G[X_1,\dots,X_k]$. Without loss of generality, suppose that $I = [r]$ (otherwise we can reorder $X_1,\dots,X_k$). We will argue that, with high probability, none of the following bad events occur:
    \begin{enumerate}
        \item[$\calE_1$:] we halt at line~\ref{line:halt-bigset} due to a set $X_{i,1}$ having size greater than $2t$;
        \item[$\calE_2$:] at least one invocation of \oldcount at line~\ref{line:smallcore-output} returns a value that is not a $(1/2)$-approximation to $e(H_\pi)$;
        \item[$\calE_3$:] we do not have $Z \le b\cdot e(H)$;
        \item[$\calE_4$:] $H$ has an $(I,\zeta)$-core for some $\zeta$ satisfying $t \ge (8k\zeta)^{1/(2(k-r))}n$ and we do not have $Z \ge e(H)/b$.
    \end{enumerate}
    If none of these events occur, then the algorithm behaves as stated. Hence
    \begin{equation}\label{eq:dlmrecurse-prob-bound}
        \pr(\mbox{\dlmrecurse fails}) \le \pr(\calE_1) + \pr(\calE_2) + \pr\big(\calE_3 \mid \neg\calE_1, \neg\calE_2)\big) + \pr\big(\calE_4 \mid \neg\calE_1, \neg\calE_2)\big).
    \end{equation}
    
    We first bound $\pr(\calE_1)$ above. For all $i \in I$, we have $|X_{i,j}| \le t$ for all $j$ by construction in line~\ref{line:dlmrecurse-A1}. For all $i \in [k]\setminus I$, we have $x_i = 1$, and $|X_{i,1}|$ follows a binomial distribution with mean $|X_i|t/n \le t$. Since $t \ge 12\log k$, it follows by a standard Chernoff bound (\cref{lem:chernoff-low} with $\delta=1$) that
    \begin{equation}\label{eq:dlmrecurse-E1-bound}
        \pr(\calE_1) \le 2e^{-t/3} \le 2/e^4 < 1/12.
    \end{equation}

    We next observe that, by \cref{thm:old-counting}, each invocation of \oldcount fails with probability at most $1/(12n^k)$. There are $|\Pi| \le \lceil n/t\rceil ^{|I|} \le n^k$ invocations in total, so by a union bound we have
    \begin{equation}\label{eq:dlmrecurse-E2-bound}
        \pr(\calE_2) \le 1/12.
    \end{equation}

    % Reordered for hbox.
    We next observe that, writing $H' = G[X_1,\dots,X_r,X_{r+1,1},\dots,X_{k,1}]$, since $H$ is $k$-partite, we have
    \[
        e(H') = \sum_{\pi\in\Pi} e(H_\pi).
    \]
    Hence, conditioned on $\neg\calE_1$ and $\neg\calE_2$, i.e.\ conditioned on \helperdlmrecurse reaching line~\ref{line:smallcore-output} and on the calls to \oldcount returning valid approximations, we have
    \begin{equation}\label{eq:dlmrecurse-H-bound}
        (n/t)^{k - r}e(H')/2 \le Z \le 2(n/t)^{k-r}e(H').
    \end{equation}
    
    We now bound $\pr(\calE_3 \mid \neg\calE_1,\neg\calE_2)$ above using~\eqref{eq:dlmrecurse-H-bound}. Each edge of $G[X_1,\dots,X_k]$ survives in $H'$ with probability $(t/n)^{k - r}$, so by linearity of expectation we have
    % Technically we have this conditioning to cart round, so I defined \mu to take the edge off. --John
    \[
        \mu \coloneqq \E(e(H')\mid \neg\calE_1,\neg\calE_2) = e(H)(t/n)^{k - r}.
    \]
    It follows by the fact that $k \ge 2$,~\eqref{eq:dlmrecurse-H-bound}, and Markov's inequality that
    \begin{align}\nonumber
        \pr\big(Z \ge e(H)\cdot b \mid \neg\calE_1,\neg\calE_2\big) 
        &\le \pr\big(Z \ge 24\mu(n/t)^{k-r} \mid \neg\calE_1,\neg\calE_2\big)\\\label{eq:dlmrecurse-E3-bound}
        &\le \pr\big(e(H') \ge 12 \mu \mid \neg\calE_1,\neg\calE_2\big) \le 1/12.
    \end{align}
    
    Finally, we bound $\pr(\calE_4\mid\neg\calE_1,\neg\calE_2)$ above. Suppose $G[X_1,\dots,X_k]$ has an $(I,\zeta)$-core for some $\zeta$ satisfying $t \ge (8k\zeta)^{(1/2(k-|I|))}n$ (since otherwise this probability is zero). Using~\eqref{eq:dlmrecurse-H-bound} and \cref{lem:colourful-algo-conc}, taking $p = t/n$, we have
    % This was a shorter duplicated version of the below, so I've deleted it. --John
    % \[
    %     \pr(\calE_4 \mid \neg\calE_1,\neg\calE_2) \le \pr(e(H') < e(H)/(2k)^{k+2}\mid \neg\calE_1,\neg\calE_2) \le 2\exp\big({-}(t/n)^{2(k-r)}/(2k\zeta)\big).
    % \]
    \begin{align*}
        \pr(\calE_4 \mid \neg\calE_1,\neg\calE_2) &= \pr(Z < e(H)/(2k)^{k+2} \mid \neg\calE_1,\neg\calE_2)\\
        &\le \pr((n/t)^{k-r}e(H')/2 < e(H)/(2k)^{k+2} \mid \neg\calE_1,\neg\calE_2)\\
        &\le \pr(e(H') < (t/n)^{k-r}e(H)/(2k)^{k+1}\mid \neg\calE_1,\neg\calE_2) \\
        &\le 2\exp\big({-}(t/n)^{2(k-r)}/(2k\zeta)\big).
    \end{align*}
    Since $t \ge (8k\zeta)^{1/(2(k-r))}n$ by hypothesis, it follows that
    \begin{equation}\label{eq:dlmrecurse-E4-bound}
        \pr(\calE_4 \mid \neg\calE_1,\neg\calE_2) \le 2e^{-4} < 1/12.
    \end{equation}
    The result now follows from~\eqref{eq:dlmrecurse-prob-bound}, \eqref{eq:dlmrecurse-E1-bound}, \eqref{eq:dlmrecurse-E2-bound}, \eqref{eq:dlmrecurse-E3-bound} and \eqref{eq:dlmrecurse-E4-bound}.
\end{proof}

\cref{lem:col-algo-small-core} now follows immediately from \cref{lma:colourcoarse-correct} together with \cref{lem:median-boosting}, which is used to reduce the failure probability from $1/3$ to $1 - \delta$ for arbitrary rational $\delta \in (0,1)$.

\subsubsection{Counting edges with a large core}\label{sec:dlm-coarse}

In this section, we adapt the algorithm \oldcoarsecount of~\cite{DLM} into \dlmimprove, which runs faster but still gives a good approximation when 
% Reworded since the running time guarantee doesn't require the core to exist.
supplied with a set $I$ and a rational number $\zeta$ such that the input graph $H=G[X_1,\dots,X_k]$ has an $(I,\zeta)$-core. The crucial ingredient in~\cite{DLM} is a subroutine \verifyguess which takes as input a guess $M$ of the number of edges in $H$, and distinguishes --- with reasonable probability -- two cases: the number of edges is at least $M$, or the number of edges is \emph{much} less than $M$. It is then relatively easy to use this subroutine to implement the algorithm of~\cite{DLM} by using binary search to find the least guess $M$ which \verifyguess accepts. Our adaptation follows exactly the same structure; we first present the analogue of \verifyguess.

%We say that $S \subset V$ is a \emph{$(\zeta,r)$-root} if $|S| = r$ and at least $\zeta|E|$ edges in $E$ contain $S$ as a subset.  Note that if $S$ is a $(\zeta,r)$-root then any subset $S'$ of $S$ is necessarily a $(\zeta,|S'|)$-root.

% \john{Again note modifications to requirements for $\zeta$ here and throughout the section, weakened from requiring it to be a power of two to bounding the denominator above to gel better with \cref{sec:coarse-approx-outline}.}

\begin{algorithm}
    \SetKwInput{Oracle}{Oracle}
    \SetKwInput{Input}{Input}
	\SetKwInput{Output}{Output}
	\DontPrintSemicolon
    \Oracle{Colourful independence oracle $\cindora(G)$ of an $n$-vertex $k$-hypergraph~$G$.}
	\Input{Positive integers~$n$ and~$M$ that are powers of two and satisfy $n \ge 32$, disjoint subsets $X_1, \dots, X_k \subseteq V(G)$, set $I \subseteq [k]$, and rational number $\zeta \in (0,1/32)$ with denominator $\OO(n^k)$.}
	\Output{Either \yes or \no. Setting 
	\[\pout=1/\big(2^{5k}\log^{|I|}(1/\zeta)\log^{k-|I|} n\big),\]
	we require the following two properties. \textit{Completeness} ensures that if $G[X_1,\dots,X_k]$ contains an $(I,\zeta)$-core and $e(G[X_1,\dots,X_k]) \ge M$, then \newverifyguess outputs \yes with probability at least $\pout$. \textit{Soundness} ensures that if $e(G[X_1,\dots,X_k]) < M \cdot \pout/((8k)^k\log^{|I|}(1/\zeta)\log^{k-|I|}n)$, then \newverifyguess outputs \no with probability at least $1-\pout/2$.
	}
 	\Begin{
	For each $i \in [k]$ and each $0 \le j \le 2\log n -1$, construct a subset $Z_{i,j}$ of $X_i$ by including each vertex independently with probability $1/2^j$. \label{line:verify-setup1}\;
	Construct the finite set $A$ of all tuples $(a_1, \dots, a_k)$ of non-negative integers satisfying: $a_i \le 2\log n$ for all $i \in [k]$; $a_i \le 2\log(1/\zeta)+1$  for all $i \in I$; and $a_1+\dots+a_k \ge \log M - k\log(2k)$.\label{line:verify-setup2}\;
	 
	 \ForAll{$(a_1, \dots, a_k) \in A$ \label{line:verify-for}}{
        \If{$\cindora(G)_{Z_{1,a_1},\dots,\allowbreak Z_{k,a_k}} = 0$}{return \yes.\;\label{line:verify-returnYES}}
 	 }
	 \Return \no.\; \label{line:verify-returnNO}
	 }
    \caption{\label{algo:newverifyguess}\newverifyguess\\
    \textit{This algorithm, with reasonable probability, distinguishes between two cases: the number of edges in the input hypergraph is much less than the guess $M$, or the input hypergraph contains an $(I,\zeta)$-core and its number of edges is at least $M$.}
    }
\end{algorithm}

The main difference between \newverifyguess and the algorithm \verifyguess of~\cite{DLM} is the choice of the set $A$, which is much smaller; this is possible because of the $(I,\zeta)$-core, and leads to an improved running time and oracle cost. The main difference in the proofs is that we must show that the presence of the core implies that completeness still holds even with this smaller set $A$, i.e.\ that if $E(G[X_1,\dots,X_k]) \ge M$ then with reasonable probability there still exists $(a_1,\dots,a_k) \in A$ such that $e(G[Z_{1,a_1},\dots,Z_{k,a_k}]) > 0$ and the algorithm outputs \yes in line~\ref{line:verify-returnYES}. The following lemma is a generalisation of part of the proof of \cite[Lemma~4.1]{DLM}, and we will use it to find this tuple $(a_1,\dots,a_k)$.

\begin{lemma}\label{lem:col-algo-find-prob}
    Let $J$ be a $k$-partite $k$-hypergraph with vertex classes $C_1,\dots,C_k$, and let $i \in [k]$. Let $\Lambda = \max\{5,\log |C_i|\}$. Then there exists a integer $0 \le a \le 2\Lambda-1$ and a set $S \subseteq C_i$ such that:
    \begin{enumerate}[(i)]
        \item for all $v \in S$, $d_J(v) \ge e(J)/2^a$; and
        \item $2^{-a}|S| \ge 1/(16\Lambda)$.
    \end{enumerate}
\end{lemma}
\begin{proof}
    Without loss of generality, suppose $i=1$ (by reordering $C_1,\dots,C_k$ if necessary). We first throw away every vertex in $C_1$ with degree significantly lower than average; let $C_1^- \coloneqq \{v \in C_1 \colon d_J(v) > e(J)/(2|C_1|)\}$, and let $J^- \coloneqq J[C_1^-, C_2,\dots,C_k]$. Observe that
    \begin{equation}\label{eq:col-algo-find-prob-J}
        e(J^-) = e(J) - \sum_{v \in C_1 \setminus C_1^-} d_J(v) \ge e(J) - |C_1|\cdot e(J)/(2|C_1|) = e(J)/2,
    \end{equation}
    and that $d_{J^-}(v) = d_J(v)$ for all $v \in C_1^-$.
    
    We now partition vertices of $C_1^-$ according to their degree. For all integers $d \ge 1$, let
    \[
        C_1^d \coloneqq \{v \in C_1^- \colon 2^{d-1} \le d_{J^-}(v) < 2^d\}.
    \]
    Thus $C_1^d$ is the set of vertices in $C_1^-$ with degree roughly $2^d$ in $J^-$ (or equivalently in $J$). By the definition of $C_1^-$, for all $v \in C_1^-$ we have $e(J)/(2|C_1|) < d_{J^-}(v) \le e(J)$; hence $C_1^d = \emptyset$ for all values of $d$ not satisfying $2^d > e(J)/(2|C_1|)$ and $2^{d-1} \le e(J)$. Since each edge in $J^-$ is incident to a vertex in exactly one set $C_1^d$, by the pigeonhole principle, we deduce that there exists
    \begin{equation}\label{eq:col-algo-find-prob-d}
        D \in \bigg[1 + \Big\lfloor \log \frac{e(J)}{2|C_1|}\Big\rfloor,\ 1 + \lfloor \log e(J) \rfloor\bigg]
    \end{equation}
    such that
    \begin{align}\nonumber
        e(J^-[C_1^D,C_2,\dots,C_k]) &\ge \frac{e(J^-)}{\lfloor \log e(J) \rfloor - \lfloor \log (e(J)/(2|C_1|)) \rfloor + 1}\\\nonumber
        &\ge \frac{e(J^-)}{\log (e(J)) - \log (e(J)/(2|C_1|)) + 2} = \frac{e(J^-)}{3+\log |C_1|}.
    \end{align}
    It follows by~\eqref{eq:col-algo-find-prob-J} that
    \begin{equation}\label{eq:col-algo-find-prob-edges}
        e(J^-[C_1^D,C_2,\dots,C_k]) \ge \frac{e(J)}{6+2\log |C_1|}.
    \end{equation}
    
    We now take $S \coloneqq C_1^D$ and $a \coloneqq \lceil \log e(J) \rceil - D + 1$.
    % \[
    %     a \coloneqq \lfloor \log e(J) \rfloor - D + 1 + \big\lceil \log (6 + 2\log |C_1|) \big\rceil.
    % \] 
    % Observe that
    % \begin{equation}\label{eq:col-algo-find-prob-logs}
    %     \big\lceil \log (6 + 2\log |C_1|) \big\rceil \le 1 + \log(6 + 2\log |C_1|) \le \Lambda.
    % \end{equation}
    It remains to prove that $S$ and $a$ satisfy the conditions required by the lemma statement.
    
    First, observe from~\eqref{eq:col-algo-find-prob-d} %and~\eqref{eq:col-algo-find-prob-logs} 
    that
    \begin{align*}
        0 \le a &\le \lceil \log e(J) \rceil - \Big\lfloor\log \frac{e(J)}{2|C_1|}\Big\rfloor%\\
        % + \big\lceil \log (6 + 2\log |C_1|) \big\rceil\\&
        \le \log e(J) - \log \frac{e(J)}{2|C_1|} + 2
        < 2\Lambda - 1,
        %2 + \Lambda \le 3\Lambda-1,
    \end{align*}
    as required in the lemma statement. We next prove (i). Since every vertex $v \in C_1^D$ has degree at least $2^{D-1}$ in $J^-$, from the definition of $a$ 
    %and from~\eqref{eq:col-algo-find-prob-edges} 
    we have
    \[
        2^a d_J(v) \ge 2^{a+D-1} = 2^{\lceil \log e(J) \rceil} \ge e(J),
        % 2^a d_J(v) \ge 2^{\lfloor \log e(J) \rfloor + 1 + \lceil \log (6 + 2\log |C_1|) \rceil} \ge e(J)(6 + 2\log |C_1|) \ge e(J^-),
    \]
    as required by (i). Finally, we prove (ii). Since every vertex in $C_1^D$ has degree at most $2^D$ in $J^-$, by~\eqref{eq:col-algo-find-prob-edges} we have
    \[
        2^D|C_1^D| \ge e(J^-[C_1^D,C_2,\dots,C_k]) \ge \frac{e(J)}{6+2\log|C_1|},
    \]
    and hence
    \[
        |C_1^D| \ge \frac{e(J)}{2^{D+1}(3+\log|C_1|)} \ge \frac{e(J)}{2^{D+2}\Lambda}.
    \]
    By the definition of $a$ we have 
    \[
    2^{-a} \ge 2^{{-}\lceil \log e(J) \rceil + D - 1} \ge  \frac{2^{D-2}}{e(J)},
    % 2^{-a} \ge 2^{{-}\log e(J) + D - 2 - \log(6+2\log |C_1|)} = \frac{2^{D-2}}{e(J)(6+2\log |C_1|)} \ge \frac{2^{D-4}}{e(J)\Lambda},
    \]
    so it follows that $2^{-a}|C_1^D| \ge 1/(16\Lambda)$ as required.
\end{proof}

The rest of the analysis of \newverifyguess is almost exactly the same as in~\cite{DLM}, with the relevant part of the proof replaced by \cref{lem:col-algo-find-prob} and with slightly different algebra. We provide full details for the benefit of the reader.

\begin{lemma}\label{lma:newverifyguess}
  \newverifyguess behaves as stated, runs in time 
  \[\OO(2^{2k}(\log n)^{k - |I|}\log^{|I|}(1/\zeta) + kn\log n),\]
  and has oracle cost
  \[ \OO(2^{2k}(\log n)^{k - |I|}\log^{|I|}(1/\zeta)\cdot \cost_k(n)). \]
\end{lemma}
\begin{proof}

Let $G,M,X_1,\dots,X_k, I, \zeta$ be the input for \newverifyguess, and write $H = G[X_1, \dots, X_k]$.  

 \medskip\noindent\textit{Running time and oracle cost.} We first observe that
 \begin{equation*}%\label{eq:newverifyguess-A}
     |A| \le 2^k(\log(1/\zeta)+1)^{|I|}(\log n)^{k-|I|}. 
 \end{equation*}
 Since $\log(1/\zeta) \ge 5$, we have $\log(1/\zeta)+1 \le (6/5)\log(1/\zeta)$ and hence 
 % Replaced 1+\log(1/\zeta) by \log(1/\zeta) --- this is both correct by the sentence above and what's used later on near the end of Soundness. --John
 \begin{equation}\label{eq:newverifyguess-A}
     |A| \le (12/5)^k\log^{|I|}(1/\zeta)(\log n)^{k-|I|} = \OO(2^{2k}(\log n)^{k-|I|}\log^{|I|}(1/\zeta))). 
 \end{equation}
 Lines~\ref{line:verify-setup1} and~\ref{line:verify-setup2} make no oracle calls; line~\ref{line:verify-setup1} takes $\OO(kn\log n)$ time to construct the sets $Z_{i,j}$, and line~\ref{line:verify-setup2} takes $\OO(|A|)$ time to construct $A$.
 % If we really wanted to we could remove that log log factor with the efficient binomial sampling lemma.
 The loop at line~\ref{line:verify-for} invokes the oracle $|A|$ times and runs in $\OO(|A|)$ time. Line~\ref{line:verify-returnNO} makes no oracle calls and runs in constant time. The claimed bounds on running time and oracle cost follow on bounding the cost of every oracle call by $\cost_k(n)$ (using \cref{lem:cost-monotone}), and it remains to prove that the soundness and completeness properties hold.

 \medskip\noindent\textit{Soundness.}
  This follows by an identical argument to that used in the proof of \cite[Lemma 4.1]{DLM}; for clarity we reproduce this reasoning with the values of $\pout$ and $|A|$ changed for the situation at hand. For notational convenience, we denote the gap in the soundness case by $\gamma$, that is, we set $\gamma\coloneqq \pout/((8k)^k\log^{|I|}(1/\zeta)\log^{k-|I|}n)$. Suppose that $e(H) < \gamma M$. 
  
  Recall that the algorithm returns \yes if and only if $e(G[Z_{1,a_1},\dots,Z_{k,a_k}])  > 0$ holds for some ${(a_1,\dots,a_k) \in A}$. Hence by a union bound over all $e \in E(H)$ and all $(a_1,\dots,a_k) \in A$, we have
  \begin{equation}\label{eq:newverifyguess-gamma}
    \pr(\mbox{Returns \no}) \ge 1 - \sum_{(a_1,\dots,a_k)\in A}e(H)\prod_{j=1}^k 2^{-a_j}.
  \end{equation}
  We now bound this product above. By the definition of $A$, we have $2^{-a_1}\dots 2^{-a_k} \le (2k)^k/M$. By hypothesis, we have $e(H) < \gamma M$ and hence $(2k)^k/M \le \gamma (2k)^k/e(H)$. It follows from \eqref{eq:newverifyguess-gamma} that
  \begin{equation}\label{lem:newverifyguess-returns-no}
    \pr(\mbox{Returns \no}) \ge 1 - |A|\gamma(2k)^k.
  \end{equation}
  By~\eqref{eq:newverifyguess-A} and the fact that $k \ge 2$ we have 
  \[
        \gamma \le \frac{\pout}{(40/12)^kk^k|A|} < \frac{\pout}{2^{k+1}k^k|A|}.
  \]
  It therefore follows from~\eqref{lem:newverifyguess-returns-no} that we return \no with probability at least $1 - \pout/2$. This establishes the soundness of the algorithm, so it remains to prove completeness. 

   \medskip\noindent\textit{Completeness.}
   Suppose now that $(Y_1,\dots,Y_k)$ is an $(I,\zeta)$-core for $H$ and that $e(H) \ge M$. We must prove that \newverifyguess outputs \yes with probability at least $\pout$.
   Let $H_0 \coloneqq H[Y_1,\dots,Y_k]$, and note that by \cref{def:core}(i), $e(H_0) \ge e(H)/(2k)^k \ge M/(2k)^k$.  It suffices to show that with probability at least $\pout$, there is at least one setting of the vector $(a_1, \dots, a_k)\in A$ such that $H_0[Z_{1,a_1},\dots,Z_{k,a_k}]$ contains at least one edge. 
   
   We will define this setting iteratively. First, with reasonable probability, we will find an integer $a_1$ and a vertex $v_1 \in Z_{1,a_1}$ such that $H_1 \coloneqq H_0[\{v_1\}, Y_2, \dots, Y_k]$ contains roughly $2^{-a_1}e(H_0)$ edges; our choice of $a_1$ will come from an application of \cref{lem:col-algo-find-prob}. In the process, we expose $Z_{1,j}$ for all $j$. We then, again with reasonable probability, find an integer $a_2$ and a vertex $v_2 \in Z_{2,a_2}$ such that $H_2 \coloneqq H_0[\{v_1\}, \{v_2\}, Y_3, \dots, Y_k]$ contains roughly $2^{-a_1 - a_2}e(H')$ edges. Continuing in this vein, we eventually find $(a_1, \dots, a_k) \in A$ and vertices $v_i \in Z_{i,a_i}$ such that $\{v_1, \dots, v_k\}$ is an edge in $H_0[Z_{1,a_1}, \dots, Z_{k,a_k}]$, proving the result. 
   
   We formalise this idea by defining a collection of events. For all $i \in [k]$, let $\calE_i$ be the event that there exist integers $a_1,\dots,a_i \ge 0$ and $v_1,\dots,v_i \in V(H)$ such that:
   \begin{enumerate}[(a)]
        \item for all $j \in [i]$, $v_j \in Z_{j,a_j}$;
   		\item for all $j \in [i] \setminus I$, $a_j \le 2\log(n)-1$;
        \item for all $j \in [i] \cap I$, $a_j \le 2\log(1/\zeta)+1$; and
   		\item %we have $d_H(v_1,\dots,v_{i}) \ge e(H)/ \prod_{j=1}^{i}2^{a_j}$, and 
   		setting $H_i \coloneqq H_0[\{v_1\},\dots,\{v_i\},Y_{i+1},\dots,Y_k]$, we have $e(H_i) \ge e(H_0)/\prod_{j=1}^i 2^{a_j}$.
   		%\item for $1 \le i \le \ell$, every edge in $H_i$ is $R$-rooted in $H_i$.
   \end{enumerate}
   We make the following \textbf{Claim:} for all $i \in [k]$,
   \[
        \pr(\calE_i \mid \calE_1, \dots, \calE_{i-1}) \ge \begin{cases}
                1/(32\log(1/\zeta)) & \mbox{ if }i \in I,\\
                1/(32\log n) & \mbox{ otherwise.}
        \end{cases}
   \]
   Note that for $i=1$, the range $\calE_1,\dots, \calE_{i-1}$ is empty.

%   We make the following \textbf{Claim}: $\Pr(\calE_1) \ge 1/(8 (2 + \log k - \log \zeta))$; for all $2 \le i \le \ell$, $\Pr(\calE_i \mid \calE_{i-1}) \ge 1/(24\log(1/\zeta))$, and for all $\ell + 1 \le i \le k$, $\Pr(\calE_i \mid \calE_{i-1}) \ge 1/(8k \log n)$.
   
   \medskip\noindent\textbf{Proof of \cref{lma:newverifyguess} from Claim:} Suppose $\calE_k$ occurs, and let $a_1, \dots, a_k$ and $v_1, \dots, v_k$ be as in the definition of $\calE_k$. By (d), we know that $\{v_1, \dots, v_k\}$ is an edge in $H$; it follows by (a) that it is also an edge in $G[Z_{1,a_1}, \dots, Z_{k,a_k}]$. Also by (d), since the number of edges containing $\{v_1,\dots,v_k\}$ cannot be more than one, we have 
   % Made this a display in order to give it a number to help the reader. --John
   \begin{equation}\label{eq:newverifyguess-from-claim}
        \prod_{j=1}^k 2^{a_j} \ge e(H_0).
   \end{equation}
   Since $e(H) \ge M$ by hypothesis, and $(Y_1,\dots,Y_k)$ is an $(I,\zeta)$-core, by \cref{def:core}(i) we have $e(H_0) \ge M/(2k)^k$; it follows from~\eqref{eq:newverifyguess-from-claim} that $a_1 + \dots + a_k \ge \log M - k\log (2k)$.  It follows from (b) and (c) that $(a_1, \dots, a_k) \in A$, so whenever $\calE_k$ occurs, \newverifyguess returns \yes on reaching $(a_1, \dots, a_k)$ in line~\ref{line:verify-for}. By the Claim, we have
   \begin{align*}
       \pr(\calE_k) &= \prod_{j=1}^k \pr(\calE_j \mid \calE_1, \dots, \calE_{j-1}) 
       \ge 1/\big(32^k\log^{|I|}(1/\zeta)\log^{k-|I|} n\big) = \pout,
       % Rewrote the RHS to better match the definition of \pout. --John
   \end{align*}
   so completeness follows. The lemma statement therefore follows as well.
   
  \medskip\noindent\textbf{Proof of Claim:} We proceed by induction. Suppose we have defined $v_1,\dots,v_{i-1}$ and $a_1,\dots,a_{i-1}$ as a deterministic function of $\{Z_{i',j}\colon i' \le i-1\}$ conditioned on $\calE_1,\dots,\calE_{i-1}$; this is vacuously true for $i=1$, and true by the induction hypothesis for $2 \le i \le k$. We then expose the values of $\{Z_{i',j}\colon i' \le i-1\}$ conditioned on $\calE_1,\dots,\calE_{i-1}$, and hence the values of $v_1,\dots,v_{i-1}$ and $a_1,\dots,a_{i-1}$; we abuse notation slightly by abbreviating the corresponding suite of events to $\calF$ and identifying $a_1,\dots,a_{i-1}$, $v_1,\dots,v_{i-1}$, and $H_{i-1}$ with their values conditioned on $\calF$. We seek to prove that for all choices of $\calF$,
   \[
    \pr(\calE_i \mid \calF) \ge \begin{cases}
                1/(32\log(1/\zeta)) & \mbox{ if }i \in I,\\
                1/(32\log n) & \mbox{ otherwise.}
        \end{cases}
   \]
   Observe that $\calF$ is independent of all sets $Z_{i,j}$. 
   
   We apply \cref{lem:col-algo-find-prob} with $J = H_{i-1}$, so that $C_a = \{v_a\}$ for all $a \le i-1$ and $C_a = Y_a$ for all $a \ge i$. We take $a_i$ to be the resulting integer $a$, and $S_i \subseteq C_i$ to be the resulting set $S$. We will take $\calE_i$ to be the event that $S_i \cap Z_{i,a_i} \ne \emptyset$, and if $\calE_i$ occurs then we will choose $v_i$ arbitrarily from $S_i \cap Z_{i,a_i}$. By construction, $v_1,\dots,v_i$ satisfy property~(a).
   
   We first note that by \cref{lem:col-algo-find-prob}, $a_i \le \max\{9, 2\log(|Y_i|)-1\}$. If $i \notin I$, then it follows that $a_i \le \max\{9, 2\log n - 1\}$; since $n \ge 32$, it follows that $a_i \le 2\log n - 1$. If $i \in I$, then since $(Y_1,\dots,Y_k)$ is an $(I,\zeta)$-core of $H_0$, by \cref{def:core}(ii), we have $|Y_i| \le 2/\zeta$ and hence $a_i \le \max\{9, 2\log(1/\zeta)+1\}$; since $\zeta \le 1/32$, it follows that $a_i \le 2\log(1/\zeta)+1$. Either way, $a_1,\dots,a_i$ satisfy properties (b) and (c).
   
   We next observe that if $\calE_i$ occurs, so that there exists some $v_i \in S_i \cap Z_{i,a_i}$, then by \cref{lem:col-algo-find-prob}(i) and property (d) of $H_{i-1}$, we have 
   \[
        e(H_i) = d_{H_{i-1}}(v_i) \ge e(H_{i-1})/2^{a_i} \ge e(H_0)/\prod_{j=1}^i 2^{a_j}.
   \]
   Thus property (d) is also satisfied for $H_i$.
   
   Finally, we observe that since all sets $Z_{i,j}$ are independent of $\calF$, we have
   \[
        \pr(\calE_i \mid \calF) = \pr(Z_{i,a_i} \cap S_i \ne \emptyset) = 1 - (1 - 2^{-a_i})^{|S_i|} \ge 1 - e^{{-}2^{-a_i}|S_i|}.
   \]
   By \cref{lem:col-algo-find-prob}(ii), it follows that
   \[
        \pr(\calE_i \mid \calF) \ge 1 - e^{{-}1/(16\max\{5, \log |Y_i|\})} \ge \frac{1}{32\max\{5,\log |Y_i|\}}.
   \]
   As before, since $(Y_1,\dots,Y_k)$ is an $(I,\zeta)$-core, $n \ge 32$ and $\zeta \le 1/32$, the required lower bound follows whether $i \in I$ or not.
\end{proof}

In \cite{DLM}, \verifyguess is used as a subroutine by an algorithm \oldcoarsecount: this algorithm makes repeated calls to \verifyguess for each $M \in \{1,2,4,8,...,n^k\}$, and outputs an estimate of the number of edges in $G$ that, with probability at least $2/3$, is a $b$-approximation. We mimic this behaviour with the following algorithm.  

\begin{algorithm}
    \SetKwInput{Oracle}{Oracle}
    \SetKwInput{Input}{Input}
	\SetKwInput{Output}{Output}
	\DontPrintSemicolon
    \Oracle{Colourful independence oracle $\cindora(G)$ of an $n$-vertex $k$-hypergraph~$G$.}
	\Input{Integer $n \ge 32$ that is a power of two, disjoint subsets $X_1, \dots, X_k \subseteq V(G)$, set $I \subseteq [k]$, and rational number $\zeta \in (0,1/32)$ with denominator $\OO(n^k)$.}
	\Output{Non-negative integer $m$ such that, setting $b\coloneqq (2k)^{5k}\log^{k - |I|}n \log^{|I|} (1/\zeta)$, $m$ satisfies both of the following properties with probability at least $2/3$. Firstly, $e(G[X_1,\dots,X_k]) \le mb$. Secondly, if $G[X_1,\dots,X_k]$ has an $(I,\zeta)$-core then $e(G[X_1,\dots,X_k]) \ge m/b$.
	}
 	\Begin{
 	    Set $\pout \coloneqq 1/\big(2^{5k}\log^{|I|}(1/\zeta)\log^{k-|I|} n\big)$.\;
 	    Calculate $N \coloneqq \ceil{24 \ln(12k \log n)/\pout}$.\;
        \ForAll{$M\in\{1,2,4,8,\dots,n^k\}$}{
            Call $\newverifyguess(\cindora(G),M,X_1,\dots,X_k,I,\zeta)$ a total of $N$ times, and let $S_M \in \{0,\dots,N\}$ be the number of calls that returned \yes. (Naturally, we use independent randomness for each call.)
        }
        \eIf{$\cindora(G)_{X_1,\dots,X_k} = 1$}{Set $m=0$.\;}{
            \eIf{there exists $M$ such that $S_M \ge 3\pout N/4$}{Let $m$ be the greatest such $M$.\;}{Set $m = n^k$.\;}
        }
        \Return $2m/b$.\;
	 } 
    \caption{\label{algo:helperdlmimprove}\helperdlmimprove\\
    \textit{This algorithm makes repeated calls to \newverifyguess with different guesses $M$, to obtain a coarse approximation to the number of edges in the input hypergraph, assuming that it contains an $(I,\zeta)$-core.}
    }
\end{algorithm}

\begin{lemma}\label{lma:coarseguess}
    \helperdlmimprove behaves as stated. Moreover, let
    \[
        T = 2^{7k}\log^{2|I|}(1/\zeta)(\log n)^{2(k-|I|)+1}n.
    \]
    Then on an $n$-vertex $k$-hypergraph $G$, $\helperdlmimprove(\cindora(G),X_1,\dots,X_t,I,\zeta)$ runs in time $\OO(Tn\log n)$, and has oracle cost $\OO(T \cost_k(n))$.
\end{lemma}
\begin{proof}
    The proof is similar to that of \cite[Lemma 4.2]{DLM}, but we include full details here for the sake of completeness.
    
    \medskip\noindent\textit{Running time.} For brevity, let $X \coloneqq \log^{|I|}(1/\zeta)\log^{k-|I|}n$. \helperdlmimprove simply executes \newverifyguess $N$ times, where $N = \OO(2^{5k} X\log n)$, and performs $o(N)$ arithmetic operations to calculate $N$. By \cref{lma:newverifyguess}, each execution takes time $\OO(2^{2k}X+kn\log n) = \OO(2^{2k}X\log(n) \cdot n)$ and has oracle cost $\OO(2^{2k}X\cdot  \cost_k(n))$.  The overall running time of \helperdlmimprove is therefore $\OO(2^{7k}X^2\log^2(n) \cdot n)$, 
    and the overall oracle cost is $\OO(2^{7k}X^2\log(n) \cdot \cost_k(n))$,
    as claimed.
    
    \medskip\noindent\textit{Correctness. } Let $H \coloneqq G[X_1,\dots,X_k]$. We will demonstrate that:
    \begin{enumerate}[(a)]
        \item $\pr(2m/b > b \cdot e(H)) \le 1/6$; and
        \item if $H$ contains an $(I,\zeta)$-core, then $\pr(2m/b < e(H)/b) \le 1/6$.
    \end{enumerate}
    Given (a) and (b), correctness follows immediately by a union bound.
    
    We first prove (a). Fix $M \in \{1,2,4,8,\dots,n^k\}$ with $e(H) < 2M/b^2$. Observe that for such an $M$, \newverifyguess outputs \yes with probability at most $\pout/2$; thus the random variable $S_M$ is a binomial variable with mean at most $N\pout/2$. A standard Chernoff bound (\cref{lem:chernoff-low} with $\delta=1/2$) then implies
    \[
        \pr(S_M \ge 3N\pout/4) \le 2e^{-N\pout/24} \le 1/(6k\log n).
    \]
    Thus on taking a union bound over all such $M$, with probability at least $5/6$, we have $e(H) \ge 2m/b^2$ and hence our output, $2m/b$, is at most $b\cdot e(H)$ as required.
    
    Next, suppose that $H$ contains an $(I,\zeta)$-core; we must prove (b). Let $M^-$ be the maximum value in $\{1,2,4,\dots,n^k\}$ with $e(H) \ge M^-$, so that $M^- \le e(H) < 2M^-$. Observe that for $M^-$, \newverifyguess outputs \yes with probability at least $\pout$; hence $\E(S_{M^-}) \ge N\pout$, and so a standard Chernoff bound (\cref{lem:chernoff-low} with $\delta=1/4$) implies
    \[
        \pr(S_{M^-} < 3N\pout/4) \le 2e^{-N\pout/48} \le 1/6.
    \]
    Thus with probability at least $5/6$, we have $m \ge M^-$ and hence $e(H) \le 2m$. It follows that our output, $2m/b$, is at least $e(H)/b$.
\end{proof}

\cref{lem:col-algo-large-core} now follows immediately from \cref{lma:coarseguess} together with \cref{lem:median-boosting}, which is used to reduce the failure probability from $1/3$ to $1 - \delta$ for arbitrary rational $\delta \in (0,1)$.

\subsection{Lower bounds on oracle algorithms for edge detection}

In order to prove the lower bound part of \cref{thm:col-main-simple} in the regimes where $k$ or $\alpha$ are large, we use a simple lower bound on \cindora-oracle algorithms for the problem of deciding whether the given $k$-hypergraph has at least one edge or whether it is empty.
In this section, we prove that simple lower bound as \cref{prop:col-dec} and state the form we will need as \cref{cor:col-dec}.

For the following lemma, recall that a deterministic \cindora-oracle algorithm makes a sequence $S_1,\dots,S_N$ of queries to the oracle, where each query $S_i=(S_{i,1},\dots,S_{i,k})$ is a tuple of disjoint vertex subsets, and recall that $S_i^{(k)}$ is the set of all possible edges in a $k$-partite $k$-hypergraph spanned by $S_{i,1},\dots,S_{i,k}$.
In the lemma, we show that~$N$ deterministic queries of overall bounded cost can detect a random edge~$e$ only with small probability.

\begin{lemma}\label{lem:col-dec-prob}
    Let $\cost_k(n) = n^{\alpha_k}$ be a cost function, where $\alpha_k \in [0,k]$. Let $G$ be an $n$-vertex $k$-hypergraph, let $S_1,\dots,S_N$ be an arbitrary sequence of queries to $\cindora(G)$ with total cost at most $C$, and let $e \in V(G)^{(k)}$ be sampled uniformly random. Then
    \[
        \pr_e\Big(e \in \bigcup_{i=1}^N S_i^{(k)}\Big) \le C/n^{\alpha_k}\,.
    \]
\end{lemma}
\begin{proof}
    For all $i \in [N]$, let $s_i \coloneqq \sum_{j=1}^k |S_{i,j}|\in[0,n]$. Then the total cost incurred by the~$N$ queries satisfies $\sum_{i=1}^{N}\cost(S_i)=\sum_{i=1}^{N} s_i^{\alpha_k}\le C$.
    By the AM-GM inequality, we have
    \[
        \abs[\big]{S_i^{(k)}} = \prod_{j=1}^k |S_{i,j}| \le s_i^k/k^k\,.
    \]
    By a union bound, it follows that
    \[
        \pr\Big(e \in \bigcup_{i=1}^N S_i^{(k)}\Big) \le \binom{n}{k}^{-1}\sum_{i=1}^N \abs[\big]{S_i^{(k)}} 
        \le \frac{k^k}{n^k}\sum_{i=1}^N \frac{s_i^k}{k^k}
        = \frac{1}{n^k}\sum_{i=1}^N s_i^k
        \,.
    \]
    By assumption, we have $\sum_{i=1}^N s_i^{\alpha_k} \le C$ and $s_i \in [0,n]$ for all $i$.
    We now apply Karamata's inequality in the form of \cref{cor:karamata}, taking $\alpha=\alpha_k$, $W=C$, $c=n$, $t=N$ and $r=k$.
    This yields:
    \[
        \pr\Big(e \subseteq \bigcup_{i=1}^N S_i^{(k)}\Big) \le\frac{Cn^{k-\alpha_k}}{n^k} = \frac{C}{n^{\alpha_k}}\,.\qedhere
    \]
\end{proof}

Next, we use \cref{lem:col-dec-prob} to show that any a randomised $\cindora$-oracle algorithm must incur relatively high cost in order to distinguish the empty $k$-hypergraph from a $k$-hypergraph that is not empty.
\begin{prop}\label{prop:col-dec}
    Let $\cost_k(n) = n^{\alpha_k}$, where $\alpha_k \in [0,k]$, and let $p \in (0,1]$. Let $\calA$ be a randomised $\cindora$-oracle algorithm with worst-case oracle cost at most~$C$ such that, for all $k$-hypergraphs $G$, with probability at least $p$, $\calA(\cindora(G))$ returns $1$ if and only if $e(G) > 0$. Then $C \ge pn^{\alpha_k}$.
\end{prop}
\begin{proof}
    Let $G_1$ be the $k$-hypergraph on $[n]$ with no edges, and let $G_2$ be the $k$-hypergraph on $[n]$ with exactly one edge, chosen uniformly at random. Let $A$ be a deterministic \cindora-oracle algorithm with worst-case oracle cost at most~$C$, and suppose $A(\cindora(G_1)) \ne A(\cindora(G_2))$ with probability at least $p$. Let $S_1(A),\dots,S_{N_A}(A)$ be the (deterministic) sequence of oracle queries issued by $A(\cindora(G_1))$.
    On input~$G_2$, the deterministic algorithm $A(\cindora(G_2))$ will make exactly the same queries until a query~$S_i(A)$ potentially contains the planted random edge~$e$.
    Thus by \cref{lem:col-dec-prob}, we have
    \[
        p \le \pr_{G_2}\Big(A(\cindora(G_1)) \ne A(\cindora(G_2))\Big)
        \le \pr_e\Big(e \in \bigcup_{i=1}^{N_A} S_i(A)^{(k)}\Big) \le C/n^{\alpha_k}\,,
    \]
    and hence $C \ge pn^{\alpha_k}$.
    Let $F$ be the family of all such algorithms $A$. 
    
    We have $\pr_{\calA,G_2}(\calA(\cindora(G_1))\ne\calA(\cindora(G_2))) \ge p$, so we must have $\calA \in F$ with non-zero probability.
    Thus the worst-case oracle cost of $\calA$ is at least $pn^{\alpha_k}$, as required.
\end{proof}
In the following corollary, we lift \cref{prop:col-dec} to worst-case \emph{expected} oracle cost.
\begin{cor}\label{cor:col-dec}
    Let $\cost_k(n) = n^{\alpha_k}$, where $\alpha_k \in [0,k]$, and let $p \in (0,1]$. Let $\calA$ be a randomised $\cindora$-oracle algorithm with worst-case expected oracle cost at most~$C$ such that, for all $k$-hypergraphs $G$, with probability at least $9/10$, $\calA(\cindora(G))$ returns $1$ if and only if $e(G) > 0$. Then $C \ge n^{\alpha_k}/100$.
\end{cor}
\begin{proof}
    Suppose for a contradiction that such an algorithm $\calA$ existed with $C \le n^{\alpha_k}/100$; then consider the algorithm $\calA'$, which simulates $\calA$, keeping track of the oracle costs.
    It aborts the simulation of~$\calA$ just before the total oracle cost would exceed $n^{\alpha_k}/10$. If $\calA$ has terminated normally, $\calA'$ copies the output, and otherwise $\calA'$ outputs an arbitrary value.
    By construction, $\calA'$ has worst-case oracle cost at most $n^{\alpha_k}/10$.
    By Markov's inequality, the probability that $\calA$ incurs cost at least $n^{\alpha_k}/10$ is at most $1/10$, and so $\calA'$ has success probability $p \ge 4/5$.
    By \cref{prop:col-dec}, such an algorithm~$\calA'$ cannot exist, a contradiction.
\end{proof}

\subsection{Lower bounds on oracle algorithms for edge estimation}%
\label{sec:col-lower}

In this section, we unconditionally prove that the worst-case oracle cost achieved by \aau is essentially optimal. As with \cref{thm:uncol-lower-distinguisher}, we will construct two correlated random \mbox{$n$-vertex} {$k$-hypergraphs} $G_1$ and $G_2$ that (with high probability) have significantly different numbers of edges, but such that any deterministic \cindora-oracle algorithm that can distinguish between $G_1$ and $G_2$ must incur a large worst-case oracle cost. This approach will yield the following result.

\begin{restatable}{theorem}{colourfulGs}\label{thm:colourful-G1-G2}
    Let $t,k \ge 1$, let $\alpha \in [0,k-3]$, and let $\cost(x)=x^\alpha$. Let $t_0 \coloneqq 2^{400k^6}$, and suppose that $t \ge t_0$. There exist two correlated distributions $\calG_1$ and $\calG_2$ on $k$-partite $k$-hypergraphs whose vertex classes~$V_1,\dots,V_k$
    each have size $t$ with the following properties:
    \begin{enumerate}[(i)]
        \item We have $\pr_{(G_1,G_2)\sim(\calG_1,\calG_2)}[e(G_2) \ge 4e(G_1)] \ge 19/20$.
        \item 
        Let 
        \[
            C \coloneqq \frac{t^\alpha}{2^{5k+7}k^{7k}}\cdot \bigg(\frac{\log t}{\log\log t}\bigg)^{k-\flalpha-3}\,.
        \]
        Suppose $A$ is a deterministic $\cindora$-oracle algorithm with
        \begin{equation}\label{eq:colourful-distinguish}
            \pr_{(G_1,G_2)\sim(\calG_1,\calG_2)}\Big(A(\cindora(G_1)) \ne A(\cindora(G_2))\Big) \ge 2/3\,,
        \end{equation}
        which only uses \cindora-queries $S=(S_1,\dots,S_k)$ with $S_i \subseteq V_i$ for all $i \in [k]$. Then the expected oracle cost of $A$ (with respect to $\cost$) under random inputs $G_1 \sim \calG_1$ satisfies $\E_{G_1\sim\calG_1}[\cost(A,G_1)] \ge C/2$.
    \end{enumerate}
\end{restatable}

Before we prove \cref{thm:colourful-G1-G2} in \cref{sec:col-lower-G1-G2,sec:col-lower-framework,sec:col-lower-bound-prob}, let us apply the minimax principle (\cref{thm:algorithm-minmax}) to it, in order to derive our main lower bound for \cindora-oracle algorithms.

\begin{theorem}\label{thm:col-lb}
    Let $n,k$ be positive integers,
    and let $\cost_k(n) = n^{\alpha_k}$ be a cost function, where $\alpha_k \in [0,k]$. Let $\calA$ be a randomised \cindora-oracle algorithm such that, for all $n$-vertex $k$-hypergraphs~$G$, $\calA(\cindora(G))$ is a $(1/2)$-approximation to $e(G)$ with probability at least $9/10$.
    Then $\calA$ has worst-case expected oracle cost at least
    $\Omega(L(n,k))$, where
    \[
        L(n,k) \coloneqq
        \frac{n^{\alpha_k}}{k^{10k}}\cdot\Big(\frac{\log n}{\log\log n}\Big)^{k-\floor{\alpha_k}-3}
        \,.
    \]
\end{theorem}

\begin{proof}
    Let $\calC(n,k)$ be the worst-case expected oracle cost of $\calA$ on $n$-vertex $k$-hypergraphs.
    We must show that $\calC(n,k) \ge \Omega(L(n,k))$ holds.
    Note that \cref{cor:col-dec} already gives a simple bound of $\calC(n,k) \ge n^\alpha/100$, where we write $\alpha=\alpha_k$.
    We bound $\calC(n,k)$ by a case distinction depending on the values of $n$, $k$, and $\alpha$.
    
    \medskip\noindent \textit{Case 1: $\alpha \ge k-3$.} In this case, we have $(\log n)^{k-\flalpha-3} \le 1$, so $L(n,k) \le n^{\alpha}$, thus the bound from \cref{cor:col-dec} already suffices.

    \medskip\noindent \textit{Case 2: $4k \ge (\log n)^{1/7}$.} In this case, we have $\log^k n\le 4^{7k} k^{6k} \le \OO(k^{10k})$,
    so now we have $L(n,k) \le \OO(n^{\alpha})$ and the bound from \cref{cor:col-dec} again suffices.
    
    \medskip\noindent \textit{Case 3: $\alpha \le k-3$ and $4k \le (\log n)^{1/7}$.}
    In this case, observe that $n \ge 2^{4^7k^7}$ holds.
    Let $t = \floor{n/k}$.
    By $k\ge 2$, we have $t \ge 2^{400k^6} = t_0$.
    Let $G_1,G_2 \sim \calG_1,\calG_2$ be $k$-partite $k$-hypergraphs with $t$ vertices in each vertex class as in \cref{thm:colourful-G1-G2}. By adding isolated vertices if necessary, we can use $\calA$ to approximately count edges in $G_1$ and $G_2$ in expected cost at most $\calC(n,k)$. 
    
    Recall from \cref{sec:oracle-algorithms} that we regard $\calA$ as a discrete probability distribution over a set $\supp(\calA)$ of deterministic algorithms based on random choices of $\calA$, and let $A \sim \calA$. By hypothesis, with probability at least $(9/10)^2 \ge 4/5$, $A$ will correctly output a $(1/2)$-approximation for both graphs $G_1$ and $G_2$. Moreover, by \cref{thm:colourful-G1-G2}(i), with probability at least $19/20$, we have $e(G_2) \ge 4e(G_1)$. If both of these events occur, then we cannot have $A(\cindora(G_1)) = A(\cindora(G_2))$; thus by a union bound, we have
    \begin{equation}\label{eq:colourful-minmax}
        \pr_{\substack{A\sim\calA\\(G_1,G_2)\sim(\calG_1,\calG_2)}} \Big(A(\cindora(G_1))\ne A(\cindora(G_2))\Big) \ge 3/4\,.
    \end{equation}
    
    Let $F$ be the family of \cindora-oracle algorithms $A \in \supp(\calA)$ that satisfy~\eqref{eq:colourful-distinguish}. By~\eqref{eq:colourful-minmax} and conditioning on the event~$A\in F$, we have
    \begin{align*}
        \frac{3}{4} &\le \pr_{\substack{A\sim\calA\\(G_1,G_2)\sim(\calG_1,\calG_2)}} \Big(A(\cindora(G_1))\ne A(\cindora(G_2))\Big)\\
        &\le 1 \cdot \pr_{A\sim\calA}(A \in F) + \frac{2}{3}\cdot \pr_{A\sim\calA}(A\notin F)
        = \frac{1}{3}\cdot \pr_{A\sim\calA}(A\in F)+\frac{2}{3}\,,
    \end{align*}
    and so $\pr_{A\sim\calA}(A\in F) \ge 1/4$. Thus by minimax (that is, \cref{thm:algorithm-minmax} with $p=1/4$ and $\calD=\calG_1$), we can lower-bound the worst-case expected oracle cost of $\calA$ by the expected oracle cost of any deterministic $A \in F$ on a random input~$G_1\sim\calG_1$:
    \[
        \calC(n,k)\ge \tfrac14\inf_{A\in F}\E_{G_1\sim\calG_1}[\cost(A,G_1)]\,.
    \]
    Observe that any such $A$ can be turned into a \cindora-algorithm which uses only queries of the form $(S_1,\dots,S_k)$ with $S_i \subseteq V_i$ by splitting each query into at most $k^k$ sub-queries of equal or lesser size; thus by \cref{thm:colourful-G1-G2}(ii), we get
    \[
        \calC(n,k) \ge \frac{t^\alpha}{2^{5k+10}k^{8k}}\cdot \bigg(\frac{\log t}{\log\log t}\bigg)^{k-\flalpha-3}\,.
    \]
    Since $2 \le k \le (\log n)^{1/7}$, we have $n/(2k) \le t \le n/k$, and so
    \[
        \calC(n,k) \ge \frac{n^\alpha}{2^{5k+10}k^{9k}}\bigg(\frac{\log t}{\log\log t}\bigg)^{k-\flalpha-3} = \frac{k^{k}}{2^{5k+10}}L(n,k)\,.
    \]
    We have $2^{5k+10} = O(k^{k})$, so we have $\calC(n,k) = \Omega(L(n,k))$ in all cases and the result follows.
\end{proof}

	\subsubsection{Defining the input graphs}\label{sec:col-lower-G1-G2}
	
	Our first step in the proof of \cref{thm:colourful-G1-G2} shall be to define the random graphs $G_1$ and $G_2$ and prove that $e(G_2) \ge 4e(G_1)$ holds with probability at least $19/20$. As in the uncoloured lower bound (see \cref{sec:uncol-lower}), we will form $G_2$ by adding a random set of ``difficult-to-detect'' edges to $G_1$, writing $G_1 = H_1$ and $G_2 = H_1 \cup H_2$ for two independent random graphs $H_1$ and $H_2$. Also as in \cref{sec:uncol-lower}, we will take $H_1$ to be an Erd\H{o}s-R\'{e}nyi graph. Our choice of $H_2$ will be heavily motivated by the following important bottleneck in our colourful approximate counting algorithm \aau.
	
	Recall from \cref{sec:col-upper} that the main subroutine of \aau is the coarse approximate counting algorithm \newcoarsecount for $k$-partite $k$-hypergraphs. The key bottleneck in \newcoarsecount arises when the input graph $G$ has an $(I,\zeta)$-core for some $|I| = \floor{\alpha_k} + 1$; for technical convenience, we will instead take $|I| = \floor{\alpha_k}+2$. Writing $V_1,\dots,V_k$ for the vertex classes of $G$, suppose for simplicity that $I = \{k-\floor{\alpha_k}-1,\dots,k\}$, and that for all $i \in I$ there exists $r_i \in V_i$ such that every edge in $G$ contains $r_i$. In this case, the $\log^{\Theta(k-\alpha_k)}n$ oracle cost arises in \dlmimprove (in \cref{sec:dlm-coarse}) because we must ``guess'' a probability vector $(q_1,\dots,q_k)$ such that $q_1\cdot \ldots \cdot q_k \approx 1/e(G)$, and such that deleting vertices from each $V_i$ with probability $1-q_i$ is still likely to leave some edges in~$G$.
    We know we should choose $q_i \approx 1$ for $i \ge k-\floor{\alpha_k}-1$ in order to avoid deleting the high-degree ``root vertices'' in the colour classes $V_{k-\floor{\alpha_k}-1},\dots,V_k$,
    but we don't know the degree distribution of vertices in $V_1,\dots,V_{k-\floor{\alpha_k}-2}$. Since \dlmimprove only needs to return a coarse approximation, roughly speaking the algorithm tries each tuple of $q_i$'s in $\{1,1/2,1/4,\dots,1/n\}^{k-\floor{\alpha_k}-2}$ for a total of $\log^{\Theta(k-\alpha_k)}n$ possible probability vectors.
	
	Motivated by this bottleneck, we will define $H_2$ to be a randomly-chosen \textbf{complete} $k$-partite graph. We define its parts in the ``non-rooted classes'' $V_1,\dots,V_{k-\floor{\alpha_k}-2}$ by binomially removing vertices according to a random probability vector $\vec{\calQ}\coloneqq(\calQ_1,\dots,\calQ_{k-\floor{\alpha_k}-2})$, and we define its parts in the ``rooted classes'' $V_{k-\floor{\alpha_k}-1},\dots,V_k$ to be uniformly random vertices $\calR_{k-\floor{\alpha_k}-1},\dots,\calR_k$. This construction has the property that a random query of the form used in \dlmimprove is likely to distinguish $H_1$ from $H_1 \cup H_2$ only if it deletes vertices in the non-rooted classes with probabilities close to the ``correct'' probability vector $\vec{\calQ}$. More formally, we take the following notation as standard throughout the remainder of \cref{sec:col-lower}.
	
	\begin{defn}\label{def:col-H1-H2}
	    Let $t$ and $k$ be integers with $k \ge 2$ and $t \ge t_0$, and let $\alpha \in [0,k-3]$. Everything in this definition will formally depend on the values of $t$, $k$ and $\alpha$; these values will always be clear from context, so we omit this dependence from the notation. Let
	    \begin{alignat}{4}\nonumber
	        p &\coloneqq 1/t^{(k+\flalpha+2)/2},\qquad &x &\coloneqq pt^{\flalpha + 2},\\\label{eq:col-H1-H2-vars}
	        \beta &\coloneqq \lfloor(\log t)/(20k^4\log\log t)\rfloor,\qquad & B&\coloneqq\lfloor\log_{x^{1/\beta}}((24\log t)/t)\rfloor\,.
	    \end{alignat}

	    For all $i \in [k]$, we write $V_i = \setc{(i,j)}{j \in [t]}$. We define $H_1\coloneqq H_1(t,k,\alpha)$ to be a $k$-partite Erd\H{o}s-R\'{e}nyi $k$-hypergraph with vertex classes $V_1,\dots,V_k$, where each edge with one vertex in each $V_i$ is present independently with probability $p$.
	    
	    We say $V_1,\dots,V_{k-\flalpha - 2}$ are \emph{non-rooted classes}, and $V_{k-\flalpha-1},\dots,V_k$ are \emph{rooted classes}. For each rooted class $V_i$, let $\calX_i$ be a singleton subset of $V_i$ containing a single uniformly random vertex $\calR_i\in V_i$; we call $\calR_i$ the \emph{root} of $V_i$. Let $\vec{\calQ} \coloneqq (\calQ_1,\dots,\calQ_{k-\flalpha-2})$ be chosen uniformly at random from the set
	    \[
	        \setc[\bigg]{
                (q_1,\dots,q_{k-\flalpha-2})
            }{
                q_i \in \{(2^5x)^{j/\beta}\colon j \in [0,B]\cap\N\}\text{ and } \prod_{i=1}^{k-\flalpha-2} q_i = 2^{5(k-\flalpha-2)}x
            }\,.
	    \]
	    For each non-rooted class~$V_i$, let $\calX_i \subseteq V_i$ be sampled at random by including each vertex of~$V_i$ independently with probability~$\calQ_i$. We then define $H_2\coloneqq H_2(t,k,\alpha)$ to be the complete $k$-partite $k$-uniform hypergraph with vertex classes $\calX_1,\dots,\calX_k$. Finally, we define $G_1 \sim \calG_1$ and $G_2 \sim \calG_2$ by
	    \[
	        G_1 \coloneqq G_1(t,k,\alpha) = H_1,\qquad 
	        G_2 \coloneqq G_2(t,k,\alpha) = H_1 \cup H_2\,.
	    \]
	\end{defn}
	
	We first observe that these graphs $G_1$ and $G_2$ have $tk$ vertices.
    By $\alpha \le k-3$ we always have $x = t^{(\flalpha+2-k)/2}\le t^{-1/2}< 1$.
    Moreover, by $t \ge t_0$, the fact that $t\mapsto\beta(t,k)$ is non-decreasing, and $k\ge 2$, we always have
    \[
        \beta \ge \lfloor 400k^6/(20k^4 \log(400k^6))\rfloor
        =
        \lfloor 20k^2/\log(400k^6)\rfloor
        \ge 2k\,.
    \]
    Finally, by $x = t^{(\flalpha+2-k)/2}$, we have
    \begin{align*}
        B &\ge \frac{\beta\log((24\log t)/t)}{\log x}-1 
        = \frac{2\beta\log(t/(24\log t))}{(k-\flalpha-2)\log t}-1\\
        &\ge \frac{3\beta\log t}{2(k-\flalpha-2)\log t} - 1 
        \ge \frac{\beta}{k-\flalpha-2}\,.
    \end{align*}
    The second inequality follows from $\log(24\log t) \le \tfrac14 \log t$ for $t\ge t_0$ and $k\ge 2$, and the third inequality follows from $\beta/({k-\flalpha-2})\ge 2k/({k-\flalpha-2})\ge 2$.
    Since $B\ge\beta/({k-\flalpha-2})$ holds, there is a choice of $k-\flalpha-2$ integer values of $j \in [0,B]$ which sum to $\beta$; hence there is at least one valid choice $\vec{q}\in\supp(\vec{\calQ})$.
    We also make some remarks to motivate the definitions of~\eqref{eq:col-H1-H2-vars}.
	
    \begin{remark}\label{rem:col-H1-H2}
        In \cref{def:col-H1-H2}:
        \begin{enumerate}[(i)]
            \item\label{rem:col-H1-H2-choice-x}
            The value of $x$ is chosen, so that the following holds for all~$\vec{q}\in\supp(\vec{\calQ})$:
            \begin{align*}
                \E_{H_2}\big(e(H_2) \,\mid\, \vec{\calQ}=\vec{q}\,\big)
                &=
                \prod_{i=1}^{k-\flalpha-2}\E_{H_2}\big(\,\abs{\calX_i} \,\mid\, \vec{\calQ}=\vec{q}\,\big)
                =
                \prod_{i=1}^{k-\flalpha-2}\paren*{q_i t}
                \\
                &=
                \paren*{2^{5(k-\flalpha-2)}x}\cdot t^{k-\flalpha-2}
                =
                2^{5(k-\flalpha-2)}pt^{k}
                \\
                &=
                2^{5(k-\flalpha-2)}\E_{H_1}\big(e(H_1)\big)
                \,.
            \end{align*}
            \item\label{rem:col-H1-H2-choice-B}
             The value of $B$ is chosen to be the largest integer with $x^{B/\beta} \ge (24\log t)/t$;
             this ensures that, for all possible values $\vec{q}\in\supp(\vec{\calQ})$ and all $i\in[k-\flalpha-2]$, we have
             \[
                \E_{H_2}\big(\,\abs{\calX_i} \,\mid\, \vec{\calQ}=\vec{q}\,\big) = q_i t \ge 24\log t\,,
             \]
             which will allow us to show concentration of $|\calX_i|$ conditioned on $\vec{\calQ}$.
            \item\label{rem:col-H1-H2-choice-p-beta}
            The values of $p$ and $\beta$ are chosen to ensure that there are $(\log t)^{\Omega(k-\alpha)}$ possible values of $\vec{\calQ}$; we will discuss this in more detail when it becomes relevant.
        \end{enumerate}
    \end{remark}
	
	The following lemma establishes \cref{thm:colourful-G1-G2}(i) for $\calG_1$ and $\calG_2$.
	
    \begin{lemma}\label{lma:col-edge-numbers}
	    We have
	    \[
	        \pr_{(G_1,G_2)\sim(\calG_1,\calG_2)} \Big(e(G_2) \ge 4e(G_1)\Big) \ge 19/20\,.
	    \]
	\end{lemma}
	\begin{proof}
	    It suffices to demonstrate that, with probability at least $19/20$, the following events occur simultaneously:
	    \begin{itemize}
	        \item $\calE_1$, the event that $e(H_1) \le 2pt^k$ holds; and
	        \item $\calE_2$, the event that $e(H_2) \ge 8pt^k$ holds.
	    \end{itemize}
	    Indeed, if these two events occur, then $e(G_2) \ge e(H_2) \ge 4e(H_1) = 4e(G_1)$ as required.
	    
	    We first show that $\calE_1$ is likely. We have $\E(e(H_1)) = pt^k$, so by the Chernoff bound of \cref{lem:chernoff-low} applied with $\delta = 1$ we have
	    \begin{align*}
	        \Pr(\calE_1) = 1 - \pr\Big(e(H_1) > 2pt^k\Big) \ge 1 - 2e^{-pt^k/3}\,.
	    \end{align*}
	    Since $\alpha \le k-3$, we have $p \ge 1/t^{k-1/2}$ and so since $t \ge 400$ we have
	    \begin{equation}\label{eqn:col-edge-numbers-E1}
	        \pr(\calE_1) \ge 1 - 2e^{-t^{1/2}/3} \ge 1 - 2e^{-20/3} > 99/100\,.
	    \end{equation}
	    
	    We now show that $\calE_2$ is likely. Let $\vec{q} = (q_1,\dots,q_{k-\flalpha-2})\in\supp(\vec{\calQ})$ be any possible value of~$\vec{\calQ}$. Then for all $j \in[k-\flalpha-2]$, we have $\E(|\calX_j|\mid\vec{\calQ}=\vec{q}) = tq_j$, so by \cref{lem:chernoff-low} applied with $\delta=1/2$ we have
	    \[
	        \pr\Big(|\calX_j| < tq_j/2\,\Big|\, \vec{\calQ} = \vec{q}\Big) \le 2e^{-tq_j/12}\,.
	    \]
	    By \cref{rem:col-H1-H2}\ref{rem:col-H1-H2-choice-B} about our choice of $B$, we have $tq_j \ge tx^{B/\beta} \ge 24\log t$; thus
	    \[
	        \pr\Big(|\calX_j| < tq_j/2\,\Big|\, \vec{\calQ} = \vec{q}\Big) \le 2e^{-2\log t} = 2/t^2\,.
	    \]
	    By a union bound over $j$, it follows that
	    \[
	        \pr\Big(\text{for all $j \in [k-\flalpha-2]$, we have }|\calX_j| \ge  tq_j/2\,\Big|\, \vec{\calQ} = \vec{q}\Big) \ge 1 - 2k/t^2\,.
	    \]
	    If this event occurs, then since $\alpha \le k-3$, we have 
	    \[
	        e(H_2) =\prod_{i=1}^{k-\flalpha-2} \abs{\calX_i}\ge (t/2)^{k-\flalpha-2}\cdot\prod_{i=1}^{k-\flalpha-2}q_j
            =
            t^{k-\flalpha-2}\cdot \frac{2^{5(k-\flalpha-2)}}{2^{k-\flalpha-2}}\cdot x \ge 8pt^k\,,
	    \]
        and so $\calE_2$ also occurs. Since $t \ge t_0$, it follows that
	    \[
	        \pr(\calE_2 \mid \vec{\calQ} = \vec{q}) \ge 1 - 2k/t^2 \ge 99/100\,.
	    \]
	    Since $\vec{q}$ was arbitrary, by a union bound with~\eqref{eqn:col-edge-numbers-E1} we arrive at $\pr(\calE_1 \cap \calE_2) \ge 19/20$, and the result follows as described above.
	\end{proof}

    \subsubsection{A framework to distinguish \texorpdfstring{$\bm{\calG_1}$}{G\_1} from \texorpdfstring{$\bm{\calG_2}$}{G\_2}}%
    \label{sec:col-lower-framework}%

    Throughout the rest of this section, let $t$, $k$, $\alpha$, $t_0$, $A$ and $C$ be as in the statement of \cref{thm:colourful-G1-G2}(ii), and let $\calG_1$, $\calG_2$, $G_1$, $G_2$, $H_1$, $H_2$, $p$, $x$, $\beta$, $B$, $V_1,\dots,V_k$, $\calX_1,\dots,\calX_k$, $\calR_{k-\flalpha-1},\dots,\calR_k$ and $\vec{\calQ}$ be as in \cref{def:col-H1-H2}. As in \cref{rem:uncol-trivial} in the uncoloured case, we first observe that we may assume without loss of generality that the number of queries made by $A$ on $G_1$ is deterministic.
    
    \begin{remark}\label{rem:col-trivial}
            Without loss of generality, we may assume that there is an~${N=N_{n,k}\in\N}$ such that $A(\cindora(G))$ makes exactly the same number~$N$ of queries for each~$n$-vertex~$k$-hypergraph~$G$. Moreover, we may assume that for some $N'(G) \le N$, the first $N'(G)$ queries have non-zero cost and the last $N-N'(G)$ queries have zero cost.
    \end{remark}
    \begin{proof}
        As in the proof of \cref{rem:uncol-trivial}, the result follows by first skipping zero-cost queries and then ``padding'' the end of $A$'s execution with zero-cost queries to the empty set.
    \end{proof}

    We now reintroduce some notation from \cref{sec:uncol-lower} to describe the sequence of oracle queries executed by $A$ on a given input graph; this is a precise analogue of \cref{defn:uncol-query-1} in the colourful setting.

    \begin{defn}
        Let $n\coloneqq kt$ and recall that~$V_1,\dots,V_k$ are disjoint sets of size~$t$ each.
        Let $G$ be an arbitrary $n$-vertex $k$-hypergraph on the vertex classes~$V_1,\dots,V_k$. 
        Let $S_1(G),\dots,S_N(G)$ be the sequence of oracle queries that $A$ makes when given input $\cindora(G)$.
        For all $i \in [N]$, write $S_i(G) \eqqcolon (S_{i,1}(G),\dots,S_{i,k}(G))$. 
    \end{defn}
    
    In the following, we outline the key idea of the proof of \cref{thm:colourful-G1-G2}(ii).
    We first recall from the assumption on~$A$ in \cref{thm:colourful-G1-G2}(ii) that $S_{i,j}(G)\subseteq V_i$ holds for all~$i\in[N]$ and~$j\in[k]$.
    As in the uncoloured case, in order for a query~$S$ to distinguish between $G_1$ and $G_2$, it must contain at least one edge of~$H_2$ but no edges of~$H_1$; this means it must contain all roots~$\calR_j\in \calX_j$ of each rooted class~$V_j$ of~$H_2$ and at least one vertex from $\calX_j$ for each non-rooted class~$V_j$ of~$H_2$. Roughly speaking, we will show that, with high probability, in order for a query $S=(S_1,\dots,S_k)$ to be useful, it must satisfy the following criteria:
    \begin{enumerate}[(C1)]
        \item There cannot be too many ``unexposed'' possible edges of $H_1$ in $S^{(k)}$, since otherwise $H_1[S]$ is likely to contain an edge and the query~$S$ will not distinguish~$G_1$ from~$G_2$. (Recall from \cref{sec:notation} that we write
        $S^{(k)} = \{\{s_1,\dots,s_k\}\colon s_j \in S_j \mbox{ for all }j \in [k]\}$
        and
        $H_1[S] = H_1[S_1,\dots,S_k]$.)
        \item For each rooted class $V_j$, the set~$S_j\subseteq V_j$ must be very large, since otherwise the query~$S$ is likely to miss the root~$\calR_j$ of $H_2$, in which case $H_2[S]$ will contain no edges and the query~$S$ will not distinguish~$G_1$ from~$G_2$. In particular, we will see that this criterion requires $\cost(S)$ to be at least roughly $n^{\alpha}$. (This is similar to \cref{cor:col-dec}.)
        \item For each non-rooted class $V_j$, the set $S_j\subseteq V_j$ must contain at least roughly $1/\calQ_j$ vertices, since otherwise~$S_j$ will contain no vertices of $\calX_j$, in which case $H_2[S]$ will contain no edges and the query~$S$ will again not distinguish~$G_1$ from~$G_2$.
    \end{enumerate}
    By combining these three properties,
    we will be able to show that, with high probability, any query distinguishing~$G_1$ and~$G_2$ must be ``accurately profiled'' in the sense that, for all non-rooted classes $V_j$, we have $\abs{S_j} \approx 1/\calQ_j$. Intuitively, this shows that~$A$ has to guess the value of~$\vec\calQ$; there are $\log^{\Theta(k-\flalpha)}t$ possible values of $\vec\calQ$, and property (C2) says that each useful query has cost at least roughly $n^{\alpha}$, so this leads naturally to the overall lower bound on the oracle cost that we are claiming. Thus in a sense, the proof is based around showing that the bottleneck of \dlmimprove described in \cref{sec:col-lower-G1-G2} (in which the algorithm essentially does guess an analogue of~$\vec\calQ$) is necessary.

    Turning these ideas into a rigorous argument is difficult, particularly since we work in an adaptive setting where queries may depend on answers received in past queries, and so we must be very careful with conditioning. We will first define formal terminology and events corresponding to (C1)--(C3), as well as the idea of ``accuracy''. We will then prove (in \cref{lem:adapt-col-framework-new-new}) that if these events all occur and all queries are inaccurate, then $A$ fails to distinguish $G_1$ from $G_2$. Finally, in \cref{sec:col-lower-bound-prob} we will show that these events are all likely to occur, and apply a union bound to prove the result. 
    
    We first formalise the idea of ``accuracy''.

    \begin{defn}
        Let $S_i = (S_{i,1},\dots,S_{i,k})$ be a query.
        \begin{itemize}
        \item
        We say $S_i$ is \emph{accurately rooted} if $\calR_j \in S_{i,j}$ holds for all $j$ with $k-\flalpha-1 \le j \le k$, and otherwise we say $S_i$ is \emph{inaccurately rooted}.
        \item
        We say $S_i$ is \emph{accurately profiled} if, for all $j$ with $1\le j \le k-\flalpha-2$, we have $|S_{i,j}| \in (x^{1/2\beta}/\calQ_j,\  x^{-1/2\beta}/\calQ_j)$, and otherwise we say $S_i$ is \emph{inaccurately profiled}.
        \item
        We say $S_i$ is \emph{accurate} if it is both accurately rooted and accurately profiled, and otherwise we say $S_i$ is \emph{inaccurate}.
        \item We define $\EInacc$ to be the event that all queries $S_1(G_1),\dots,S_N(G_1)$ are inaccurate.\label{def:EInacc}
        \end{itemize}
    \end{defn}

    We next formalise the notion of an ``unexposed'' edge from (C1), in a similar fashion to the proof of \cref{lem:uncol-size} in the uncoloured setting.
    \begin{defn}\label{def:col-unexposed}
        For all $i \in [N]$, we define the set~$F_i$ of unexposed edges via
        \[
            F_i \coloneqq S_i(G_1)^{(k)} \setminus \bigcup_{\substack{\ell \le i-1\\e(G_1[S_\ell])=0}} S_\ell(G_1)^{(k)}\,.
        \]
    \end{defn}
    Note for intuition that, before the~$i$\th query, the algorithm has already performed the queries~$S_\ell(G_1)$ for $\ell\le i-1$. Each time it received $e(G_1[S_\ell(G_1)])=0$ as a response from the query~$\cindora(G_1)_{S_\ell(G_1)}$, the algorithm can know for certain that none of the edges in~$S_\ell(G_1)^{(k)}$ exist. Thus, $F_i$ is the set of all possible, unexposed edges of~$G_1$ that the $i$\th query can hope to newly detect.

    We next define three events which, if they occur, formalise (C1)--(C3) respectively.

    \def\Filower{(\log^{k^3} t)/(2p)}
    \begin{defn}[Formalisation of (C1)]\label{def:EEdge}
        Let $\EEdge$ be the event that, for all $i \in [N]$ with $|F_i| \ge \Filower$, we have $e(G_1[S_i(G_1)]) > 0$.
    \end{defn}
    \begin{defn}[Formalisation of (C2)]\label{def:ERoot}
        We say a query $S = (S_1,\dots,S_k)$ is \emph{safely rooted} if 
        \[
            \min\setc{|S_j|}{k-\flalpha-1\le j \le k} > t/\log^{2k}t\,.
        \]
        Let $\ERoot$ be the event that, for all $i \in [N]$, the query $S_i(G_1)$ is safely rooted or inaccurately rooted.
    \end{defn}
    \begin{defn}[Formalisation of (C3)]\label{def:ENroot}
        Let $\xi \coloneqq x^{-1/(2k\beta)}$. Let $\ENroot$ be the event that, for all $i \in [N]$ such that $S_i(G_1)$ is safely rooted and such that there exists $j\in[k-\flalpha-2]$ with $|S_{i,j}(G_1)| \le 1/(\xi\calQ_j)$, we have $e(H_2[S_i(G_1)]) = 0$.
    \end{defn}

    Note that $\xi > 1$ holds, since $x < 1$ and $\beta \ge 1$. Finally, recall from the assumption on~$A$ in \cref{thm:colourful-G1-G2}(ii) that~$A(\cindora(G_1))$ has expected oracle cost at most~$C/2$.
    We now define an event for the cost of $A(\cindora(G_1))$ not being too much larger than this expectation.
    \begin{defn}\label{def:ECost}
        Let $\ECost$ be the event that $\cost(A,G_1) \le C$ holds.
    \end{defn}

    Our next goal is to prove \cref{lem:adapt-col-framework-new-new}, which says that, if $\EInacc$, $\EEdge$, $\ERoot$, $\ENroot$, and~$\ECost$ all occur, then $A$ fails to distinguish $G_1$ from $G_2$, as required by \cref{thm:colourful-G1-G2}(ii). To this end we first prove an ancillary lemma, which says that any sequence of individually-cheap queries cannot ``cover'' many possible tuples of root locations in $V_{k-\flalpha-1}, \dots, V_k$.

    \begin{lemma}\label{lem:correct-number-roots}
	    Let $\gamma \in (0,1]$, and let $S_1,\dots,S_z$ be a sequence of queries of total cost at most~$C$.
        Further suppose $|S_{i,k-\flalpha-1}| + \dots + |S_{i,k}| \le \gamma kt$ for all~$i\in[z]$. Then 
		\[
			\sum_{i\in [z]} |S_{i,k-\flalpha-1} \times \dots \times S_{i,k}| \le C(\gamma kt)^{\flalpha + 2 - \alpha}\,.
		\]
	\end{lemma}
	\begin{proof}
	    Note first that, by the AM-GM inequality, for all $i \in [z]$ we have
	    \begin{align*}
	        \prod_{j=k-\flalpha-1}^k |S_{i,j}| &\le \left( \frac{\sum_{j=k-\flalpha-1}^k |S_{i,j}|}{\flalpha + 2} \right)^{\flalpha + 2} 
	         \le \bigg( \sum_{j=k-\flalpha-1}^k |S_{i,j}| \bigg)^{\flalpha + 2}
             \,.
	    \end{align*}
	    Setting $s_i \coloneqq |S_{i,k-\flalpha-1}| + \dots + |S_{i,k}|$, this gives
	    \begin{equation}\label{eq:correct-number-roots}
	        \sum_{i\in [z]} |S_{i,k-\flalpha-1} \times \dots \times S_{i,k}| \le \sum_{i \in [z]} s_i^{\flalpha + 2}\,.
	    \end{equation}
        By assumption, we have $\sum_{i \in [z]} s_i^\alpha \le \sum_{i\in[z]}\cost(S_i) \le C$ and $s_i \in [0,\gamma kt]$ for all $i\in[z]$. Moreover, $\flalpha+2 > \alpha$.
        We now apply Karamata's inequality in the form of \cref{cor:karamata}, taking $W=C$, $c=\gamma kt$, $t=z$ and $r=\flalpha+2$.
        This yields:
        \[
            \sum_{i \in [z]} s_i^{\flalpha + 2} \le C (\gamma kt)^{\flalpha + 2 - \alpha}\,, 
        \]
        and so the result follows from~\eqref{eq:correct-number-roots}.
	\end{proof}
	
	We now prove the main lemma of this subsection: If the five key events defined above occur simultaneously, then the algorithm~$A$ is unable to distinguish the two possible input graphs~$G_1$ and~$G_2$ from each other.
	
	\begin{lemma}\label{lem:adapt-col-framework-new-new}
		If $\EInacc$, $\EEdge$, $\ERoot$, $\ENroot$, and $\ECost$ all occur, then we have $A(\cindora(G_1)) = A(\cindora(G_2))$.
	\end{lemma}
	\begin{proof}
	    Throughout, we assume that $\EInacc$, $\EEdge$, $\ERoot$, $\ENroot$ and $\ECost$ occur and work deterministically. For all $i \in [N]$, let $S_i$ be the value of $S_i(G_1)$. 
	
	    We will prove by induction that, for all $i \in [N]$, we have $\cindora(G_1)_{S_i} = \cindora(G_2)_{S_i}$; since $A$ is deterministic, this immediately implies $A(\cindora(G_1)) = A(\cindora(G_2))$.
        Fix $i \in [N]$, and suppose that $\cindora(G_1)_{S_\ell} = \cindora(G_2)_{S_\ell}$ holds for all $\ell \in [i-1]$. 
	    
	    We first define some useful notation. Let
	    \begin{align*}
	        \yset_i &= \setc{S_\ell}{\ell \in [i-1] \text{ and } e(G_1[S_\ell]) > 0}\,,\\
	        \nset_i &= \setc{S_\ell}{\ell \in [i-1] \text{ and } e(G_2[S_\ell]) = 0}\,.
	    \end{align*}
	    Recall that $G_1$ is a subgraph of $G_2$; thus $\yset_i$ is the set of all queries in $\{S_1,\dots,S_{i-1}\}$ which return that an edge exists, regardless of whether the input graph is $G_1$ or $G_2$; and $\nset_i$ is the set of all queries which return that no edge exists, regardless of whether the input graph is $G_1$ or $G_2$.
        By our induction hypothesis, we have $[i-1] = \yset_i \cup \nset_i$.
		
		Since $\EInacc$ occurs by assumption, $S_i$ is inaccurate.
        We now consider three cases depending on how~$S_i$ fails to be accurate: by definition, either it is inaccurately rooted (Case 1), or it is inaccurately profiled and there is at least one vertex class whose intersection with the query is small (Case 2), or it is inaccurately profiled and the intersection of every vertex class with the query is reasonably large (Case 3).
		
		\textit{Case 1: $S_i$ is inaccurately rooted.} In this case, by definition of $H_2$, we have $e(H_2[S_i]) = 0$. Since $G_2 = G_1 \cup H_2$, it follows that $e(G_1[S_i]) = e(G_2[S_i])$, and we have $\cindora(G_1)_{S_i} = \cindora(G_2)_{S_i}$ as required.
		
		\textit{Case 2: $S_i$ is accurately rooted and there exists a non-rooted class $V_j$ such that $|S_{i,j}| \le 1/(\xi\calQ_j)$.}
        In this case, since $\ENroot$ occurs, we again have $e(H_2[S_i]) = 0$ and thus $\cindora(G_2)_{S_i} = \cindora(G_1)_{S_i}$ as required.
		
		\textit{Case 3: $S_i$ is accurately rooted and, for all non-rooted classes $V_j$, we have $|S_{i,j}| \ge 1/(\xi\calQ_j)$.} This is the difficult case of the proof. Our aim is to prove $F_i\ge\Filower$, which by $\EEdge$ implies $e(G_1[S_i]) > 0$; since $G_1$ is a subgraph of $G_2$, this immediately implies $\cindora(G_1)_{S_i} = \cindora(G_2)_{S_i} = 0$.
		
		We divide the queries in $\nset_i$ into two sets according to their size; let
		\begin{align*}
			\nset_i^+ &= \Big\{S_m \in \nset_i \colon \max\big\{|S_{m,j}| \colon j \in [k]\big\} \ge t/(\log t)^{4k^2/(\flalpha + 2 - \alpha)} \Big\}\,,\\
			\nset_{i}^- &= \nset_{i} \setminus \nset_{i}^+\,.
		\end{align*}
		Observe that by the definition of $F_i$, we have
		\begin{align}\nonumber
            F_i &
            =S_i^{(k)} \setminus \bigcup_{S_m \in \nset_i} S_m^{(k)}
            =S_i^{(k)} \setminus \paren*{
                S_i^{(k)} \cap \bigcup_{S_m \in \nset_i} \paren[\Big]{
                    S_m^{(k)}
                    \setminus
                    \bigcup_{S_r \in \nset_m}
                    S_r^{(k)}
                }
            }\,,
		\end{align}
        and thus by~$\nset_i=\nset_i^-\cup\nset_i^+$, we have
        \begin{align}\nonumber
			|F_i| &\ge \abs[\big]{S_i^{(k)}} - \Bigg|S_i^{(k)} \cap \bigcup_{S_m \in \nset_i^-} S_m^{(k)} \Bigg|
            -\Bigg|\bigcup_{S_m \in \nset_i^+}
            \paren[\Big]{
                S_m^{(k)}
                \setminus
                \bigcup_{S_r \in \nset_m}
                S_r^{(k)}}
            \Bigg|\\\label{eqn:adapt-col-frame-4}
			&\ge \abs[\big]{S_i^{(k)}} - \Bigg|S_i^{(k)} \cap \bigcup_{S_m \in \nset_i^-} S_m^{(k)} \Bigg| - \sum_{S_m \in \nset_i^+} \Bigg|S_m^{(k)} \setminus \bigcup_{\substack S_r \in \nset_m} S_r^{(k)}\Bigg|\,.
		\end{align}
		We will show that the second and the third term each are at most~$\tfrac14\abs[\big]{S_i^{(k)}}$. We will bound each term of~\eqref{eqn:adapt-col-frame-4} in turn.
		
        \emph{First term of~\eqref{eqn:adapt-col-frame-4}:}
		We bound $\abs[\big]{S_i^{(k)}}$ below.
        Since $\EInacc$ occurs, $S_i$ is an inaccurate query; as we are assuming in Case~3 that $S_i$ is accurately rooted, it must be inaccurately profiled, so that there exists $\ell\in [k-\flalpha-2]$ with $|S_{i,\ell}| \notin (x^{1/(2\beta)}/\calQ_\ell, x^{-1/(2\beta)}/\calQ_\ell)$.
        Moreover, in Case~3 we have $|S_{i,\ell}| \ge 1/(\xi\calQ_{\ell})$.
        Note that $1/\xi = x^{1/(2k\beta)} > x^{1/(2\beta)}$ holds by $x < 1$, which implies $|S_{i,\ell}| > x^{1/(2\beta)}/\calQ_{\ell}$;
        due to the inaccuracy of~$\abs{S_{i,\ell}}$, we must have $|S_{i,\ell}| \ge x^{-1/(2\beta)}/\calQ_{\ell}$.
        Moreover, since $\ERoot$ occurs and we are assuming that~$S_i$ is accurately rooted, $S_i$ must also be safely rooted; thus $|S_{i,j}| \ge t/\log^{2k} t$ for all $j \ge k-\flalpha-1$. Putting these three bounds together, it follows that
		\begin{align*}
			\abs[\big]{S_i^{(k)}} &= \prod_{j=1}^k |S_{i,j}| \ge |S_{i,\ell}| \cdot \prod_{\substack{j \le k-\flalpha - 2\\j\ne \ell}} |S_{i,j}| \cdot \prod_{j=k-\flalpha-1}^k |S_{i,j}|\\
			&\ge
			\frac{x^{-1/(2\beta)}}{\calQ_\ell} \cdot \prod_{\substack{j \le k-\flalpha-2\\j\ne\ell}} \frac{1}{\xi\calQ_j} \cdot \prod_{j= k-\flalpha-1}^k \frac{t}{\log^{2k} t}\\
			&=\frac{x^{-1/(2\beta)}}{\xi^{k-\flalpha-3}} \cdot \frac{t^{\flalpha+2}}{\log^{2k(\flalpha+2)}n}  \cdot\prod_{j =1}^{k-\flalpha-2} \frac{1}{\calQ_j}
			\,.
		\end{align*}
		Thus by $\prod_j \calQ_j = 2^{5(k-\flalpha-2)}x \le 2^{5k}pt^{\flalpha+2}$ and $\xi = x^{-1/(2k\beta)}>1$ and $\flalpha+2\le k$, we have
		\begin{align}\label{eq:col-framework-abc}
		    \abs[\big]{S_i^{(k)}}&\ge \frac{x^{-1/(2\beta)}}{\xi^{k-2}} \cdot \frac{t^{\flalpha+2}}{\log^{2k^2}t} \cdot \frac{1}{2^{5k}p\cdot t^{\flalpha+2}}
			= \frac{x^{-1/(k\beta)}}{2^{5k}p\log^{2k^2}t}\,.
		\end{align}
		Recall that $x = pt^{\flalpha + 2} = t^{(\flalpha + 2 - k)/2}$. Since $\alpha \le k-3$ we have $k-\flalpha-2\ge1$, so $x^{-1} \ge t^{1/2}$. By the definition of $\beta$ in \cref{def:col-H1-H2}, it follows that
		\[
		    x^{-1/(k\beta)}  \ge t^{1/(2k\beta)}\ge t^{(20k^4\log \log t)/(2 k \log t)} = (\log t)^{10k^3}\,.
		\]
		Since $t \ge t_0$ we have $(\log t)^{k^3}\ge 2^{5k}$, and so by~\eqref{eq:col-framework-abc} we have
		\begin{equation}\label{eqn:adapt-col-frame-2-new}
			\abs[\big]{S_i^{(k)}} \ge \frac{(\log t)^{7k^3}}{p}\,.
    \end{equation}
		
        \emph{Second term of~\eqref{eqn:adapt-col-frame-4}:}
		We show that the second term is at most $\tfrac14\abs[\big]{S_i^{(k)}}$, by showing that almost all possible edges of $G_1$ covered by $S_i$ intersect $V_{k-\flalpha-1},\dots,V_k$ at vertices not covered by any query in $\nset_i^-$.

        Indeed, we will apply \cref{lem:correct-number-roots} with $\gamma = 1/(\log t)^{4k^2/(\flalpha + 2 - \alpha)}$ to the queries in $\nset_i^-$.
        Note that the conditions of the lemma are met, because, by $\ECost$, the total cost of all queries is at most~$C$ and, by definition of~$\nset_i^-$, the total size $\sum_{j=1}^k \abs{S_{m,j}}$ of queries~$S_m\in\nset_i^-$ is at most $\gamma kt$.
        Thus we get
		\begin{align*}
		    \Big|\bigcup_{S_m \in \nset_i^-} (S_{m,k-\flalpha-1} \times \dots \times S_{m,k}) \Big|
            &\le \sum_{S_m\in \nset_i^-} |S_{m,k-\flalpha-1} \times \dots \times S_{m,k}|\\
		    & \le C \left(\frac{kt}{(\log t)^{4k^2/(\flalpha + 2 - \alpha)}}\right)^{\flalpha + 2 - \alpha}\,.
		\end{align*}
		By the definition of $C$ in \cref{thm:colourful-G1-G2}, we get
		\begin{equation}\label{eqn:adapt-col-frame-root}
		    \Big|\bigcup_{S_m \in \nset_i^-} (S_{m,k-\flalpha-1} \times \dots \times S_{m,k}) \Big| \le C\cdot \frac{t^{\flalpha + 2-\alpha}k^2}{\log^{4k^2}t} \le \frac{t^{\flalpha + 2}\log^k t}{\log^{4k^2}t} \le \frac{t^{\flalpha + 2}}{\log^{3k^2}t}\,.
		\end{equation}
		Next, we observe that since $\ERoot$ occurs and $S_i$ is accurately rooted by hypothesis of Case~3, we have $|S_{i,j}| \ge t/\log^{2k} t$ for all $j \ge k-\flalpha-1$; thus by~\eqref{eqn:adapt-col-frame-root} we have
		\begin{align*}
		    |S_{i,k-\flalpha-1} \times \dots \times S_k| &
		    =\prod_{j=k-\flalpha-1}^k |S_{i,j}| 
			\ge \frac{t^{\flalpha+2}}{\log^{2k^2} t}\\ 
			&\ge (\log^{k^2} t)\cdot\Big|\bigcup_{S_m \in \nset_i^-} (S_{m,k-\flalpha-1} \times \dots \times S_{m,k}) \Big|\,.
		\end{align*}
		This implies that almost all $k$-tuples of $S_i^{(k)}$ intersect $V_{k-\flalpha-1},\dots,V_k$ at vertices not covered by any query in $\nset_i^-$; we therefore have
		\begin{equation}\label{eqn:adapt-col-frame-5}
			\Bigg|S_i^{(k)} \cap \bigcup_{S_m \in \nset_i^-} S_m^{(k)} \Bigg| \le \frac{1}{4}\abs[\big]{S_i^{(k)}}\,.
		\end{equation}
		
        \emph{Third term of~\eqref{eqn:adapt-col-frame-4}:}
        We show that the third term is at most $\tfrac14\abs[\big]{S_i^{(k)}}$.
        Observe that by the definition of $F_m$, we conveniently have
		\begin{equation}\label{eq:adapt-col-frame-last-term}
			\sum_{S_m \in \nset_i^+} \Bigg|S_m^{(k)} \setminus \bigcup_{S_r \in \nset_m} S_r^{(k)}\Bigg| = \sum_{S_m \in \nset_i^+} |F_m|\,.
		\end{equation}
		We proceed by proving an upper bound on $|\nset_i^+|$ as well as on each term $|F_m|$.
        By the definition of~$\nset_i^+$, for each query $S_m \in \nset_i^+$, we have 
		\[
		    \cost(S_m) \ge (\max_j |S_{m,j}|)^{\alpha} \ge (t/(\log t)^{4k^2/(\flalpha + 2 - \alpha)})^\alpha \ge (t/\log^{4k^2}t)^\alpha\,.
		\]
		Since $\ECost$ occurs, the total cost of all queries is at most $C$; by the definition of~$C$, it follows that
		\[
		    |\nset_i^+| \le \frac{C}{(t/\log^{4k^2}t)^\alpha} \le \frac{t^\alpha\log^k t}{(t/\log^{4k^2}t)^\alpha} \le \log^{5k^3}t\,.
		\]
		Moreover, since $\EEdge$ occurs, we have $|F_m| \le \Filower$ for all queries $S_m \in \nset_i$. It therefore follows from~\eqref{eq:adapt-col-frame-last-term} that
		\[
			\sum_{S_m \in \nset_i^+} \bigg|S_m^{(k)} \setminus \bigcup_{S_r \in \nset_m} S_r^{(k)}\bigg| \le (\log t)^{5k^3} \cdot \frac{\log^{k^3} t}{2p} \le \frac{\log^{6k^3} t}{p}\,.
		\]
		By~\eqref{eqn:adapt-col-frame-2-new}, it follows that 
		\begin{equation}\label{eqn:adapt-col-frame-6}
			\sum_{S_m \in \nset_i^+} \bigg|S_m^{(k)} \setminus \bigcup_{S_r \in \nset_m} S_r^{(k)}\bigg| \le \frac{1}{\log^{k^3}n}\abs[\big]{S_i^{(k)}} \le \frac{1}{4}\abs[\big]{S_i^{(k)}}\,.
		\end{equation}
		
        \emph{Conclusion of the proof of \cref{lem:adapt-col-framework-new-new}.}
		We are now essentially done. Combining~\eqref{eqn:adapt-col-frame-4}, \eqref{eqn:adapt-col-frame-5}, and~\eqref{eqn:adapt-col-frame-6} yields $|F_i| \ge \abs[\big]{S_i^{(k)}}/2$; applying~\eqref{eqn:adapt-col-frame-2-new} then yields
		\[
			|F_i| \ge \frac{1}{2}\abs[\big]{S_i^{(k)}} \ge \frac{\log^{7k^3} t}{2p}\,.
		\]
		Since $\EEdge$ occurs, we have $e(H_1[S_i]) > 0$, and so $\cindora(G_1)_{S_i} = \cindora(G_2)_{S_i}=0$ as required.
	\end{proof}
	
	\subsubsection{Bounding probabilities}\label{sec:col-lower-bound-prob}
    
	We are now going to \textbf{separately} bound the probability that each event $\EEdge$, $\ERoot$, $\ENroot$ and~$\EInacc$ fails to occur in terms of the cost~$C$.
    The event $\ECost$ will be easy to handle using Markov's inequality.
    The final result then follows using \cref{lem:adapt-col-framework-new-new} and the assumption that~$A$ \emph{does} distinguish~$\calG_1$ from~$\calG_2$ with probability at least $2/3$; together, this shows that the cost must be high.

    For every event except $\EEdge$, we will be able to work with $H_1$ (and hence $S_1,S_2,\dots$) exposed, using the fact that $H_1$ and $H_2$ are independent.
	
	\begin{defn}
	    For the rest of the section and as in the proof of \cref{lem:adapt-col-framework-new-new}, we write $S_i \coloneqq S_i(G_1)$ for all $i \in [N]$.
	\end{defn}
	
	\begin{lemma}\label{lem:adapt-col-E1}
		We have
		\[
			\Pr(\overline{\ECost} \vee \EEdge) \ge 1 - Ce^{-(\log^{k^3} t)/2}\,.
		\]
	\end{lemma}
	\begin{proof}
	    For all $i$, let $\EEdgeSub{i}$\label{def:EEdgeSub}
        be the event that either $|F_i| \le \Filower$
        or $e(G_1[S_i]) > 0$. Let $\EEdgeSub{\le i} = \EEdgeSub{1} \wedge \dots \wedge \EEdgeSub{i}$. Recall from \cref{rem:col-trivial} that all queries have non-zero cost --- and hence cost at least $1$ --- except for a final segment of ``padding'' at the end of the algorithm. Thus if $\EEdgeSub{\le \floor{C}}$ occurs, then either $\EEdge$ occurs or the total query cost of $A(\cindora(G_1))$ is greater than $C$; we therefore have
	    \begin{align}\nonumber
	        \pr(\overline{\ECost} \vee \EEdge) &
	        \ge \pr(\EEdgeSub{\le \floor{C}}) 
	        =
            1 - \sum_{i=1}^{\floor{C}}\pr\bigg(\overline{\EEdgeSub{i}} \wedge \bigwedge_{j=1}^{i-1}\EEdgeSub{j}\bigg)\\\label{eq:adapt-col-E1}
	        &\ge 1 - \sum_{i=1}^{\floor{C}}\pr\bigg(\overline{\EEdgeSub{i}} \,\,\bigg|\,\, \bigwedge_{j=1}^{i-1}\EEdgeSub{j}\bigg)\,.
	    \end{align}
	    
	    We will bound this sum term-by-term by exposing the results of past queries. Let $i \le \floor{C}$. Let $T_{<i}(G_1) = (S_1,\dots,S_i,\cindora(G_1)_{S_1},\dots,\cindora(G_1)_{S_{i-1}})$ be the information to which $A$ has access after its $(i-1)$\st query to $\cindora(G_1)$; thus $T_{<i}(G_1)$ is a deterministic function of $G_1$. Let $t_{<i}(G_1) = (s_1,\dots,s_i,b_1,\dots,b_{i-1})$ be any possible value for $T_{<i}(G_1)$ consistent with the conditioning of~\eqref{eq:adapt-col-E1}, and let $f_i$ be the value of $F_i$ implied by conditioning on $T_{<i}(G_1) = t_{<i}(G_1)$. Let $\yset_i = \{s_j \colon b_j = 0\}$ and $\nset_i = \{s_j \colon b_j = 1\}$; thus conditioned on $T_{<i}(G_1) = t_{<i}(G_1)$, we have $s_j \in \yset_i$ if $e(G_1[s_j]) > 0$ and $s_j \in \nset_i$ if $e(G_1[s_j]) = 0$. Then we have
	    \begin{align*}
	        &\pr\big(\overline{\EEdgeSub{i}} \mid T_{<i}(G_1) = t_{<i}(G_1)\big)\\
	        &\qquad= \pr\Big(\overline{\EEdgeSub{i}} \,\Big|\, e(G_1[s_j]) > 0 \mbox{ for all }s_j \in \yset_i\mbox{ and }e(G_1[s_j]) = 0\mbox{ for all }s_j \in \nset_i\Big)\,.
	    \end{align*}
	    
	    If $|f_i| \le \Filower$
        then this probability is zero, so suppose $|f_i| > \Filower$. In this case $\overline{\EEdgeSub{i}}$ occurs if and only if $e(G_1[s_i]) = 0$. This event is a monotonically decreasing function of the indicator variables of $G_1$'s (independently-present) edges, and  for all~$j$, the events $e(G_1[s_j]) > 0$ are monotonically increasing functions of these variables. Thus by the FKG inequality (\cref{lem:FKG}), we obtain
	    \begin{align*}
	        \pr\Big(\overline{\EEdgeSub{i}} \mid T_{<i}(G_1) = t_{<i}(G_1)\Big)
	        &\le \pr\Big(e(G_1[s_i]) = 0 \,\Big|\, e(G_1[s_j]) = 0 \mbox{ for all }s_j \in \nset_i\Big)\\
	        &= (1-p)^{\Big|s_i^{(k)} \,\setminus\, \bigcup_{s_j \in \nset_i} s_j^{(k)}\Big|}\,.
	    \end{align*}
	    By the definition of $f_i$, the term in the exponent here is simply $\abs{f_i}$; we therefore have
	    \[
	        \pr\Big(\overline{\EEdgeSub{i}} \mid T_{<i}(G_1) = t_{<i}(G_1)\Big) \le e^{-p\abs{f_i}} \le e^{-(\log^{k^3}t)/2}\,.
	    \]
	    By~\eqref{eq:adapt-col-E1}, it follows that
	    \[
	        \pr(\overline{\ECost} \vee \EEdge) \ge 1 - \floor{C} e^{-(\log^{k^3}t)/2} \ge 1 - Ce^{-(\log^{k^3}t)/2}\,,
	    \]
	    as required.
	\end{proof}

    We will next show that $\ERoot$ is likely to occur, in \cref{lem:adapt-col-E3}.
    In order to do so, we will apply Karamata's inequality to show in \cref{lem:adapt-col-unsafe-roots} that an arbitrary sequence of unsafely-rooted queries cannot cover too many possible tuples of roots.
    In order to prove \cref{lem:adapt-col-unsafe-roots}, we first bound the effectiveness an unsafely-rooted query in terms of its cost when $\alpha \ge 1$ and the query is large.

    \begin{lemma}\label{lem:adapt-col-pre-unsafe-roots}
        Suppose $\alpha \ge 1$. If $S_i$ is not safely rooted and there exists $j \ge k-\flalpha-1$ with $|S_{i,j}| \ge t/\log^k t$, then 
        \[
            \abs[\big]{S_{i,k-\flalpha-1} \times \dots \times S_{i,k}}
            \le \frac{\cost(S_i)^{(\flalpha+2)/\alpha}}{\log^k t}\,.
        \]
    \end{lemma}
    \begin{proof}
        Without loss of generality, suppose $|S_{i,k-\flalpha-1}| \ge |S_{i,k-\flalpha}| \ge \dots \ge |S_{i,k}|$.
        If $|S_{i,k}|=0$, the claim is trivially true, so suppose $|S_{i,k}|\ge1$.
        For brevity, let $\sigma=|S_{i,k-\flalpha-1}|^\alpha + \dots + |S_{i,k}|^\alpha$. We now set out parameters for an application of Karamata's inequality.
        To this end, we define $\tau_j$ and $\sigma_j$ for $j\ge k-\flalpha-1$ via
        \begin{align*}
            \tau_j &= \frac{1}{\flalpha+1}\Big(\sigma - \frac{t^\alpha}{\log^{2k\alpha}t}\Big)\mbox{ for all }j \le k-1\,,\qquad
            \tau_{k} = \frac{t^\alpha}{\log^{2k\alpha}t}\,,\\
            \sigma_j &= |S_{i,j}|^\alpha\mbox{ for all }j \ge k-\flalpha-1\,.
        \end{align*}
        Observe that $\sum_j \tau_j =\sigma= \sum_j \sigma_j$, that $\sigma_{k-\flalpha-1} \ge \dots \ge \sigma_k$, and that $\tau_{k-\flalpha-1} = \dots = \tau_{k-1} > \tau_k$ since $\sigma \ge |S_{i,k-\flalpha-1}|^\alpha \ge (t/\log^k t)^\alpha$ by hypothesis.
        
        It remains to show that  $(\sigma_{k-\flalpha-1},\dots,\sigma_k)$ majorises $(\tau_{k-\flalpha-1},\dots,\tau_k)$. Since $S_i$ is not safely rooted, we have $|S_{i,k}| \le t/\log^{2k}t$.
        Hence $\sigma_k \le \tau_k$, and so $\sigma_{k-\flalpha-1} + \dots + \sigma_{k-1} \ge \tau_{k-\flalpha-1} + \dots + \tau_{k-1}$. Since the $\sigma_j$'s are decreasing and $\tau_{k-\flalpha-1} = \dots = \tau_{k-1}$, it follows that for all $x \in \{0,\dots,\flalpha\}$,
        \[
            \sum_{j=k-\flalpha-1}^{k-\flalpha-1+x} \sigma_j \ge \frac{x+1}{\flalpha+1}\sum_{j=k-\flalpha-1}^{k-1} \sigma_j \ge \frac{x+1}{\flalpha+1}\sum_{j=k-\flalpha-1}^{k-1} \tau_j = \sum_{j=k=\flalpha-1}^{k-\flalpha-1+x} \tau_j\,.
        \]
        (Here the first inequality follows from the fact that the average of the $x+1$ largest values in a set is no smaller that the average of the entire set.) Thus $(\sigma_{k-\flalpha-1},\dots,\sigma_k)$ majorises $(\tau_{k-\flalpha-1},\dots,\tau_k)$.
        
        Finally, for all $y\ge1$, let $\phi(y) = -\log(y^{1/\alpha})$, and observe that $\phi$ is a convex function. It now follows by Karamata's inequality (\cref{lem:karamata}) that
        \[
            \sum_{j=k-\flalpha-1}^k \phi(\tau_j) \le \sum_{j=k-\flalpha-1}^k \phi(\sigma_j)\,.
        \]
        Substituting in the definition of $\phi$ and negating both sides yields
        \[
            \log\bigg(\prod_{j=k-\flalpha-1}^k \sigma_i^{1/\alpha}\bigg) \le \log\bigg(\prod_{j=k-\flalpha-1}^k \tau_i^{1/\alpha}\bigg)\,.
        \]
        Exponentiating both sides and substituting in the definitions of the $\sigma_j$'s and $\tau_j$'s then yields
        \[
            \prod_{j=k-\flalpha-1}^k |S_{i,j}| \le \frac{t}{\log^{2k} t}\cdot \Big(\sigma - \frac{t^\alpha}{\log^{2k\alpha} t}\Big)^{(\flalpha+1)/\alpha}\le \frac{t\sigma^{(\flalpha+1)/\alpha}}{\log^{2k}t} 
        \]
        Since $\sigma \ge t^\alpha/\log^{k\alpha}t$ by hypothesis and $\sigma \le \cost(S_i)$ since $\alpha \ge 1$, it follows that
        \[
            \prod_{j=k-\flalpha-1}^k |S_{i,j}| \le \frac{\sigma^{(\flalpha+2)/\alpha}}{\log^{k}t}  \le \frac{\cost(S_i)^{(\flalpha+2)/\alpha}}{\log^{k} t}\,,
        \]
        and the result follows immediately.
    \end{proof}
    
    \begin{lemma}\label{lem:adapt-col-unsafe-roots}
        Writing $X$ for the set of all unsafely-rooted queries in $\{S_1,\dots,S_N\}$, we have
        \[
            \sum_{S_i \in X}\prod_{j=k-\flalpha-1}^k |S_{i,j}| \le \frac{k^kt^{\flalpha+2-\alpha}}{\log^k t}\sum_{i \in X}\cost(S_i)\,.
        \]
    \end{lemma}
    \begin{proof}
        We first split the terms of the sum according to the size of the largest part of the corresponding query among rooted vertex classes. Let
        \begin{align*}
            X^+ = \Big\{S_i \in X\colon \max\{|S_{i,j}|\colon j \ge k-\flalpha-1\} > t/\log^k t\Big\},\qquad 
            X^- = X \setminus X^+\,,
        \end{align*}
        so that
        \begin{equation}\label{eq:adapt-col-unsafe-roots-0}
            \sum_{S_i \in X} \prod_{j=k-\flalpha-1}^k |S_{i,j}| = \sum_{S_i \in X^-} \prod_{j=k-\flalpha-1}^k |S_{i,j}| + \sum_{S_i \in X^+} \prod_{j=k-\flalpha-1}^k |S_{i,j}|\,.
        \end{equation}
        
        We first bound the $X^-$ term. By \cref{lem:correct-number-roots} applied with $\gamma = 1/\log^{k} t$ (which is less than $1$ since $t \ge t_0$), we have
        \begin{equation}\label{eq:adapt-col-unsafe-roots-1}
            \sum_{S_i \in X^-} \prod_{j=k-\flalpha-1}^k |S_{i,j}| \le \Big(\frac{kt}{\log^k t}\Big)^{\flalpha+2-\alpha}\sum_{S_i \in X^-}\cost(S_i) \le \frac{k^{k}t^{\flalpha+2-\alpha}}{\log^k t}\sum_{S_i \in X^-}\cost(S_i)\,.
        \end{equation}
        
        We next bound the $X^+$ term, splitting into two cases depending on the value of $\alpha$.
        
        \medskip\noindent \textit{Case 1: $\alpha \le 1$.} For all $S_i \in X^+$, we have $\cost(S_i) \ge (\max_i |S_i|)^\alpha \ge (t/\log^k t)^\alpha$; since $S_i$ is safely rooted, it follows that
        \[
            \prod_{j=k-\flalpha-1}^k |S_{i,j}| \le \frac{t}{\log^{2k}t}\cdot t^{\flalpha+1} = \frac{t^{\flalpha+2}}{\log^{2k}t} \le \frac{t^{\flalpha+2-\alpha}}{\log^{2k-\alpha k}t}\cost(S_i) \le \frac{t^{\flalpha+2-\alpha}}{\log^{k}t}\cost(S_i)\,.
        \]
        The last inequality holds since~$\alpha\le1$.
        Summing over all $S_i \in X^+$, we obtain
        \begin{equation}\label{eq:LACU-1}
            \sum_{S_i \in X^+}\prod_{j=k-\flalpha-1}^k |S_{i,j}| \le \frac{t^{\flalpha+2-\alpha}}{\log^{k}t}\sum_{S_i \in X^+}\cost(S_i)\,.
        \end{equation}
        
        \medskip\noindent \textit{Case 2: $\alpha > 1$.}
        In this case, by definition of~$X^+$, we can apply \cref{lem:adapt-col-pre-unsafe-roots} to each term in the sum; this yields
        \[
            \sum_{S_i \in X^+} \prod_{j=k-\flalpha-1}^k |S_{i,j}|
            \le \frac{1}{\log^k t}\sum_{S_i \in X^+}\cost(S_i)^{(\flalpha+2)/\alpha}
            \le \frac{1}{\log^k t}\sum_{S_i \in X^+}\cost(S_i)^{\flalpha+2}\,.
        \]
        We now apply Karamata's inequality in the form of \cref{cor:karamata}, taking $s_i = |S_{i,1}| + \dots + |S_{i,k}|$, $W = \sum_{S_i \in X^+}\cost(S_i)$, $c = kt$, and $r=\flalpha + 2$. This yields:
        \begin{equation}\label{eq:LACU-2}
            \sum_{S_i \in X^+} \prod_{j=k-\flalpha-1}^k |S_{i,j}| \le \frac{(kt)^{\flalpha+2-\alpha}}{\log^k t}\sum_{S_i \in X^+} \cost(S_i) \le \frac{k^kt^{\flalpha+2-\alpha}}{\log^k t} \sum_{S_i \in X^+}\cost(S_i)\,.
        \end{equation}
        The result therefore follows from~\eqref{eq:adapt-col-unsafe-roots-0} and~\eqref{eq:adapt-col-unsafe-roots-1} combined with
        \eqref{eq:LACU-1} and \eqref{eq:LACU-2}.
    \end{proof}
 
	\begin{lemma}\label{lem:adapt-col-E3}
	    We have
	    \[
	        \pr(\overline{\ECost} \vee \ERoot) \ge 1 - \frac{Ck^k}{t^\alpha\log^{k}t}\,.
	    \]
	\end{lemma}
	\begin{proof}
	    Let $s_1,\dots,s_N$ be any possible sequence of values for the query sets $S_1,\dots,S_N$, and let~$\calT$ be the event that $(S_1,\dots,S_N) = (s_1,\dots,s_N)$ holds. We will prove $\pr(\ECost \wedge \overline{\ERoot} \mid \calT) \le Ck^k/(t^{\alpha}\log^{k}t)$, from which the result follows immediately.
	    
	    Note that $\ECost$ is a function of~$S_1,\dots,S_N$; if $\overline{\ECost}$ occurs under $\calT$ then we are done, so suppose not. Let $X\subseteq\set{s_1,\dots,s_N}$ be the set of all queries that are not safely rooted. Recall that $S_1,\dots,S_N$ are functions of $G_1$, and that the roots $\calR_{k-\flalpha-1},\dots,\calR_k$ are independent of $G_1$; thus
	    \begin{align}\nonumber
	        \pr(\ECost \wedge \overline{\ERoot} \mid \calT) 
	        &= \pr\Big(\mbox{Some $s_i \in \{s_1,\dots,s_N\}$ is neither safely nor inaccurately rooted}\Big)\\\nonumber
	        &\le \sum_{s_i \in X} \pr\big(s_i\mbox{ is accurately rooted}\big) = \sum_{s_i \in X}\prod_{j=k-\flalpha-1}^k \pr(\calR_j \in s_{i,j}) \\\nonumber
	        &= \frac{1}{t^{\flalpha+2}}\sum_{s_i \in X}\prod_{j=k-\flalpha-1}^k |s_{i,j}|\,.
	    \end{align}
        Since $\ECost$ occurs, the total cost of $s_1,\dots,s_N$ is at most $C$; by applying \cref{lem:adapt-col-unsafe-roots}, we arrive at
	    \[
	        \pr(\ECost \wedge \overline{\ERoot} \mid \calT) \le \frac{1}{t^{\flalpha+2}}\sum_{s_i \in X}\prod_{j=k-\flalpha-1}^k |s_{i,j}|
            \le Ck^k/(t^\alpha\log^k t)
            \,.\qedhere
	    \]
	\end{proof}
	
	\begin{lemma}\label{lem:adapt-col-E2}
		We have
		\[
		    \pr(\overline{\ECost} \vee \ENroot) \ge 1 - \frac{C}{t^\alpha\log^{2k}t}\,.
		\] 
	\end{lemma}
	\begin{proof}
	    Let $s_1,\dots,s_N$ be any possible sequence of values for $S_1,\dots,S_N$, let $\vec{q}$ be any possible value for $\vec{\calQ}$, and let $\calT$ be the event that $(S_1,\dots,S_N) = (s_1,\dots,s_N)$ and $\vec{\calQ} = \vec{q}$. We will prove
     \[
     \pr(\ECost \wedge \overline{\ENroot} \mid \calT) \le C/(t^{\alpha}\log^{2k}t)\,.
     \]
     From this, the result follows immediately. 
	    
	    If $\overline{\ECost}$ occurs under $\calT$ then we are done, so suppose not; then $\sum_i \cost(s_i) \le C$. Let $X$ be the set of all queries in $s_1,\dots,s_N$ which are safely rooted and for which there exists $j \le k-\flalpha-2$ with $|s_{i,j}| \le 1/(\xi q_j)$. Recall that $S_1,\dots,S_N$ are deterministic functions of $G_1$, and that $\calQ$ is independent of $G_1$; thus
	    \begin{align*}
	        \pr(\ECost \wedge \overline{\ENroot} \mid \calT) 
	        &= \pr\Big(\mbox{Some $s_i \in X$ has $e(H_2[s_i]) > 0$}\mid \vec{\calQ} = \vec{q}\Big)\\
	        &\le \sum_{s_i \in X} \pr\Big(e(H_2[s_i]) > 0\mid \vec{\calQ} = \vec{q}\Big)\,.
	    \end{align*}
	    We have $e(H_2[s_i]) = 0$ whenever $\calX_j \cap s_{i,j} = \emptyset$ for any $j$, so it follows that
	    \[
	        \pr(\ECost \wedge \overline{\ENroot} \mid \calT) \le \sum_{s_i \in X} \min_{j \le k-\flalpha-2} \pr\big(\calX_j \cap s_{i,j} \ne \emptyset\mid \vec{\calQ} = \vec{q}\big)\,.
	    \]
	    By the definition of $X$, for all $s_i \in X$ there exists $j \le k-\flalpha-2$ such that $|s_{i,j}| \le 1/(\xi q_j)$, and each vertex in $s_{i,j}$ lies in $\calX_j$ with probability $q_j$. It follows by a union bound over all vertices in $s_{i,j}$ that
	    \[
	        \pr(\ECost \wedge \overline{\ENroot} \mid \calT) \le \sum_{s_i \in X} \frac{1}{\xi q_j}\cdot q_j = \frac{|X|}{\xi}\,.
	    \]
	    
	    Now, since $\ECost$ occurs under $\calT$, the total cost of $s_1,\dots,s_N$ is at most $C$. Since each query in $X$ is safely rooted, it has cost at least $(t/\log^{2k} t)^\alpha$. It follows that the total number of queries in $X$ is at most $C/(t/\log^{2k} t)^\alpha$, and so
	    \begin{equation}\label{eq:adapt-col-E2-1}
	        \pr(\ECost \wedge \overline{\ENroot} \mid \calT) \le \frac{C\log^{2k\alpha}t}{\xi t^\alpha} \le \frac{C\log^{2k^2}t}{\xi t^\alpha}\,.
	    \end{equation}
	    By definition, and using the fact that $k-\flalpha+2 \ge 1$, we have
	    \[
	        \xi = (pn^{\flalpha+2})^{-1/(2k\beta)} = t^{\frac{k-\flalpha+2}{4k\beta}} \ge t^{1/(4k\beta)} \ge  t^{\frac{20k^4\log\log t}{4k\log t}} \ge (\log t)^{5k^3}\,,
	    \]
	    and the result follows from~\eqref{eq:adapt-col-E2-1}.
	\end{proof}
	
	\begin{lemma}\label{lem:adapt-col-I}
		We have
		\[
		    \pr(\overline{\ECost} \vee \EInacc) \ge 1 - \frac{2^{5k}k^{7k}C}{t^{\alpha}}\cdot \bigg(\frac{\log\log t}{\log t}\bigg)^{k-\flalpha-3}\,.
		\]
	\end{lemma}
	\begin{proof}
	    Let $s_1,\dots,s_N$ be any possible sequence of values for $S_1,\dots,S_N$, and let $\calT$ be the event that $(S_1,\dots,S_N) = (s_1,\dots,s_N)$. We will prove that 
	    \[
	        \pr(\overline{\ECost} \vee \EInacc \mid \calT) \ge 1 - \frac{2^{5k}k^{3k}C}{t^{\alpha}}\cdot \bigg(\frac{\log\log t}{\log t}\bigg)^{k-\flalpha-3}\,,
	    \]
	    from which the result follows immediately. 
	    
	    If $\overline{\ECost}$ occurs under $\calT$ then we are done, so suppose not; then $\sum_i \cost(s_i) \le C$. Recall that $S_1,\dots,S_N$ are deterministic functions of $G_1$, and that $H_2$ is independent of $G_1$; thus 
	    \begin{equation}\label{eq:adapt-col-I-0}
	        \pr(\overline{\ECost} \vee \EInacc \mid \calT) = 1 - \pr(\mbox{some $s_i$ is accurate}) \ge 1 - \sum_{i=1}^N \pr(s_i\mbox{ is accurate})\,.
	    \end{equation}
	    A given query $s_i$ is accurate precisely when it is accurately rooted (which depends only on $\vec{\calR} \coloneqq (\calR_{k-\flalpha-1},\dots,\calR_k)$) and accurately profiled (which depends only on $\vec{\calQ}$). These events are independent, so we have
	    \begin{equation}\label{eq:adapt-col-I-1}
	        \pr(\mbox{some $s_i$ is accurate}) = \pr(s_i\mbox{ is accurately rooted})\cdot \pr(s_i\mbox{ is accurately profiled})\,.
	    \end{equation}
	    
	    Observe that $s_i$ is accurately profiled for at most one choice of $\calQ$. Let $m$ be the number of ways in which $\beta$ can be decomposed into a sum of $k-\flalpha-2$ ordered integers in $\{0,\dots,B\}$. Each such sequence of integers corresponds to the numerators of the exponents of a possible value of $\calQ$, and so there are $m$ possible values of $\calQ$ in total. We therefore have
	    \begin{align*}
	         \pr(s_i\mbox{ is accurately rooted}) &= |s_{i,k-\flalpha-1} \times \dots \times s_{i,k}|/t^{\flalpha+2}\,,\\
	         \pr(s_i\mbox{ is accurately profiled}) &= 1/m\,.
	    \end{align*}
	    It follows from~\eqref{eq:adapt-col-I-0} and~\eqref{eq:adapt-col-I-1} that
	    \[
	        \pr(\overline{\ECost} \vee \EInacc \mid \calT) \ge 1 - \frac{1}{mt^{\flalpha+2}}\sum_{i=1}^N |s_{i,k-\flalpha-1} \times \dots \times s_{i,k}|\,.
	    \]
	    Recall that $\ECost$ occurs and so $s_1,\dots,s_N$ have total cost at most $C$; applying \cref{lem:correct-number-roots} with $\gamma=1$, we obtain
	    \begin{equation}\label{eq:adapt-col-I-2}
	        \pr(\overline{\ECost} \vee \EInacc \mid \calT) \ge 1 -  \frac{Ck^k}{mt^\alpha}\,.
	    \end{equation}
	
	    It remains to bound $m$ below. We will first bound $B$ below. Recall that $x = pt^{\flalpha+2} = t^{-(k-\flalpha-2)/2}$; thus by the definition of $B$, we have
	    \[
	        B = \Big\lfloor\frac{\log((24\log t)/t)}{\log (x^{1/\beta})}\Big\rfloor = \Big\lfloor{-}\frac{2\beta\log((24\log t)/t)}{(k-\flalpha-2)\log t}\Big\rfloor = \Big\lfloor\frac{2\beta(\log t - \log(24\log t))}{(k-\flalpha-2)\log t}\Big\rfloor\,.
	    \]
	    Since $t \ge t_0 \ge e^{640}$, it follows that
	    \begin{equation}\label{eq:adapt-col-I-3}
	        B \ge \frac{3}{2}\cdot \frac{\beta}{k-\flalpha-2}\,.
	    \end{equation}
	    
	    We now exploit this bound on $B$ in order to bound $m$ below. For every choice of integers $z_1,\dots,z_{k-\flalpha-3}$ in $[(1-\tfrac{1}{2k})\beta/(k-\flalpha-2), (1+\tfrac{1}{2k})\beta/(k-\flalpha-2)]$, we have
	    \begin{align*}
	        \sum_{\ell=1}^{k-\flalpha-3} z_\ell &\le \frac{(k-\flalpha-3)\beta}{k-\flalpha-2} + \frac{(k-\flalpha-3)\beta}{2k(k-\flalpha-2)} 
	        \le \frac{(k-\flalpha-3)\beta}{k-\flalpha-2} + \frac{\beta}{2(k-\flalpha-2)}\\
	        &= \beta - \frac{\beta}{2(k-\flalpha-2)}\,,
	    \end{align*}
	    and similarly
	    \[
	        \sum_{\ell=1}^{k-\flalpha-3} z_\ell \ge \beta - \frac{3\beta}{2(k-\flalpha-2)}\,.
	    \]
	    By~\eqref{eq:adapt-col-I-3}, it follows that for all such choices of integers $z_1,\dots,z_{k-\flalpha-3}$ there is a unique integer $z_{k-\flalpha-2} \in [0,B]$ such that $\sum_i z_i = \beta$. Thus
	    \[
	        m \ge \bigg|\mathbb{Z} \cap \bigg[\Big(1-\frac{1}{2k}\Big)\frac{\beta}{k-\flalpha-1},\  \Big(1+\frac{1}{2k}\Big)\frac{\beta}{k-\flalpha-1}\bigg] \bigg|^{k-\flalpha-3}\,.
	    \]
	    Since $t \ge t_0$, it follows that
	    \[
	        m \ge \Big(\frac{\beta}{2k(k-\flalpha-2)}\Big)^{k-\flalpha-3}\,.
	    \]
	    Again since $t \ge t_0$, we have $\beta \ge 1$, and so we can bound away the floor in its definition to obtain
	    \[
	        m \ge \Big(\frac{\log t}{40k^5(k-\flalpha-2)\log\log t}\Big)^{k-\flalpha-3} \ge \frac{(\log t)^{k-\flalpha-3}}{2^{5k}k^{6k}(\log\log t)^{k-\flalpha-3}}\,.
	    \]
	    The result therefore follows immediately from~\eqref{eq:adapt-col-I-2}.
	\end{proof}
	
	With these probability bounds in place, \cref{thm:colourful-G1-G2} now follows easily.
	
	\colourfulGs*
	\begin{proof}
	    Part (i) of the theorem is immediate from~\cref{lma:col-edge-numbers}. Suppose $A$ is a deterministic \cindora-oracle algorithm with
	    \begin{equation}
            \pr_{(G_1,G_2)\sim(\calG_1,\calG_2)}\Big(A(\cindora(G_1)) \ne A(\cindora(G_2))\Big) \ge 2/3\,.
        \end{equation}
        By \cref{lem:adapt-col-framework-new-new}, it follows that
        \begin{equation}\label{eq:col-final-1}
            \pr(\ECost \wedge \EEdge \wedge \ERoot \wedge \ENroot \wedge \EInacc) \le 1/3\,.
        \end{equation}
        
        Observe that since $k \ge 2$ and $t \ge t_0 \ge \exp(640)$, we have $\log^{k^3} t \ge 2k(\ln t + \ln\ln t)$, and hence $e^{-(\log^{k^3}t)/2} \le 1/(t^\alpha\log^{k} t)$. By Lemmas~\ref{lem:adapt-col-E1}, \ref{lem:adapt-col-E3}, \ref{lem:adapt-col-E2} and \ref{lem:adapt-col-I} together with a union bound, it follows that
        \begin{align*}
            &\pr\Big(\overline{\ECost} \vee (\EEdge \wedge \ERoot \wedge \ENroot \wedge \EInacc)\Big)\\
            &\qquad\qquad\ge 1 - \bigg(Ce^{-(\log^{k^3}t)/2} + \frac{Ck^k}{t^\alpha\log^k t} + \frac{C}{t^\alpha\log^{2k}t} + \frac{2^{5k}k^{7k}C}{t^\alpha}\cdot\Big(\frac{\log\log t}{\log t}\Big)^{k-\flalpha-3}\bigg)\\
            &\qquad\qquad\ge 1 - \frac{Ck^k}{t^\alpha}\bigg(\frac{3}{\log^{k}t} + 2^{5k}k^{6k}\Big(\frac{\log\log t}{\log t}\Big)^{k-\flalpha-3}\bigg)\\
            &\qquad\qquad\ge 1 - \frac{4C}{t^\alpha}\cdot 2^{5k}k^{7k}\Big(\frac{\log\log t}{\log t}\Big)^{k-\flalpha-3}\,.
        \end{align*}
        By the definition of $C$, it follows that
        \[
            \pr\big(\overline{\ECost} \vee (\EEdge \wedge \ERoot \wedge \ENroot \wedge \EInacc)\big) \ge 31/32\,,
        \]
        and hence
        \[
            \pr\big(\ECost \vee \overline{(\EEdge\wedge\ERoot\wedge\ENroot\wedge\EInacc)}\big) \le 1/32\,.
        \]
        By~\eqref{eq:col-final-1}, it follows that
        \begin{equation}\label{eq:col-final-2}
            \pr(\ECost) \le 1/32 + 1/3 < 1/2.
        \end{equation}
        
        By Markov's inequality, with probability at least $1/2$, $\cost(A,G_1)$ is at most twice its expected value. If $\E(\cost(A,G_1))$ were at least $C/2$ then this would contradict~\eqref{eq:col-final-2}, so we must have $\E(\cost(A,G_1)) \le C/2$ as required.
	\end{proof}

    \bibliographystyle{plainurl}
    \bibliography{references}
\end{document}